\documentclass[letterpaper,11pt]{article}
\usepackage{fullpage}
\usepackage[ansinew]{inputenc}
\usepackage{csquotes}
\usepackage{verbatim}
\usepackage{graphicx}
\usepackage{tabularx}
\usepackage{array}
\usepackage{delarray}
\usepackage{dcolumn}
\usepackage{float}
\usepackage{amsmath}
\usepackage{psfrag}
\usepackage{booktabs}
\usepackage{multirow,hhline}
\usepackage{amstext}
\usepackage{amssymb}
\usepackage{fancyhdr}
\usepackage{setspace}
\usepackage{amsthm}
\usepackage{makeidx}
\usepackage{tikz}
\usepackage{url}
\usepackage{adjustbox}
\usepackage{enumitem}
\usepackage{mathtools}
\usepackage{mathabx}
\usepackage{authblk}
\usepackage{caption}

\usepackage[mathscr]{euscript}

\DeclarePairedDelimiter\norm{\lVert}{\rVert}

\usepackage{hyperref}
\hypersetup{citecolor=true,colorlinks=true,allcolors=blue}

\usepackage{cleveref}
\crefname{appsec}{appendix}{appendices}
\crefname{ass}{assumption}{assumptions}
\Crefname{ass}{Assumption}{Assumptions}
\Crefname{rem}{Remark}{Remarks}
\Crefname{lem}{Lemma}{Lemmata}
\usepackage[
 bibencoding=utf8,
 maxbibnames=99, 
 minbibnames=1, 
 backend=biber,%
 sortlocale=en_US,%
 style=authoryear,%
 doi=true,%
 url=true,
 eprint=true%
]{biblatex}
\usetikzlibrary{arrows,calc}

\makeindex
\allowdisplaybreaks
\numberwithin{equation}{section}
\newtheorem{lem}{Lemma}[section]
\newtheorem{thm}{Theorem}[section]
\newtheorem{cor}{Corollary}[section]

\newtheorem{ass}{Assumption}

\newcommand{\convP}{\stackrel{p}{\longrightarrow}}
\newcommand{\convD}{\rightsquigarrow}

\newcommand{\N}{\mathcal{N}}
\newcommand{\eps}{\varepsilon}
\newcommand{\tr}{\text{tr}}
\renewcommand{\epsilon}{\varepsilon}

\DeclareMathOperator{\Bdiag}{Bdiag} 
\DeclareMathOperator{\bvec}{bvec} 
\DeclareMathOperator{\mkvec}{vec} 
\usepackage[english]{babel}
\newtheoremstyle{boldremark}
    {\dimexpr\topsep/2\relax} 
    {\dimexpr\topsep/2\relax} 
    {}          
    {}          
    {\bfseries} 
    {.}         
    {.5em}      
    {}          

\theoremstyle{boldremark}

\newtheorem{rem}{Remark}[section]

\newcommand{\hidefromtoc}{%
  \setcounter{oldtocdepth}{\value{tocdepth}}%
  \addtocontents{toc}{\protect\setcounter{tocdepth}{-10}}%
}
\newcommand{\unhidefromtoc}{%
  \addtocontents{toc}{\protect\setcounter{tocdepth}{\value{oldtocdepth}}}%
}
\newcounter{oldtocdepth}

\usepackage[nolist]{acronym}
\begin{acronym}
  \acro{RR}{Riesz representer}
  \acro{L1CO}{leave-one-cluster-out}
  \acro{L3CO}{leave-three-cluster-out}
  \acro{L2CO}{leave-two-cluster-out}
  \acro{CI}{confidence interval}
  \acro{IV}{instrumental variable}
  \acro{TSLS}{two-stage least squares}
  \acro{OLS}{ordinary least squares}
\end{acronym}

\makeatletter
\newcommand*{\rom}
[1]{\expandafter\@slowromancap\romannumeral#1@}
\makeatother

\title{Cluster-Robust Inference for Quadratic Forms\thanks{We are grateful to Ivan Canay, Dachuan Chen, Xu Cheng, Andrew Chesher, Liyu Dou, Qingliang Fan, Sheng Chao Ho, Jia Li, Didier Nibbering, Xiaoxia Shi, and all the participants at the SMU Econometrics Workshop for their valuable comments. Zhang acknowledges the financial support from the NSFC under grant No. 72133002. This research is supported by the Ministry of Education, Singapore under its MOE Academic Research Fund Tier 2 (Project ID\@: MOE-000767-00). Any opinions, findings and conclusions or recommendations expressed in this material are those of the authors and do not reflect the views of the Ministry of Education, Singapore.  Any and all errors are our own. Corresponding author: Wenjie Wang, Email: wang.wj@ntu.edu.sg.\vspace{1.3mm}}
        \\ \vspace{2mm}
}

\author[a]{Michal Koles\'{a}r}
\author[b]{Pengjin Min}
\author[c]{Wenjie Wang}
\author[b]{Yichong Zhang}
\affil[a]{\small \emph{Princeton University, United States}}
\affil[b]{\small \emph{Singapore Management University, Singapore}}
\affil[c]{\small \emph{Nanyang Technological University, Singapore}}
\date{\today}
\begin{document}

\maketitle

\def\baselinestretch{1.05}
\begin{abstract}
  This paper studies inference for quadratic forms of linear regression
  coefficients with clustered data and many covariates. Our framework covers
  three important special cases: instrumental variables regression with many
  instruments and controls, inference on variance components, and testing
  multiple restrictions in a linear regression. Na\"{\i}ve plug-in estimators
  are known to be biased. We study a leave-one-cluster-out estimator that is
  unbiased, and provide sufficient conditions for its asymptotic normality. For
  inference, we establish the consistency of a leave-three-cluster-out variance
  estimator under primitive conditions. In addition, we develop a novel
  leave-two-cluster-out variance estimator that is computationally simpler and
  guaranteed to be conservative under weaker conditions. Our analysis allows
  cluster sizes to diverge with the sample size, accommodates strong
  within-cluster dependence, and permits the dimension of the covariates to
  diverge with the sample size, potentially at the same rate.\bigskip

\noindent \textbf{Keywords:} Many instruments, many covariates, clustered data, cross-fit, judge design \bigskip

\noindent \textbf{JEL codes:} C12, C36, C55
\end{abstract}
\clearpage

\hidefromtoc%

\section{Introduction}

We study inference for quadratic forms of linear regression coefficients from two regressions with clustered data, given by
\begin{align}\label{eq:model_outcome}
    Y_{i, g} &= W_{i, g}' \gamma \;+\; U_{i, g},
    & \mathbb{E}[U_{i, g}] &= 0,\\
    \label{eq:model_treatment}
    X_{i, g} &= W_{i, g}' \pi \;+\; V_{i, g},
    & \mathbb{E}[V_{i, g}] &= 0,
\end{align}
where $g$ indexes clusters, $i$ indexes observations within a cluster, and
$W_{i, g}$ is a $d$-dimensional vector of covariates, which is treated as fixed. The
object of interest,
\begin{equation}\label{eq:mom}
    \theta \;=\; \pi' A_0 \gamma,
\end{equation}
combines the two sets of regression coefficients using a known nonrandom $d\times d$ matrix $A_0$. Two key features of this setting complicate inference. First, the covariates may be high-dimensional, with $d$ allowed to grow as fast as the sample size. Second, the error terms $(U_{i, g}, V_{i, g})$ may be heteroskedastic and exhibit arbitrary within-cluster correlation, with cluster sizes that may diverge with the sample size.

This framework covers several important applications: if we set $Y=X$, variance
components can be written as quadratic forms in regression coefficients
\parencite[e.g.,][]{KSS2020}; tests of many linear restrictions can likewise be
cast as a restriction on a quadratic form \parencite{AS23}. The framework also
covers inference in \ac{IV} regression models with many potentially weak
instruments or controls, and heterogeneous treatment effects---in this case
\cref{eq:model_outcome,eq:model_treatment} correspond to reduced-form and
first-stage equations, $W$ contains both IVs and controls, and the effect of $X$
on $Y$ may be heterogeneous.

In the IV setting, $W$ is commonly high-dimensional in the popular IV leniency
design, which leverages quasi-random assignment of judges or other
decision-makers who differ in their leniency---the propensity to grant treatment
\parencite[see][for recent surveys]{CFL2024,GHK25}. Since these leniency
measures are unknown, they must be estimated in a first-stage regression of
treatment on decision-maker indicators. Furthermore, since their assignment is
typically only random conditional on time and location, time-by-location fixed
effects need to be included. Commonly, these designs feature a large number of
decision-makers and fine-grained fixed effects, resulting in a high-dimensional
first-stage regression. More generally, to avoid restrictive first-stage
assumptions under nonrandom instrument assignment, ensuring causal
interpretation requires the inclusion of interactions between the instruments
and controls, as well as a flexible control specification
\parencite[e.g.,][]{BBMT22}. This can generate a high-dimensional first-stage
regression even if the number of baseline instruments and covariates is small.
Whenever cases are assigned in batches, this induces clustering in the data.

Similarly, variance component estimation commonly features high-dimensional
specifications. For instance, consider the popular two-way worker--firm fixed
effects model of \textcite{AKM99}, where $Y$ denotes log wages, $X=Y$, and $W$
contains both worker and firm dummies. In this setting, firms' contributions to
wages, measured by the variance component of firm effects, can be formulated as
a quadratic form in firm fixed effects, and the sorting of high-wage workers to
high-wage firms can be measured by a quadratic form involving worker and firm
fixed effects. Unless one has access to a long panel, the number of workers is
necessarily large relative to the sample size; often the firm dimension is large
as well. To allow for match-specific unobservables, it is desirable to cluster
at the worker--firm match level \parencite{kline2024}. Consequently, valid
inference must account for both the high dimensionality of the regressors and
the clustered dependence structure of the data.

Finally, tests of linear restrictions, such as testing whether a subset of the
$\gamma$ coefficients is zero, feature high-dimensional specifications in
several applications. For example, value-added models can be validated by
testing whether a set of regressors is excluded from the structural regression
model \parencite{ahpw24}. Likewise, tests for endogenous peer effects can be
cast as testing whether peer characteristics have zero coefficients
\parencite{JL26}. As discussed in \textcite{AS23}, testing for the presence of
heterogeneity amounts to testing whether a set of fixed effects is zero
\parencite[e.g.,][]{krw22,fgw16}. These zero restrictions can be
high-dimensional, and it is desirable to cluster the standard errors to allow
for unobserved common shocks: depending on the application, one may want to
cluster at the job level \parencite{krw22} or by classroom (in value-added or
peer-effects applications).

Although least squares estimators of linear regression coefficients are unbiased
under general conditions, their quadratic forms are biased (the bias is positive
by Jensen's inequality if $A_{0}$ is positive semidefinite). The bias scales
with the coefficient dimension, so that plugging in least squares estimates of
$\gamma$ and $\pi$ into \cref{eq:mom} yields an estimator with non-negligible
asymptotic bias if the coefficient dimension is large. In the IV setting, this
bias corresponds to the well-known many instrument bias of two-stage least
squares \parencite[e.g.,][]{bjb95,bekker94}. The bias can be purged by using a
\ac{L1CO} estimator proposed by \textcite{KSS2020}.

The contribution of this paper is threefold. First, we establish the asymptotic
normality of the \ac{L1CO} estimator. Second, we develop a \ac{L3CO}
cluster-robust variance estimator and derive conditions for its consistency.
Third, we introduce a novel \ac{L2CO} variance estimator that yields valid, but
conservative inference under a weaker set of assumptions. These results are
derived under primitive conditions that allow for growing cluster sizes, both
weak and strong within-cluster dependence, and only impose weak conditions on
the regressors, allowing their dimension $d$ to grow potentially as fast as the
sample size $n$. In particular, under suitable further regularity conditions,
the proposed inference based on the \ac{L3CO} variance estimator is valid
provided that the largest cluster size~\(n_G\) satisfies
\(n_G^{\max\left(\frac{2q}{q-1},3\right)} = o(n)\) and \(n_G^{5} = o(n)\) under
weak and strong within-cluster dependence, respectively, where \(2q\) denotes
the number of finite moments of the regression errors \((U_{i,g}, V_{i,g})\) for
some $q\geq 2$. In contrast, the rate requirements associated with the \ac{L2CO}
variance estimator can be relaxed to \(n_G^{\frac{2q}{q-1}} = o(n)\) under weak
within-cluster dependence. We also show that when $d$ is not very large and
cluster sizes are fixed, inference based on the \ac{L2CO} variance estimator is
exact.

These results generalize and unify existing results from two strands of
literature. The first strand studies inference in high-dimensional linear
regressions. \Textcite{KSS2020} propose the \ac{L1CO} estimator we study.
However, their formal results focus on the case with independent data, and their
variance estimator is based on sample splitting. While least-squares estimators
are unbiased for \emph{linear functions} of regression coefficients, estimating
their asymptotic variance boils down to estimating a quadratic form with a
particular $A_{0}$ matrix. The Eicker-Huber-White estimator is a plug-in
estimator of this quadratic form, and is thus biased in high-dimensional
settings. \Textcite{CJN18_ET,CJN18} show consistency of a Hadamard-type
estimator of the quadratic form, first proposed by \textcite{hrk69}, for the
case with independent data or clustered data with bounded cluster sizes.
\Textcite{CGJN22} extend its consistency to settings with diverging cluster
sizes. In this context, under independent sampling, \textcite{J22} establishes
the consistency of the \ac{L1CO} unbiased estimator we study, and
\textcite{AnNg26} generalize the consistency result to the case with clustering.

In contrast, this paper is concerned with inference on quadratic forms, not just their consistent estimation. The variance of the \ac{L1CO} estimator we study depends on products of second moments, and can be written as a quartic form.
As a result, unbiased variance estimation using leave-out methods requires leaving three clusters out. \Textcite{MSJ25} study estimation in linear models with
high-dimensional controls, clustered data, and weak exogeneity. The weak
exogeneity motivates formulating the estimation problem as a quadratic form
problem, and they use this representation to construct a just-identified
internal-IV estimator. They propose a jackknife variance estimator and show that
its expectation is conservative.

The \ac{L3CO} variance estimator we study is inspired by \textcite{AS23}, who
provide primitive sufficient conditions for its consistency in the special case
with independent data, and $Y=X$.
Extending these results to clustered data is substantially more involved, as it
requires handling multilevel sums of products of block matrices, where matrix
multiplication is non-commutative. We address this challenge through two key
technical innovations: (1) a new representation of the L3CO projection matrix,
and (2) a decomposition of sums of products of two L3CO projection matrices,
each leaving out different clusters. These innovations yield primitive
consistency conditions for the L3CO variance estimator in terms of the sample
size $n$, the largest cluster size $n_{G}$, the number of moments of the
residuals $2q$, the regressor dimension $d$, and the properties of $A_{0}$. The
\ac{L2CO} variance estimator that we consider is, to our knowledge, new.

The second strand concerns inference in linear IV models with many weak
instruments; see \textcite{MS24} for a recent survey. Among studies assuming
independent data, our results are most closely related to \textcite{Yap24}, who
shows that ignoring the additional variation generated by $V_{i, g}$ under
heterogeneous treatment effects can lead to underestimation of the asymptotic
variance, and proves consistency of the L3CO variance estimator we study under
high-level assumptions. The \ac{L1CO} unbiased estimator we study
can be thought of as a generalization of the UJIVE estimator in \textcite{K13}
to the clustered setting. \Textcite{BN24,EK2018} also study IV inference with
heterogeneous treatment effects, but focus on independent data.

We are only aware of a few IV studies that allow for clustering.
\Textcite{CNT23} study many-IV regression under clustering but impose strong
restrictions, including a correctly specified structural equation (homogeneous
treatment effects), few controls (of order $o(\sqrt{n})$, excluding cluster
fixed effects), bounded cluster sizes, and independence of errors within
clusters. \Textcite{FLM23} allow within-cluster dependence and propose a
cluster-jackknife estimator, but they do not accommodate heterogeneous treatment
effects, weak identification, or many controls, and they do not provide a
consistent cluster-robust variance estimator for high-dimensional settings.
\Textcite{L23} considers the Anderson--Rubin test in many-IV regression under
clustering, and \textcite{LW24} extend it to multidimensional clustering,
adapting bias corrections from \textcite{CNT23} but without formal
distributional theory when there are many
controls.


\section{Setup}\label{sec:setup}

We sample data from $G$ clusters indexed by $g\in[G]$, where $[G]$ is a
shorthand for $\{1, \dotsc, G\}$. Each cluster contains $n_g$ units indexed by
$i\in [n_g]$, such that the total sample size is given by
$n=\sum_{g \in [G]}n_g$. Without loss of generality, we assume the last cluster
is the largest, so that $n_{G}$ gives the largest cluster size. For each unit,
we observe a vector of covariates $W_{i, g} \in \Re^{d}$, and two scalar outcome
variables, $Y_{i, g} \in \Re$, and $X_{i, g} \in \Re$. We treat the covariates as
deterministic (or, equivalently, we condition on them), and assume the
regressions of $Y_{i, g}$ and $X_{i, g}$ onto $W_{i, g}$, defined in
\cref{eq:model_outcome,eq:model_treatment}, are linear, so that the error terms
$U_{i, g}$ and $V_{i, g}$ are mean zero. The matrix $W_g \in \Re^{n_g \times d}$
stacks the row vectors $\{W_{i, g}'\}_{i \in [n_g]}$, and the matrix $W$ stacks
the matrices $\{W_g\}_{g\in[G]}$. The matrices $Y$, $X$, $U$, $V$, as well as
their cluster-specific counterparts $Y_g$, $X_g$, $U_g$, $V_g$ are defined
analogously.

A natural estimator of the quadratic form $\theta$, defined in \cref{eq:mom}, is the plug-in estimator that replaces $\pi$ and $\gamma$ with the \ac{OLS} estimates $\hat{\pi}=(W' W)^{-1}W' X$ and $\hat{\gamma}=(W' W)^{-1}W' Y$,
\begin{equation}\label{eq:theta_pi}
    \hat{\theta}_{\rm PI} = \hat{\pi}' A_0 \hat{\gamma}
    =\sum_{g\in[G]} \hat{\pi}' A_0 (W' W)^{-1} W_g' Y_g .
\end{equation}
This plug-in estimator displays an overfitting bias given by
\begin{equation*}
    \mathbb{E} \hat{\theta}_{\rm PI}-\theta = \sum_{g \in [G]} \mathbb E V_{g}' A_{g, g} U_{g}, \qquad\text{where}\quad A=W(W' W)^{-1} A_0 (W' W)^{-1} W' \in \Re^{n \times n}.
\end{equation*}

The bias arises because the estimation error in $\hat{\pi}$ is correlated with $Y_{g}$. If we replace $\hat{\pi}$ in \cref{eq:theta_pi} with an \ac{OLS} estimator that leaves out cluster $g$, $\hat{\pi}_{-g}=(W'W-W_{g}'W_{g})^{-1}(W'X-W_{g}'X_{g})$, we obtain the leave-out estimator
\begin{equation*}
    \hat \theta_{\rm LO} =\sum_{g\in[G]} \hat{\pi}_{-g}' A_0 (W' W)^{-1} W_g' Y_g
\end{equation*}
that is unbiased by construction. This estimator was first proposed by \textcite{KSS2020} in the setting where $X=Y$, though their formal results focus on the setting with independent data. As discussed in \textcite{KSS2020}, the leave-out estimator can alternatively be thought of as a debiased version of the plug-in estimator. In particular, let \(P = W (W'W)^{-1} W'\) and \(M = I_n - P\) denote the projection and annihilator matrices, respectively. Let \(P_{g,h} \in \mathbb{R}^{n_g \times n_h}\) and \(M_{g,h} \in \mathbb{R}^{n_g \times n_h}\) denote the submatrices of \(P\) and \(M\) corresponding to rows and columns associated with clusters \(g\) and \(h\). We may then equivalently write
\begin{equation}\label{eq:theta_ub}
    \hat \theta_{\rm LO} = \hat{\theta}_{\rm PI}
-\sum_{g\in[G]} (M X)_g'  M_{g, g}^{-1} A_{g, g} Y_{g}=X' B Y,
\end{equation}
where, $B = A - M D$ and $D = \Bdiag(M_{g,g}^{-1}A_{g,g}) \in \Re^{n \times n}$ is the block diagonal matrix with blocks $M_{g, g}^{-1}A_{g, g}$ on its diagonal (so that the diagonal blocks $B_{g, g}$ are all zero). Here, $A_{g, g}$ is defined analogously to $M_{g, g}$, and $(MX)_g \in \Re^{n_g}$ denotes the subvector of $MX$ corresponding to cluster $g$. \Cref{eq:theta_ub} follows from the definition of $\hat \theta_{\rm LO}$ since, by the Woodbury identity,
 $\hat{\pi}_{-g}=\hat{\pi}-(W'W)^{-1}W_{g}'M_{g, g}^{-1}(MX)_{g}$. Since
\begin{equation*}
    \mathbb{E}\sum_{g\in[G]} (MX)_g' M_{g, g}^{-1} A_{g, g} Y_{g}=\mathbb{E}\sum_{g\in[G]} V_{g}' A_{g, g} U_g,
\end{equation*}
\cref{eq:theta_ub} implies that $\hat{\theta}_{LO}$ obtains by subtracting off a
bias estimate from the plug-in estimator.

The bias correction is not unique: any estimator of the form $X'(A-C)Y$, where
$C$ is a matrix such that $W'CW=0$ and $\Bdiag(C_{g,g})=\Bdiag(A_{g,g})$ will be unbiased. To
motivate the bias correction matrix $C_{\rm LO}=M\Bdiag(M_{g,g}^{-1}A_{g,g})$ we
consider, \Cref{lemma:rao} below generalizes the classic results in
\textcite{rao70} showing that $C_{\rm LO}$ has certain optimality properties. To
state the result, let $Q\ast S$ denote the Khatri-Rao product of two matrices
$Q,S$, so that the $(g, h)$ block of the matrix is given by the Kronecker
product $(Q\ast S)_{g, h}=Q_{gh}\otimes S_{gh}$. Under independent sampling, the
Khatri-Rao product reduces to the Hadamard product, $Q\ast S=Q\odot S$. Let
$\bvec(Q)=(\mkvec(Q_{1,1})',\dotsc,\mkvec(Q_{G,G}))'$ vectorize the diagonal
blocks of an $n\times n$ matrix $Q$.

\begin{lem}\label{lemma:rao}
  The bias-correction matrix $C$ that minimizes $\tr(C'C)$ subject to (i)
  $\Bdiag(C_{g,g})=\Bdiag(A_{g,g})$ (ii) $CW=0$, and (iii) $W'C=0$ is given by
  $C_{\rm KR}=M\Bdiag(\Lambda)M$, where $\bvec(\Lambda)$ solves
  $\bvec(A)=(M\ast M)\bvec(\Lambda)$. If property (ii) is not imposed, the
  solution is given by $C_{\rm LO}=M\Bdiag(M_{g,g}^{-1}A_{g,g})$, while if
  property (iii) is not imposed, it is given by $C_{\rm LO}'$.
\end{lem}

To interpret this result, note that property (i) is necessary for unbiasedness,
while properties (ii) and (iii) ensure that the estimator is invariant to the
signals $W\gamma$ and $W\pi$, respectively, in the regressions
in~\cref{eq:model_outcome,eq:model_treatment}. If these conditions hold and the
errors $U,V$ are homoskedastic, then the variance of the bias correction is
proportional to the square the Frobenius norm, $\tr(C'C)$. Thus, the
bias-correction $C_{KR}$ is variance-minimizing among all bias-corrections that
are invariant to the signal in both regressions. The resulting estimator
$\hat{\theta}_{\rm KR}=X'(B-C_{\rm KR})Y$ (or its version under independent sampling) has been previously studied by, among
others, \textcite{hrk69,CJN18,CGJN22,CNT23,L23}. If the invariance requirement is
weakened, so that we only require the bias-correction to be invariant to the
signal in one of the regressions, the minimum norm unbiased estimator is given
by $C_{\rm LO}$ or $C_{\rm LO}'$, which yields the leave-out estimator in
\cref{eq:theta_ub} that this paper focuses on (using $C_{\rm LO}'$ yields a
numerically equivalent expression if $X=Y$, but not in general).

While the invariance properties of the estimator $\hat{\theta}_{KR}$ are clearly
desirable, they come with three limitations relative to the estimator
$\hat{\theta}_{LO}$. First, since
$\tr(C_{\rm LO}'C_{\rm LO})\leq \tr(C_{\rm KR}'C_{\rm KR})$, it can display
greater variability. Second, computing the estimator $\hat{\theta}_{KR}$ is not
feasible in large datasets or in datasets with large clusters, as it requires
solving a linear system defined by the matrix $M\ast M$, which has dimension
$\sum_{g} n_{g}^2$, and may not even be storable in memory. In contrast, the
leave-out estimator only requires computation of least squares estimators.
Third, the estimator may not exist. A sufficient condition for the existence of
$\hat{\theta}_{KR}$ is that the matrix $M\ast M$ be invertible. Under
independent sampling, the matrix $M\ast M$ corresponds to the Hadamard product
$M\odot M$ (so that $(M\ast M)_{ij}=M_{ij}^{2}$), and \textcite{HoHo75} show
that a sufficient condition for its invertibility is that $\min_i M_{ii} > 1/2$,
which holds if $d < n/2$ and the design is sufficiently balanced. While this is
only a sufficient condition, simulations reported in \textcite{J22} show that,
with independent sampling, the estimator $\hat{\theta}_{KR}$ often fails to
exist when the regression is high-dimensional. Under clustered sampling, general
sufficient conditions for the invertibility of the Khatri-Rao product are, to
our knowledge, unknown. For these reasons, we focus on the leave-out estimator.

To test a particular null hypothesis $\theta=\theta_0$ at the significance level
$\alpha$, we reject if the $t$-statistic based on the leave-out estimator
exceeds the usual critical value ${z}_{1-\alpha/2}$, the $1-\alpha/2$ quantile
of a standard normal distribution. That is, we reject if
\begin{equation*}
\left \vert \frac{\hat \theta_{\rm LO} - \theta_0}{\hat \omega_n} \right \vert \geq {z}_{1-\alpha/2},
\end{equation*}
where $\hat \omega_n^2$ is either a consistent or a conservative estimator of the variance of the estimator $\omega_n^2 = \mathbb V(X' B Y)$. We consider variance estimation in \Cref{sec:variance_estimation}; asymptotic validity of this test follows from the fact that $\hat{\theta}_{\rm LO}$ is asymptotically normal, as shown in \Cref{sec:asymtptic_normality}. Before giving these results, the remainder of this section fleshes out three important special cases that fit this general setup.

\subsection{Instrumental Variables Regression with Many Instruments and Controls}\label{sec:IV}

To see how our setup covers instrumental variables regressions, decompose the covariate vector $W_{i, g}=(\mathcal W_{i, g}', Z_{i, g}')'$ into a vector of controls $\mathcal{W}_{i, g}\in \Re^{d_w}$ and a vector of instruments $Z_{i, g}\in \Re^{d_z}$, with $d= d_w + d_z$. We allow both $d_w$ and $d_z$ to diverge to infinity, potentially as fast as $n$. Let $\mathcal{Y}_{i, g}$ denote an outcome variable of interest, with $X_{i, g}$ denoting the treatment. Letting
\begin{equation*}
\mathcal{Y}_{i, g} = W_{i, g}\pi_{Y} + e_{i, g}, \quad \mathbb E [e_{i, g}]=0
\end{equation*}
denote the reduced-form regression and \cref{eq:model_treatment} the first-stage regression, the \ac{TSLS} estimand is given by
\begin{equation}\label{eq:beta}
    \beta = \frac{\pi' W' P_{\tilde Z} W \pi_Y}{\pi' W' P_{\tilde Z} W\pi},
\end{equation}
where $P_{\tilde Z}=\tilde Z(\tilde Z' \tilde Z)^{-1}\tilde Z'$ is the
projection matrix of $\tilde Z$, $\tilde Z = M_{\mathcal W} Z$, and
$M_{\mathcal W} = I_n - \mathcal W (\mathcal W' \mathcal W)^{-1}\mathcal W'$ is
the annihilator matrix of $\mathcal W$. For testing a particular value $\beta_0$
of $\beta$, let $Y=\mathcal{Y}-X\beta_0$ denote the treatment-adjusted outcome,
so that \cref{eq:model_outcome} corresponds to the treatment-adjusted outcome
regression with $\gamma=\pi_Y -\pi\beta_0$ and
$U_{i, g}=e_{i, g}-V_{i, g}\beta_0$. Our setup allows for testing the null value
$\beta_0$ by testing whether $\pi'A_0 \gamma=0$, with
\begin{equation*}
 A_0=W' P_{\tilde Z} W.
\end{equation*}

Hypothesis tests on $\beta$ are of interest because if the effect of $X_{i, g}$
on $\mathcal{Y}_{i, g}$ is constant, and the instrument satisfies an exclusion
restriction, then $\beta$ identifies this constant treatment effect, provided
that the regression of $\mathcal{Y}_{i, g}-X_{i, g}\beta$ on
$\mathcal{W}_{i, g}$ is linear. Our setting allows for heterogeneous treatment
effects and, in this case, a causal interpretation of $\beta$ requires an
instrument monotonicity assumption, and the assumption that the reduced-form and
first-stage regressions are (approximately) correctly specified. We refer to
\textcite{K13,EK2018,BN24,Yap24,BBMT22} for further discussion and statement of
the precise conditions. The requirement that the reduced-form and first-stage
regressions be approximately linear gives one motivation for why the covariate
vector $W_{i, g}$ may be high-dimensional even in settings where a baseline set
of instruments $Q_{i, g}$ and covariates $C_{i, g}$ is low dimensional (as
discussed in the introduction, in leniency IV designs, these vectors may already
be high-dimensional). In particular suppose we set
$\mathcal W_{i, g} = \mathcal W(C_{i, g})$, and
$Z_{i, g} = Z(Q_{i, g}, C_{i, g})$ to correspond to technical transformations of
these baseline variables ensuring that the linear specifications
$W_{i, g}'\pi_Y$ and $W_{i, g}'\pi$ provide good approximations to the
reduced-form $\mathbb{E}(\mathcal Y_{i, g} \mid Q_{i, g}, C_{i, g})$ and first
stage $\mathbb{E}(X_{i, g} \mid Q_{i, g}, C_{i, g})$, respectively. Typically,
such technical transformations involve interactions and series expansions, which
can lead to high dimensionality. Our analysis remains valid if the model in
\cref{eq:model_outcome,eq:model_treatment} contains an asymptotically negligible
approximation error, as in, for example, \textcite{CJN18}. For notational
simplicity, we abstract from any such approximation errors and assume that the
linearity holds exactly.

The variance $\omega_n^2$ for $X' B Y$ and its estimator $\hat{\omega}_n^2$ implicitly depend on $\beta_0$ through the relation $Y = \mathcal{Y} - X \beta_0$, so the resulting $t$-statistic corresponds to the weak-identification-robust Lagrange multiplier statistic. Our regularity conditions require that either the concentration parameter $\pi' W' P_{\tilde{Z}} W \pi$ or the number of instruments $d_z$ diverges to infinity. This framework encompasses both the strong-identification case with a finite number of instruments and the strong- and weak-identification cases with many instruments, including situations with many weak IVs in which $\pi' W' P_{\tilde{Z}} W \pi / \sqrt{d_z}$ remains bounded.

If the instruments are collectively strong, so that $\pi' W' P_{\tilde{Z}} W \pi / \sqrt{d_z}$ diverges, the leave-out estimator of $\beta$,
\begin{equation*}
    \hat \beta = \frac{X' B \mathcal Y}{X' B X},
\end{equation*}
will be consistent, even in the presence of many instruments (this stands in
contrast with the inconsistency of TSLS, which can be thought of as a plug-in
estimator of $\beta$). In this case, we do not have to impose the null when
computing the variance of $X'BY$, and can instead base inference on the Wald
test, rejecting whenever the absolute value of the  $t$-statistic
\begin{equation*}
\frac{X' B X (\hat \beta - \beta)}{\hat \omega_n(\hat \beta)}=\frac{X' B (\mathcal Y - \beta X)}{\hat \omega_n(\hat \beta)}
\end{equation*}
exceeds $z_{1-\alpha/2}$, where $\hat \omega_n^2(\hat \beta)$ is just $\hat \omega_n^2$ with $\beta_0$ in $Y = \mathcal Y - X \beta_0$ replaced by $\hat \beta$.

The construction of our estimator $\hat{\beta}$ (and its corresponding weak-IV robust Wald test) can be thought of as a cluster-robust version of the UJIVE estimator studied by \textcite{K13} and further advocated by \textcite{GHK25}, which employs a leave-one-out technique to eliminate the overfitting bias of TSLS\@.

\begin{rem}
  Validity of the Wald test requires consistency of the variance estimator
  $\hat{\omega}_{n}^2(\hat{\beta})$ in the sense that
  $\hat \omega_n^2(\hat \beta)/\omega_n^2(\beta) \convP 1$, where
  $\omega_n^2(\beta) := Var \left(X' B (\mathcal Y - \beta X) \right)$. Our
  results below show that $\hat \omega_n^2(\beta) /\omega_n^2(\beta) \convP 1$.
  Since $\hat{\beta}$ is consistent, a stochastic equicontinuity argument can
  then be used to show consistency of the variance estimator
  $\hat \omega_n^2(\hat \beta)$, using the fact that
  $\hat \omega_n^2(\beta_0)/\omega_n^2(\beta)$ is Lipschitz continuous in
  $\beta_0$ over a neighborhood of $\beta$.
\end{rem}

\subsection{Variance and Covariance Components in Linear Regressions}\label{sec:var_comp}

Consider a two-way fixed effect model of log wage determination proposed by
\textcite{AKM99}, in which the log-wage $Y_{\ell, t}$ of a worker $\ell\in[L]$ in year
$t \in [T_{\ell}]$ is given by
\begin{equation}\label{eq:akm}
  Y_{\ell, t} =
  \alpha_{\ell} + \psi_{\mathcal J(\ell, t)} + \mathcal{W}_{\ell, t}'\delta +
  U_{\ell, t}.
\end{equation}
Here $\alpha_{\ell}$ is a worker fixed effect, $\mathcal J(\ell, t) \in [J]$
returns the identity of the firm that employs $\ell$ in year $t$, $\psi_{j}$ is
a firm fixed effect, $\mathcal{W}_{\ell, t}$ contains time-varying control
variables, and $U_{\ell, t}$ is a time-varying noise term. The error
$U_{\ell, t}$ is assumed to be mean zero, which imposes strict exogeneity:
workers and firms are allowed to match based on firm and worker fixed effects,
but not on time-varying factors influencing wages. As discussed in a recent
survey by \textcite{kline2024}, it is desirable to allow this noise to exhibit
arbitrary correlation within a worker-firm match to allow for match-specific
unobservables, so that $U_{\ell, t}$ and $U_{\ell, t'}$ may be correlated if
$J(\ell, t)=J(\ell, t')$. To map this to our setup, following \textcite[Section
4.2]{kline2024}, let $g$ index worker-firm matches, with $j(g)$ returning the
firm and $\ell(g)$ the worker that form the match. Then $n_{g}$ corresponds to
the duration of the match (so if a worker $\ell(g)$ spends 5 years at a firm
$j(g)$, say, then $n_{g}=5$). Then we may rewrite \cref{eq:akm} as
\begin{equation*}
  Y_{i, g} =  F_{j(g)}'\psi + D_{\ell(g)}'\alpha + \mathcal{W}_{i,g}'\delta + U_{i, g},
\end{equation*}
where $i\in[n_{g}]$ indexes years during the match, $D_{\ell}\in\Re^{L}$ is the
$\ell$-th basis vector, and $F_{j}\in\Re^{J-1}$ is the $j$-th basis vector (we
normalize the last firm's effect to zero). This maps to our general framework in
\cref{eq:model_outcome,eq:model_treatment} by setting $Y_{i, g}=X_{i, g}$,
$W_{i, g}=(F_{j(g)}, D_{\ell(g)}, \mathcal{W}_{i, g})$, and
$\gamma=\pi=(\psi', \alpha', \delta')'$.

The key parameters of interest in this model are the variances of firm and
worker effects, as well as their covariance. A person-year weighted variance of
the firm effects is given by
$\sigma^{2}_{\psi}=\frac{1}{n}\sum_{g\in[G]}n_{g}(\psi_{j(g)}-\bar{\psi})^{2}$,
where $\bar{\psi}=\frac{1}{n}\sum_{g\in[G]}n_{g} \psi_{j(g)}$, and measures the
direct contribution of firms to wage inequality. This can be written as a
quadratic form in \cref{eq:mom}, $\sigma_{\psi}^2 = \gamma'A_\psi \gamma$, with
the matrix $A_{\psi}$ given by
\begin{equation*}
  A_\psi = \frac{1}{n}
  \begin{pmatrix}
    (F - 1_{n}\bar{F}')'(F - 1_{n}\bar{F}') & 0 \\
    0 & 0
  \end{pmatrix},
\end{equation*}
where $\bar{F} = \frac{1}{n}\sum_{g \in [G]}n_{g}F_{g}$, and $1_{n}$ is a vector
of ones. Variance of worker effects, given by
$\sigma^{2}_{\alpha}=\frac{1}{n}\sum_{g\in[G]}n_{g}(\alpha_{\ell(g)}-\bar{\alpha})^{2}$,
maps to \cref{eq:mom} analogously.

The covariance between worker and firm effects is given by
$\sigma_{\alpha,\psi} = \frac{1}{n}\sum_{g \in [G]}n_{g}(\psi_{j(g)} - \bar\psi)
 \alpha_{\ell(g)}=\gamma' A_{\alpha,\psi} \gamma$, with
\begin{equation*}
    A_{\alpha,\psi} = \frac{1}{2 n }
    \begin{pmatrix}
        0 & (F - 1_{n}\bar{F}')'(D - 1_{n}\bar{D}') & 0 \\
        (D - 1_{n}\bar{D}')'(F - 1_{n}\bar{F}') & 0 & 0 \\
        0 & 0 & 0
    \end{pmatrix},
\end{equation*}
and $\bar{D} = \frac{1}{n} \sum_{g \in [G]}n_{g}{D}_{\ell(g)}$, so that
$\sigma_{\alpha, \psi}$ again takes the form of \cref{eq:mom}. The covariance
$\sigma_{\alpha, \psi}$ measures the contribution of systematic sorting of high
wage workers to high wage firms to wage
inequality.

While we focus on the \textcite{AKM99} model for concreteness, variance and
covariance components are of interest in numerous other settings that exhibit
potentially high-dimensional fixed effects as well. Examples include determining
the importance of neighborhoods for intergenerational mobility
\parencite{ChHe18i}, the importance of geography for healthcare utilization
\parencite{fgw16}, or the importance of classroom assignment in determining
student outcomes \parencite{chetty2011}. In some cases, covariance components
involve fixed effects estimated from different regressions (so that $Y\neq X$):
for instance, \textcite{chetty2011} are interested in the covariance between
classroom fixed effects in an earnings regression and that in a test score
regression.

\subsection{Testing Many Linear Restrictions}\label{sec:testing-many-linear}

Consider the linear regression~\eqref{eq:model_outcome}, so that
\( X_{i,t} = Y_{i,t} \) in \cref{eq:model_treatment}. As in \textcite{AS23}, we
are interested in testing the linear restriction
\begin{equation*}
    R \gamma = q,
\end{equation*}
where \( q \) may be high-dimensional. This restriction implies
\begin{equation*}
     \gamma' \left[ R' (R(W' W)^{-1} R')^{-1} R \right] \gamma
     = q' (R(W' W)^{-1} R')^{-1} q,
\end{equation*}
which fits into our framework with \( Y = X \), and
$A_0 = R' (R(W' W)^{-1} R')^{-1} R$, and the hypothesis of interest is that
$\theta = q' (R(W' W)^{-1} R')^{-1} q$. Note that the plug-in estimator
$X'A_{0}X$ corresponds to the numerator of the classic $F$-statistic under
homoskedasticity.

By letting $q=0$ and letting $R$ correspond to a matrix that selects a subset of
the coefficients $\gamma$, this setup nests testing for the presence of fixed
effects, which is how tests of heterogeneity can be cast
\parencite[e.g.,][]{krw22,fgw16}.

The setup also covers validation of value-added models. In particular, value
added models are typically estimated by a linear regression specification, where
$i$ indexes students within clusters $g$ (such as classrooms), and the vector of
regressors $\mathcal{W}^{0}_{i, g}=(\mathcal{W}_{i,g}',\mathcal{D}_{i, g}')'$,
consists of a set of school dummies $\mathcal{W}_{i, g}$ indicating which school
$i$ attends, and a set of controls $\mathcal{D}_{i, g}$ (that may include lagged
test scores). The outcome $Y_{i, g}=X_{i, g}$ denotes student test scores.
Provided that $\mathcal{W}^{0}_{i, g}$ is exogenous, the coefficients on the
school dummies (i.e., estimates of school fixed effects), can be interpreted as
school value added estimates. \Textcite{ahpw24} point out that the exogeneity
assumption is testable if we have a set of covariates $\mathcal{Z}_{i, g}$ that
affect school attendance, but not test scores directly (such as indicators for
winning a lottery to attend oversubscribed schools). \Textcite{ahpw24} then
develop a test of exogeneity by adapting the classic Sargan test by viewing
$\mathcal{Z}_{i, g}$ as a set of instruments. Their framework requires the
instrument and covariate dimension to be fixed, and the sampling to be
independent. However, the exogeneity assumption is equivalent to the assumption
that the coefficients on $\mathcal{Z}_{i,g}$ in the long regression with
$\mathcal{W}_{i, g}=(\mathcal{W}_{i,g}',\mathcal{D}_{i, g}',\mathcal{Z}_{i,
  g})'$ as the set of regressors all equal zero. This fits our framework by
letting $R$ denote the selector matrix that selects the coefficients on
$\mathcal{Z}_{i, g}$, and setting $q=0$. Doing so allows us to accommodate
high-dimensional instruments and many schools, as well as clustering.

This framework also nests testing for endogenous peer effects. In particular,
\textcite{JL26} point out that in a panel data setting, one can test for peer
effects without specifying the network structure by using an Anderson--Rubin
test in an \ac{IV} regression of own outcomes on outcomes of potential peers,
instrumenting with peer characteristics. However, implementing the test is
complicated by the presence of individual and time fixed effects. The
implementation in \textcite{JL26} assumes the regression errors are
homoskedastic and independent across both time and individuals. However, an
Anderson--Rubin test of a zero effect of endogenous variables is equivalent to
an $F$ test on the reduced form. Thus, the null hypothesis is equivalent to
testing whether coefficients on the instruments in a reduced-form regression of
own outcomes on controls (that may include individual and time effects) and
instruments equal zero. This again fits the above framework, if we set $R$ to be
a matrix that selects the regressor coefficients on the instruments. Doing so
allows us to accommodates both heteroskedasticity and cluster dependence, either
in the time dimension or in the cross-section (e.g., by classroom).

\section{Asymptotic Normality}\label{sec:asymtptic_normality}
This section shows that the estimator $\hat{\theta}_{\textrm{LO}}$ is
asymptotically normal. To state the assumptions needed for this result, let
$\lambda_{\max}(Q)$ and $\lambda_{\min}(Q)$ denote the largest and smallest
eigenvalues of a symmetric matrix $Q$, and for any matrix $Q$, let
$\norm{Q}_{op}=\lambda_{\max}(Q'Q)^{1/2}$, and $\norm{Q}_{F}=\tr(Q'Q)^{1/2}$.
Finally, $\norm{q}_{2}$ denotes the Euclidean norm of a vector $q$.

\begin{ass}\label{ass:dgp}
  \begin{enumerate}
  \item\label{item:ass_clust} \Cref{eq:model_outcome,eq:model_treatment} hold with $W$ nonrandom and
    $\{(U_{g},V_{g})\}_{g \in [G]}$ independent across $g \in [G]$.
  \item\label{item:ass_lo} $ \min_{g \in [G]} \lambda_{\min} \left( M_{g,g} \right) \geq c$ for
    some constant $c>0$.
  \item\label{item:ass_moments}
    $\max_{g \in [G]}\max_{i \in [n_{g}]} \left( \mathbb E U_{i,g}^{2q}+\mathbb E
      V_{i,g}^{2q} \right) \leq C$ for some constants $C<\infty$ and $q\geq 2$.
    In addition, there exists a constant $c>0$ and a sequence $u_{n}<\infty$ such that
    \begin{equation*}
      c \leq             \min_{g \in [G]}\lambda_{\min}\left( \Omega_{g}\right)
      \leq               \max_{g \in [G]}\lambda_{\max}\left( \Omega_{g}\right)  \leq u_n,
    \end{equation*}
    where
    \begin{equation*}
      \Omega_{g} = \begin{pmatrix}
        \Omega_{U,g} & \Omega_{U,V,g} \\
        \Omega_{U,V,g}' & \Omega_{V,g}
      \end{pmatrix} = \begin{pmatrix}
        \mathbb E U_{g} U_{g}' & \mathbb E U_{g} V_{g}' \\
        \mathbb E V_{g} U_{g}' & \mathbb E V_{g} V_{g}'
      \end{pmatrix}.
    \end{equation*}
\end{enumerate}
\end{ass}

\begin{rem}
  \Cref{ass:dgp}.\ref{item:ass_clust} describes the data-generating process for
  clustered observations. \Cref{ass:dgp}.\ref{item:ass_lo} ensures that the
  model remains estimable after leaving out any particular cluster; otherwise
  the estimator $\hat{\theta}_{\textrm{LO}}$ would not be well-defined. When the
  observations are independent, the condition is equivalent to
  $\min_{i \in [n]} M_{i, i} \geq c$ for some small constant $c>0$, which is
  necessary for a leave-observation out regression to be feasible, and standard
  in the literature \parencite[e.g.,][]{KSS2020,J22,AS23}. In the context of IV
  leniency designs or estimation of variance components, one can ensure this
  condition holds by ``pruning'' the sample and dropping observations associated
  with singleton decision-makers or firms \parencite[see, e.g.,][for
  details]{KSS2020}.

  Finally, \Cref{ass:dgp}.\ref{item:ass_moments} imposes mild restrictions on
  the moments and the within-cluster covariance structure of the error terms.
  The upper bound $u_{n}$ captures the strength of within-cluster dependence.
  While we can always take $u_{n}$ to equal a constant times the largest cluster
  size $n_{G}$, which is the best bound when the errors share a common
  cluster-specific component, a tighter upper bound $u_{n}$ obtains when the
  within-cluster errors do not exhibit such strong dependence. For example, in a
  panel data setting, where $i$ indexes the time periods for which we observe
  individuals indexed by $g$, the errors may satisfy weak dependence,
  $cov(U_{i,g},U_{i',g}) = \rho^{|i'-i|}$, in which case $u_n$ is bounded. A
  tighter upper bound $u_{n}$ and the existence of higher-order moments (higher
  $q$) allow us to weaken the conditions on cluster sizes (see \Cref{rem:rate1}
  below).
\end{rem}

\begin{ass}\label{ass:reg}
  Let $\Pi = W \pi $, $\Gamma = W \gamma$, $r_n = \operatorname{rank}(A)$, $h_n = \norm{(W'W)^{-1/2} A_0 (W'W)^{-1/2}}_{op}$,
  \begin{align*}
  H & = B' \Pi/h_n,& \tilde H& = B \Gamma/h_n, & \kappa_n& = \norm{B}_{F}^{2}/h_n^2, \\
    \lambda_n& = \max_{g \in [G]} \norm{P_{g,g}}_{op},&  \zeta_{H,n}& = \max_{g \in [G]}\norm{H_{g}}_2^2, &  \zeta_{\tilde H,n}& = \max_{g \in [G]}
      \norm{\tilde H_{g}}_2^2, \\
      \eta_{n}&  = \max \left\{1, \log^2 (r_n) + \log^2 (n) \lambda_n^2 \right\}&
     \mu_n^2 &= \norm{H}_2^2, &\text{and}\quad \tilde \mu_n^2 &= \norm{\tilde H}_2^2,
  \end{align*}

\begin{enumerate}
    \item\label{item:signal_bound} $\max_{g \in [G]}||\Pi_{g}||_2^2 + \max_{g \in [G]}||\Gamma_{g}||_2^2 \lesssim n_G$.
    \item\label{item:clt_rate} $G \to \infty$, and the following condition holds
    \begin{multline*}
      u_n^3 \eta_n (\mu_n^2 + \tilde \mu_n^2)+u_n^{\frac{3q-4}{q-1}}
      n_G^{\frac{q}{q-1}} \lambda_n \kappa_n +u_n^4 \eta_n \kappa_n+u_n^{\frac{2(q-2)}{q-1}} n_G^{\frac{2q}{q-1}} \lambda_n^2 \kappa_n\\
      +u_n^{\frac{q-2}{q-1}} n_G^{\frac{q}{q-1}} (\zeta_{H,n}\mu_n^2 + \zeta_{\tilde H,n} \tilde \mu_n^2 + \kappa_n)
      = o(       (\mu_n^2 + \tilde \mu_n^2 + \kappa_n)^2).
    \end{multline*}
\end{enumerate}
\end{ass}

\begin{rem}
A sufficient condition for \Cref{ass:reg}.\ref{item:signal_bound}
to hold is that the signal in \cref{eq:model_outcome,eq:model_treatment},
$\Pi_{i, g}$ and $\Gamma_{i, g}$, is bounded, which is mild.
\end{rem}

\begin{rem}\label{rem:kappa}
Across the three examples, many-IV regression, variance components, and many restrictions, the rank \(r_n\) of \(A\) equals the number of instruments, firms, and restrictions, respectively. In the first and third cases, \(\kappa_n \asymp r_n\). In the variance-components case, \textcite{KSS2020} show \(\kappa_n \ge \tfrac{1}{4}\, J\, \dot{\lambda}_J^{\,2}\) for a stochastic block model for the firm connectivity network,\footnote{Two firms are considered connected if at least one worker moves from one firm to the other.} where $J$ is the number of firms,  \(\dot{\lambda}_J\) is the smallest eigenvalue of \(E^{1/2}\mathcal{L}E^{1/2}\), \(\mathcal{L}\) is the normalized graph Laplacian of the employer mobility network, and \(E\) collects employer-specific churn rates. If the network's Cheeger constant is bounded away from zero, then \(\kappa_n \asymp J\).
\end{rem}

\begin{rem}\label{rem:lambda}
  The parameter \(\lambda_n\) represents the maximum leverage in the cluster
  setting, and we naturally have \(\lambda_n \le 1\). When the design matrix
  \(W\) is well balanced, we should have \(\lambda_n \lesssim d/n\). In the
  first two examples in \Cref{sec:setup}, \(r_n\) denotes, respectively, the
  number of instruments and the number of firms, and \(d\) equals \(r_n\) plus
  the number of controls. If the number of controls does not dominate \(r_n\),
  then \(d \lesssim r_n\), which, combined with the discussion in
  \Cref{rem:kappa}, implies $\lambda_n \lesssim \kappa_n/n$.
\end{rem}

\begin{rem}\label{rem:eta}
  To establish the asymptotic distribution of the quadratic form, we bound the
  operator norm of the upper-triangular part of \(A\), denoted \(\nabla(A)\).
  When \(A\) is a projection matrix, \textcite{Chao12} show that
  \(\lVert \nabla(A)\rVert_{\mathrm{op}} \lesssim r_n^{1/4}\), which may be
  loose when \(r_n \asymp n\). Instead, we provide a sharper bound,
  \begin{equation*}
      \norm{\nabla(A)}_{\mathrm{op}} \lesssim \log(r_n)\,\norm{A}_{\mathrm{op}},
  \end{equation*}
  for a general matrix \(A\), which motivates the definition of \(\eta_n\).

\end{rem}

\begin{rem}\label{rem:linear+quadratic}
  In the proof, we show that the estimation error
  \(\hat{\theta}_{\mathrm{LO}} - \theta\) is decomposed into linear and
  quadratic terms. The variability of the linear component is governed by
  \(\mu_n^2 + \tilde{\mu}_n^2\); in the many-IV setting, \(\mu_n^2\) coincides
  with the usual concentration parameter in IV regression (with \(Y\) the
  outcome and \(X\) the endogenous regressor), while \(\tilde{\mu}_n^2\) is the
  concentration parameter for the reverse IV regression of \(X\) on \(Y\). The
  variability of the quadratic component is captured by \(\kappa_n\), which is
  proportional to \(r_n\) in many cases (as discussed above).
  \Cref{ass:reg}.\ref{item:ass_moments} implies that at least one of
  \(\mu_n^2 + \tilde{\mu}_n^2\) and \(\kappa_n\) must diverge.

  The first condition \((\mu_n^2 + \tilde{\mu}_n^2 \to \infty)\) holds in our
  three leading examples when, respectively: (i) the concentration parameter
  diverges (many-IV), (ii) the total variation across firms diverges (variance
  components), and (iii) the total variation of \(W\) projected onto
  \((W'W)^{-1/2}R'\) diverges (many restrictions). If the number of IVs is
  fixed, then in the first example, this corresponds to strong identification.
  Indeed, it is well known that the usual Lagrange multiplier test under weak
  identification is not asymptotically normal \parencite{SS97}. The second
  condition \((\kappa_n \to \infty)\) holds when the number of instruments
  (many-IV), firms (variance components), or restrictions (many restrictions)
  diverges, i.e., \(r_n \to \infty\). This accommodates many-weak-IV settings in
  which no consistent test for \(\theta\) exists when
  \(\mu_n^2/\sqrt{\kappa_n}\) is bounded \parencite{MS22}. Moreover, because
  $\kappa_{n}$ is normalized by the operator norm $h_{n}$, this second condition
  also entails the Lindeberg-type condition in \textcite[Section~5]{KSS2020},
  yielding a Gaussian limit for \(\hat{\theta}_{\mathrm{LO}}\).

  Analyzing the limiting distribution of \(\hat{\theta}_{\mathrm{LO}}\) when the
  Lindeberg-type condition fails encounters the challenges highlighted by
  \textcite[Section~6]{KSS2020}, especially when \(d \asymp n\) (see
  \textcite{LWZ24b} for a valid bootstrap procedure when $d= o(n)$). This issue
  is further complicated by within-cluster dependence in our setting, and we
  leave a full treatment to future work.
\end{rem}

\begin{rem}\label{rem:rate1}
\Cref{ass:reg}.\ref{item:clt_rate} accommodates several scenarios provided that the regularization parameter satisfies $\lambda_n \lesssim \kappa_n / n$, and $r_n \asymp \kappa_n$, as discussed in \Cref{rem:kappa,rem:lambda}.

\textbf{Scenario 1.} Suppose that \(\kappa_n \asymp n\) and $\zeta_{H,n} + \zeta_{\tilde H,n} \lesssim n_G$. If the within-cluster dependence is weak (i.e., \(u_n \lesssim 1\)), then \Cref{ass:reg}.\ref{item:clt_rate} holds provided \(n_G^{\frac{2q}{q-1}} = o(n)\), regardless of the order of \(\mu_n^2 + \tilde{\mu}_n^2\). If the errors have moments of all orders, a sufficient condition is \(n_G^2 = o(n)\). Under strong within-cluster dependence (i.e., \(u_n \lesssim n_G\)), it suffices that \(n_G^4 \log^2 n = o(n)\), again irrespective of \(\mu_n^2 + \tilde{\mu}_n^2\). The same rate requirements apply under strong identification, in the sense that \(\mu_n^2 + \tilde{\mu}_n^2 \asymp n\), irrespective of the order of \(\kappa_n\).

\textbf{Scenario 2.} Suppose instead that \(\kappa_n \asymp G\), which holds by
construction if IVs are assigned at the cluster level or the number of firms is
less than the number of workers, and $H_g$ and $\tilde H_g$, are balanced across
clusters, i.e.,
\[
\frac{\zeta_{H,n}}{\mu_n^2} \lesssim \frac{1}{G}
\qquad \text{and} \qquad
\frac{\zeta_{\tilde H,n}}{\tilde \mu_n^2} \lesssim \frac{1}{G}.
\]
Then, under weak within-cluster dependence, \Cref{ass:reg}.\ref{item:clt_rate} holds if \(n_G^{\frac{2q}{q-1}} G = o(n^2)\) and \(n_G^{\frac{q}{q-1}} = o(G)\), regardless of the order of \(\mu_n^2 + \tilde{\mu}_n^2\). Under strong within-cluster dependence, a sufficient condition is \(n_G^{4} \log^2 n = o(G)\), again regardless of the order of \(\mu_n^2 + \tilde{\mu}_n^2\).

\textbf{Scenario 3.} If the clusters have a bounded size such that $n_G$ is bounded and $G \asymp n$, then \Cref{ass:reg}.\ref{item:clt_rate} holds as long as \(\mu_n^2 + \tilde{\mu}_n^2 + \kappa_n \to \infty \) and $\zeta_{H,n} + \zeta_{\tilde H,n} = o(\mu_n^2 + \tilde{\mu}_n^2 + \kappa_n)$. This setting includes independent observations as a special case with $n_G = 1$. Under this configuration, our condition is sufficient to guarantee the Lindeberg conditions imposed in Theorems~1 and~2 of \textcite{KSS2020}, which in turn imply the asymptotic normality of the proposed estimator.

\end{rem}

\begin{thm}\label{thm:clt}
Suppose \Cref{ass:dgp,ass:reg} hold. Then
\begin{equation*}
    \frac{\hat \theta_{\rm LO} - \theta}{ \omega_n} \convD \N(0,1),
\end{equation*}
where
\begin{multline*}
  \omega_n^2 =  \sum_{g,h \in [G]^2} tr\left( \Omega_{V,g} B_{g,h} \Omega_{U,h} B_{g,h}'\right)  + \sum_{g,h \in [G]^2}  tr\left( \Omega_{U,V,g} B_{g,h}  \Omega_{U,V,h} B_{h,g} \right)\\
  + \sum_{g \in [G]} H_g' \Omega_{U,g} H_g  +  \sum_{g \in [G]} \tilde H_g' \Omega_{V,g} \tilde H_g  + 2  \sum_{g \in [G]} H_g' \Omega_{U,V,g} \tilde H_g .
\end{multline*}
\end{thm}

\section{Variance Estimator}\label{sec:variance_estimation}

To make use of \Cref{thm:clt} for inference, we need a consistent or a
conservative estimator of the asymptotic variance $\omega_{n}$. We now consider
two such estimators. The first estimator we consider, a \acf{L3CO} estimator, is
shown to be consistent, while the second one, a \acf{L2CO} estimator, is shown
to be conservative.

\subsection{Leave-three-clusters-out Variance Estimator}

The L3CO variance estimator we consider generalizes the leave-three-observations
out estimators considered in \textcite{AS23,Yap24} to the case with clustered
data, and is defined as
\begin{equation*}
\hat \omega^2_{n,\rm L3CO}    = \hat \omega^2_{n,\rm L3CO,1} + 2\hat \omega^2_{n,\rm L3CO,2} + \hat \omega^2_{n,\rm L3CO,3} - (\hat \omega^2_{n,\rm L3CO,4} + \hat \omega^2_{n,\rm L3CO,5}),
\end{equation*}
where
\begin{align*}
\hat \omega^2_{n,\rm L3CO,1} &    = \sum_{g, h, k \in [G]^3} \left(X_h ' B_{h,g} Y_g \right) \left(X_k ' B_{k,g} \tilde Y_{g,-hk} \right), \\
\hat \omega^2_{n,\rm L3CO,2} &    = \sum_{g, h, k \in [G]^3} \left(X_h ' B_{h,g} Y_g \right) \left(Y_k ' B_{g,k}' \tilde X_{g,-hk} \right), \\
\hat \omega^2_{n,\rm L3CO,3} &   = \sum_{g, h, k \in [G]^3} \left(Y_h ' B_{g,h}' X_g \right) \left(Y_k ' B_{g,k}' \tilde X_{g,-hk} \right), \\
\hat \omega^2_{n,\rm L3CO,4} &   = \sum_{g, h, k \in [G]^3} \left(\tilde{Y}_{h, -gk}' B_{g,h}' X_g \right) \left(Y_h' B_{g,h}' \tilde M_{g,k,-gh} X_k \right), \\
\hat \omega^2_{n,\rm L3CO,5} &    = \sum_{g, h, k \in [G]^3}
                               \left(\tilde{X}_{h, -gk}' B_{h,g} Y_g \right) \left(Y_h' B_{g,h}' \tilde M_{g,k,-gh} X_k \right).
\end{align*}
Here $\tilde Y_{g,-hk} = Y_g - W_g \hat \gamma_{-ghk}$, and
$\tilde X_{g,-hk} = X_g - W_g \hat \pi_{-ghk}$ denote leave-three-clusters-out
OLS residuals,
\begin{align}\label{eq:gamma-ghk}
  \hat{\gamma}_{-ghk}&=(W'W-\sum_{l \in (g,h,k)} W_{l}'W_{l})^{-1}(W'Y-\sum_{l
                       \in (g,h,k)} W_{l}'Y_{l}) \quad\text{and}\\
  \hat{\pi}_{-ghk}&=(W'W-\sum_{l \in (g,h,k)} W_{l}'W_{l})^{-1}(W'X-\sum_{l \in (g,h,k)} W_{l}'X_{l})\label{eq:pi-ghk}
\end{align}
denote leave-three-clusters-out OLS estimators, and
\begin{equation*}
  \tilde M_{{g,k}, -gh} = \left(M_{g,g} - M_{g,h} M_{h,h}^{-1} M_{h,g} \right)^{-1} \left(M_{g,k} - M_{g,h} M_{h,h}^{-1} M_{h,k} \right).
\end{equation*}

The definition of $\tilde M_{{g,k}, -gh}$ implies that when $g \neq h$,
\begin{equation*}
  \tilde M_{g,k,-gh}   = \begin{cases}
     I_{n_g}, \quad k = g; \\
     0_{n_g \times n_h}, \quad k = h; \\
     - W_{g} (W'W-\sum_{l \in (g,h)} W_{l}'W_{l})^{-1} W_{k}', \quad k\neq g, k\neq h
\end{cases}
\end{equation*}
and thus $\tilde M_{g,k,-gh}$ can partial out $W$ in the sense that
$\sum_{k \in [G]} \tilde M_{g,k,-gh} W_k = 0_{n_g \times d}$. The key property
of this variance estimator is that
$\mathbb E \hat \omega^2_{n,\rm L3CO} = \omega_n^2$, which leads to consistency
even when $d$ is large. To establish this formally, we impose the following
assumption:
\begin{ass}\label{ass:var_L3O}
For $g \neq h \neq k$, define
\begin{align*}
  S_{k,g}& =    M_{k,k} - M_{k,g} M_{g,g}^{-1} M_{g,k}, \\
  \tilde S_{k,gh}& = S_{k,g} - \left(M_{k,h} - M_{k,g} M_{g,g}^{-1} M_{g,h}\right)  S_{h,g} ^{-1} \left( M_{h,k} - M_{h,g} M_{g,g}^{-1} M_{g,k} \right),
\end{align*} and let \(\phi_n = \max_{g \in [G]}\sum_{h \in [G]} ||P_{g,h}||_{op}^2\).
\begin{enumerate}
\item\label{item:l3o_well_defined} There exists a finite constant $C>0$ such that
\begin{equation*}
 \max_{k,g \in [G]^2, k \neq g}  \norm{S_{k,g}^{-1}}_{op} + \max_{g,h,k \in [G]^3, k \neq g, h \neq g, k \neq h}
  \norm{\tilde S_{k,gh}^{-1}}_{op}  \leq C.
\end{equation*}
\item\label{item:l3o_rates} We have
\begin{multline*}
  u_n n_G^3 (\phi_n \lambda_n^2 + \phi_n^2 \lambda_n + \phi_n^3 \lambda_n^2 + \lambda_n^2 )\kappa_n +
  u_n^2  n_G^2 \phi_n \lambda_n  \kappa_n+
  u_n n_G^2 \phi_n \lambda_n^3 (\mu_n^2 + \tilde \mu_n^2)\\
  +u_n^{\frac{2q-3}{q-1}} n_G^{\frac{q}{q-1}} (\zeta_{H,n} + \zeta_{\tilde H,n}) \kappa_{n}
  +u_n^2 n_G \phi_n (\zeta_{H,n} + \zeta_{\tilde H,n}) \kappa_n
  + u_n^4 n_G \lambda_n^2 \kappa_n = o\left( (\mu_n^2 + \tilde \mu_n^2 + \kappa_n)^2\right).
\end{multline*}
\end{enumerate}
\end{ass}

\begin{rem}
  \Cref{ass:var_L3O}.\ref{item:l3o_well_defined} ensures that L3CO least-squares
  estimator is well-defined and numerically stable. It is also sufficient, but
  not necessary, for the L2CO estimator, defined in \Cref{sec:cons-vari-estim}
  below, to be well-defined.
\end{rem}

\begin{rem}\label{rem:phi}
  The quantity \(\phi_n\) is naturally bounded by \(\lambda_n n_G\). If cluster
  sizes are uniformly bounded, or if
  \(\lambda_n \lesssim \kappa_n/n \lesssim G/n\) as discussed in
  \Cref{rem:kappa,rem:lambda}, then \(\phi_n \lesssim 1\). In addition, if the
  projection matrix \(P\) is sparse in the sense that there are only finitely
  many nonzero blocks in the row block matrix \([P_{g,1},\ldots,P_{g,G}]\) for
  each \(g \in [G]\)---which occurs when \(W\) partitions the clusters---then
  \(\phi_n \lesssim 1\). Finally, if the eigenvalues of \(P_{g,h}P_{h,g}\) are
  well balanced for all \(g,h \in [G]^2\) in the sense that
  \begin{equation*}
      \lVert P_{g,h}P_{h,g}\rVert_{\mathrm{op}} \lesssim \frac{\operatorname{tr}(P_{g,h}P_{h,g})}{n_g},
  \end{equation*}
  then \(\phi_n \lesssim \lambda_n \le 1\) provided that cluster sizes are well
  balanced (i.e., \(n_g \asymp n/G\) for all \(g \in [G]\)). These rate
  conditions are directly verifiable, since $\phi_n$ is a known function of the
  data.
\end{rem}

\begin{rem}\label{rem:rate2}
  We now derive the implications of the rate requirements in
  \Cref{ass:var_L3O}.\ref{item:l3o_rates} for the three scenarios in
  \Cref{rem:rate1}, assuming $\phi_n \lesssim 1$.

\textbf{Scenario 1.} Suppose that \(\kappa_n \asymp n\) and $\zeta_{H,n} + \zeta_{\tilde H,n} \lesssim n_G$. If the within-cluster dependence is weak (i.e., \(u_n \lesssim 1\)), then \Cref{ass:var_L3O}.\ref{item:l3o_rates} holds provided \(n_G^{3} = o(n)\), regardless of the order of \(\mu_n^2 + \tilde{\mu}_n^2\). Under strong within-cluster dependence (i.e., \(u_n \lesssim n_G\)), it suffices that \(n_G^{5} = o(n)\), again irrespective of \(\mu_n^2 + \tilde{\mu}_n^2\). The same rate requirements apply under strong identification, in the sense that \(\mu_n^2 + \tilde{\mu}_n^2 \asymp n\), regardless of the order of \(\kappa_n\).

\textbf{Scenario 2.} Following Scenario 2 in \Cref{rem:rate1}, if the
within-cluster dependence is weak, then \Cref{ass:var_L3O}.\ref{item:l3o_rates} holds if
\(n_G^{3} = o(n)\) and \(n_G^{\frac{q}{q-1}} = o(G) \), regardless of the order
of \(\mu_n^2 + \tilde{\mu}_n^2\). Under strong within-cluster dependence, a
sufficient condition is \(n_G^{5} G = o(n^2)\), again regardless of
\(\mu_n^2 + \tilde{\mu}_n^2\).

\textbf{Scenario 3.} If the clusters have a bounded size such that \(n_G\) is bounded and \(G \asymp n\), then \Cref{ass:var_L3O}.\ref{item:l3o_rates} holds provided \(\mu_n^2 + \tilde{\mu}_n^2 + \kappa_n \to \infty\) and $\zeta_{H,n} + \zeta_{\tilde H,n} = o(\mu_n^2 + \tilde{\mu}_n^2 + \kappa_n)$.
\end{rem}

The next theorem establishes the unbiasedness and consistency for the L3CO variance estimator.

\begin{thm}\label{thm:var_l3co}
  Suppose \Cref{ass:dgp,,ass:reg,ass:var_L3O}.\ref{item:l3o_well_defined} hold. Then
  $\mathbb E \hat \omega^2_{n,\rm L3CO} = \omega_n^2$. If, in addition,
  \Cref{ass:var_L3O}.\ref{item:l3o_rates} holds, then
  \begin{equation*}
    \hat \omega_{n,\rm L3CO}^2/\omega_n^2 \convP 1.
  \end{equation*}
\end{thm}

\begin{rem}
We establish the consistency of \(\hat{\omega}_{n,\mathrm{L3CO}}^{2}\) by showing \(\mathbb{V}(\hat{\omega}_{n,\mathrm{L3CO}}^{2}) = o(\omega_n^{4})\). For this purpose, we need to evaluate sums of the form
\[
\sum_{l \in[G]} P_{g,l,-ghk}\, P_{l,g',-g'h'k'},
\]
where \(P_{g,l,-ghk}\) is the \((g,l)\) block of the leave-\((g,h,k)\)-clusters-out projection matrix for \(W\). Evaluating the summation is not trivial because the sets of clusters omitted in the two projection matrices (i.e., $(g,h,k)$ and $(g',h',k')$) need not coincide. It is further complicated in the clustered setting because \(P_{g,l,-ghk}\) is a block matrix rather than a scalar, and matrix multiplication is non-commutative. There are two key technical innovations in our proof:
(1) a new representation of $P_{g,l,-ghk}$ that enables the exact calculation of the sum, and
(2) a detailed decomposition of the summand into matrices that depend on $(g,g',h,h',k,k')$ only through
\[
s \in \{g, gh, gk, ghk\} \times \{g', g'h', g'k', g'h'k'\},
\]
which implies that they are invariant with respect to the remaining indices $(g,g',h,h',k,k')\setminus s$.\footnote{For example, if a matrix is indexed by $s = (g,h,g',k')$, then it is invariant to the remaining indices $(g,g',h,h',k,k')/s = (k,h')$.}
This invariance allows the summation over the remaining indices to pass through these matrices.
\end{rem}

\begin{rem}
The main difficulty in the variance estimation stems from the leave-out construction, which ensures that the estimators of the linear coefficients $\gamma$ and $\pi$ are independent of the other observations appearing in the same summand. Although such independence could also be achieved by sample splitting, this approach does not directly apply in our setting.

To illustrate, consider the estimation of $\omega^2_{n,1}$. Suppose that the clusters are split into two subsets, $I_1$ and $I_2$, and $(\gamma,\pi)$ is estimated using clusters in $I_1$. A natural sample-splitting estimator is
\begin{equation*}
\hat{\omega}^2_{n,\mathrm{SS},1}
= \sum_{g,h,k \in I_2^3}
\left( X_h' B_{h,g} Y_g \right)
\left( X_k' B_{k,g} \bigl( Y_g - W_g \hat{\gamma}_{I_1} \bigr) \right).
\end{equation*}
However, the summation over $(g,h,k)\in I_2^3$ excludes cross-split terms (e.g., $g,h\in I_2$ but $k\in I_1$). This issue remains even if the roles of $I_1$ and $I_2$ are reversed, showing that simple sample splitting is inadequate here. \Textcite{KSS2020} propose a more complicated sample-splitting variance estimator for independent data. Extending their construction to clustered settings is nontrivial, both theoretically and computationally.
\end{rem}

\begin{rem}
The L3CO variance estimator $\hat \omega_{n,\rm L3CO}^2$ is not necessarily nonnegative in finite samples. However, it is possible to construct a variant of the L3CO estimator that is guaranteed to be nonnegative and consistent. Specifically, let
\begin{align*}
 & \tilde \omega^2_{n,\rm L3CO,1} = \hat \omega^2_{n,\rm L3CO,1} + 2\hat \omega^2_{n,\rm L3CO,2} + \hat \omega^2_{n,\rm L3CO,3} - 2(\hat \omega^2_{n,\rm L3CO,4} + \hat \omega^2_{n,\rm L3CO,5}),\\
& \tilde \omega^2_{n,\rm L3CO,2} = \hat \omega^2_{n,\rm L3CO,4} + \hat \omega^2_{n,\rm L3CO,5}.
\end{align*}
We note that \(\mathbb{E}\tilde{\omega}^2_{n,\mathrm{L3CO},1}\) corresponds to the variance of the linear component of \(X'BY\), while \(\mathbb{E}\tilde{\omega}^2_{n,\mathrm{L3CO},2}\) corresponds to the variance of its quadratic component. Since
\[
\hat{\omega}^2_{n,\mathrm{L3CO}}
=
\tilde{\omega}^2_{n,\mathrm{L3CO},1}
+
\tilde{\omega}^2_{n,\mathrm{L3CO},2},
\]
and both \(\tilde{\omega}^2_{n,\mathrm{L3CO},1}\) and \(\tilde{\omega}^2_{n,\mathrm{L3CO},2}\) are asymptotically nonnegative, this motivates the following variant of the L3CO estimator:
\[
\tilde{\omega}^2_{n,\mathrm{L3CO}}
=
\bigl|\tilde{\omega}^2_{n,\mathrm{L3CO},1}\bigr|
+
\bigl|\tilde{\omega}^2_{n,\mathrm{L3CO},2}\bigr|.
\]
\end{rem}

\begin{cor}\label{cor:L3CO}
  Suppose \Cref{ass:dgp,,ass:reg,ass:var_L3O} hold. Then
  \begin{equation*}
    \tilde \omega^2_{n,\rm L3CO}/\omega_n^2 \convP 1.
  \end{equation*}
\end{cor}

\subsection{Leave-two-clusters-out Variance Estimator}\label{sec:cons-vari-estim}

The \ac{L3CO} variance estimator may be computationally expensive in large
datasets, since computing it involves looping over three cluster indices. This
motivates an alternative \ac{L2CO} variance estimator, given by
\begin{equation*}
  \hat \omega^2_{n, \rm L2CO} = \sum_{g \in [G]} \left( \sum_{h \in [G]} \left( X_h ' B_{h,g} \tilde Y_{g,-h} + Y_h ' B_{g,h}' \tilde X_{g,-h} \right) \right)^2,
\end{equation*}
where
\begin{align*}
    \tilde Y_{g,-h}& = Y_g - W_g \hat \gamma_{-gh},&
    \tilde X_{g,-h}& = X_g - W_g \hat \pi_{-gh} 
\end{align*}
denote leave-two-clusters-out OLS residuals (if $g=h$, we set
$\tilde Y_{g,-h} = 0$ and $\tilde X_{g,-h} = 0$; the leave-two-clusters out
estimators $\gamma_{-gh}$ and $\pi_{-gh}$ are defined analogously to
\cref{eq:gamma-ghk,eq:pi-ghk}). Since its computation only involves looping over
two cluster indices, it is computationally cheaper than the \ac{L3CO} variance
estimator. The second advantage of this estimator is that it is guaranteed to be
non-negative. Third, as \Cref{thm:var_l2co} shows, valid inference based on
$\hat{\omega}_{n, \rm L2CO}$ can be conducted under less restrictive
assumptions. The inference is exact when $d$ is not very large and cluster sizes
are fixed, but it can be conservative otherwise.

\begin{ass}\label{ass:var_l2co}
For $S_{k,g}$ defined in \Cref{ass:var_L3O}.\ref{item:l3o_well_defined}, the
following holds:
    \begin{enumerate}
    \item\label{item:l2o_well_defined} There exists a finite constant $C > 0$
      such that
      $\max_{k,g \in [G]^2, k \neq g} \norm{S_{k,g}^{-1}}_{op}
      \leq C$.
\item\label{item:l2o_rates} We have
 \begin{multline*}
         u_n^{\frac{2q-3}{q-1}} n_G^{\frac{q}{q-1}}(\zeta_{H,n} +
         \zeta_{\tilde{H}, n}) \kappa_n +u_n^2 n_G \phi_n \lambda_n^3
         (\zeta_{H,n} + \zeta_{\tilde{h}, n}) \kappa_n +
      u_n^{\frac{2q-3}{q-1}} n_G^{\frac{2q-1}{q-1}} \phi_n^2 \lambda_n^2 \kappa_n\\
+u_n^3 n_G(\phi_n^2 \lambda_n + \phi_n \lambda_n^2) \kappa_n+
u_n^4 n_G \lambda_n^2 \kappa_n
= o\left( (\mu_n^2 + \tilde \mu_n^2 + \kappa_n)^2\right),
    \end{multline*}
    where $\phi_n$ is defined in \Cref{ass:var_L3O}.
\end{enumerate}
\end{ass}

\begin{thm}\label{thm:var_l2co}
  Suppose \Cref{ass:dgp,,ass:reg,ass:var_l2co}.\ref{item:l2o_well_defined} hold. Then
  $\mathbb E \hat \omega^2_{n, \rm L2CO} \geq \omega_n^2$. If, in addition,
  \Cref{ass:var_l2co}.\ref{item:l2o_rates} holds, then
  \begin{equation}\label{eq:l2o_conservative}
    \limsup_{n \to \infty } \mathbb P \left( \left  \vert
        \frac{\hat \theta_{\rm LO} - \theta}{\hat \omega_{n,\rm L2CO}} \right \vert \geq {z}_{1-\alpha/2} \right) \leq \alpha.
  \end{equation}
  Suppose further that there exists a finite constant $C > 0$ such that
  $\max_{k,g \in [G]^2, k \neq g} \norm{S_{k,g}^{-1}}_{op} \leq
  C$, and that $n_G = O(1)$, $\lambda_n = o(1)$ and
  $\kappa_n = o (\mu_n^2 + \tilde \mu_n^2)$. Then~\cref{eq:l2o_conservative}
  holds with equality.
\end{thm}

With large cluster sizes and/or high-dimensional covariates, the upward bias
$\mathbb E \hat \omega_{n,\rm L2CO}^2 - \omega_n^2$ could be large. In such
settings, we recommend using $\hat \omega^2_{n,\rm L2CO}$ only when computing
$\hat \omega_{n,\rm L3CO}^2$ is infeasible.

\begin{rem}\label{rem:rate3}
    One can easily check that the rate conditions in \Cref{ass:var_l2co}.\ref{item:l2o_rates} are weaker than those in \Cref{ass:var_L3O}.\ref{item:l3o_rates}. In particular, consider again the three scenarios in \Cref{rem:rate2}.

    \textbf{Scenario 1.} Suppose that \(\kappa_n \asymp n\) and
    $\zeta_{H,n} + \zeta_{\tilde H,n} \lesssim n_G$. If the within-cluster
    dependence is weak, then \Cref{ass:var_l2co}.\ref{item:l2o_rates} holds provided
    \(n_G^{\frac{2q-1}{q-1}} = o(n)\), which is weaker than \(n_G^{3} = o(n)\)
    if $q > 2$. Under strong within-cluster dependence, for \Cref{ass:var_l2co}.\ref{item:l2o_rates}
    to hold, we still require \(n_G^{5} = o(n)\).

    \textbf{Scenario 2.} Following the setting of Scenario~2 in
    Remark~\ref{rem:rate1}, Assumption~\ref{ass:var_l2co}.\ref{item:l2o_rates} holds under weak
    within-cluster dependence if \( n_G^{\frac{q}{q-1}} = o(G)\) and
    \(n_G^{\frac{2q-1}{q-1}} G = o(n^2)\), which is weaker than the requirement
    \(n_G^3 = o(n)\) imposed by \Cref{ass:var_L3O}.\ref{item:l3o_rates} in this scenario,
    provided that $q>2$. Under strong within-cluster dependence, both
    \Cref{ass:var_L3O}.\ref{item:l3o_rates} and \Cref{ass:var_l2co}.\ref{item:l2o_rates} require
    $n_G^{5} G = o(n^2)$.

    \textbf{Scenario 3.} If clusters have bounded size so that \(n_G\) is
    bounded and \(G \asymp n\), then both
    Assumption~\ref{ass:var_L3O}.\ref{item:l3o_rates} and
    Assumption~\ref{ass:var_l2co}.\ref{item:l2o_rates} hold provided
    \(\mu_n^2 + \tilde{\mu}_n^2 + \kappa_n \to \infty\) and
    $\zeta_{H,n} + \zeta_{\tilde H,n} = o(\mu_n^2 + \tilde{\mu}_n^2 +
    \kappa_n)$.
\end{rem}

\begin{rem}
  Note that \(\hat{\omega}^2_{n,\mathrm{L2CO}}\) remains computable even when
  leaving three clusters out is infeasible, in the sense that some of the
  matrices \(\tilde S_{k,gh}\) in \Cref{ass:var_L3O}.\ref{item:l3o_well_defined}
  are not invertible. One may also construct conservative variance estimators
  based on ``HC3''-type residuals \parencite[see, e.g.,][]{CJN18}, which can be
  computed even when leaving two clusters out is infeasible, that is, when some
  of the matrices \(S_{k,g}\) in \Cref{ass:var_L3O}.\ref{item:l3o_well_defined}
  are not invertible. Indeed, the computation of HC3-type residuals only
  requires \Cref{ass:dgp}.\ref{item:ass_lo}. Establishing formal theoretical
  guarantees for such estimators is left for future research.
\end{rem}

\subsection{Practical Guidance}

In practice, it is not necessary to run OLS regression when computing
$\hat \gamma_{-ghk}$ and $\hat \pi_{-ghk}$ for each combination of
$g,h,k \in [G]^3$, which can be computationally demanding for large $d$.
Instead,  the leave-out algebra from \Cref{sec:setup} implies that the
leave-out residuals may be computed directly as
\begin{align*}
 \tilde Y_{g,-hk}& = \left[M_{ghk,ghk}^{-1} (MY)_{ghk} \right]_g,& \text{and}&&
 \tilde X_{g,-hk}& = \left[M_{ghk,ghk}^{-1} (MX)_{ghk} \right]_g,
\end{align*}
where $M_{ghk,ghk}$ is the block diagonal matrix of $M$ corresponding to
clusters $(g,h,k)$. Similarly,
\begin{align*}
 \tilde Y_{g,-h}& = \left[M_{gh,gh}^{-1} (MY)_{gh} \right]_g & \text{and}&& \tilde X_{g,-h}& = \left[M_{gh,gh}^{-1} (MX)_{gh} \right]_g.
\end{align*}

Another potential issue is the numerical instability when solving the system
$M_{ghk,ghk}\tilde Y_{ghk,-hk} = (MY)_{ghk}$ that defines $\tilde Y_{g,-hk}$
above (and similarly for the other leave-out residuals). We propose the
following approach: whenever the minimum eigenvalue of $M_{ghk,ghk}$ is below a
threshold $t_n$, such that $t_n \downarrow 0$ as $n$ goes to infinity, replace
it with a ridge regularizer, and instead solve
$(M_{ghk,ghk}+t_{n}I_{n_{g}+n_{h}+n_{k}})\tilde Y_{ghk,-hk} = (MY)_{ghk}$. We find that
$t_n = 1/\log(n^2)$ performs well in our simulation. Our theory can be adapted
to deal with such shrinking regularizer, but for ease of exposition, we do not
pursue this extension.

\section{Simulation}

In this section, we compare our test statistic with L3CO variance estimator with several existing methods in the literature: the two-stage least squares estimator (TSLS) with cluster-robust variance estimator; the cluster jackknife instrumental variable estimator (CJIVE) proposed by \textcite{FLM23} with cluster-robust variance estimator; and the three test statistics (CSW-JIV, CSW-LIM and CSW-FUL) proposed in \textcite{CNT23}. For all methods, we impose the null when computing their corresponding asymptotic variances to guard against weak identification. All the results below are based on $1,000$ simulations.

\subsection{Design \texorpdfstring{\uppercase\expandafter{\romannumeral1}}{I}: Homogeneous Treatment Effect}

We first consider the following panel IV regression adapted from \textcite{CNT23}, which assumes a homogeneous treatment effect:
\begin{align*}
        \mathcal Y_{i,g} &= X_{i,g}\beta + \mathcal W_{i,g}'\gamma + \alpha_g + U_{i,g}, \\
        X_{i,g} &= Z_{i,g}'\pi + \mathcal W_{i,g}'\delta + \xi_g + V_{i,g},
    \end{align*}
where $\alpha_g$ and $\xi_g$ are cluster-level fixed effects. The fixed effects are generated by $\alpha_g = u_{1g} + g/G, \xi_g = u_{2g} + g/G, g=1, \dots, G$ where $u_{1g}$ and $u_{2g}$ are independent standard normal random variables. These fixed effects are partialled out from the model by demeaning at the cluster level. The instruments in $Z$ are normally distributed with mean $0$ and we allow for cluster-level dependence: for each cluster the covariance matrix is given by
\begin{equation*}
  \Omega_{1g}= \begin{bmatrix}
    1 & \theta_1 & \cdots & \theta_1 \\
    \theta_1 & 1 & \cdots & \theta_1 \\
    \vdots & \vdots & \ddots & \vdots \\
    \theta_1 & \theta_1 & \cdots & 1
  \end{bmatrix}_{n_g \times n_g} \;\; g=1, \dots, G
\end{equation*}
and across clusters these instruments are independent from each other; we set $\theta_1 = 0.5$ in our simulation. 
The controls in $W$ are generated by $\mathcal W_{i,g} = (z_{i,g}, z_{i,g}^2-1, z_{i,g}^3-3z_{i,g}, z_{i,g}^4-6z_{i,g}^2+3, z_{i,g}(D_{i,g}^{(1)}-0.5), \dots, z_{i,g}(D_{i,g}^{(d_w-4)}-0.5))$ where $z_{i,g}$ are independent standard normal random variables and $D_{i,g}^{(k)}, k=1,\dots,d_w-4$ are independent Bernoulli random variables with success probability $0.5$. The error terms are generated by
\begin{align*}
    \tilde{U}_{i,g} &= \rho \epsilon_{i,g} + \sqrt{1-\rho^2} \sigma_{i,g} v_{i,g}, \\
    \tilde{V}_{i,g} &= \rho \eta_{i,g} + \sqrt{1-\rho^2} \sigma_{i,g} v_{i,g},
\end{align*}
with $\sigma_{i,g} = \sqrt{(0.2+z_{i,g}^2)/2.4}$ and $\rho=0.5$, where $\epsilon_{i,g}$, $\eta_{i,g}$ and $v_{i,g}$ are independent standard normal random variables. We further collect $\tilde{U}_{i,g}$ and $\tilde{V}_{i,g}$ for each cluster and left multiply them by
\begin{equation*}
  \Omega_{2g}= \begin{bmatrix}
    1 & 0 & \cdots & 0 \\
    \theta_2 & 1 & \cdots & 0 \\
    \vdots & \vdots & \ddots & \vdots \\
    \theta_2^{n_g-1} & \theta_2^{n_g-2} & \cdots & 1
  \end{bmatrix}_{n_g \times n_g} \;\; g=1, \dots, G,
\end{equation*}
to generate $U_{i,g}$ and $V_{i,g}$; we set $\theta_2 = 0.7$ in our simulation. For the parameters, we set $\beta = 0.3$, $\gamma = \delta = (1/\sqrt{d_w}) \times \iota_{d_w}$ where $\iota_{d_w}$ is a $d_w \times 1$ vector of ones, and $\pi = t_n \times \iota_{d_z}$ where $\iota_{d_z}$ is a $d_z \times 1$ vector of ones and $t_n = \sqrt{30/(\sqrt{d_z} \times n)}$.

Finally, the number of observations is $n = 600$ and the number of clusters is $G = 150$, where all the clusters have equal cluster size. We consider three different designs for the number of instruments and the number of controls: $d_z = d_w = 50$,  $d_z = d_w = 100$ and $d_z = d_w = 150$. We report the empirical size as well as the power curve for each design below.

\begin{table}[tp]
\centering
\begin{tabular}{ccccccc}
 & \multicolumn{3}{c}{5\% significance level} & \multicolumn{3}{c}{10\% significance level} \\
  \cmidrule(rl){2-4}  \cmidrule(rl){5-7}
 & $d_z=d_w=50$ & $100$ & $150$ & $d_z=d_w=50$ & $100$ & $150$ \\
\midrule
\text{TSLS}     & 53.7\% & 92.3\% & 99.7\% & 66.5\% & 96.8\% & 99.9\% \\
\text{CJIVE}    & 7.9\%  & 22.3\% & 53.4\% & 15.5\% & 31.4\% & 63.3\% \\
\text{CSW-JIV}  & 7.8\%  & 13.3\% & 22.7\% & 14.2\% & 20.8\% & 31.0\% \\
\text{CSW-LIM}  & 7.1\%  & 10.5\% & 16.6\% & 14.6\% & 16.8\% & 25.8\% \\
\text{CSW-FUL}  & 7.6\%  & 10.5\% & 16.4\% & 14.7\% & 16.9\% & 25.6\% \\
  \text{L3CO}     & 6.3\%  & 5.0\%  & 5.0\%  & 10.4\% & 9.1\%  & 9.4\%  \\
  \bottomrule
\end{tabular}
\caption{Empirical size at 5\% and 10\% significance levels for the \textcite{CNT23} simulation.}\label{tab:size_CSW_combined}
\end{table}

\begin{figure}[tp]
\centering
\begin{tabular}{ccc}
\multicolumn{3}{c}{\textbf{5\% significance level}} \\[2mm]
\includegraphics[width=0.31\textwidth]{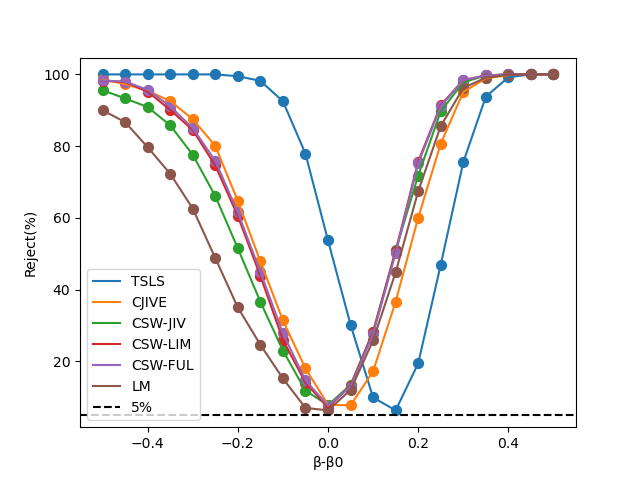} &
\includegraphics[width=0.31\textwidth]{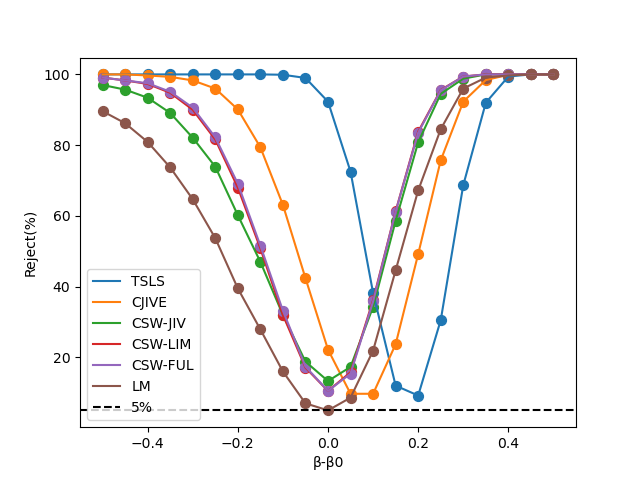} &
\includegraphics[width=0.31\textwidth]{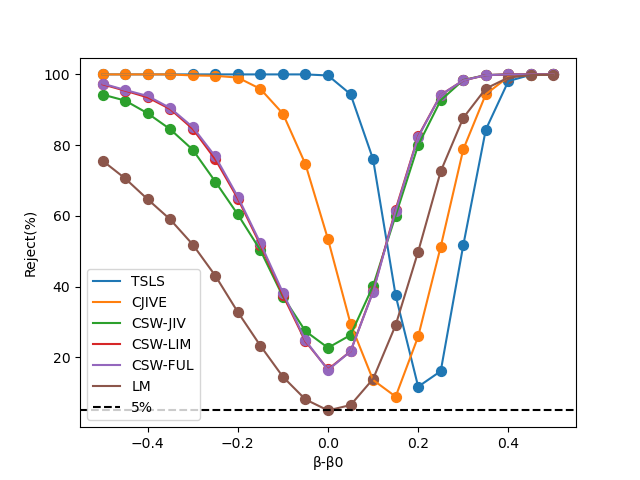} \\
$d_z=50,\ d_w=50$ & $d_z=100,\ d_w=100$ & $d_z=150,\ d_w=150$ \\[4mm]

\multicolumn{3}{c}{\textbf{10\% significance level}} \\[2mm]
\includegraphics[width=0.31\textwidth]{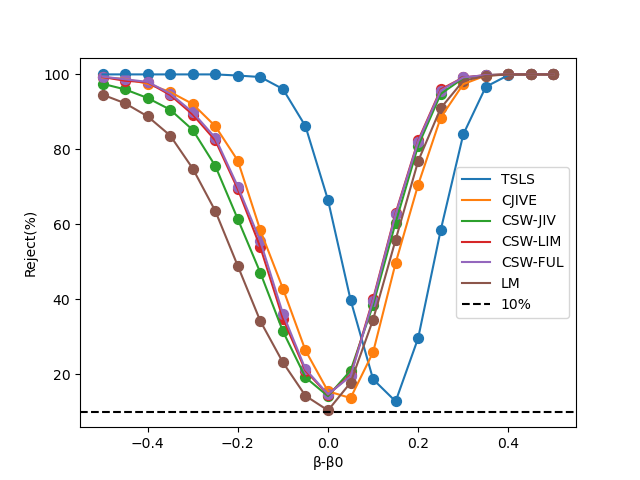} &
\includegraphics[width=0.31\textwidth]{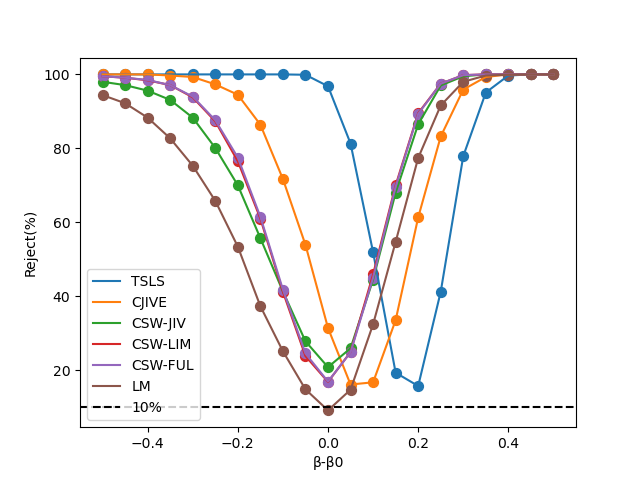} &
\includegraphics[width=0.31\textwidth]{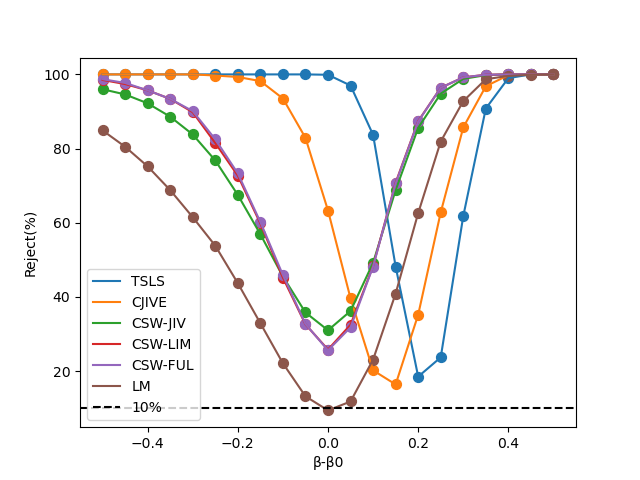} \\
$d_z=50,\ d_w=50$ & $d_z=100,\ d_w=100$ & $d_z=150,\ d_w=150$
\end{tabular}
\caption{Power curves for the \textcite{CNT23} simulation. Top row: 5\% significance level. Bottom row: 10\% significance level.}\label{fig:power_CSW_combined}
\end{figure}

\subsection{Design \texorpdfstring{\uppercase\expandafter{\romannumeral2}}{II}: Heterogeneous Treatment Effect, Saturated}

We consider a simulation setup similar to that in \textcite{Yap24} where we have many fixed effects, the effect of $X$ on $Y$ is heterogeneous, and the linear regressions for $Y$ and $X$ are correctly specified. Let $t = 1, \dots, K$ index the state and suppose that individuals within the same cluster belong to the same state. We have in total $K = 48$ states, where each state contains $4$ clusters and each cluster contains $4$ individuals. There is also another binary exogenous variable $B \in \{0,1\}$, which is equally distributed within each cluster. The structural equations are given by
\begin{align*}
    \mathcal Y_{i,g} &= X_{i,g} (\beta + \xi_{i,g}) + \mathcal W_{i,g}'\gamma + U_{i,g} \\
    X_{i,g} &= \mathbf{1} \left\{ Z_{i,g}' \pi + W_{i,g}'\delta \geq V_{i,g} \right\}
\end{align*}
where the control variables $\mathcal W_{i,g}$ contain indicators for states  with $d_w = 48$ and the instrumental variables $Z_{i,g}$ contain indicators for $k = t \times B$ with $d_z = 48$ (the baseline instrument for $k=0$ is dropped to avoid multicollinearity). For the parameters, we set $\beta = 0.5$, $\pi(k) = 0$ if $k = 0$, $\pi(k) = \sqrt{15\sqrt{K}/n}$ for half of the states and $\pi(k) = -\sqrt{15\sqrt{K}/n}$ for the other half. In addition, we set $\gamma(t) = \delta(t) = 1/\sqrt{K}$ for half of the states and $\gamma(t) = \delta(t) = -1/\sqrt{K}$ for the other half.

The error terms are generated as follows: $V_{i,g}$ is generated as
\begin{align*}
    V_{i,g} = 2 \Phi (\rho u_g + \sqrt{1-\rho^2} v_{i,g}) - 1
\end{align*}
where $\rho = 0.4$, $u_g$ and $v_{i,g}$ are independent standard normal random variables, and $\Phi$ is the CDF of the standard normal distribution so that marginally $V_{i,g}$ is uniformly distributed on $[-1,1]$. Given $V_{i,g}$, we generate
\begin{align*}
    U_{i,g} | V_{i,g} \sim
    \begin{cases}
        \mathcal N (\mu, \sigma_{\eps}^2) & V_{i,g} \geq 0 \\
        \mathcal N (-\mu, \sigma_{\eps}^2) & V_{i,g} < 0
    \end{cases}
\end{align*}
where we set $\mu = 0.4$ and $\sigma_{\eps} = 0.2$. Lastly, given $V_{i,g}$, $\xi_{i,g}$ is generated by
\begin{align*}
    \xi_{i,g} | V_{i,g} \sim
    \begin{cases}
         \mathcal B(\sigma_{\xi}^{(k)}, -\sigma_{\xi}^{(k)}, p) & V_{i,g} \geq 0 \\
         \mathcal B(\sigma_{\xi}^{(k)}, -\sigma_{\xi}^{(k)}, 1-p) & V_{i,g} < 0
    \end{cases}
\end{align*}
where we use $\mathcal B(a,b,p)$ to denote the binary random variable that equals to $a$ with probability $p$ and $b$ with probability $1-p$, and we set $p = 2/3$. Here $\sigma_{\xi}^{(k)} = 0$ if $k = 0$, and for those states with $\pi(k) = \sqrt{15\sqrt{K}/n}$, we set $\sigma_{\xi}^{(k)} = \sqrt{30\sqrt{K}/n}$ for half of them and $\sigma_{\xi}^{(k)} = -\sqrt{30\sqrt{K}/n}$ for the other half; the same procedure is applied to those states with $\pi(k) = -\sqrt{15\sqrt{K}/n}$.

As shown in \textcite{EK2018}, the TSLS estimand admits a valid (conditional) causal interpretation in this context, since the monotonicity condition is satisfied. In addition, following the same steps as in \textcite{Yap24}, it can be shown that this estimand equals to $\beta$ under the design above. We report the empirical size as well as the power curve below.

\begin{table}[tp]
\centering
\begin{tabular}{ccc}
 & 5\% & 10\% \\ \midrule
\text{TSLS}    & 59.6\% & 71.7\% \\
\text{CJIVE}   & 10.0\%  & 15.8\% \\
\text{CSW-JIV} & 13.8\% & 21.5\% \\
\text{CSW-LIM} & 33.0\% & 45.8\% \\
\text{CSW-FUL} & 30.4\% & 43.0\% \\
\text{L3CO}    & 4.8\%  & 9.4\% \\ \bottomrule
\end{tabular}
\caption{Empirical size at 5\% and 10\% significance levels for the \textcite{Yap24} simulation.}\label{tab:size_LY_combined}
\end{table}

\begin{figure}[tp]
\centering
\begin{tabular}{cc}
\includegraphics[width=0.47\textwidth]{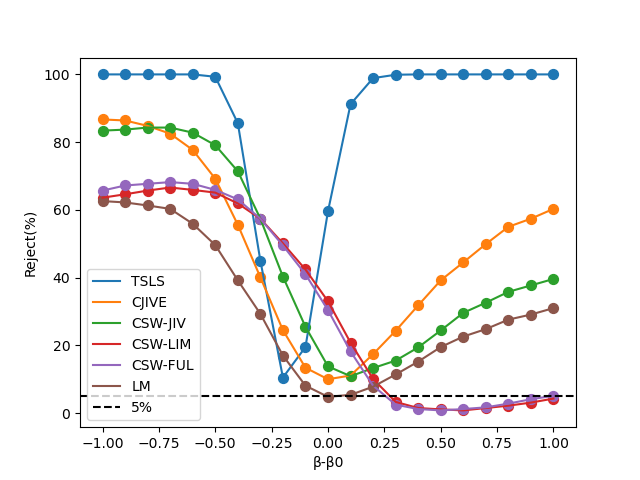} &
\includegraphics[width=0.47\textwidth]{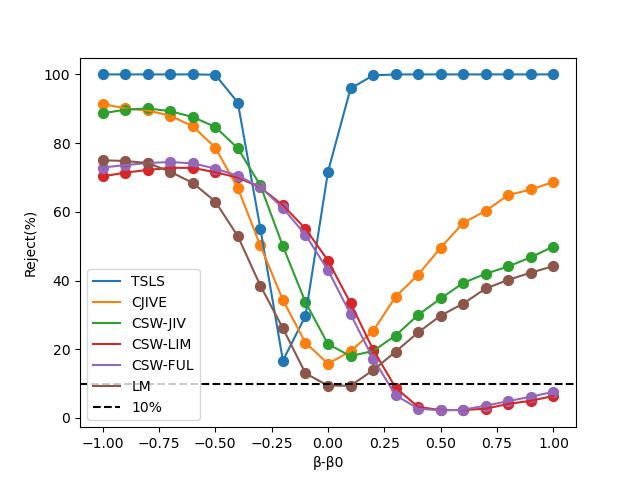} \\
5\% significance level & 10\% significance level
\end{tabular}
\caption{Power curves for the \textcite{Yap24} simulation.}\label{fig:power_LY_combined}
\end{figure}

\subsection{Design \texorpdfstring{\uppercase\expandafter{\romannumeral3}}{III}: Heterogeneous Treatment Effect, Approximated}

We consider a simulation setup similar to the judge design in which the treatment effect of $X$ on $Y$ is heterogeneous and the linear regressions for $X$ and $Y$ are approximately correctly specified due to the binning method. In this setup, individuals within the same cluster are assigned to the same judge. We have in total $G = 200$ clusters, where each cluster contains 4 individuals (so $n = 800$) and the clusters are distributed evenly across $J$ judges.

If cluster $g$ is assigned to judge $j$ for $j = 1, \dots, J$, then the first- and second-stage equations are given by
\begin{align*}
    X_{i,g} &= \mathbf{1} \left\{U_{i,g} \leq \alpha_0 + \left(\alpha_1 + \alpha_2 S_{1,i,g} \right) \times \frac{j}{J} \right\}, \\
   \mathcal Y_{i,g} &= \mathbf{1} \left\{V_{i,g} \leq \beta_0 + \beta_1 X_{i,g} + \beta_2 S_{1,i,g} + \beta_3 S_{2,g} \right\},
\end{align*}
where $S_{1,i,g} \sim \text{Uniform}[0,1]$ is an individual-level exogenous variable, and $S_{2,g} \in \{0, 1, \ldots, K-1\}$ is a cluster-level discrete exogenous variable with $K = 5$ levels. The error terms are specified as
\begin{align*}
    U_{i,g} &= \rho_1 \xi_{g}^{(1)} + \sqrt{1 - \rho_1^2} \, \varepsilon_{i,g}^{(1)}, \\
    V_{i,g} &= \rho_2 U_{i,g} + \sqrt{1 - \rho_2^2} \left( \rho_3 \xi_{g}^{(2)} + \sqrt{1 - \rho_3^2} \, \varepsilon_{i,g}^{(2)} \right),
\end{align*}
where $\{\xi_{g}^{(1)}\}_{g \in [G]}$, $\{\xi_{g}^{(2)}\}_{g \in [G]}$, $\{\{\varepsilon_{i,g}^{(1)}\}_{i \in [n_g]}\}_{g \in [G]}$, and $\{\{\varepsilon_{i,g}^{(2)}\}_{i \in [n_g]}\}_{g \in [G]}$ are independent sequences of standard normal random variables.  We set $\alpha_0 = -0.8$, $\alpha_1 = 1.0$, $\alpha_2 = 0.6$, $\beta_0 = -0.8$, $\beta_1 = 1$, $\beta_2 = 0.6$, $\beta_3 = 1.0$, $\rho_1 = 0.8$, $\rho_2 = 0.8$ and $\rho_3 = 0.3$ in our simulation.

We generate piecewise constant basis functions by first defining cells based on $(j, S_{2,g})$, then within each cell, we distribute observations evenly into $n_{\text{bins}} = 6$ bins based on $S_{1,i,g}$. Specifically, within each $(j, S_{2,g})$ cell, we sort observations by $S_{1,i,g}$ and assign them to bins such that each bin has approximately the same number of observations. The control variables $\mathcal{W}_{i,g}$ are then constructed as all interactions between these $S_{1,i,g}$ bin indicators and the dummy variables of $S_{2,g}$, i.e., $\{\mathbf{1}\{S_{2,g} = k\}\}_{k=0}^{K-1}$, resulting in up to $6 \times 5 = 30$ control variables.

Next, we consider $J=4$ judges so that all clusters are distributed equally among the four judges. We create the instrumental variables (IVs) by interacting the dummy for the first three judges, $\{\mathbf{1}\{j(i,g)=l\}\}_{l=1,2,3}$, with the control variables, i.e., $Z_{i,g} = \{\mathbf{1}\{j(i,g)=l\} \mathcal{W}_{i,g}\}_{l=1,2,3}$. Therefore, the entire set of regressors $W_{i,g}$ is defined as $W_{i,g} = (Z_{i,g}', \mathcal{W}_{i,g}')'$.

Given that $(\alpha_1 + \alpha_2 S_{1,i,g}) > 0$, the IV monotonicity condition is satisfied. The linear reduced form regressions are approximately correctly specified due to the use of piecewise constant basis functions and full interactions. Consequently, the TSLS estimand admits a valid (conditional) causal interpretation, as established by \textcite{EK2018}. The value of this causal estimand in our simulation is computed via numerical integration. We report the empirical size as well as the power curve below.

\begin{table}[tp]
\centering
\begin{tabular}{ccc}
 & 5\% & 10\% \\ \midrule
\text{TSLS}    & 47.5\% & 60.7\% \\
\text{CJIVE}   & 11.3\% & 19.7\% \\
\text{CSW-JIV} & 21.8\% & 31.2\% \\
\text{CSW-LIM} & 16.5\% & 23.3\% \\
\text{CSW-FUL} & 16.7\% & 23.4\% \\
\text{L3CO}    & 4.1\%  & 8.0\%  \\
\bottomrule
\end{tabular}
\caption{Empirical size at 5\% and 10\% significance levels for the judge design simulation.}\label{tab:size_judge_combined}
\end{table}

\begin{figure}[tp]
\centering
\begin{tabular}{cc}
\includegraphics[width=0.48\textwidth]{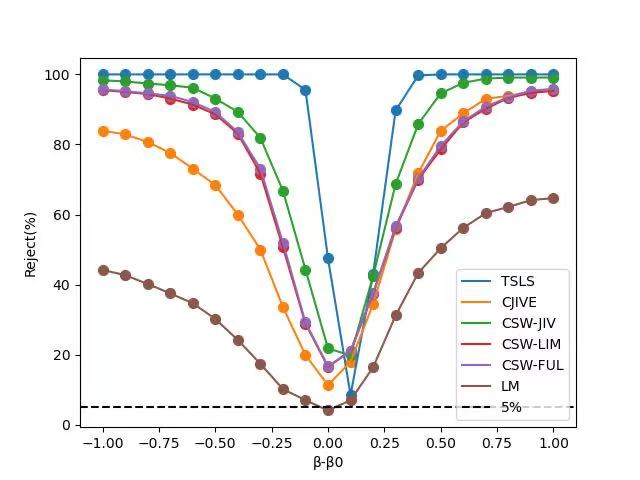} &
\includegraphics[width=0.48\textwidth]{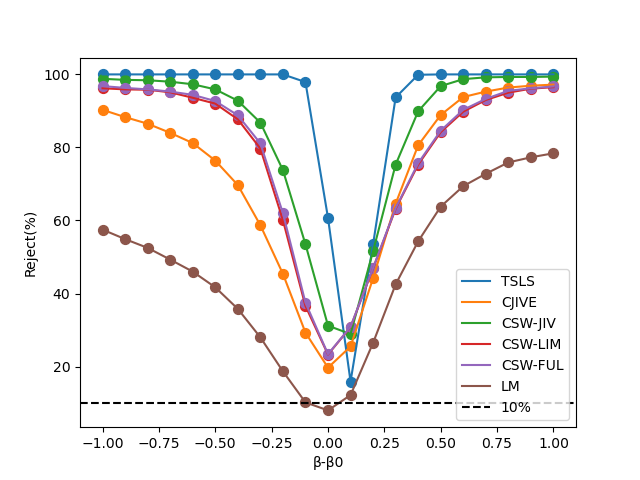} \\
5\% significance level & 10\% significance level
\end{tabular}
\caption{Power curves for the judge design simulation.}\label{fig:power_judge_combined}
\end{figure}

\subsection{Remarks}
Based on the simulation results in Designs
\uppercase\expandafter{\romannumeral1}--\uppercase\expandafter{\romannumeral3},
only our proposed inference method controls asymptotic size in settings with
many instruments and controls, heterogeneous treatment effects, and clustered
data. Existing methods fail for different reasons. TSLS does not correct for the
many-instrument bias, and both TSLS and CJIVE fail to address the many-control
bias. The three CSW-type methods correct for both many-instrument and
many-control biases, but their corrections rely on independence of observations,
even within clusters. Moreover, the variance estimators used by TSLS, CJIVE, and
the CSW-type methods are not asymptotically valid in our designs for several
reasons: (i) they are developed for homogeneous treatment effect models, whereas
Designs \uppercase\expandafter{\romannumeral2} and
\uppercase\expandafter{\romannumeral3} feature heterogeneous treatment effects;
(ii) TSLS and CJIVE do not account for many controls in variance estimation;
(iii) CSW-type methods allow a diverging number of controls only at rates slower
than \(\sqrt n\); and (iv) CSW-type variance estimators ignore within-cluster
dependence. As a result, existing methods exhibit substantial size distortions,
while our method performs well. Since our procedure is the only one that controls size, power comparisons are not particularly meaningful. Nevertheless, the power curves indicate that our method retains substantial power. Developing alternative tests that maintain size control while improving on the power of our test represents an interesting direction for future research.


\printbibliography

\newpage
\appendix
\noindent{\huge\textbf{Appendix}}

\unhidefromtoc
\tableofcontents

\newpage

Through the appendix, we use the following notation. For any matrix $Q$, we
denote by $Q'_{g, h} \in \Re^{n_h \times n_g} $ the transpose of the submatrix
$Q_{g, h}$. This implies $Q'_{g, h}=[Q']_{h, g}$. In addition, for two
deterministic and nonnegative sequences $\{a_n\}_{n \geq 1}$ and
$\{b_n\}_{n \geq 1}$, we write $a_n \lesssim b_n$ if there exists a constant $c$
independent of $n$ such that for all $n \geq 1$, $a_n \leq c b_n$. We further impose the normalization that $h_n = \norm{(W'W)^{-1/2} A_0 (W'W)^{-1/2}}_{op} = 1$.

\section{Proof of Lemma~\ref{lemma:rao}}
  The Lagrangian for the optimization problem may be written
  $\tr(C'C)/2-\tr(\Lambda'\Bdiag(C-A))-
  \tr(\tilde{\Lambda}'CW)-\tr(\ddot{\Lambda}'W'C)$, where
  $\Lambda=\Bdiag(\Lambda_{g,g})$ is the matrix of Lagrange multipliers
  associated with the first constraint and $\tilde{\Lambda}$ and
  $\ddot{\Lambda}$ are the Lagrange multiplier matrices associated with the
  other two constraints. Taking a derivative yields the first-order condition
  \begin{equation*}
    C=\Bdiag(\Lambda_{g,g})+\tilde{\Lambda}W'+W\ddot{\Lambda}.
  \end{equation*}
  Using (ii) then yields
  $\tilde{\Lambda}=-\Bdiag(\Lambda_{g,g})
  W(W'W)^{-1}-W\ddot{\Lambda}W(W'W)^{-1}$, and plugging this back into the
  first-order condition and using (iii) yields
  \begin{equation*}
    C=M\Bdiag(\Lambda_{g,g})M.
  \end{equation*}
  Using (i) then gives $A_{g,g}=\sum_{h}M_{gh}\Lambda_{hh}M_{hg}$, or, in
  vector form,
  $\mkvec(A_{g,g})=\sum_{h}(M_{hg}\otimes M_{gh})\mkvec(\Lambda_{hh})$.
  Stacking these equations yields $\bvec(A)=(M\ast M)\bvec(\Lambda)$,
  which yields the result. If condition (ii) is not imposed, then $\tilde{\Lambda}=0$
  in the first-order condition, which yields
  $\ddot{\Lambda}=-(W'W)^{-1}W'\Bdiag(\Lambda_{g,g})$. Plugging this back into
  the first-order condition yields $C=M\Bdiag(\Lambda_{g,g})$. Hence, by (i),
  $M_{g,g}^{-1}A_{g,g}=\Lambda_{g,g}$ for all $g$, which yields the result.

\section{Proof of Theorem~\ref{thm:clt}}\label{sec:clt_pf}
Decompose,
\begin{align*}
 X' B Y  - \theta
& = (W \pi + V)' B (W \gamma + U) - \theta \\
& = H' U + \tilde H' V  + V' B U \\
& = \sum_{g \in [G]} \left[ H_g' U_g + \tilde H_g' V_g \right]  + \sum_{g \in [G]}\left[ \left(\sum_{h \leq g}  B_{h,g}'V_h \right)' U_g + \left(\sum_{h \leq g} B_{g,h}U_h \right)' V_{g}\right] \\
& = \sum_{g \in [G]} \left[ T_{g}' U_g + \tilde T_{g}' V_g \right],
\end{align*}
where
\begin{equation*}
  T_{g} = H_g + \sum_{h \leq g}  B_{h,g}' V_h , \quad \tilde T_{g} = \tilde H_g + \sum_{h \leq g}  B_{g,h} U_h .
\end{equation*}
Observe that the sequence $\{T_{g}' U_g + \tilde T_{g}' V_g\}_{g\in[G]}$ is a
martingale difference sequence with respect to the filtration
$\{\mathcal F_g: g\geq 0\}$, where $F_{0}=\emptyset$, and for $g\in[G]$,
$\mathcal{F}_{g}$ is the $\sigma$-field generated by $(U_h,V_h)_{h \leq g}$.
Furthermore, $\mathcal{\omega}_{n}^{2}=E[\sigma^{2}_{n}]$, where
\begin{equation*}
\sigma_n^2 = \sum_{g \in [G]} \begin{pmatrix}
T_{g}'  & \tilde T_{g}'
\end{pmatrix} \Omega_{g} \begin{pmatrix}
T_{g}  \\
\tilde T_{g}
\end{pmatrix}
\end{equation*}
sums the conditional variances. Hence, the result follows by a central limit
theorem for martingale difference sequences \textcite[Corollary 3.1]{HH1980},
provided we can show that the sum of the conditional variances converges,
\begin{equation}\label{eq:omega_n}
\sigma_n^2/ \omega_n^2 \convP 1
\end{equation}
and that the Lindeberg condition holds, that is for any $\eps>0$,
\begin{equation}\label{eq:Lindeberg}
\sum_{g \in [G]} \mathbb E \left[\frac{(T_{g}' U_g + \tilde T_{g}' V_g)^2 }{\omega_n^2}1\left\{\frac{\left\vert T_{g}' U_g + \tilde T_{g}' V_g \right\vert }{\omega_n }  \geq \eps \right\}  \mid \mathcal{F}_{g-1}   \right] \convP 0.
\end{equation}
We show these claims in the two steps below.

\paragraph{Step 1: Establishing \cref{eq:omega_n}.} We first establish the lower bound
\begin{equation}\label{eq:omega_lower}
  \omega_n^2 \geq  c(\mu_n^2 + \tilde \mu_n^2 +   \kappa_n).
\end{equation}
To show \cref{eq:omega_lower}, note that
\begin{align*}
\omega_n^2 & =  \sum_{g \in [G]} tr \left[\mathbb E\left( \Omega_{g} \begin{pmatrix}
T_{g}  \\
\tilde T_{g}
\end{pmatrix} \begin{pmatrix}
T_{g}'  & \tilde T_{g}'
\end{pmatrix}  \right)\right] \\
& \geq c \sum_{g \in [G]} \mathbb E \left( \norm{T_{g}}_2^2 + \norm{\tilde T_{g}}_2^2\right) \\
& = c\sum_{g \in [G]} \left( \norm{H_{g}}_2^2 + \norm{\tilde H_{g}}_2^2 \right) +
c\sum_{g \in [G]}
\mathbb E \norm*{\sum_{h \leq g}  B_{h,g}' V_h}_2^2
+c\sum_{g \in [G]}\mathbb{E} \norm*{\sum_{h \leq g}  B_{g,h} U_h}_2^2.
\end{align*}
The second term can be further lower-bounded as
\begin{multline*}
\sum_{g \in [G]}\mathbb E \left\Vert \sum_{h \leq g}  B_{h,g}' V_h \right\Vert_2^2  = \mathbb E \sum_{g \in [G]} \sum_{h \leq g} V_h' B_{h,g} B_{h,g}' V_h
 = \sum_{g \in [G]} \sum_{h \leq g} \tr\left[\left(B_{h,g} B_{h,g}'\right)\Omega_{V,h}\right]\\
 \geq  c \sum_{g \in [G]} \sum_{h \leq g} ||B_{h,g}||_F^2,
\end{multline*}
and by analogous arguments, the third term can be lower-bounded by $c \sum_{g \in [G]} \sum_{h \geq g} ||B_{h,g}||_F^2$.
Since the first term equals $c(||H||_2^2 + ||\tilde H||_2^2)$, it follows that
\begin{equation*}
  \omega_{n}^{2}\geq
  c(||H||_2^2 + ||\tilde H||_2^2)
  +\sum_{g \in [G]}\mathbb E \left\Vert \sum_{h \leq g}  B_{h,g}' V_h
  \right\Vert_2^2  +  \sum_{g \in [G]} \mathbb E \left\Vert \sum_{h \leq g}
  B_{g,h} U_h \right\Vert_2^2 \geq c(||H||_2^2 + ||\tilde H||_2^2+  ||B||_F^2),
\end{equation*}
which establishes \cref{eq:omega_lower}. Therefore, to show \cref{eq:omega_n},
it suffices to show
\begin{align*}
\frac{\mathbb V(\sigma_n^2)}{(\mu_n^2 + \tilde \mu_n^2 +   \kappa_n)^2 } \to 0.
\end{align*}

Note that
\begin{align}\label{eq:Var}
  \mathbb V(\sigma_n^2)
  & = \mathbb V\left( \sum_{g \in [G]} \begin{pmatrix}
    T_{g}'  & \tilde T_{g}'
\end{pmatrix} \Omega_{g} \begin{pmatrix}
T_{g}  \\
\tilde T_{g}
\end{pmatrix} \right) \notag \\
        & \lesssim \mathbb V\left( \sum_{g \in [G]} \begin{pmatrix}
(H_{g}'  & \tilde H_{g}'
\end{pmatrix} \Omega_{g} \begin{pmatrix}
\sum_{h \leq g}  B_{h,g}' V_h   \\
\sum_{h \leq g}  B_{g,h} U_h
\end{pmatrix} \right) \notag \\
& + \mathbb V\left( \sum_{g \in [G]} \begin{pmatrix}
\sum_{h \leq g} V_h' B_{h,g}  & \sum_{h \leq g}  U_h' B_{g,h}'
\end{pmatrix} \Omega_{g} \begin{pmatrix}
\sum_{h \leq g}  B_{h,g}' V_h   \\
\sum_{h \leq g}  B_{g,h} U_h
\end{pmatrix} \right) \notag \\
& \lesssim \sum_{h \in [G]} \mathbb V\left( \sum_{g \geq h}
(H_{g}'\Omega_{U,g} + \tilde H_{g}'\Omega_{V,U,g}) B_{h,g}' V_h\right)  + \sum_{h \in [G]} \mathbb V\left( \sum_{g \geq h} (H_g' \Omega_{U,V,g} + \tilde H_g' \Omega_{V,g} ) B_{g,h} U_h \right) \notag \\
& + \mathbb V\left( \sum_{g \in [G]}
\left(\sum_{h \leq g} V_h' B_{h,g}\right) \Omega_{U,g}  (\sum_{h \leq g}  B_{h,g}' V_h ) \right)+ \mathbb V\left( \sum_{g \in [G]} (
\sum_{h \leq g} V_h' B_{h,g}) \Omega_{U,V,g} (
\sum_{h \leq g}  B_{g,h} U_h )\right)  \notag \\
& + \mathbb V\left( \sum_{g \in [G]} ( \sum_{h \leq g}  U_h' B_{g,h}')  \Omega_{V,g} (\sum_{h \leq g}  B_{g,h} U_h) \right).
\end{align}

We bound each of the five terms on the RHS of the above display. Let
$\Omega_U = \Bdiag(\Omega_{U,g})$, $\Omega_V$,
$\Omega_{U,V}$, and $\Omega_{V,U}$ are similarly defined, $\nabla(B)$ be the
upper-triangular part of $B$ defined as $\nabla(B) = T \circ B$,
$ [T]_{ij} = 1\{i \geq j\}$, $ \Delta(B) = (\iota_n\iota_n' - T)\circ B$ be the
lower-triangular part of $B$, and for any square matrix
$U \in \Re^{n \times n}$,
\begin{equation*}
 \boxbackslash(U) = \begin{cases}
     U_{g,h}  \quad \text{if $g \neq h$} \\
     0_{n_g \times n_g} \quad \text{otherwise}
 \end{cases},
\end{equation*}
where $\circ$ represents the Hadamard product and $\iota_n \in \Re^n$ is a vector of ones.

Then, we have  $\nabla(B') = (\Delta(B))'$ and
\begin{align}\label{eq:Var12}
        \sum_{h \in [G]} \mathbb V\left( \sum_{g \geq h}
(H_{g}'\Omega_{U,g} + \tilde H_{g}'\Omega_{V,U,g}) B_{h,g}' V_h\right) & = Var \left( (H' \Omega_U + \tilde H' \Omega_{V,U}) \Delta(B)' V\right) \notag \\
        & = (H' \Omega_U + \tilde H' \Omega_{V,U}) \Delta(B)' \Omega_{V} \Delta(B) (H' \Omega_U + \tilde H' \Omega_{V,U})' \notag \\
        & \leq u_n^3 ||\Delta(B)||_{op}^2 (\mu_n^2 + \tilde \mu_n^2)  \notag \\
        & \lesssim u_n^3 \eta_n (\mu_n^2  + \tilde \mu_n^2) = o\left((\mu_n^2 + \tilde \mu_n^2+   \kappa_n)^2\right),
\end{align}
where the last inequality is by \Cref{lem:nablaB} and the last equality is by \Cref{ass:reg}.

For the same reason, we have
\begin{align}\label{eq:Var12'}
\sum_{h \in [G]} \mathbb V\left( \sum_{g \geq h} (H_g' \Omega_{U,V,g} + \tilde H_g' \Omega_{V,g} ) B_{g,h} U_h \right)  = o\left((\mu_n^2 + \tilde \mu_n^2+   \kappa_n)^2\right).
\end{align}

Next, we consider $ \mathbb V\left( \sum_{g \in [G]}
\left(\sum_{h \leq g} V_h' B_{h,g}\right) \Omega_{U,g}  (\sum_{h \leq g}  B_{h,g}' V_h ) \right)$. We note that
\begin{align*}
\sum_{g \in [G]}
\left(\sum_{h \leq g} V_h' B_{h,g}\right) \Omega_{U,g}  (\sum_{h \leq g}  B_{h,g}' V_h ) & = V' \Delta(B) \Omega_{U} \Delta(B)' V \\
& = \left( \sum_{g \in [G]} V_{g}' [\Delta(B) \Omega_{U} \Delta(B)']_{g,g}V_{g} \right) + V' \boxbackslash (\Delta(B) \Omega_{U} \Delta(B)') V.
\end{align*}

Therefore, we have
\begin{align}\label{eq:Var22}
        & \mathbb V\left(\sum_{g \in [G]}
\left(\sum_{h \leq g} V_h' B_{h,g}\right) \Omega_{U,g}  (\sum_{h \leq g}  B_{h,g}' V_h )\right) \notag \\
         & \lesssim \mathbb V\left( \sum_{g \in [G]} V_{g}' [\Delta(B) \Omega_{U} \Delta(B)']_{g,g}V_{g}\right) + \mathbb V\left(V' \boxbackslash (\Delta(B) \Omega_{U} \Delta(B)') V\right) \notag \\
         & \lesssim \sum_{g \in [G]} \mathbb V(V_{g}' [\Delta(B) \Omega_{U} \Delta(B)']_{g,g}V_{g}) \notag \\
         & + \sum_{g \in [G]} \sum_{h \neq g} \mathbb E V'_{g}[\Delta(B) \Omega_{U} \Delta(B)']_{g,h} \Omega_{V,h} [\Delta(B) \Omega_{U} \Delta(B)']_{g,h}' V_{g} \notag \\
         & \lesssim \sum_{g \in [G]} \mathbb E \left(V_{g}' [\Delta(B) \Omega_{U} \Delta(B)']_{g,g}V_{g} \right)^2 \notag \\
         &+ \sum_{g \in [G]} \sum_{h \neq g} \tr\left[ [\Delta(B) \Omega_{U} \Delta(B)']_{g,h} \Omega_{V,h} [\Delta(B) \Omega_{U} \Delta(B)']_{g,h}' \Omega_{V,g} \right].
\end{align}
For the first term on the RHS of \cref{eq:Var22}, we have
\begin{align*}
\left\Vert [\Delta(B) \Omega_{U} \Delta(B)']_{g,g} \right\Vert_{F} & = \left\Vert \sum_{h \leq g} B_{g,h} \Omega_{U,h}  B_{g,h}' \right\Vert_{F} \\
& \lesssim u_n \left\Vert\sum_{h \leq g} B_{g,h}  B_{g,h}'\right\Vert_{F} \\
& \lesssim u_n \left\Vert\sum_{h \in [G]} B_{g,h}  B_{g,h}'\right\Vert_{F} \\
& \lesssim u_n \left\Vert (BB')_{g,g}\right\Vert_{F},
\end{align*}
by \Cref{lem:nablaB}, and thus,
\begin{align}\label{eq:Var22_1}
&       \sum_{g \in [G]} \mathbb E \left(V_{g}' [\Delta(B) \Omega_{U} \Delta(B)']_{g,g}V_{g} \right)^2 \notag \\
    & = \sum_{g \in [G]} \mathbb E tr\begin{pmatrix}
             [\Delta(B) \Omega_{U} \Delta(B)']_{g,g}V_{g} V_{g}' [\Delta(B) \Omega_{U} \Delta(B)']_{g,g}V_{g}V_{g}'
    \end{pmatrix} \notag \\
& \lesssim \sum_{g \in [G]} \mathbb E tr\begin{pmatrix}
             [\Delta(B) \Omega_{U} \Delta(B)']_{g,g}V_{g} ||V_{g}||_2^2 V_{g}' [\Delta(B) \Omega_{U} \Delta(B)']_{g,g}
    \end{pmatrix} \notag \\
& \lesssim u_n^{\frac{q-2}{q-1}} n_G^{\frac{q}{q-1}} \sum_{g \in [G]} ||[\Delta(B) \Omega_{U} \Delta(B)']_{g,g}||_F^2 \notag \\
& \lesssim u_n^{\frac{3q-4}{q-1}} n_G^{\frac{q}{q-1}} \sum_{g \in [G]} tr ((BB')_{g,g}^2 ) \notag \\
& \lesssim u_n^{\frac{3q-4}{q-1}} n_G^{\frac{q}{q-1}} \lambda_n \kappa_n = o\left((\mu_n^2 + \tilde \mu_n^2 +  \kappa_n)^2\right),
\end{align}
where the second inequality is by \Cref{lem:4mom} and the last equality is by \Cref{ass:reg}.


For the second term on the RHS of \cref{eq:Var22}, we have
\begin{align}\label{eq:Var22_2}
& \sum_{g \in [G]} \sum_{h \neq g} \tr\left[ [\Delta(B) \Omega_{U} \Delta(B)']_{g,h} \Omega_{V,h} [\Delta(B) \Omega_{U} \Delta(B)']_{g,h}' \Omega_{V,g} \right]  \notag \\
& \leq \tr\left[\Delta(B) \Omega_{U} \Delta(B)' \Omega_V \Delta(B) \Omega_{U} \Delta(B)' \Omega_{V}  \right] \notag \\
& \leq u_n^4 \tr\left[\Delta(B)'\Delta(B)\Delta(B)'\Delta(B)\right] \notag \\
& \leq u_n^4 ||\Delta(B)||_{op}^2 \tr(\Delta(B)'\Delta(B)) \notag \\
& \lesssim u_n^4  \eta_n ||\Delta(B)||_F^2 \notag \\
& \lesssim u_n^4 \eta_n  ||B||_F^2 \notag \\
& \lesssim u_n^4 \eta_n  \kappa_n = o\left((\mu_n^2 +  \tilde \mu_n^2 + \kappa_n)^2\right),
\end{align}
where the third and last inequalities are due to \Cref{lem:nablaB} and the last equality is by \Cref{ass:reg}.
Combining \cref{eq:Var22} with \cref{eq:Var22_1} and \cref{eq:Var22_2}, we have
\begin{align}\label{eq:Var22_3}
        \mathbb V\left(\sum_{g \in [G]}
\left(\sum_{h \leq g} V_h' B_{h,g}\right) \Omega_{U,g}  (\sum_{h \leq g}  B_{h,g}' V_h )\right)=o\left((\mu_n^2 +  \tilde \mu_n^2 + \kappa_n)^2\right).
\end{align}
Following the same argument, we have
\begin{align}\label{eq:Var22_4}
& \mathbb V\left( \sum_{g \in [G]} (
\sum_{h \leq g} V_h' B_{h,g}) \Omega_{U,V,g} (
\sum_{h \leq g}  B_{g,h} U_h )\right)  =o\left((\mu_n^2 +  \tilde \mu_n^2 + \kappa_n)^2\right) \quad \text{and} \notag \\
& \mathbb V\left( \sum_{g \in [G]} ( \sum_{h \leq g}  U_h' B_{g,h}')  \Omega_{V,g} (\sum_{h \leq g}  B_{g,h} U_h) \right) =o\left((\mu_n^2 +  \tilde \mu_n^2 + \kappa_n)^2\right).
\end{align}

Combining \cref{eq:Var}, \cref{eq:Var12}, \cref{eq:Var12'}, \cref{eq:Var22_3}, and \cref{eq:Var22_4}, we have
\begin{align*}
        \mathbb V(\sigma_n^2) = o\left((\mu_n^2 +  \tilde \mu_n^2 + \kappa_n)^2\right),
\end{align*}
which is the desired result.

\paragraph{Step 2: Establishing \cref{eq:Lindeberg}.} We note that
\begin{align}\label{eq:Lindeberg1}
        & \sum_{g \in [G]} \mathbb E \left[\frac{(T_{g}' U_g + \tilde T_{g}' V_g)^2  }{\omega_n^2 }1\left\{\frac{\left\vert T_{g}' U_g + \tilde T_{g}' V_g \right\vert }{\omega_n } \geq \eps  \right\} \right] \notag \\
        & \lesssim \sum_{g \in [G]} \mathbb E \biggl\{ \left[\frac{(T_{g}' U_g)^2}{\omega_n^2} + \frac{(\tilde T_{g}' V_g)^2 }{\omega_n^2 }\right] \left[ 1\left\{\frac{\left\vert T_{g}' U_g \right\vert}{\omega_n} \geq \eps/2 \right\} + 1\left\{ \frac{\left\vert \tilde T_{g}' V_g \right\vert }{\omega_n } \geq \eps/2  \right\}\right] \biggr\} \notag \\
        & \lesssim \sum_{g \in [G]} \mathbb E  \left[\frac{(T_{g}' U_g)^2}{\omega_n^2}1\left\{\frac{\left\vert T_{g}' U_g \right\vert}{\omega_n} \geq \eps/2 \right\} + \frac{( \tilde T_{g}' V_g)^2 }{\omega_n^2 }1\left\{ \frac{\left\vert\tilde T_{g}' V_g \right\vert }{\omega_n } \geq \eps/2  \right\} \right] \notag \\
         & \lesssim  \sum_{g \in [G]} \mathbb E \frac{(T_{g}' U_{g})^{4}}{\omega_n^{4} \eps^2} +  \sum_{g \in [G]} \mathbb E \frac{(\tilde T_{g}' V_{g})^{4}}{\omega_n^{4} \eps^2} ,
\end{align}
where the second inequality is by the Chebyshev's sum inequality that if $a_1 \geq a_2 \geq \cdot \geq a_J \geq 0$ and $b_1 \geq b_2 \geq \cdot \geq b_J \geq 0$, then
\begin{align}\label{eq:chebyshev}
 \left(\frac{1}{J}\sum_{j=1}^J a_j\right)\left(\frac{1}{J}\sum_{j=1}^J
  b_j\right) \leq \frac{1}{J} \sum_{j = 1}^J a_{j} b_j.
\end{align}

For the first term on the RHS of \cref{eq:Lindeberg1}, we have
\begin{align*}
        \sum_{g \in [G]} \mathbb E (T_{g}' U_{g})^{4} & \lesssim \sum_{g \in [G]} \mathbb E (H_{g}' U_{g})^{4} + \sum_{g \in [G]} \mathbb E \left[ (\sum_{h \leq g}  B_{h,g}' V_h)' U_{g}\right]^{4}
\end{align*}
and
\begin{align*}
\sum_{g \in [G]} \mathbb E (H_{g}' U_{g})^{4} & \lesssim \sum_{g \in [G]} ||H_{g}||_2^4  u_n^{\frac{q-2}{q-1}} n_G^{\frac{q}{q-1}} \\
    & \lesssim \zeta_{H,n} \mu_n^2  u_n^{\frac{q-2}{q-1}} n_G^{\frac{q}{q-1}} \\
    & =  o\left( (\mu_n^2 + \tilde \mu_n^2 + \kappa_n)^2\right) = o(\omega_n^4),
\end{align*}
where the second inequality is by \Cref{lem:4mom}, the second last equality is by \Cref{ass:reg}.4, and the last equality is by \cref{eq:omega_lower}. In addition, we have
\begin{align*}
\sum_{g \in [G]} \mathbb E \left[ (\sum_{h \leq g}  B_{h,g}' V_h)'
  U_{g}\right]^{4} & = \sum_{g \in [G]} \sum_{h_1,h_2,h_3,h_4 \geq g}\mathbb E
                     \prod_{l=1}^4(V_{h_l}' B_{h_l, g} U_{g})   \\
        & \lesssim \sum_{g \in [G]} \left[\sum_{h \geq g}\mathbb E (V_h' B_{h,g} U_{g})^4 + \sum_{h_1,h_2 \geq g, h_1 \neq h_2}\mathbb E (V_{h_1}' B_{h_1g} U_{g})^2(V_{h_2}' B_{h_2g} U_{g})^2 \right] \notag \\
    & \lesssim u_n^{\frac{2(q-2)}{q-1}} n_G^{\frac{2q}{q-1}} \lambda_n^2 \kappa_n + u_n^{\frac{3q-4}{q-1}} n_G^{\frac{q}{q-1}}  \lambda_n \kappa_n  \notag \\
    & =  o\left( (\mu_n^2 + \tilde \mu_n^2 + \kappa_n)^2\right) = o(\omega_n^4),
\end{align*}
where the second inequality is by \Cref{lem:E4} and the last equality is by \cref{eq:omega_lower} and \Cref{ass:reg}.4.

Therefore, we have
\begin{align}\label{eq:H1^4}
\sum_{g \in [G]} \mathbb E (T_{g}' U_{g})^{4} = o(\omega_n^4).
\end{align}
Following the same argument, we have
\begin{align}\label{eq:tildeH1^4}
  \sum_{g \in [G]} \mathbb E (\tilde T_{g}' V_{g})^{4} = o(\omega_n^4).
\end{align}
Combining \cref{eq:Lindeberg1} with \cref{eq:H1^4} and \cref{eq:tildeH1^4}, we have
\begin{align*}
\sum_{g \in [G]} \mathbb E \left[\frac{(T_{g}' U_g + \tilde T_{g}' V_g)^2  }{\omega_n^2 }1\left\{\frac{\left\vert T_{g}' U_g + \tilde T_{g}' V_g \right\vert }{\omega_n } \geq \eps  \right\} \right] = o(1).
\end{align*}
This concludes the proof.

\section{Proof of Theorem~\ref{thm:var_l3co}}
\textbf{Step 1.} We compute the expectation of  $ \hat \omega^2_{n,\rm L3CO}$. For any random variable $R_n$, let $\mathbb E_{g} R_n$ be $\mathbb E( R_n| U_g,V_g)$ and define $\mathbb E_{g,h} R_n$ and $\mathbb E_{ghk} R_n$ in the same manner. Then, we have
\begin{align*}
    \mathbb E_{ghk} (X_k' B_{k,g} (Y_g - W_g \hat \gamma_{-ghk})) = X_k' B_{k,g} U_g
\end{align*}
and
\begin{align*}
   & \mathbb E \sum_{g,h,k \in [G]^3} \left(X_h ' B_{h,g} Y_g \right) \left(X_k ' B_{k,g} (Y_g - W_g \hat \gamma_{-ghk}) \right) \\
   & =  \mathbb E \sum_{g,h,k \in [G]^3} \left(X_h ' B_{h,g} Y_g \right) \left(X_k ' B_{k,g} U_g \right) \\
    & = \sum_{g,h,k \in [G]^3, h \neq k} \mathbb E \left(X_h ' B_{h,g} U_g \right) \left(X_k ' B_{k,g} U_g \right)  + \sum_{g,h \in [G]^2} \mathbb E  \left(X_h ' B_{h,g} U_g \right)^2 \\
    & = \sum_{g,h,k \in [G]^3} \Pi_h ' B_{h,g} \Omega_{U,g} B_{k,g}' \Pi_k  + \sum_{g,h \in [G]^2} tr  \left( B_{h,g} \Omega_{U,g} B_{h,g}' \Omega_{V,h} \right)  \\
    & =  H' \Omega_U H + \sum_{g,h \in [G]^2} tr  \left( B_{h,g} \Omega_{U,g} B_{h,g}' \Omega_{V,h} \right).
\end{align*}

Similarly, we have
\begin{align*}
 & \mathbb E   \sum_{g, h, k \in [G]^3} \left(X_h ' B_{h,g} Y_g \right) \left(Y_k ' B_{g,k}' (X_g - W_g \hat \pi_{-ghk}) \right) \\
 &  = \sum_{g, h, k \in [G]^3}  \mathbb E   \left(X_h ' B_{h,g} U_g \right) \left(Y_k ' B_{g,k}' V_g \right) \\
 &  = \sum_{g, h, k \in [G]^3}  \Pi_h ' B_{h,g} \Omega_{U,V,g} B_{g,k} \Gamma_k  + \sum_{g, h \in [G]^2}  \mathbb E   \left(V_h ' B_{h,g} U_g \right) \left(U_h ' B_{g,h}' V_g \right) \\
 & =  H' \Omega_{U,V} \tilde H + \sum_{g, h \in [G]^2}  tr   \left(B_{h,g} \Omega_{U,V,g} B_{g,h} \Omega_{U,V,h}  \right)
\end{align*}
and
\begin{align*}
   & \mathbb E \sum_{g, h, k \in [G]^3} \left(Y_h ' B_{g,h}' X_g \right) \left(Y_k ' B_{g,k}' (X_g - W_g \hat \pi_{-ghk}) \right) \\
   & =   \mathbb E \sum_{g, h, k \in [G]^3} \left(Y_h ' B_{g,h}' V_g \right) \left(Y_k ' B_{g,k}' V_g \right) \\
   & = \tilde H'\Omega_{V} \tilde H +\mathbb E \sum_{g, h \in [G]^2} \left(U_h ' B_{g,h}' V_g \right) \left(U_h ' B_{g,h}' V_g \right) \\
   & = \tilde H'\Omega_{V} \tilde H +\sum_{g, h \in [G]^2} tr \left(B_{g,h}' \Omega_{V,g} B_{g,h} \Omega_{U,h} \right) .
\end{align*}

Next, we have
\begin{align*}
   & \mathbb E  \sum_{g, h, k \in [G]^3} \left((Y_h - W_h \hat \gamma_{-ghk})' B_{g,h}' X_g \right) \left(Y_h' B_{g,h}' \tilde M_{gk,-gh} X_k \right)  \\
   & =  \mathbb E  \sum_{g, h, k \in [G]^3} \left(U_h' B_{g,h}' X_g \right) \left(Y_h' B_{g,h}' \tilde M_{gk,-gh} X_k \right) \\
   & =  \mathbb E  \sum_{g, h, k \in [G]^3} \left(U_h' B_{g,h}' X_g \right) \left(Y_h' B_{g,h}' \tilde M_{gk,-gh} V_k \right) \\
   & =  \mathbb E  \sum_{g, h \in [G]^2} \left(U_h' B_{g,h}' X_g \right) \left(Y_h' B_{g,h}'  V_g \right) \\
   & =  \mathbb E  \sum_{g, h \in [G]^2} \left(U_h' B_{g,h}' V_g \right) \left(U_h' B_{g,h}'  V_g \right) \\
   & = \sum_{g, h \in [G]^2} tr \left(B_{g,h}' \Omega_{V,g} B_{g,h}  \Omega_{U,h} \right),
\end{align*}
where the first equality is by
\begin{align*}
    \mathbb E_{ghk}(Y_h - W_h \hat \gamma_{-ghk}) = U_h,
\end{align*}
the second equality is by
\begin{align*}
\sum_{k \in [G]} \tilde M_{gk,-gh} W_k = 0_{n_g \times n_g},
\end{align*}
the third equality is by
\begin{align*}
\tilde M_{gk,-gh}   = \begin{cases}
     I_{n_g}, \quad k = g; \\
     0_{n_h \times n_g}, \quad k = h;
\end{cases} \quad \text{and} \quad  \mathbb E  \sum_{g, h, k \in [G]^3, k \neq g,h} \left(U_h' B_{g,h}' X_g \right) \left(Y_h' B_{g,h}' \tilde M_{gk,-gh} V_k \right) = 0.
\end{align*}
Similarly, we have
\begin{align*}
& \mathbb E \sum_{g, h, k \in [G]^3} \left((X_h - W_h \hat \pi_{-ghk})' B_{h,g} Y_g \right) \left(Y_h' B_{g,h}' \tilde M_{gk,-gh} X_k \right)   \\
& = \mathbb E \sum_{g, h \in [G]^2} \left(V_h' B_{h,g} U_g \right) \left(U_h' B_{g,h}'  V_g \right)   \\
& =\sum_{g, h \in [G]^2} tr \left( B_{h,g} \Omega_{U,V,g} B_{g,h} \Omega_{U,V,h}  \right).
\end{align*}

Combining the above five expectations, we have
\begin{align*}
    \mathbb E \hat \omega^2_{n,\rm L3CO} &= H' \Omega_U H + \sum_{g,h \in [G]^2} tr  \left( B_{h,g} \Omega_{U,g} B_{h,g}' \Omega_{V,h} \right) \\
    & + 2 H' \Omega_{U,V} \tilde H +2\sum_{g, h \in [G]^2}  tr   \left(B_{h,g} \Omega_{U,V,g} B_{g,h} \Omega_{U,V,h}  \right) \\
    & + \tilde H'\Omega_{V} \tilde H +\sum_{g, h \in [G]^2} tr \left(B_{g,h}' \Omega_{V,g} B_{g,h} \Omega_{U,h} \right)  \\
    & - \sum_{g, h \in [G]^2} tr \left(B_{g,h}' \Omega_{V,g} B_{g,h}  \Omega_{U,h} \right) \\
    & - \sum_{g, h \in [G]^2} tr \left( B_{h,g} \Omega_{U,V,g} B_{g,h} \Omega_{U,V,h}  \right) \\
    & = \sum_{g,h \in [G]^2} tr  \left(  \Omega_{V,h} B_{h,g} \Omega_{U,g} B_{h,g}' \right) + \sum_{g, h \in [G]^2}  tr   \left(\Omega_{U,V,g} B_{g,h} \Omega_{U,V,h} B_{h,g}  \right) \\
    & +  H' \Omega_U H + \tilde H'\Omega_{V} \tilde H + 2 H' \Omega_{U,V} \tilde H \\
    & = \sum_{g,h \in [G]^2} tr  \left(  \Omega_{V,g} B_{g,h} \Omega_{U,h} B_{g,h}' \right) + \sum_{g, h \in [G]^2}  tr   \left(\Omega_{U,V,g} B_{g,h} \Omega_{U,V,h} B_{h,g}  \right) \\
    & +  H' \Omega_U H + \tilde H'\Omega_{V} \tilde H + 2 H' \Omega_{U,V} \tilde H \\
    & = \omega_n^2.
\end{align*}

\textbf{Step 2.} We aim to show $\mathbb V (\hat \omega^2_{n,\rm L3CO}) = o(\omega_n^4)$, which leads to the desired result that
\begin{align*}
   \hat \omega^2_{n,\rm L3CO}/ \omega_n^2 \convP 1.
\end{align*}

Denote
\begin{align*}
 P_{g,l,-ghk} & = \begin{cases}
    & W_g (W'W - W_g W_g' - W_h W_h' - W_k W_k') W_l', \quad g \neq h, g \neq k, h \neq k, \\
    & W_g (W'W - W_g W_g' - W_h W_h') W_l' =  P_{g,l,-gh}, \quad g \neq h, h = k, \\
    & 0_{n_g,n_l}, \quad \text{otherwise}
 \end{cases} \\
 P_{l,g,-ghk} & = P_{g,l,-ghk}'.
 \end{align*}
Here we simply set $ P_{g,l,-ghk} = 0_{n_g,n_l}$ if either $g = h$ or $g = k$ because in all our analysis below, these two cases will not occur due to the fact that $B_{g,g} = 0_{n_g,n_g}$.

Given this definition, we have
\begin{align}\label{eq:var_l3co_1}
  &  \mathbb V\left(\sum_{g, h, k \in [G]^3} \left(X_h ' B_{h,g} Y_g \right) \left(X_k ' B_{k,g} (Y_g - W_g \hat \gamma_{-ghk}) \right)  \right) \notag \\
  & \lesssim \underbrace{ \mathbb V \left(\sum_{g, h, k \in [G]^3} \left(X_h ' B_{h,g} Y_g \right) \left(X_k ' B_{k,g} U_g \right)  \right) }_{T_1}+  \underbrace{ \mathbb V \left(\sum_{g, h, k\in [G]^3, l \in [-ghk]} \left(X_h ' B_{h,g} Y_g \right) \left(X_k ' B_{k,g} P_{g,l,-ghk} U_l \right)  \right) }_{T_2}.
\end{align}
The first term $T_1$ on the RHS of the above display can be bounded by standard calculation. \Cref{lem:V1} shows
\begin{equation*}
T_1 =   o(\omega_n^4).
\end{equation*}

Next, we focus on the second term $T_2$ on the RHS of \cref{eq:var_l3co_1}. We have
\begin{align*}
T_2 & \lesssim  \underbrace{ \mathbb V \left(\sum_{g, h, k \in [G]^3, l \in [-ghk]} \left(\Pi_h' B_{h,g} \Gamma_g \right) \left(\Pi_k ' B_{k,g} P_{g,l,-ghk} U_l \right)  \right) }_{R_{1}} \\
& + \underbrace{  \mathbb V \left(\sum_{g, h, k \in [G]^3, l \in [-ghk]} \left(\Pi_h' B_{h,g} \Gamma_g \right) \left(V_k ' B_{k,g} P_{g,l,-ghk} U_l \right)  \right) }_{R_{2}} \\
& + \underbrace{  \mathbb V \left( \sum_{g, h, k \in [G]^3, l \in [-ghk]} \left(V_h' B_{h,g} \Gamma_g \right) \left(\Pi_k ' B_{k,g} P_{g,l,-ghk} U_l \right)  \right)  }_{R_{3}} \\
& + \underbrace{  \mathbb V \left( \sum_{g, h, k \in [G]^3, l \in [-ghk]} \left(\Pi_h' B_{h,g} U_g \right) \left(\Pi_k ' B_{k,g} P_{g,l,-ghk} U_l \right)  \right) }_{R_{4}} \\
& + \underbrace{  \mathbb V \left( \sum_{g, h, k \in [G]^3, l \in [-ghk]} \left(V_h' B_{h,g} U_g \right) \left(\Pi_k ' B_{k,g} P_{g,l,-ghk} U_l \right)  \right) }_{R_{5}} \\
& + \underbrace{  \mathbb V \left( \sum_{g, h, k \in [G]^3, l \in [-ghk]} \left(V_h' B_{h,g} \Gamma_g \right) \left(V_k ' B_{k,g} P_{g,l,-ghk} U_l \right)  \right) }_{R_{6}} \\
& + \underbrace{  \mathbb V \left( \sum_{g, h, k \in [G]^3, l \in [-ghk]} \left(\Pi_h' B_{h,g} U_g \right) \left(V_k ' B_{k,g} P_{g,l,-ghk} U_l \right)  \right) }_{R_{7}} \\
& + \underbrace{  \mathbb V \left( \sum_{g, h, k \in [G]^3, l \in [-ghk]} \left(V_h' B_{h,g} U_g \right) \left(V_k ' B_{k,g} P_{g,l,-ghk} U_l \right)  \right) }_{R_{8}}.
\end{align*}

In the following sections, we will show that $R_i = o(\omega_n^4)$ for $i
=1, \dotsc, 8$, which combined with \cref{eq:var_l3co_1}, implies that
\begin{align*}
 \mathbb V\left(\sum_{g, h, k \in [G]^3} \left(X_h ' B_{h,g} Y_g \right) \left(X_k ' B_{k,g} (Y_g - W_g \hat \gamma_{-ghk}) \right)  \right) = o(\omega_n^4).
\end{align*}
Following the same argument, we can show
\begin{align*}
   & \mathbb V \left( 2 \sum_{g, h, k \in [G]^3} \left(X_h ' B_{h,g} Y_g \right) \left(Y_k ' B_{g,k}' (X_g - W_g \hat \pi_{-ghk}) \right) \right) = o(\omega_n^4) \quad \text{and} \\
   & \mathbb V \left( \sum_{g, h, k \in [G]^3} \left(Y_h ' B_{g,h}' X_g \right) \left(Y_k ' B_{g,k}' (X_g - W_g \hat \pi_{-ghk}) \right) \right)  = o(\omega_n^4).
\end{align*}

In addition, we have
\begin{align*}
    & \mathbb V\left( \sum_{g, h, k \in [G]^3} \left((Y_h - W_h \hat \gamma_{-ghk})' B_{g,h}' X_g \right) \left(Y_h' B_{g,h}' \tilde M_{gk,-gh} X_k \right) \right) \\
    & \lesssim \mathbb V\left( \sum_{g, h, k \in [G]^3} \left(U_h' B_{g,h}' X_g \right) \left(Y_h' B_{g,h}' \tilde M_{gk,-gh} X_k \right) \right) \\
    & + \mathbb V\left( \sum_{g, h, k \in [G]^3} \sum_{l \in [-ghk]} \left( U_l' P_{l,h,-ghk} B_{g,h}' X_g \right) \left(Y_h' B_{g,h}' \tilde M_{gk,-gh} X_k \right) \right) \\
    & \lesssim \underbrace{ \mathbb V\left( \sum_{g, h \in [G]^2} \left(U_h' B_{g,h}' X_g \right) \left(Y_h' B_{g,h}'  X_g \right) \right) }_{T_3} \\
    & +  \underbrace{ \mathbb V\left( \sum_{g, h \in [G]^2} \sum_{k \in [-gh]} \left(U_h' B_{g,h}' X_g \right) \left(Y_h' B_{g,h}' P_{g,k,-gh} X_k \right) \right) }_{T_4} \\
    & +  \underbrace{ \mathbb V\left( \sum_{g, h \in [G]^2} \sum_{l \in [-gh]} \left( U_l' P_{l,h,-gh} B_{g,h}' X_g \right) \left(Y_h' B_{g,h}'  X_g \right) \right) }_{T_5} \\
    & +  \underbrace{ \mathbb V\left( \sum_{g, h \in [G]^2} \sum_{k \in [-gh]} \sum_{l \in [-ghk]} \left( U_l' P_{l,h,-ghk} B_{g,h}' X_g \right) \left(Y_h' B_{g,h}' P_{g,k,-gh} X_k \right) \right) }_{T_6},
\end{align*}
where \Cref{lem:V1} shows $T_3 = o(\omega_n^4)$ and $T_5 = o(\omega_n^4)$ following the same argument as $T_2$ with $k = h$.

For $T_4$, we have
\begin{align*}
    T_4 & \lesssim  \mathbb V\left( \sum_{g, h \in [G]^2}  \left(U_h' B_{g,h}' X_g \right) \left(Y_h' B_{g,h}' \Pi_k \right) \right) + \mathbb V\left( \sum_{g, h \in [G]^2} \sum_{k \in [-gh]} \left(U_h' B_{g,h}' X_g \right) \left(Y_h' B_{g,h}' P_{g,k,-gh} V_k \right) \right) \\
    & \lesssim \mathbb V\left( \sum_{g, h \in [G]^2} \sum_{k \in [-gh]} \left(U_h' B_{g,h}' X_g \right) \left(Y_h' B_{g,h}' P_{g,k,-gh} V_k \right) \right) + o(\omega_n^4)\\
    & \lesssim \underbrace{ \mathbb V\left( \sum_{g, h \in [G]^2} \sum_{k \in [-gh]} \left(U_h' B_{g,h}' \Pi_g \right) \left(\Gamma_h' B_{g,h}' P_{g,k,-gh} V_k \right) \right) }_{T_{4,1}} \\
    & +  \underbrace{ \mathbb V\left( \sum_{g, h \in [G]^2} \sum_{k \in [-gh]} \left(U_h' B_{g,h}' \Pi_g \right) \left(U_h' B_{g,h}' P_{g,k,-gh} V_k \right) \right) }_{T_{4,2}} \\
    & +  \underbrace{ \mathbb V\left( \sum_{g, h \in [G]^2} \sum_{k \in [-gh]} \left(U_h' B_{g,h}' V_g \right) \left(\Gamma_h' B_{g,h}' P_{g,k,-gh} V_k \right) \right) }_{T_{4,3}} \\
    & +  \underbrace{ \mathbb V\left( \sum_{g, h \in [G]^2} \sum_{k \in [-gh]} \left(U_h' B_{g,h}' V_g \right) \left(U_h' B_{g,h}' P_{g,k,-gh} V_k \right) \right) }_{T_{4,4}} + o(\omega_n^4).
\end{align*}
where the first inequality is by
\begin{align*}
    \sum_{k \in [-gh]}P_{g,k,-gh}\Pi_k =  \sum_{k \in [-gh]}P_{g,k,-gh}W_k\pi = W_g \pi = \Pi_g,
\end{align*}
second inequality follows \Cref{lem:V1}.
In addition, we can show
\begin{align*}
    T_{4,i} = o(\omega_n^4), \quad i = 1, \dotsc, 4
\end{align*}
following the same arguments of  $R_3$, $R_6$, $R_5$, and $R_8$ with $h = k$, respectively. This implies
\begin{align*}
    T_4 = o(\omega_n^4).
\end{align*}

Last, for $T_6$, we have
\begin{align*}
T_6 & \lesssim \underbrace{ \mathbb V\left( \sum_{g, h \in [G]^2} \sum_{k \in [-gh]} \sum_{l \in [-ghk]} \left( U_l' P_{l,h,-ghk} B_{g,h}' \Pi_g \right) \left(\Gamma_h' B_{g,h}' P_{g,k,-gh} \Pi_k \right) \right) }_{R_9}\\
& + \underbrace{  \mathbb V\left( \sum_{g, h \in [G]^2} \sum_{k \in [-gh]} \sum_{l \in [-ghk]} \left( U_l' P_{l,h,-ghk} B_{g,h}' V_g \right) \left(\Gamma_h' B_{g,h}' P_{g,k,-gh} \Pi_k \right) \right) }_{R_{10}} \\
& + \underbrace{  \mathbb V\left( \sum_{g, h \in [G]^2} \sum_{k \in [-gh]} \sum_{l \in [-ghk]} \left( U_l' P_{l,h,-ghk} B_{g,h}' \Pi_g \right) \left(U_h' B_{g,h}' P_{g,k,-gh} \Pi_k \right) \right) }_{R_{11}} \\
& + \underbrace{  \mathbb V\left( \sum_{g, h \in [G]^2} \sum_{k \in [-gh]} \sum_{l \in [-ghk]} \left( U_l' P_{l,h,-ghk} B_{g,h}' \Pi_g \right) \left(\Gamma_h' B_{g,h}' P_{g,k,-gh} V_k \right) \right) }_{R_{12}} \\
& + \underbrace{  \mathbb V\left( \sum_{g, h \in [G]^2} \sum_{k \in [-gh]} \sum_{l \in [-ghk]} \left( U_l' P_{l,h,-ghk} B_{g,h}' V_g \right) \left(U_h' B_{g,h}' P_{g,k,-gh} \Pi_k \right) \right) }_{R_{13}} \\
& + \underbrace{  \mathbb V\left( \sum_{g, h \in [G]^2} \sum_{k \in [-gh]} \sum_{l \in [-ghk]} \left( U_l' P_{l,h,-ghk} B_{g,h}' V_g \right) \left(\Gamma_h' B_{g,h}' P_{g,k,-gh} V_k \right) \right) }_{R_{14}} \\
& + \underbrace{  \mathbb V\left( \sum_{g, h \in [G]^2} \sum_{k \in [-gh]} \sum_{l \in [-ghk]} \left( U_l' P_{l,h,-ghk} B_{g,h}' \Pi_g \right) \left(U_h' B_{g,h}' P_{g,k,-gh} V_k \right) \right) }_{R_{15}} \\
& + \underbrace{  \mathbb V\left( \sum_{g, h \in [G]^2} \sum_{k \in [-gh]} \sum_{l \in [-ghk]} \left( U_l' P_{l,h,-ghk} B_{g,h}' V_g \right) \left(U_h' B_{g,h}' P_{g,k,-gh} V_k \right) \right) }_{R_{16}}.
\end{align*}

In the following, we also show $R_9, \dotsc, R_{16}$ are $ o(\omega_n^4)$, which implies
\begin{equation*}
\mathbb V\left( \sum_{g, h, k \in [G]^3} \left((Y_h - W_h \hat \gamma_{-ghk})' B_{g,h}' X_g \right) \left(Y_h' B_{g,h}' \tilde M_{gk,-gh} X_k \right) \right)  =      o(\omega_n^4).
\end{equation*}
Following the same argument, we can show
\begin{equation*}
\mathbb V\left( \sum_{g, h, k \in [G]^3} \left((X_h - W_h \hat \pi_{-ghk})' B_{h,g} Y_g \right) \left(Y_h' B_{g,h}' \tilde M_{gk,-gh} X_k \right) \right) =o(\omega_n^4),
\end{equation*}
which concludes the proof of \Cref{thm:var_l3co}.

Terms $R_1, \dotsc, R_{16}$ involve summation of $P_{g,l,-ghk}$ across indexes $(g,h,k,l)$, which is calculated based on the following lemma, whose proof is provided in \Cref{sec:P_l3o_1_pf} below.

\begin{lem}\label{lem:P_l3o_1}
Define $M_{ghk,ghk}$ as the submatrix of $M$ whose rows and columns correspond to clusters $(g,h,k)$ in this order, and $P_{ghk,ghk}$, $M_{gh,gh}$, and $P_{gh,gh}$ in the same manner. Let
\begin{align*}
\tilde  P_{ghk} =
\begin{cases}
    & M_{ghk,ghk}^{-1}   P_{ghk,ghk}, \quad g \neq h, g \neq k, h \neq k, \\
    & M_{gh,gh}^{-1}   P_{gh,gh} = \tilde  P_{g,h} , \quad g \neq h, h = k, \\
    & 0, \quad \text{with conformable dimensions otherwise}
 \end{cases},
\end{align*}
and $(\tilde  P_{ghk})_{g,ghk} = (\tilde  P_{g,h})_{g,gh} $and $P_{g,ghk} = P_{g,gh}$ when $h = k$, where $(\tilde  P_{ghk})_{g,ghk}$ denotes the submatrix of $\tilde  P_{ghk}$ whose rows and columns correspond to the $g$-th cluster and the $(g,h,k)$-th cluster (in this order), respectively.

Then, the following statements are true:
\begin{enumerate}
    \item  Suppose $g \neq h$ and $g \neq k$. Then,   $\tilde  P_{ghk} $ is symmetric and
    \begin{align*}
    P_{l,g,-ghk} = P_{l,g} + P_{l,ghk} (\tilde  P_{ghk})_{ghk,g}
    \end{align*}
    \item Suppose $g \neq h$ and $g \neq k$. Then, we have
    \begin{align}\label{eq:p_l30_1}
    P_{ghk,ghk} \tilde  P_{ghk} = \tilde  P_{ghk} - P_{ghk,ghk} = \tilde  P_{ghk}  P_{ghk,ghk}.
\end{align}
    \item When  $g \neq h$, $g \neq k$,  $ h \neq k$, we have
\begin{align*}
& (\tilde P_{ghk})_{g,g} =  \Xi_{1,g} + \Xi_{2,gh} + \Xi_{3,gk} + \Xi_{4,ghk},     \\
& (\tilde P_{ghk})_{g,h} = \Xi_{5,gh} + \Xi_{6,ghk}, \\
& (\tilde P_{ghk})_{g,k} = \Xi_{7,gk} + \Xi_{8,ghk},
\end{align*}
where
\begin{align*}
& ||\Xi_{1,g}||_{op} \lesssim  \lambda_n, \quad ||\Xi_{2,gh}||_{op} \lesssim ||P_{g,h}||_{op}^2, \quad  ||\Xi_{3,gk}||_{op} \lesssim ||P_{g,k}||_{op}^2,  \\
& ||\Xi_{4,ghk}||_{op} \lesssim  (||P_{g,h}||_{op}^2 + ||P_{g,k}||_{op}^2)  ||P_{h,k}||_{op}^2 + \lambda_n ||P_{g,h}||_{op} ||P_{g,k}||_{op}, \\
& ||\Xi_{5,gh}||_{op} \lesssim ||P_{g,h}||_{op}, \\
& || \Xi_{6,ghk}||_{op} \lesssim  \left(||P_{g,h}||_{op} ||P_{g,k}||_{op} + ||P_{h,k}||_{op} \right) \left(||P_{g,h}||_{op} + ||P_{g,k}||_{op}\right), \\
& ||\Xi_{7,gk}||_{op} \lesssim ||P_{g,k}||_{op}, \quad \text{and}\\
& || \Xi_{8,ghk}||_{op} \lesssim \left(||P_{g,h}||_{op} ||P_{g,k}||_{op} + ||P_{h,k}||_{op} \right) \left(||P_{g,h}||_{op} + ||P_{g,k}||_{op}\right).
\end{align*}

\item We have
\begin{align*}
    ||P_{h,k,-gh}||_{op} = ||P_{h,k} + (\tilde P_{gh})_{h,gh} P_{gh,k}||_{op} \lesssim ||P_{h,k}||_{op} + ||P_{g,k}||_{op}.
\end{align*}

\item When $g \neq h$,   $k = h$, we denote $(\tilde P_{ghk})_{g,g}$ and $(\tilde P_{ghk})_{g,h}$ as $(\tilde P_{gh})_{g,g}$ and $(\tilde P_{gh})_{g,h}$, respectively, and have
\begin{align*}
& (\tilde P_{gh})_{g,g} =  \Xi_{1,g} + \Xi_{9,gh}, \quad ||\Xi_{9,gh}||_{op} \lesssim ||P_{g,h}||_{op}^2, \quad \text{and} \quad  || (\tilde P_{gh})_{g,h} ||_{op}\lesssim ||P_{g,h}||_{op}.
\end{align*}
\item When $g \neq h$, $g \neq k$, $g' \neq h'$, $g' \neq k'$, we have
\begin{align*}
& \sum_{l \in [-gg'hh'kk']} P_{g,l,-ghk} P_{l,g',-g'h'k'} \\
& = P_{g,g'} -  (\tilde  P_{ghk})_{g,ghk} P_{ghk,g'h'k'} (\tilde  P_{g'h'k'})_{g'h'k',g'} + \sum_{l \in (g,h,k) \cap (g',h',k')} P_{g,l,-ghk} P_{l,g',-g'h'k'}.
\end{align*}
\end{enumerate}
\end{lem}

\subsection{Bound for \texorpdfstring{$R_1$}{R1}}
By \Cref{lem:P_l3o_1}.6, we have
\begin{align*}
    R_1 & \leq u_n \sum_{l \in [G]} \left\Vert \sum_{g, h, k \in [-l]^3} \left(\Pi_h' B_{h,g} \Gamma_g \right) \Pi_k ' B_{k,g} P_{g,l,-ghk} \right\Vert_2^2 \\
    & \lesssim  u_n \sum_{ \substack{g,g', h,h', k,k' \in [G]^6}    }  \left(\Pi_h' B_{h,g} \Gamma_g \right) \Pi_k ' B_{k,g}  \left( \sum_{l \in [-gg'hh'kk']} P_{g,l,-ghk}P_{l,g',-g'h'k'} \right) B_{k',g'}' \Pi_{k'} \left(\Pi_{h'}' B_{h',g'} \Gamma_{g'} \right) \\
    & \lesssim \underbrace{ \left\vert u_n \sum_{ \substack{g,g', h,h', k,k' \in [G]^6}    }  \left(\Pi_h' B_{h,g} \Gamma_g \right) \Pi_k ' B_{k,g}  \left(P_{g,g'}\right) B_{k',g'}' \Pi_{k'} \left(\Pi_{h'}' B_{h',g'} \Gamma_{g'} \right) \right\vert }_{R_{1,1}} \\
    & + \underbrace{ \left\vert u_n \sum_{ \substack{g,g', h,h', k,k' \in [G]^6 ,\\ s \in (g,h,k),s' \in (g',h',k')  }    }  \left(\Pi_h' B_{h,g} \Gamma_g \right) \Pi_k ' B_{k,g}  \left( (\tilde  P_{ghk})_{g,s} P_{s,s'} (\tilde  P_{g'h'k'})_{s',g'} \right) B_{k',g'}' \Pi_{k'} \left(\Pi_{h'}' B_{h',g'} \Gamma_{g'} \right) \right\vert }_{R_{1,2}} \\
    & + \underbrace{ \left\vert u_n \sum_{ \substack{g,g', h,h', k,k' \in [G]^6 }    }  \left(\Pi_h' B_{h,g} \Gamma_g \right) \Pi_k ' B_{k,g}  \left( \sum_{l \in (g,h,k) \cap (g',h',k')} P_{g,l,-ghk} P_{l,g',-g'h'k'} \right) B_{k',g'}' \Pi_{k'} \left(\Pi_{h'}' B_{h',g'} \Gamma_{g'} \right) \right\vert }_{R_{1,3}},
\end{align*}
where
\begin{align*}
    R_{1,1} & =  \left\vert u_n \sum_{ g,g', h,h' \in [G]^4   }  \left(\Pi_h' B_{h,g} \Gamma_g \right) H_g'  \left(P_{g,g'}\right) H_{g'} \left(\Pi_{h'}' B_{h',g'} \Gamma_{g'} \right) \right\vert \\
    & \leq \sum_{g \in [G]} u_n  \left\Vert \Gamma_g'  H_{g} H_g' \right\Vert_2^2 \\
    & \lesssim u_n n_G \zeta_{H,n} \mu_n^2  = o ((\mu_n^2 + \kappa_n + \tilde \mu_n^2)^2) = o(\omega_n^4).
\end{align*}
To bound $R_{1,2}$, we only need to show
for $s \in  \{g,h,k\}$
\begin{align}
& u_n    \sum_{s \in [G]} \left\Vert \sum_{ (g,h,k)/\{s\} \in [G]^2}\left((\tilde P_{ghk})_{s,g}  B_{k,g}' \Pi_{k}\right)  \left(\Pi_{h}' B_{h,g} \Gamma_{g} \right)   \right\Vert^2_2 = o(\omega_n^4). \label{eq:L3O_10_2}
\end{align}

To see \cref{eq:L3O_10_2}, we first consider the case with $s =g$. By \Cref{lem:P_l3o_1}, we have
\begin{align*}
&  u_n \sum_{g \in [G]} \left\Vert \sum_{ h,k \in [G]^2}\left((\tilde P_{ghk})_{g,g}  B_{k,g}' \Pi_{k}\right)  \left(\Pi_{h}' B_{h,g}  \Gamma_g\right)   \right\Vert_2^2\\
& \lesssim \underbrace{ u_n \sum_{g \in [G]} \left\Vert \sum_{ h,k \in [G]^2}\left(\Xi_{1,g}  B_{k,g}' \Pi_{k}\right)  \left(\Pi_{h}' B_{h,g}  \Gamma_g\right)   \right\Vert_2^2 }_{R_{1,2,1}}\\
 & + \underbrace{  u_n \sum_{g \in [G]} \left\Vert \sum_{ h,k \in [G]^2, h \neq k}\left(\Xi_{2,gh}  B_{k,g}' \Pi_{k}\right)  \left(\Pi_{h}' B_{h,g}  \Gamma_g\right)   \right\Vert_2^2 }_{R_{1,2,2}} \\
 & + \underbrace{  u_n \sum_{g \in [G]} \left\Vert \sum_{ h,k \in [G]^2, h \neq k}\left(\Xi_{3,gk}  B_{k,g}' \Pi_{k}\right)  \left(\Pi_{h}' B_{h,g}  \Gamma_g\right)   \right\Vert_2^2 }_{R_{1,2,3}}\\
 & + \underbrace{  u_n \sum_{g \in [G]} \left\Vert \sum_{ h,k \in [G]^2, h \neq k}\left(\Xi_{4,ghk}  B_{k,g}' \Pi_{k}\right)  \left(\Pi_{h}' B_{h,g}  \Gamma_g\right)   \right\Vert_2^2 }_{R_{1,2,4}}\\
 & + \underbrace{  u_n \sum_{g \in [G]} \left\Vert \sum_{ h\in [G]}\left(\Xi_{9,gh}  B_{h,g}' \Pi_{h}\right)  \left(\Pi_{h}' B_{h,g}  \Gamma_g\right)   \right\Vert_2^2 }_{R_{1,2,5}},
\end{align*}

where
\begin{align*}
R_{1,2,1} & \lesssim u_n  \sum_{g \in [G]} \left\Vert \Xi_{1,g}  H_{g}H_{g}'  \Gamma_g   \right\Vert_2^2 \\
& \lesssim u_n  n_G \lambda_n \zeta_{H,n} \mu_n^2  = o ((\mu_n^2 + \kappa_n + \tilde \mu_n^2)^2) = o(\omega_n^4),
\end{align*}
\begin{align*}
R_{1,2,2} & \lesssim  u_n \sum_{g \in [G]} \left\Vert \sum_{ h \in [G]}\left(\Xi_{2,gh}  H_g\right)  \left(\Pi_{h}' B_{h,g}  \Gamma_g\right)   \right\Vert_2^2  \\
& +  u_n \sum_{g \in [G]} \left\Vert \sum_{ h \in [G]}\left(\Xi_{2,gh}  B_{h,g}' \Pi_{h}\right)  \left(\Pi_{h}' B_{h,g}  \Gamma_g\right)   \right\Vert_2^2 \\
& \lesssim u_n n_G \lambda_n^2 \sum_{g \in [G]} \left( \sum_{ h \in [G]} ||P_{g,h}||_{op} || H_g||_2  ||B_{h,g}  \Gamma_g||_2 \right)^2  \\
& +  u_n n_G^2 \lambda_n^4 \sum_{g \in [G]} \left( \sum_{ h \in [G]} ||P_{g,h}||_{op}  || B_{h,g}  \Gamma_g||_2 \right)^2 \\
& \lesssim u_n n_G  \lambda_n^2 \sum_{g \in [G]} \left( \sum_{ h \in [G]} ||P_{g,h}||_{op}^2 || H_g||_2^2 \right) \left( \sum_{h \in [G]}  ||B_{h,g}  \Gamma_g||_2^2 \right)  \\
& +  u_n n_G^2 \lambda_n^4 \sum_{g \in [G]} \left( \sum_{ h \in [G]} ||P_{g,h}||_{op}^2 \right) \left( \sum_{h \in [G]}  || B_{h,g}  \Gamma_g||_2^2 \right) \\
& \lesssim u_n n_G^2 \phi_n \lambda_n^3 \mu_n^2 + u_n n_G^3 \phi_n \lambda_n^4 \kappa_n  = o(\omega_n^4)
\end{align*}
\begin{align*}
R_{1,2,3}  = o(\omega_n^4),
\end{align*}
following the same manner as $R_{1,2,2}$.

In addition, we have
\begin{align*}
R_{1,2,4}  &  \lesssim u_n \sum_{g \in [G]} \left( \sum_{ h,k \in [G]^2} ||\Xi_{4,ghk}||_{op}  ||\Pi_{k}' B_{k,g}||_{2}  ||\Pi_{h}' B_{h,g}||_2 || \Gamma_g ||_{2}   \right)^2 \\
 & \lesssim u_n \sum_{g \in [G]} n_G \left(
     \sum_{ h,k \in [G]^2}  ||P_{g,h}||_{op}^2 ||P_{h,k}||_{op}^2  ||\Pi_{k}' B_{k,g}||_{2}  ||\Pi_{h}' B_{h,g}||_2  \right)^2 \\
 & +  u_n \sum_{g \in [G]} n_G \left(
       \sum_{ h,k \in [G]^2}  ||P_{g,k}||_{op}^2  ||P_{h,k}||_{op}^2 ||\Pi_{k}' B_{k,g}||_{2}  ||\Pi_{h}' B_{h,g}||_2
 \right) ^2 \\
 & + u_n \sum_{g \in [G]} n_G \lambda_n^2 \left( \sum_{ h,k \in [G]^2} ||P_{g,h}||_{op} ||P_{g,k}||_{op}    ||\Pi_{k}' B_{k,g}||_{2}  ||\Pi_{h}' B_{h,g}||_2     \right)^2 \\
 & \lesssim u_n \sum_{g \in [G]} n_G^2 \phi_n \lambda_n^2 \begin{pmatrix}
     \sum_{ h \in [G]}  ||P_{g,h}||_{op}^2   ||\Pi_{h}' B_{h,g}||_2
 \end{pmatrix} ^2 \\
 & +  u_n \sum_{g \in [G]} n_G^2 \phi_n \lambda_n^2 \begin{pmatrix}
     \sum_{ k \in [G]}  ||P_{g,k}||_{op}^2  ||\Pi_{k}' B_{k,g}||_{2}
 \end{pmatrix} ^2 \\
 & + u_n \sum_{g \in [G]} n_G \lambda_n^2 \left(\sum_{h \in [G]} ||P_{g,h}||_{op}^2 \right)^2 \left(\sum_{h \in [G]} ||\Pi_{h}' B_{h,g}||_{2}^2 \right)^2 \\
 & \lesssim u_n n_G^3 \phi_n^2 \lambda_n^2 \kappa_n + u_n n_G^3 \phi_n^3 \lambda_n^2 \kappa_n = o ((\mu_n^2 + \kappa_n + \tilde \mu_n^2)^2) = o(\omega_n^4),
\end{align*}
where the last inequality is by \Cref{lem:nablaB} and
\begin{equation}\label{eq:phi_B}
\sum_{h \in [G]} ||\Pi_{h}' B_{h,g}||_{2}^2 \leq n_G \sum_{h \in [G]} ||B_{h,g}||_{op}^2 \lesssim n_G (\phi_n + \lambda_n^2).
\end{equation}

Last, we have
\begin{align*}
R_{1,2,5}  & \lesssim u_n  n_G \sum_{g \in [G]}  \left( \sum_{h \in [G]} ||P_{g,h}||_{op}^2 || B_{h,g}' \Pi_{h}||_2^2 \right)^2 \\
& \lesssim  u_n n_G^2 \phi_n \lambda_n^2 \left( \sum_{g,h \in [G]^2} ||P_{g,h}||_{op}^2 || B_{h,g}' \Pi_{h}||_2^2 \right) \\
& \lesssim  u_n n_G^3 \phi_n \lambda_n^4 \kappa_n = o ((\mu_n^2 + \kappa_n + \tilde \mu_n^2)^2) = o(\omega_n^4).
\end{align*}
This verifies \cref{eq:L3O_10_2} with $s = g$.

Next, we consider \cref{eq:L3O_10_2} with $s = k$. By \Cref{lem:P_l3o_1}, we have
\begin{align*}
 & u_n    \sum_{k \in [G]} \left\Vert \sum_{ g,h \in [G]^2}\left((\tilde P_{ghk})_{k,g}  B_{k,g}' \Pi_{k}\right)  \left(\Pi_{h}' B_{h,g}  \Gamma_g\right)   \right\Vert^2_2 \\
 & \lesssim \underbrace{ u_n    \sum_{k \in [G]} \left\Vert \sum_{ g,h \in [G]^2, h \neq k}\left(\Xi_{7,gk}'  B_{k,g}' \Pi_{k}\right)  \left(\Pi_{h}' B_{h,g}  \Gamma_g\right)   \right\Vert^2_2}_{R_{1,2,6}} \\
 & + \underbrace{ u_n   \sum_{k \in [G]} \left\Vert \sum_{ g,h \in [G]^2, h \neq k}\left(\Xi_{8,ghk}'  B_{k,g}' \Pi_{k}\right)  \left(\Pi_{h}' B_{h,g}  \Gamma_g\right)   \right\Vert^2_2}_{R_{1,2,7}}  \\
 & + \underbrace{ u_n    \sum_{k \in [G]} \left\Vert \sum_{ g \in [G]}\left((\tilde P_{gk})_{k,g}  B_{k,g}' \Pi_{k}\right)  \left(\Pi_{k}' B_{k,g}  \Gamma_g\right)   \right\Vert^2_2}_{R_{1,2,8}},
\end{align*}
where
\begin{align*}
R_{1,2,6} & = u_n    \sum_{k \in [G]} \left\Vert \sum_{ g \in [G]}\left(\Xi_{7,gk}'  B_{k,g}' \Pi_{k}\right)  \left(H_{g} - B_{k,g}' \Pi_{k}\right)'  \Gamma_g   \right\Vert^2_2    \\
& \lesssim u_n    \sum_{k \in [G]} \left\Vert \sum_{ g \in [G]}\left(\Xi_{7,gk}'  B_{k,g}' \Pi_{k}\right)  H_{g}'  \Gamma_g   \right\Vert^2_2 \\
& + u_n    \sum_{k \in [G]} \left(  \sum_{ g \in [G]}||P_{g,k}||_{op}  ||B_{k,g}' \Pi_{k}||_{2} ||B_{k,g}' \Pi_{k}||_2 || \Gamma_g ||_2   \right)^2 \\
& \lesssim u_n   n_G \sum_{k \in [G]} \left( \sum_{ g \in [G]} ||P_{g,k}||_{op}^2 || H_{g}||_2\right) \left( \sum_{g \in [G]}  ||B_{k,g}' \Pi_{k}||_2^2     \right) \\
& + u_n n_G^2 \lambda_n^2  \sum_{k \in [G]} \left(  \sum_{ g \in [G]}||P_{g,k}||_{op}^2 \right) \left( \sum_{g \in [G]}  ||B_{k,g}' \Pi_{k}||_{2}^2   \right) \\
& \lesssim u_n   n_G^2 \phi_n \lambda_n \mu_n^2  + u_n n_G^3 \phi_n \lambda_n^2   \kappa_n = o ((\mu_n^2 + \kappa_n + \tilde \mu_n^2)^2) = o(\omega_n^4).
\end{align*}


In addition, for $R_{1,2,7}$, we have
\begin{align*}
R_{1,2,7} & \lesssim u_n \sum_{k \in [G]} \left ( \sum_{ g,h \in [G]^2} ||P_{g,h}||_{op}^2 ||P_{g,k}||_{op} ||B_{k,g}' \Pi_{k}||_2  \left\vert \Pi_{h}' B_{h,g} \Gamma_g \right\vert   \right)^2\\
& + u_n \sum_{k \in [G]} \left ( \sum_{ g,h \in [G]^2} ||P_{h,k}||_{op} ||P_{g,h}||_{op} ||B_{k,g}' \Pi_{k}||_2  \left\vert \Pi_{h}' B_{h,g} \Gamma_g \right\vert    \right)^2\\
& + u_n \sum_{k \in [G]} \left ( \sum_{ g,h \in [G]^2} ||P_{g,h}||_{op} ||P_{g,k}||_{op}  ||P_{g,k}||_{op}  ||B_{k,g}' \Pi_{k}||_2  \left\vert \Pi_{h}' B_{h,g} \Gamma_g \right\vert    \right)^2\\
& + u_n \sum_{k \in [G]} \left ( \sum_{ g,h \in [G]^2} ||P_{h,k}||_{op} ||P_{g,k}||_{op}  ||B_{k,g}' \Pi_{k}||_2  \left\vert \Pi_{h}' B_{h,g} \Gamma_g \right\vert    \right)^2\\
& \lesssim u_n n_G^2 \phi_n^2 \lambda_n^2  \sum_{k \in [G]} \left ( \sum_{ g \in [G]}  ||P_{g,k}||_{op}  ||B_{k,g}' \Pi_{k}||_2   \right)^2\\
& + u_n \sum_{k \in [G]} \left ( \sum_{ g,h \in [G]^2} ||P_{h,k}||_{op}^2 ||P_{g,h}||_{op}^2\right)  \left( \sum_{ g,h \in [G]^2}||B_{k,g}' \Pi_{k}||_2^2  \left( \Pi_{h}' B_{h,g} \Gamma_g \right)  ^2 \right)\\
& + u_n \lambda_n^2 \sum_{k \in [G]} \left ( \sum_{ g,h \in [G]^2} ||P_{g,k}||_{op}^2 ||P_{g,h}||_{op}^2\right)  \left( \sum_{ g,h \in [G]^2}||B_{k,g}' \Pi_{k}||_2^2  \left( \Pi_{h}' B_{h,g} \Gamma_g \right)  ^2 \right)\\
& + u_n \sum_{k \in [G]} \left ( \sum_{ g,h \in [G]^2} ||P_{h,k}||_{op}^2 ||P_{g,k}||_{op}^2\right)  \left( \sum_{ g,h \in [G]^2}||B_{k,g}' \Pi_{k}||_2^2  \left( \Pi_{h}' B_{h,g} \Gamma_g \right)  ^2 \right)\\
& \lesssim  u_n n_G^3  \phi_n^3 \lambda_n^2 \kappa_n + u_n n_G^3  \phi_n^2 \lambda_n \kappa_n = o ((\mu_n^2 + \kappa_n + \tilde \mu_n^2)^2) = o(\omega_n^4),
\end{align*}
where the last two inequalities are by \cref{eq:phi_B}.

Last, we have
\begin{align*}
R_{1,2,8} &  \lesssim u_n  \lambda_n^2 n_G  \sum_{k \in [G]} \left( \sum_{ g \in [G]}||B_{k,g}' \Pi_{k}||_2^2 \right)^2  \\
& \lesssim u_n n_G^3  \lambda_n^3   \kappa_n = o ((\mu_n^2 + \kappa_n + \tilde \mu_n^2)^2) = o(\omega_n^4),
\end{align*}
which, combined with previous bounds for $R_{1,2,6}$ and $R_{1,2,7}$, further implies \cref{eq:L3O_10_2} holds with $s = k$. We can show \cref{eq:L3O_10_2} with $s = h$ in the same manner as above. This implies
\begin{align*}
    R_{1,2} = o(\omega_n^4).
\end{align*}

Next, we turn to $R_{1,3}$. Denote
\begin{align*}
f_{gg'hh'kk'}(l) = u_n
& \left( B_{k',g'}' \Pi_{k'} \Pi_{h'}' B_{h',g'}  \Gamma_{g'} \right)'\left( P_{g',l,-g'h'k'} P_{l,g,-ghk} \right) \left( B_{k,g}' \Pi_{k}\Pi_{h}' B_{h,g}  \Gamma_g\right),
\end{align*}
such that $f_{gg'hh'kk'}(l) = f_{g'gh'hk'k}(l)$.

Then, we have
\begin{align*}
 R_{1,3}   & = \left\vert \sum_{g, g',h, h', k, k' \in [G]^6}   \sum_{l \in (g,h,k) \cap (g',h',k')} f_{gg'hh'kk'}(l) \right\vert \\
    & =   \left\vert \sum_{g, g',h, h', k, k' \in [G]^6}   \sum_{s \in \{g,h,k\}}  \sum_{l \in (g,h,k) \cap (g',h',k')} f_{gg'hh'kk'}(l) 1\{l = s\} \right\vert \\
    & \lesssim  \left\vert \sum_{g, g',h, h', k, k' \in [G]^6}   \sum_{g \in (g',h',k')} f_{gg'hh'kk'}(g) \right\vert+  \left\vert \sum_{g, g',h, h', k, k' \in [G]^6}   \sum_{h \in (g',h',k')} f_{gg'hh'kk'}(h)  \right\vert \\
   & +  \left\vert \sum_{g, g',h, h', k, k' \in [G]^6}   \sum_{k \in (g',h',k')} f_{gg'hh'kk'}(k)  \right\vert \\
   & \lesssim  \left\vert \sum_{g, g',h, h', k, k' \in [G]^6}   \sum_{g \in (g',h',k')} f_{gg'hh'kk'}(g) \right\vert+  \left\vert \sum_{g, g',h, h', k, k' \in [G]^6}   \sum_{h \in [h'k']} f_{gg'hh'kk'}(h)  \right\vert \\
   & +  \left\vert \sum_{g, g',h, h', k, k' \in [G]^6}   \sum_{k =  k'} f_{gg'hh'kk'}(k)  \right\vert \\
   & \lesssim \underbrace{ \left\vert \sum_{g, g',h, h', k, k' \in [G]^6}   1\{g = g'\} f_{gg'hh'kk'}(g) \right\vert }_{R_{1,3,1}} + \underbrace{ \left\vert \sum_{g, g',h, h', k, k' \in [G]^6}   1\{g = h'\} f_{gg'hh'kk'}(g) \right\vert }_{R_{1,3,2}}\\
   & + \underbrace{ \left\vert \sum_{g, g',h, h', k, k' \in [G]^6}   1\{g = k'\} f_{gg'hh'kk'}(g) \right\vert }_{R_{1,3,3}} + \underbrace{  \left\vert \sum_{g, g',h, h', k, k' \in [G]^6}   1\{g = h' = k'\} f_{gg'hh'kk'}(g) \right\vert }_{R_{1,3,4}} \\
   & + \underbrace{  \left\vert \sum_{g, g',h, h', k, k' \in [G]^6}   1\{h = h'\} f_{gg'hh'kk'}(h)  \right\vert }_{R_{1,3,5}} + \underbrace{  \left\vert \sum_{g, g',h, h', k, k' \in [G]^6}   1\{h = k'\} f_{gg'hh'kk'}(h)  \right\vert }_{R_{1,3,6}} \\
   & + \underbrace{  \left\vert \sum_{g, g',h, h', k, k' \in [G]^6}   1\{h = h' = k'\} f_{gg'hh'kk'}(h)  \right\vert}_{R_{1,3,7}} + \underbrace{  \left\vert \sum_{g, g',h, h', k, k' \in [G]^6}   1\{k =  k'\} f_{gg'hh'kk'}(k)  \right\vert }_{R_{1,3,8}},
\end{align*}
where the second inequality is by the triangle inequality and the symmetry of $f_{gg'hh'kk'}(\cdot)$ w.r.t.\ indexes $(g,h,k)$ and $(g',h',k')$ and the last inequality is by the fact that $f_{gg'hh'kk'}(\cdot) = 0 $ if $g' = h'$ or $g' = k'$.

Next, we have
\begin{align*}
R_{1,3,1} & = \left\vert \sum_{g,h, h', k, k' \in [G]^5}    u_n  \left( B_{k',g}' \Pi_{k'} \Pi_{h'}' B_{h',g}  \Gamma_g \right)' \left((\tilde P_{gh'k'})_{g,g} (\tilde P_{ghk})_{g,g} \right)  \left( B_{k,g}' \Pi_{k}\Pi_{h}' B_{h,g}  \Gamma_g\right) \right\vert \\
& = u_n \sum_{g \in [G]} \left\Vert \sum_{h,k \in [G]^2} (\tilde P_{ghk})_{g,g}\left( B_{k,g}' \Pi_{k}\Pi_{h}' B_{h,g}  \Gamma_g\right)  \right\Vert_2^2 \\
& =o\left( \omega_n^4\right),
\end{align*}
where the last equality is due to \cref{eq:L3O_10_2} with $s = g$.

Similarly, by \cref{eq:L3O_10_2} with $s = g$ and $s = h$, we have
\begin{align*}
R_{1,3,2} & = \left\vert \sum_{g,g',h, k, k' \in [G]^5}   u_n  \left( B_{k',g'}' \Pi_{k'} \Pi_{g}' B_{g,g'}  \Gamma_{g'} \right)' \left((\tilde P_{g'gk'})_{g',g} (\tilde P_{ghk})_{g,g}  \right)  \left( B_{k,g}' \Pi_{k}\Pi_{h}' B_{h,g}  \Gamma_g\right) \right\vert \\
& = u_n \sum_{g \in [G]} \left\Vert \sum_{g',k' \in [G]^2} (\tilde P_{g'gk'})_{g,g'}\left( B_{k',g'}' \Pi_{k'}\Pi_{g}' B_{g,g'}  \Gamma_{g'}\right)  \right\Vert_2^2 \\
& + u_n \sum_{g \in [G]} \left\Vert \sum_{h,k \in [G]^2} (\tilde P_{ghk})_{g,g}\left( B_{k,g}' \Pi_{k}\Pi_{h}' B_{h,g}  \Gamma_g\right)  \right\Vert_2^2 \\
& = u_n \sum_{h \in [G]} \left\Vert \sum_{g,k \in [G]^2} (\tilde P_{ghk})_{h,g}\left( B_{k,g}' \Pi_{k}\Pi_{h}' B_{h,g}  \Gamma_g\right)  \right\Vert_2^2 \\
& + u_n \sum_{g \in [G]} \left\Vert \sum_{h,k \in [G]^2} (\tilde P_{ghk})_{g,g}\left( B_{k,g}' \Pi_{k}\Pi_{h}' B_{h,g}  \Gamma_g\right)  \right\Vert_2^2 \\
& = o\left( \omega_n^4\right).
\end{align*}

For the same reason, we can show that
\begin{align*}
R_{1,3,i}  = o\left( \omega_n^4\right), \quad i = 3,5,6,8.
\end{align*}

Next, we have
\begin{align*}
R_{1,3,4} & = \left\vert \sum_{g, g',h, k \in [G]^4}  u_n  \left( B_{g,g'}' \Pi_{g} \Pi_{g}' B_{g,g'}  \Gamma_{g'} \right)'  \left( P_{g',g,-g'g} P_{g,g,-ghk} \right)  \left( B_{k,g}' \Pi_{k}\Pi_{h}' B_{h,g}  \Gamma_g\right) \right\vert \\
& \lesssim u_n \sum_{g \in [G]} \left\Vert \sum_{g' \in [G]}\left( B_{g,g'}' \Pi_{g} \Pi_{g}' B_{g,g'}  \Gamma_{g'} \right)' P_{g',g,-g'g} \right\Vert_2^2 \\
& + u_n \sum_{g \in [G]} \left\Vert \sum_{h,k \in [G]^2} (\tilde P_{ghk})_{g,g}\left( B_{k,g}' \Pi_{k}\Pi_{h}' B_{h,g}  \Gamma_g\right)  \right\Vert_2^2 \\
& \lesssim  u_n n_G \phi_n \lambda _n^2 \sum_{g \in [G]} (\Pi_{g}' [BB']_{g,g}\Pi_{g})^2 + o\left((\mu_n^2 + \kappa_n)^2\right)\\
& \lesssim  u_n n_G^3 \phi_n \lambda _n^3  \kappa_n + o\left(\omega_n^4\right)\\
& = o\left(\omega_n^4\right).
\end{align*}
Similarly, we can show $R_{1,3,7} = o\left(\omega_n^4\right)$, which concludes the proof.  This implies $R_{1,3} = o\left(\omega_n^4\right)$, which, combined with the bounds for $R_{1,1}$ and $R_{1,2}$, leads to the desired result that $R_1 = o\left(\omega_n^4\right)$.

\subsection{Bound for \texorpdfstring{$R_2$}{R2}}
We have
\begin{align*}
    R_2 & =  \mathbb V \left(\sum_{ k \in [G], l \in [-k]}  \sum_{g,h \in [-l]^2}\left(\Pi_h' B_{h,g} \Gamma_g \right) \left(V_k ' B_{k,g} P_{g,l,-ghk} U_l \right)  \right) \\
    & \lesssim u_n^2 \sum_{ k \in [G], l \in [-k]} \sum_{g,g',h,h' \in [-l]^4} \left(\Pi_h' B_{h,g} \Gamma_g \right)\left(\Pi_{h'}' B_{h',g'} \Gamma_{g'} \right) tr\left(  B_{k,g} P_{g,l,-ghk}P_{l,g',-g'h'k} B_{k,g'}'   \right)\\
    & = u_n^2 \sum_{g,g',h,h',k \in [G]^5}   \left(\Pi_h' B_{h,g} \Gamma_g \right)\left(\Pi_{h'}' B_{h',g'} \Gamma_{g'} \right) tr\left(  B_{k,g} \left( \sum_{ l \in [-gg'hh'k]} P_{g,l,-ghk}P_{l,g',-g'h'k} \right) B_{k,g'}'   \right).
\end{align*}
In addition, by \Cref{lem:P_l3o_1}.6, we have
\begin{align*}
\sum_{l \in [-gg'hh'k]} P_{g,l,-ghk} P_{l,g',-g'h'k} & = P_{g,g'} - (\tilde  P_{ghk})_{g,ghk} P_{ghk,g'h'k} (\tilde  P_{g'h'k})_{g'h'k,g'} \\
& + \sum_{l \in (g,h,k) \cap (g',h',k)} P_{g,l,-ghk} P_{l,g',-g'h'k},
\end{align*}
which implies
\begin{align*}
R_2 & \lesssim \underbrace{ \left\vert u_n^2 \sum_{g,g',h,h',k \in [G]^5}   \left(\Pi_h' B_{h,g} \Gamma_g \right)\left(\Pi_{h'}' B_{h',g'} \Gamma_{g'} \right) tr\left(  B_{k,g}  P_{g,g'} B_{k,g'}'   \right) \right\vert }_{R_{2,1}} \\
& + \underbrace{ \left\vert  u_n^2 \sum_{g,g',h,h',k \in [G]^5}   \left(\Pi_h' B_{h,g} \Gamma_g \right)\left(\Pi_{h'}' B_{h',g'} \Gamma_{g'} \right) tr\left(  B_{k,g} \left( (\tilde  P_{ghk})_{g,ghk} P_{ghk,g'h'k} (\tilde  P_{g'h'k})_{g'h'k,g'}  \right) B_{k,g'}'   \right) \right\vert }_{R_{2,2}} \\
& + \underbrace{ \left\vert u_n^2 \sum_{g,g',h,h',k \in [G]^5}   \left(\Pi_h' B_{h,g} \Gamma_g \right)\left(\Pi_{h'}' B_{h',g'} \Gamma_{g'} \right) tr\left(  B_{k,g} \left( \sum_{l \in (g,h,k) \cap (g',h',k)} P_{g,l,-ghk} P_{l,g',-g'h'k}  \right) B_{k,g'}'   \right) \right\vert }_{R_{2,3}}.
\end{align*}


In addition, we have
\begin{align*}
    R_{2,1}  & \lesssim  \left\vert u_n^2 \sum_{g,g',k \in [G]^3}   \left(H_g' \Gamma_g \right)\left(H_{g'}' \Gamma_{g'} \right) tr\left(  B_{k,g}  P_{g,g'} B_{k,g'}'   \right) \right\vert \\
    & \lesssim u_n^2 \sum_{g,k \in [G]^2}  (H_g' \Gamma_g ) ^2 ||  B_{k,g}||_F^2 \\
& \lesssim u_n^2 n_G \zeta_{H,n} \kappa_n = o ((\mu_n^2 + \kappa_n + \tilde \mu_n^2)^2) = o(\omega_n^4).
\end{align*}

In addition, we have
\begin{align}\label{eq:R22}
R_{2,2} = \left\vert  u_n^2 \sum_{g,g',h,h',k \in [G]^5} \sum_{\substack{s \in (g,h,k) \\ s' \in (g',h',k)}} \left(\Pi_h' B_{h,g} \Gamma_g \right)\left(\Pi_{h'}' B_{h',g'} \Gamma_{g'} \right) tr\left(  B_{k,g} \left( (\tilde  P_{ghk})_{g,s} P_{s,s'} (\tilde  P_{g'h'k})_{s',g'}  \right) B_{k,g'}'   \right) \right\vert.
\end{align}
We aim to bound $R_{2,2}$ for cases $(s,s') = (g,g')$, $(s,s') = (g,h')$, $(s,s') = (g,k)$, $(s,s') = (h,k)$, and $(s,s') = (k,k)$. The other cases can be bounded in the same manner.

\textbf{ For the case $(s,s')=(g,g')$}, we denote
\begin{align*}
    \mathcal A_{k,\bullet} = (\mathcal A_{k,1}, \dotsc, \mathcal A_{k,G}) \quad \text{and} \quad \mathcal A_{k,g} = \sum_{h \in [G]} \left(\Pi_h' B_{h,g} \Gamma_g \right)  B_{k,g}  (\tilde  P_{ghk})_{g,g}.
\end{align*}
Then, we have
\begin{align*}
& \left\vert  u_n^2 \sum_{g,g',h,h',k \in [G]^5}  \left(\Pi_h' B_{h,g} \Gamma_g \right)\left(\Pi_{h'}' B_{h',g'} \Gamma_{g'} \right) tr\left(  B_{k,g} \left( (\tilde  P_{ghk})_{g,g} P_{g,g'} (\tilde  P_{g'h'k})_{g',g'}  \right) B_{k,g'}'   \right) \right\vert \\
& = \sum_{k \in [G]} u_n^2 tr(\mathcal A_{k,\bullet}  P \mathcal A_{k,\bullet} ' ) \\
& \lesssim \sum_{k \in [G]} u_n^2 ||\mathcal A_{k,\bullet}||_F^2 \\
& = \sum_{g,k \in [G]^2} u_n^2 \left\Vert  \sum_{h \in [G]} \left(\Pi_h' B_{h,g} \Gamma_g \right)  B_{k,g}  (\tilde  P_{ghk})_{g,g} \right\Vert_F^2\\
& \lesssim  \underbrace{ \sum_{g,k \in [G]^2} u_n^2 \left\Vert  \sum_{h \in [G]} \left(\Pi_h' B_{h,g} \Gamma_g \right)  B_{k,g}  \Xi_{1,g} \right\Vert_F^2 }_{R_{2,2,1} }+ \underbrace{ \sum_{g,k \in [G]^2} u_n^2 \left\Vert  \sum_{h \in [-k]} \left(\Pi_h' B_{h,g} \Gamma_g \right)  B_{k,g}  \Xi_{3,gk} \right\Vert_F^2  }_{R_{2,2,2} }\\
& + \underbrace{ \sum_{g,k \in [G]^2} u_n^2 \left\Vert  \sum_{h \in [-k]} \left(\Pi_h' B_{h,g} \Gamma_g \right)  B_{k,g}  (\Xi_{2,gh} + \Xi_{4,ghk}) \right\Vert_F^2  }_{R_{2,2,3} }+ \underbrace{ \sum_{g,h \in [G]^2} u_n^2 \left\Vert  \left(\Pi_h' B_{h,g} \Gamma_g \right)  B_{h,g}  \Xi_{9,gh} \right\Vert_F^2 }_{R_{2,2,4} },
\end{align*}
where
\begin{align*}
R_{2,2,1} = \sum_{g,k \in [G]^2} u_n^2 (H_g' \Gamma_g)^2 ||  B_{k,g}  \Xi_{1,g} ||_F^2 \lesssim u_n^2 n_G  \lambda_n^2 \kappa_n = o ((\mu_n^2 + \kappa_n + \tilde \mu_n^2)^2) = o(\omega_n^4),
\end{align*}
\begin{align*}
R_{2,2,2} & =    \sum_{g,k \in [G]^2} u_n^2( (H_g - B_{k,g}'\Pi_k)' \Gamma_g)^2 ||  B_{k,g}  \Xi_{3,gk} ||_F^2 \\
& \lesssim o(\omega_n^4) +  \sum_{g,k \in [G]^2} u_n^2 (\Pi_k'B_{k,g} \Gamma_g)^2 ||  B_{k,g}  \Xi_{3,gk} ||_F^2 \\
& \lesssim o(\omega_n^4) +  \sum_{g,k \in [G]^2} u_n^2 (\Pi_k'B_{k,g} \Gamma_g)^2 ||  B_{k,g}  \Xi_{3,gk} ||_F^2 \\
& \lesssim o(\omega_n^4) +  u_n^2 n_G^2 \lambda_n^6 \kappa_n = o(\omega_n^4),
\end{align*}
\begin{align*}
R_{2,2,3} & \lesssim \sum_{g,k \in [G]^2} u_n^2 \left( \sum_{h \in [G]} ||\Pi_h' B_{h,g}||_2 ||\Gamma_g||_2  || B_{k,g} ||_F ||\Xi_{2,gh} + \Xi_{4,ghk}||_{op} \right)^2 \\
& \lesssim u_n^2  n_G^2 \lambda_n^2 \sum_{g,k\in [G]^2} \left(  \sum_{h \in [G]} ||B_{h,g}||_{op}  ||P_{g,h}||_{op} \right)^2  ||B_{k,g}||_F^2 \\
& \lesssim u_n^2  n_G^2 \lambda_n^2 \sum_{g,k\in [G]^2} \left(  \sum_{h \in [G]} ||B_{h,g}||_{op}^2 \right) \left(\sum_{h \in [G]} ||P_{g,h}||_{op}^2 \right)  ||B_{k,g}||_F^2 \\
& \lesssim u_n^2  n_G^2 \lambda_n^2 (\phi_n + \lambda_n^2) \phi_n \kappa_n = o(\omega_n^4),
\end{align*}
and
\begin{align*}
R_{2,2,4} \lesssim u_n^2 n_G^2 \sum_{g,h \in [G]^2} ||B_{h,g}||_F^2 ||\Xi_{9,gh}||_{op}^2 \lesssim     u_n^2 n_G^2  \lambda_n^4 \kappa_n = o(\omega_n^4).
\end{align*}
This concludes the case of $(s,s')=(g,g')$.

\textbf{ For the case $(s,s')=(g,h')$}, we denote
\begin{align*}
    \mathcal B_{k,\bullet} = [\mathcal B_{k,1}, \dotsc, \mathcal B_{k,G}] \in \Re^{n_k \times n} \quad \text{and}\quad \mathcal B_{k,h} = \sum_{g \in [G]} \left(\Pi_h' B_{h,g} \Gamma_g \right)  B_{k,g}  (\tilde  P_{ghk})_{g,h}.
\end{align*}
Then, we have
\begin{align*}
& \left\vert  u_n^2 \sum_{g,g',h,h',k \in [G]^5}  \left(\Pi_h' B_{h,g} \Gamma_g \right)\left(\Pi_{h'}' B_{h',g'} \Gamma_{g'} \right) tr\left(  B_{k,g} \left( (\tilde  P_{ghk})_{g,g} P_{g,h'} (\tilde  P_{g'h'k})_{h',g'}  \right) B_{k,g'}'   \right) \right\vert \\
& = \left\vert u_n^2 \sum_{k \in [G]} tr \left( \mathcal A_{k,\bullet} P \mathcal B_{k,\bullet}' \right)\right\vert \\
& \leq u_n^2 \sum_{k \in [G]} \left( \left\Vert \mathcal A_{k,\bullet}\right\Vert_F^2  + \left\Vert \mathcal B_{k,\bullet} \right\Vert_F^2 \right) \\
& \lesssim u_n^2 \sum_{h,k \in [G]^2} \left\Vert \sum_{g \in [G]} \left(\Pi_h' B_{h,g} \Gamma_g \right)  B_{k,g}  (\tilde  P_{ghk})_{g,h} \right\Vert_F^2 + o(\omega_n^4) \\
& \lesssim u_n^2 \sum_{h,k \in [G]^2} \left\Vert \sum_{g \in [G]} \left(\Pi_h' B_{h,g} \Gamma_g \right)  B_{k,g}  (\tilde  P_{ghk})_{g,h} \right\Vert_F^2 + o(\omega_n^4) \\
& \lesssim u_n^2 \sum_{h,k \in [G]^2}  n_G \left( \sum_{g \in [G]} ||B_{h,g}\Gamma_g ||_{2} ||B_{k,g}||_F ( ||P_{g,h}||_{op} + ||P_{g,k}||_{op}) \right)^2 + o(\omega_n^4) \\
& \lesssim u_n^2 \sum_{h,k \in [G]^2}  n_G  \left(\sum_{g \in [G]}  ||B_{h,g}\Gamma_g ||_{2}^2 ||B_{k,g}||_F^2\right)\left(\sum_{g \in [G]}  ( ||P_{g,h}||_{op} + ||P_{g,k}||_{op}) ^2\right)  + o(\omega_n^4) \\
& \lesssim u_n^2  n_G^2 \phi_n \lambda_n \kappa_n + o(\omega_n^4) \\
& = o(\omega_n^4),
\end{align*}
where the second inequality is by $u_n^2 \sum_{k \in [G]} \left\Vert \mathcal A_{k,\bullet}\right\Vert_F^2 = o(\omega_n^4)$ as shown above and the third inequality is by
\begin{align*}
||(\tilde  P_{ghk})_{g,h}||_{op} \lesssim ||P_{g,h}||_{op} + ||P_{g,k}||_{op}
\end{align*}
as shown in \Cref{lem:P_l3o_1} for both $h \neq k$ and $h = k$.
This concludes the case of $(s,s')=(g,h')$.

\textbf{ For the case $(s,s')=(g,k)$}, we denote
\begin{align*}
    \mathcal C_{k} = \sum_{g,h \in [G]^2} \left(\Pi_h' B_{h,g} \Gamma_g \right)  B_{k,g}  (\tilde  P_{ghk})_{g,k} \in \Re^{n_k \times n_k}.
\end{align*}
Then, we have
\begin{align*}
& \left\vert  u_n^2 \sum_{g,g',h,h',k \in [G]^5}  \left(\Pi_h' B_{h,g} \Gamma_g \right)\left(\Pi_{h'}' B_{h',g'} \Gamma_{g'} \right) tr\left(  B_{k,g} \left( (\tilde  P_{ghk})_{g,g} P_{g,k} (\tilde  P_{g'h'k})_{k,g'}  \right) B_{k,g'}'   \right) \right\vert \\
& = \left\vert u_n^2 \sum_{g,k \in [G]^2} tr \left( \mathcal A_{k,g} P_{g,k} \mathcal C_{k}' \right)\right\vert \\
& \lesssim u_n^2 \sum_{k \in [G]} \left(||\mathcal C_k||_F^2 + || \sum_{g \in [G]}\mathcal A_{k,g} P_{g,k} ||_{F}^2 \right) \\
& \lesssim u_n^2 \sum_{k \in [G]^2} ||\mathcal C_k||_F^2 + o(\omega_n^4) \\
& \lesssim u_n^2 \sum_{k \in [G]} \left\Vert \sum_{g,h \in [G]^2} \left(\Pi_h' B_{h,g} \Gamma_g \right)  B_{k,g}  (\tilde  P_{ghk})_{g,k} \right\Vert_F^2 + o(\omega_n^4)\\
& \lesssim \underbrace{ u_n^2 \sum_{k \in [G]} \left\Vert \sum_{g,h \in [G]^2, h \neq k} \left(\Pi_h' B_{h,g} \Gamma_g \right)  B_{k,g}  \Xi_{7,gk}\right\Vert_F^2 }_{R_{2,2,5}} + \underbrace{  u_n^2 \sum_{k \in [G]} \left\Vert \sum_{g,h \in [G]^2, h \neq k} \left(\Pi_h' B_{h,g} \Gamma_g \right)  B_{k,g}  \Xi_{8,ghk}\right\Vert_F^2 }_{R_{2,2,6}} \\
& + \underbrace{  u_n^2 \sum_{h \in [G]} \left\Vert \sum_{g \in [G]} \left(\Pi_h' B_{h,g} \Gamma_g \right)  B_{h,g} (\tilde P_{gh})_{g,h}  \right\Vert_F^2 }_{R_{2,2,7}} + o(\omega_n^4),
\end{align*}
where the second inequality is by
\begin{align*}
u_n^2 \sum_{k \in [G]} || \sum_{g \in [G]}\mathcal A_{k,g} P_{g,k} ||_{F}^2 &  =  u_n^2 \sum_{k \in [G]} tr(\mathcal A_{k,\bullet} P_{\bullet,k}P_{k,\bullet} \mathcal A_{k,\bullet}') \\
& \lesssim  u_n^2 \sum_{k \in [G]} tr(\mathcal A_{k,\bullet} \mathcal A_{k,\bullet}') = o(\omega_n^4).
\end{align*}

In addition, we have
\begin{align*}
R_{2,2,5} & \lesssim     u_n^2  \sum_{k \in [G]} \left\Vert \sum_{g \in [G]} \left( (H_g - B_{k,g}'\Pi_k)' \Gamma_g \right)  B_{k,g}  \Xi_{7,gk}\right\Vert_F^2 \\
& \lesssim    u_n^2 \sum_{k \in [G]} \left\Vert \sum_{g \in [G]} \left( H_g' \Gamma_g \right)  B_{k,g}  \Xi_{7,gk}\right\Vert_F^2  \\
& + u_n^2 n_G \sum_{k \in [G]} \left( \sum_{g \in [G]} ||B_{k,g}'\Pi_k||_2 || B_{k,g}||_F  ||P_{g,k}||_{op} \right)^2 \\
& \lesssim   u_n^2  \sum_{k \in [G]} \left( \sum_{g \in [G]} \left\vert H_g' \Gamma_g \right\vert  ||B_{k,g}||_F  ||P_{g,k}||_{op}\right)^2  \\
& + u_n^2 n_G \lambda_n^2 \sum_{k \in [G]} \left( \sum_{g \in [G]} ||B_{k,g}'\Pi_k||_2^2 \right) \left( \sum_{g \in [G]}|| B_{k,g}||_F^2 \right) \\
& \lesssim   u_n^2 \sum_{k \in [G]} \left( \sum_{g \in [G]} \left\vert H_g' \Gamma_g \right\vert^2  ||B_{k,g}||_F^2 \right) \left( \sum_{g \in [G]}  ||P_{g,k}||_{op}^2\right) + u_n^2 n_G^2 \lambda_n^3 \kappa_n \\
& \lesssim u_n^2 n_G \phi_n \zeta_{H,n} \kappa_n + u_n^2 n_G^2 \lambda_n^3 \kappa_n =o(\omega_n^4)
\end{align*}

\begin{align*}
R_{2,2,6} & \lesssim  u_n^2  \sum_{k \in [G]} \left( \sum_{g,h \in [G]^2} || B_{h,g}\Gamma_g||_2 ||\Pi_h||_2   ||B_{k,g}||_F  ||\Xi_{8,ghk}||_{op}\right)^2    \\
& \lesssim  u_n^2 n_G   \sum_{k \in [G]} \left( \sum_{g,h \in [G]^2} ||B_{h,g} \Gamma_g||_{2}  ||B_{k,g}||_F \begin{pmatrix}
    & ||P_{g,h}||_{op} ||P_{g,k}||_{op} \\
    &+ ||P_{h,k}||_{op}||P_{g,h}||_{op} \\
    &+||P_{h,k}||_{op}||P_{g,k}||_{op}
\end{pmatrix}  \right)^2 \\
& \lesssim  u_n^2 n_G  \sum_{k \in [G]} \left( \sum_{g,h \in [G]^2} ||B_{h,g} \Gamma_g||_{2}^2  ||B_{k,g}||_F^2 \right) \\
& \times  \left( \sum_{g,h \in [G]^2} (||P_{g,h}||_{op} ||P_{g,k}||_{op} + ||P_{h,k}||_{op}||P_{g,h}||_{op} +||P_{h,k}||_{op}||P_{g,k}||_{op}  )^2 \right) \\
& \lesssim  u_n^2 n_G^2 \phi_n^2 \lambda_n \kappa_n  =  o(\omega_n^4)
\end{align*}
and
\begin{align*}
R_{2,2,7}  & \lesssim    u_n^2 n_G^2 \phi_n \sum_{h \in [G]} \left(  \sum_{g \in [G]} || B_{h,g}||_{op} || B_{h,g}||_F ||P_{g,h}||_{op}  \right)^2   \\
& \lesssim   u_n^2 n_G^2 \phi_n^2 \lambda_n^2 \kappa_n  =  o(\omega_n^4).
\end{align*}
This concludes the case of $(s,s')=(g,k)$.

\textbf{For the case $(s,s') = (h,k)$}, we have
\begin{align*}
& \left\vert  u_n^2 \sum_{g,g',h,h',k \in [G]^5} \left(\Pi_h' B_{h,g} \Gamma_g \right)\left(\Pi_{h'}' B_{h',g'} \Gamma_{g'} \right) tr\left(  B_{k,g} \left( (\tilde  P_{ghk})_{g,h} P_{h,k} (\tilde  P_{g'h'k})_{k,g'}  \right) B_{k,g'}'   \right) \right\vert    \\
& \lesssim \left\vert u_n^2 \sum_{k,h \in [G]^2} tr(\mathcal B_{k,h} P_{h,k} \mathcal C_k'  )  \right\vert \\
& \lesssim  u_n^2 \sum_{k \in [G]} (|| \sum_{h \in [G]}\mathcal B_{k,h} P_{h,k}||_F^2  + || \mathcal C_k||_F^2  )  =  o(\omega_n^4).
\end{align*}

\textbf{For the case $(s,s') = (k,k)$}, we have
\begin{align*}
& \left\vert  u_n^2 \sum_{g,g',h,h',k \in [G]^5} \left(\Pi_h' B_{h,g} \Gamma_g \right)\left(\Pi_{h'}' B_{h',g'} \Gamma_{g'} \right) tr\left(  B_{k,g} \left( (\tilde  P_{ghk})_{g,k} P_{k,k} (\tilde  P_{g'h'k})_{k,g'}  \right) B_{k,g'}'   \right) \right\vert    \\
& \lesssim \left\vert u_n^2 \sum_{k \in [G]} tr(\mathcal C_k P_{k,k} \mathcal C_k'  )  \right\vert \\
& \lesssim  u_n^2 \lambda_n \sum_{k \in [G]} || \mathcal C_k||_F^2 =  o(\omega_n^4).
\end{align*}

Combining the above bounds, we can conclude that
\begin{align*}
    R_{2,2} =  o(\omega_n^4).
\end{align*}

Last, following the same argument for $R_{1,3}$, we can show that
\begin{align*}
    R_{2,3} =  o(\omega_n^4).
\end{align*}
This concludes the proof.

\subsection{Bounds for \texorpdfstring{$R_3 \text{ and }R_4$}{R3 and R4}}
By \Cref{lem:P_l3o_1}.6, we have
\begin{align*}
R_3 & =  \mathbb V \left( \sum_{h, l \in [G]^2,  h \neq l} \sum_{g,k \in [-l]^2} V_h' \left( B_{h,g} \Gamma_g \Pi_k ' B_{k,g} P_{g,l,-ghk}\right) U_l   \right) \\
& \lesssim \sum_{h, l \in [G]^2,  h \neq l}  u_n^2 \left\Vert \sum_{g,k \in [-l]^2} B_{h,g} \Gamma_g \Pi_k ' B_{k,g} P_{g,l,-ghk} \right\Vert_F^2 \\
& = \sum_{h, l \in [G]^2,  h \neq l} \sum_{g,g',k,k' \in [-l]^4}  u_n^2 tr\left(  B_{h,g} \Gamma_g \Pi_k ' B_{k,g} P_{g,l,-ghk}P_{l,g',-g'hk'} B_{k',g'}' \Pi_{k'} \Gamma_{g'} B_{h,g'}' \right) \\
& = \sum_{g,g',h,k,k' \in [G]^5}  u_n^2 tr\left(  B_{h,g} \Gamma_g \Pi_k ' B_{k,g} \left( \sum_{l \in [-gg'hkk']} P_{g,l,-ghk}P_{l,g',-g'hk'} \right) B_{k',g'}' \Pi_{k'} \Gamma_{g'} B_{h,g'}' \right)\\
& \lesssim \underbrace{ \left\vert \sum_{g,g',h,k,k' \in [G]^5}  u_n^2 tr\left(  B_{h,g} \Gamma_g \Pi_k ' B_{k,g} P_{g,g'}  B_{k',g'}' \Pi_{k'} \Gamma_{g'} B_{h,g'}' \right) \right\vert }_{R_{3,1}} \\
& + \underbrace{ \left\vert \sum_{g,g',h,k,k' \in [G]^5}  u_n^2 tr\left(  B_{h,g} \Gamma_g \Pi_k ' B_{k,g} \left(  (\tilde  P_{ghk})_{g,ghk} P_{ghk,g'hk'} (\tilde  P_{g'hk'})_{g'hk',g'} \right) B_{k',g'}' \Pi_{k'} \Gamma_{g'} B_{h,g'}' \right) \right\vert }_{R_{3,2}} \\
& + \underbrace{ \left\vert \sum_{g,g',h,k,k' \in [G]^5}  u_n^2 tr\left(  B_{h,g} \Gamma_g \Pi_k ' B_{k,g} \left( \sum_{l \in (g,h,k) \cap (g',h,k')} P_{g,l,-ghk} P_{l,g',-g'hk'} \right) B_{k',g'}' \Pi_{k'} \Gamma_{g'} B_{h,g'}' \right) \right\vert }_{R_{3,3}},
\end{align*}
where
\begin{align*}
    R_{3,1} & \lesssim u_n^2 \sum_{g,h \in [G]^2} ||B_{h,g} \Gamma_g H_g'||_F^2 \\
    & \lesssim u_n^2 \sum_{g,h \in [G]} tr ( B_{h,g}' B_{h,g} \Gamma_g    \Gamma_g') ||H_g||_2^2 \\
    & \lesssim u_n^2 n_G \zeta_{H,n} \kappa_n = o ((\mu_n^2 + \kappa_n + \tilde \mu_n^2)^2) = o(\omega_n^4).
\end{align*}

Next, by a slight abuse of notation, denote
\begin{align*}
& \mathcal A_{h,g} = \sum_{k \in [G]} B_{h,g} \Gamma_g \Pi_k ' B_{k,g}  (\tilde  P_{ghk})_{g,g}, \quad \mathcal B_{h,k} = \sum_{g \in [G]} B_{h,g} \Gamma_g \Pi_k ' B_{k,g}  (\tilde  P_{ghk})_{g,k}, \quad \text{and} \\
& \mathcal C_h = \sum_{g,k \in [G]^2} B_{h,g} \Gamma_g \Pi_k ' B_{k,g}  (\tilde  P_{ghk})_{g,h}.
\end{align*}

To show $R_{3,2} = o(\omega_n^4) $ following the same argument as $R_{2,2}$, we only need to show
\begin{align}
 &   u_n^2 \sum_{g,h \in [G]^2}||\mathcal A_{h,g}||_F^2 = o(\omega_n^4), \label{eq:R3A}\\
 &   u_n^2 \sum_{h,k \in [G]^2}||\mathcal B_{h,k}||_F^2 = o(\omega_n^4), \label{eq:R3B}\\
 &   u_n^2 \sum_{h \in [G]}||\mathcal C_{h}||_F^2 = o(\omega_n^4). \label{eq:R3C}
\end{align}

For \cref{eq:R3A}, we have
\begin{align*}
u_n^2 \sum_{g,h \in [G]^2}||\mathcal A_{h,g}||_F^2 & = u_n^2 \sum_{g,h \in [G]^2} \left\Vert \sum_{k \in [G]} B_{h,g} \Gamma_g \Pi_k ' B_{k,g}  (\tilde  P_{ghk})_{g,g}\right\Vert_F^2    \\
& \lesssim \underbrace{ u_n^2 \sum_{g,h \in [G]^2} \left\Vert \sum_{k \in [G]} B_{h,g} \Gamma_g \Pi_k ' B_{k,g}  \Xi_{1,g}\right\Vert_F^2  }_{R_{3,2,1}}+   \underbrace{ u_n^2 \sum_{g,h \in [G]^2} \left\Vert \sum_{k \in [-h]} B_{h,g} \Gamma_g \Pi_k ' B_{k,g}  \Xi_{2,gh}\right\Vert_F^2 }_{R_{3,2,2}} \\
& +  \underbrace{ u_n^2 \sum_{g,h \in [G]^2} \left\Vert \sum_{k \in [-h]} B_{h,g} \Gamma_g \Pi_k ' B_{k,g}  (\Xi_{3,gk} + \Xi_{4,ghk})\right\Vert_F^2 }_{R_{3,2,3}} \\
& +  \underbrace{ u_n^2 \sum_{g,h \in [G]^2} \left\Vert B_{h,g} \Gamma_g \Pi_h ' B_{h,g} (\tilde  P_{gh})_{g,h} \right\Vert_F^2  }_{R_{3,2,4}},
\end{align*}
where
\begin{align*}
    R_{3,2,1} & = u_n^2 \sum_{g,h \in [G]^2} \left\Vert B_{h,g} \Gamma_g H_g'  \Xi_{1,g}\right\Vert_F^2 \\
    & \lesssim u_n^2 n_G \zeta_{H,n} \lambda_n^2 \kappa_n = o(\omega_n^4),
\end{align*}
\begin{align*}
   R_{3,2,2} & = u_n^2 \sum_{g,h \in [G]^2} \left\Vert B_{h,g} \Gamma_g (H_g - B_{h,g}'\Pi_h)'  \Xi_{2,gh}\right\Vert_F^2 \\
   & \lesssim u_n^2 \sum_{g,h \in [G]^2} \left\Vert B_{h,g} \Gamma_g  \Pi_h' B_{h,g}  \Xi_{2,gh}\right\Vert_F^2 + o(\omega_n^4) \\
   & \lesssim u_n^2 n_G^2 \lambda_n^4 \kappa_n + o(\omega_n^4)  = o(\omega_n^4),
\end{align*}
\begin{align*}
 R_{3,2,3} & \lesssim u_n^2 \sum_{g,h \in [G]^2} \left( \sum_{k \in [G]} ||B_{h,g} \Gamma_g||_2 ||\Pi_k ' B_{k,g}||_2  ||\Xi_{3,gk} + \Xi_{4,ghk}||_{op} \right)^2 \\
 & \lesssim u_n^2 n_G \lambda_n^2 \sum_{g,h \in [G]^2} ||B_{h,g} \Gamma_g||_2^2 \left( \sum_{k \in [G]} ||B_{k,g}||_{op}^2 \right)\left( \sum_{k \in [G]} (||P_{k,g}||_{op}^2 + ||P_{k,h}||_{op}^2) \right) \\
 & \lesssim u_n^2 n_G^2  \phi_n(\phi_n + \lambda_n^2) \lambda_n^2  \kappa_n = o(\omega_n^4),
\end{align*}
and
\begin{align*}
R_{3,2,4} & \lesssim  u_n^2 \sum_{g,h \in [G]^2} \left\Vert B_{h,g} \Gamma_g \right\Vert_2^2 \left\Vert \Pi_h ' B_{h,g} \right\Vert_2^2 \left\Vert(\tilde  P_{gh})_{g,h} \right\Vert_{op}^2  \\
& \lesssim u_n^2 n_G^2 \lambda_n^4 \kappa_n = o(\omega_n^4).
\end{align*}
This establishes \cref{eq:R3A}.

For \cref{eq:R3B}, following the same argument as $R_2$ with case $(s,s') = (g,h')$, we have
\begin{align*}
u_n^2 \sum_{h,k \in [G]^2}||\mathcal B_{h,k}||_F^2 &  =  u_n^2 \sum_{h,k \in [G]^2} \left\Vert   \sum_{g \in [G]} B_{h,g} \Gamma_g \Pi_k ' B_{k,g}  (\tilde  P_{ghk})_{g,k} \right\Vert_F^2   \\
& \lesssim u_n^2 n_G \sum_{h,k \in [G]^2} \left(   \sum_{g \in [G]} ||B_{h,g}||_{op} || \Pi_k 'B_{k,g}||_{2} (||P_{g,k}||_{op}  + ||P_{g,h}||_{op} ) \right)^2  = o(\omega_n^4).
\end{align*}

For \cref{eq:R3C}, we have
\begin{align*}
& u_n^2 \sum_{h \in [G]}||\mathcal C_{h}||_F^2 \\
& = u_n^2 \sum_{h \in [G]}\left\Vert \sum_{g,k \in [G]^2} B_{h,g} \Gamma_g \Pi_k ' B_{k,g}  (\tilde  P_{ghk})_{g,h} \right\Vert_F^2    \\
& \lesssim \underbrace{ u_n^2 \sum_{h \in [G]}\left\Vert \sum_{g,k \in [G]^2, k \neq h} B_{h,g} \Gamma_g \Pi_k ' B_{k,g}  \Xi_{5,gh} \right\Vert_F^2 }_{R_{3,2,5}}  \\
& +  \underbrace{  u_n^2 \sum_{h \in [G]}\left\Vert \sum_{g,k \in [G]^2, k \neq h} B_{h,g} \Gamma_g \Pi_k ' B_{k,g}  \Xi_{6,ghk} \right\Vert_F^2 }_{R_{3,2,6}}  \\
& +  \underbrace{  u_n^2 \sum_{h \in [G]}\left\Vert \sum_{g \in [G]} B_{h,g} \Gamma_g \Pi_h ' B_{h,g}   (\tilde  P_{gh})_{g,h}  \right\Vert_F^2 }_{R_{3,2,7}},
\end{align*}
where
\begin{align*}
R_{3,2,5} & \lesssim  u_n^2 \sum_{h \in [G]}\left\Vert \sum_{g \in [G]} B_{h,g} \Gamma_g (H_g - B_{h,g}'\Pi_h)'   \Xi_{5,gh} \right\Vert_F^2    \\
& \lesssim u_n^2  (\zeta_{H,n} + n_G\lambda_n^2) \sum_{h \in [G]}\left( \sum_{g \in [G]} ||B_{h,g}\Gamma_g||_{2}  ||P_{g,h}||_{op} \right)^2 \\
& \lesssim u_n^2 \phi_n (\zeta_{H,n} + n_G\lambda_n^2) (\sum_{g,h \in [G]^2}  ||B_{h,g}\Gamma_g ||_2^2) \\
& \lesssim u_n^2 n_G \phi_n (\zeta_{H,n} + n_G\lambda_n^2) \kappa_n = o(\omega_n^4),
\end{align*}
\begin{align*}
R_{3,2,6} & \leq      u_n^2 \sum_{h \in [G]}\left\Vert \sum_{g,k \in [G]^2} ||B_{h,g} \Gamma_g||_2 || \Pi_k ' B_{k,g}||_2  \Xi_{6,ghk} \right\Vert_F^2 \\
& \lesssim u_n^2 \sum_{h \in [G]}\left( \sum_{g,k \in [G]^2} ||B_{h,g} \Gamma_g||_{2}  ||\Pi_k ' B_{k,g}||_{2}    \begin{pmatrix}
    & ||P_{g,h}||_{op} ||P_{g,k}||_{op} \\
    &+ ||P_{h,k}||_{op}||P_{g,h}||_{op} \\
    &+||P_{h,k}||_{op}||P_{g,k}||_{op}
\end{pmatrix} \right)^2 \\
& \lesssim u_n^2 \phi_n^2 \sum_{g,h,k \in [G]^3} ||B_{h,g} \Gamma_g||_{2}^2  || \Pi_k 'B_{k,g}||_{op}^2 \\
& \lesssim u_n^2 n_G \phi_n^2  \lambda_n \sum_{g,k \in [G]^2} ||B_{k,g}' \Pi_k||_{2}^2  \\
& \lesssim u_n^2 n_G^2 \phi_n^2 \lambda_n \kappa_n = o(\omega_n^4),
\end{align*}
and
\begin{align*}
R_{3,2,7} & \lesssim    u_n^2 n_G  \lambda_n^2 \sum_{h \in [G]}\left( \sum_{g \in [G]} ||B_{h,g} \Gamma_g||_2    ||P_{g,h}||_{op}  \right)^2 \\
& \lesssim  u_n^2 n_G^2 \phi_n  \lambda_n^2 \kappa_n =  o(\omega_n^4).
\end{align*}
This leads to \cref{eq:R3C}, which further implies $R_{3,2} = o(\omega_n^4) $.

We can show $R_{3,3} = o(\omega_n^4) $ following the same argument for $R_{1,3}$, which concludes that $R_3 = o(\omega_n^4) $. We can show $R_4 = o(\omega_n^4) $ in the same manner as $R_2$ and $R_3$.

\subsection{Bounds for \texorpdfstring{$R_5 \text{ and }R_7$}{R5 and R7}}
By \Cref{lem:P_l3o_1}.6, we have
\begin{align*}
R_5 & =     \mathbb V \left( \sum_{g, h \in [G]^2, l \in [-gh]}\sum_{k \in [-l]} \left(V_h' B_{h,g} U_g \right) \left( \Pi_k ' B_{k,g} P_{g,l,-ghk} U_l \right)  \right) \\
& \lesssim  \sum_{g, h \in [G]^2, l \in [-gh]} \mathbb V \left(  \left(V_h' B_{h,g} U_g \right) \left( \sum_{k \in [-l]} \Pi_k ' B_{k,g} P_{g,l,-ghk}\right) U_l   \right) \\
& \lesssim  \sum_{g, h \in [G]^2, l \in [-gh]} \sum_{k,k' \in [-l]^2} u_n^3 ||B_{h,g}||_F^2 \left( \Pi_k ' B_{k,g} P_{g,l,-ghk} P_{l,g,-ghk'} B_{k',g}' \Pi_{k'} \right) \\
& =  \sum_{g, h,k,k' \in [G]^4} u_n^3 ||B_{h,g}||_F^2 \left( \Pi_k ' B_{k,g} \left( \sum_{l \in [-ghkk']} P_{g,l,-ghk} P_{l,g,-ghk'} \right) B_{k',g}' \Pi_{k'} \right) \\
& \lesssim \underbrace{ \left\vert  \sum_{g, h,k,k' \in [G]^4} u_n^3 ||B_{h,g}||_F^2 \left( \Pi_k ' B_{k,g} P_{g,g} B_{k',g}' \Pi_{k'} \right)\right\vert }_{R_{5,1}} \\
& + \underbrace{  \left\vert  \sum_{g, h,k,k' \in [G]^4} u_n^3 ||B_{h,g}||_F^2 \left( \Pi_k ' B_{k,g} (\tilde  P_{ghk})_{g,ghk} P_{ghk,ghk'} (\tilde  P_{ghk'})_{ghk',g} B_{k',g}' \Pi_{k'} \right)\right\vert }_{R_{5,2}} \\
& + \underbrace{  \left\vert  \sum_{g, h,k,k' \in [G]^4} u_n^3 ||B_{h,g}||_F^2 \left( \Pi_k ' B_{k,g} \left( \sum_{l \in (g,h,k) \cap (g,h,k')} P_{g,l,-ghk} P_{l,g,-ghk'} \right) B_{k',g}' \Pi_{k'} \right)\right\vert }_{R_{5,3}},
\end{align*}
where
\begin{align*}
R_{5,1} & =     \left\vert  \sum_{g, h \in [G]^4} u_n^3 ||B_{h,g}||_F^2 \left( H_g' P_{g,g} H_{g} \right)\right\vert \\
& \lesssim u_n^3 \zeta_{H,n} \lambda_n \kappa_n = o(\omega_n^4).
\end{align*}

Following the arguments for $R_{2,2}$ and $R_{3,2}$, in order to bound $R_{5,2}$, it suffices to show
\begin{align*}
 \left\vert  \sum_{g, h,k,k' \in [G]^4} u_n^3 ||B_{h,g}||_F^2 \left( \Pi_k ' B_{k,g} (\tilde  P_{ghk})_{g,s} P_{s,s'} (\tilde  P_{ghk'})_{s',g} B_{k',g}' \Pi_{k'} \right)\right\vert = o(\omega_n^4)
\end{align*}
for $(s,s') = (g,g)$, $(s,s') = (g,h)$, $(s,s') = (g,k')$, and $(s,s') = (k,k')$. The other cases can be bounded in the same manner.

\textbf{For case $(s,s') = (g,g)$}, we have
\begin{align*}
& \left\vert  \sum_{g, h,k,k' \in [G]^4} u_n^3 ||B_{h,g}||_F^2 \left( \Pi_k ' B_{k,g} (\tilde  P_{ghk})_{g,g} P_{g,g} (\tilde  P_{ghk'})_{g,g} B_{k',g}' \Pi_{k'} \right)\right\vert \\
& \lesssim \sum_{g,h} u_n^3 ||B_{h,g}||_F^2 \left\Vert \sum_{k \in [G]} \Pi_k ' B_{k,g} (\tilde  P_{ghk})_{g,g} \right\Vert_2^2 \\
& \lesssim \underbrace{ \sum_{g,h} u_n^3 ||B_{h,g}||_F^2 \left\Vert \sum_{k \in [G]} \Pi_k ' B_{k,g} \Xi_{1,g} \right\Vert_2^2 }_{R_{5,2,1}} + \underbrace{  \sum_{g,h} u_n^3 ||B_{h,g}||_F^2 \left\Vert \sum_{k \in [-h]} \Pi_k ' B_{k,g} \Xi_{2,gh} \right\Vert_2^2 }_{R_{5,2,2}} \\
& + \underbrace{  \sum_{g,h} u_n^3 ||B_{h,g}||_F^2 \left\Vert \sum_{k \in [-h]} \Pi_k ' B_{k,g} (\Xi_{3,gk} + \Xi_{4,ghk}) \right\Vert_2^2 }_{R_{5,2,3}} +  \underbrace{  \sum_{g,h} u_n^3 ||B_{h,g}||_F^2 \left\Vert \Pi_h ' B_{h,g} \Xi_{9,gh} \right\Vert_2^2 }_{R_{5,2,4}},
\end{align*}
where
\begin{align*}
    R_{5,2,1} & = \sum_{g,h} u_n^3 ||B_{h,g}||_F^2 \left\Vert  H_g' \Xi_{1,g} \right\Vert_2^2 \lesssim u_n^3 \zeta_{H,n} \lambda_n^2 \kappa_n =  o(\omega_n^4),
\end{align*}
\begin{align*}
   R_{5,2,2} &  =  \sum_{g,h} u_n^3 ||B_{h,g}||_F^2 \left\Vert  (H_g - B_{h,g}' \Pi_h)' \Xi_{2,g} \right\Vert_2^2 \lesssim   u_n^3 n_G \lambda_n^4 \kappa_n   + o(\omega_n^4) = o(\omega_n^4),
\end{align*}
\begin{align*}
R_{5,2,3} &  \lesssim   \sum_{g,h} u_n^3 ||B_{h,g}||_F^2 \left( \sum_{k \in [G]} ||\Pi_k ' B_{k,g}||_2 ||\Xi_{3,gk} + \Xi_{4,ghk}||_{op} \right )^2  \\
& \lesssim   \sum_{g,h} u_n^3 n_G \lambda_n^2 ||B_{h,g}||_F^2 \left( \sum_{k \in [G]} || B_{k,g}||_{op}^2 \right) \left( \sum_{k \in [G]} ||P_{g,k}||_{op}^2 + ||P_{h,k}||_{op}^2 \right )   \\
& \lesssim u_n^3 n_G \phi_n (\phi_n + \lambda_n^2) \lambda_n^2 \kappa_n = o(\omega_n^4),
\end{align*}
and
\begin{align*}
R_{5,2,4} & \lesssim u_n^3 n_G \lambda_n^4 \kappa_n = o(\omega_n^4).
\end{align*}

\textbf{For case $(s,s') = (g,h)$}, we have
\begin{align*}
   & \left\vert  \sum_{g, h,k,k' \in [G]^4} u_n^3 ||B_{h,g}||_F^2 \left( \Pi_k ' B_{k,g} (\tilde  P_{ghk})_{g,g} P_{g,h} (\tilde  P_{ghk'})_{h,g} B_{k',g}' \Pi_{k'} \right)\right\vert \\
   & \lesssim  \sum_{g, h \in [G]^2} u_n^3 ||B_{h,g}||_F^2 \left\Vert \sum_{k' \in [G]} (\tilde  P_{ghk'})_{h,g} B_{k',g}' \Pi_{k'} \right\Vert_2^2 + o(\omega_n^4) \\
   & \lesssim \underbrace{ \sum_{g, h \in [G]^2} u_n^3 ||B_{h,g}||_F^2 \left\Vert \sum_{k' \in [-h]} \Xi_{5,gh}' B_{k',g}' \Pi_{k'} \right\Vert_2^2 }_{R_{5,2,5}}+  \underbrace{  \sum_{g, h \in [G]^2} u_n^3 ||B_{h,g}||_F^2 \left\Vert \sum_{k' \in [-h]} \Xi_{6,ghk'}' B_{k',g}' \Pi_{k'} \right\Vert_2^2 }_{R_{5,2,6}} \\
   & +  \underbrace{  \sum_{g, h \in [G]^2} u_n^3 ||B_{h,g}||_F^2 \left\Vert (\tilde  P_{gh})_{h,g} B_{h,g}' \Pi_{h} \right\Vert_2^2 }_{R_{5,2,7}}  + o(\omega_n^4),
\end{align*}
where
\begin{align*}
R_{5,2,5} & =    \sum_{g, h \in [G]^2} u_n^3 ||B_{h,g}||_F^2 \left\Vert  \Xi_{5,gh}' (H_g - B_{h,g}' \Pi_{h}) \right\Vert_2^2 \lesssim u_n^3 (\zeta_{H,n} + n_G \lambda_n^2) \lambda_n \kappa_n = o(\omega_n^4),
\end{align*}
\begin{align*}
R_{5,2,6} & \lesssim \sum_{g, h \in [G]^2} u_n^3 n_G \lambda_n^2 ||B_{h,g}||_F^2 \left( \sum_{k' \in [G]} (||P_{g,k'}||_{op} +||P_{h,k'}||_{op})  ||B_{k',g}||_{op} \right)^2 \\
& \lesssim  \sum_{g, h \in [G]^2} u_n^3 n_G \lambda_n^2 ||B_{h,g}||_F^2 \left( \sum_{k' \in [G]} (||P_{g,k'}||_{op}^2 +||P_{h,k'}||_{op}^2) \right) \left(  \sum_{k' \in [G]} ||B_{k',g}||_{op}^2 \right) \\
& \lesssim u_n^3 n_G \phi_n (\phi_n + \lambda_n^2) \lambda_n^2 \kappa_n = o(\omega_n^4),
\end{align*}
and
\begin{align*}
R_{5,2,7} & \lesssim     u_n^3 n_G \lambda_n^4 \kappa_n  = o(\omega_n^4).
\end{align*}

\textbf{For case $(s,s') = (g,k')$}, we have
\begin{align*}
   & \left\vert  \sum_{g, h,k,k' \in [G]^4} u_n^3 ||B_{h,g}||_F^2 \left( \Pi_k ' B_{k,g} (\tilde  P_{ghk})_{g,g} P_{g,k'} (\tilde  P_{ghk'})_{k',g} B_{k',g}' \Pi_{k'} \right)\right\vert \\
   & \lesssim  \sum_{g, h \in [G]^2} u_n^3 ||B_{h,g}||_F^2 \left\Vert \sum_{k' \in [G]}  P_{g,k'} (\tilde  P_{ghk'})_{k',g} B_{k',g}' \Pi_{k'} \right\Vert_2^2 + o(\omega_n^4) \\
   & \lesssim  \sum_{g, h, k' \in [G]^3} u_n^3 ||B_{h,g}||_F^2 \left\Vert (\tilde  P_{ghk'})_{k',g} B_{k',g}' \Pi_{k'} \right\Vert_2^2 + o(\omega_n^4) \\
   & \lesssim  \sum_{g, h, k' \in [G]^3} u_n^3 n_G \lambda_n^2 ||B_{h,g}||_F^2 \left\Vert B_{k',g}' \right\Vert_{op}^2 + o(\omega_n^4) \\
   & \lesssim  u_n^3 n_G (\phi_n + \lambda_n^2)\lambda_n^2 \kappa_n +  o(\omega_n^4) =  o(\omega_n^4).
\end{align*}

\textbf{For case $(s,s') = (k,k')$}, we have
\begin{align*}
 & \left\vert  \sum_{g, h,k,k' \in [G]^4} u_n^3 ||B_{h,g}||_F^2 \left( \Pi_k ' B_{k,g} (\tilde  P_{ghk})_{g,k} P_{k,k'} (\tilde  P_{ghk'})_{k',g} B_{k',g}' \Pi_{k'} \right)\right\vert \\
 & \lesssim  \sum_{g, h,k \in [G]^3} u_n^3 ||B_{h,g}||_F^2 \left\Vert \Pi_k ' B_{k,g} (\tilde  P_{ghk})_{g,k} \right\Vert_2^2   = o(\omega_n^4).
\end{align*}
This leads to the desired result that $R_{5,2} = o(\omega_n^4)$.

We can show $R_{5,3} = o(\omega_n^4) $ following the same argument for $R_{1,3}$, which concludes that $R_5 = o(\omega_n^4) $. We can show $R_5 = o(\omega_n^4) $ in the similar manner.

\subsection{Bound for \texorpdfstring{$R_6$}{R6}}
We have
\begin{align*}
R_6 & \lesssim  \mathbb V \left( \sum_{h, k \in [G]^2, h \neq k, l \in [-hk]} \sum_{g \in [-l]} \left(V_h' B_{h,g} \Gamma_g \right) \left(V_k ' B_{k,g} P_{g,l,-ghk} U_l \right)  \right) \\
& +  \mathbb V \left( \sum_{h,l \in [G]^2, h \neq l} \sum_{g \in [-l]}  \left(
  \Gamma_g'B_{h,g}' V_{h} V_h' B_{h,g} P_{g,l,-gh} U_l \right)  \right) \\
& \lesssim \sum_{h, k \in [G]^2, h \neq k, l \in [-hk]}  \mathbb V \left( \sum_{g \in [-l]} \left(V_h' B_{h,g} \Gamma_g \right) \left(V_k ' B_{k,g} P_{g,l,-ghk} U_l \right)  \right) \\
& + \sum_{h,l \in [G]^2, h \neq l} \mathbb V \left(  \sum_{g \in [-l]}  \left(
  \Gamma_g'B_{h,g}' (V_{h} V_h'- \Omega_{V,h}) B_{h,g} P_{g,l,-gh} U_l \right)  \right) \\
& + \sum_{l \in [G]}  \mathbb V \left( \sum_{g,h \in [-l]^2}  \left( \Gamma_g'B_{h,g}' \Omega_{V,h} B_{h,g} P_{g,l,-gh} U_l \right)  \right) \\
& \lesssim \sum_{h, k \in [G]^2, h \neq k, l \in [-hk]} u_n^2 \mathbb E \left( \left\Vert\sum_{g \in [-l]}  B_{h,g} \Gamma_g V_k ' B_{k,g} P_{g,l,-ghk} \right\Vert_F^2   \right) \\
& + \sum_{h,l \in [G]^2, h \neq l} u_n \mathbb E \left\Vert  \sum_{g \in [-l]}
  \left( \Gamma_g'B_{h,g}' (V_{h} V_h'- \Omega_{V,h}) B_{h,g} P_{g,l,-gh} \right)  \right\Vert_2^2  \\
& + \sum_{l \in [G]} u_n \left\Vert \sum_{g,h \in [-l]^2}  \left( \Gamma_g'B_{h,g}' \Omega_{V,h} B_{h,g} P_{g,l,-gh} \right)  \right\Vert_2^2 \\
& \lesssim \sum_{h, k \in [G]^2, h \neq k, l \in [-hk]} u_n^2 \left(\sum_{g,g' \in [-l]^2} \left( \Gamma_{g'}' B_{h,g'}' B_{h,g} \Gamma_g \right) tr\left(   B_{k,g} P_{g,l,-ghk} P_{l,g',-g'hk} B_{k,g'}' \Omega_{V,k} \right)   \right) \\
& + \sum_{h,l \in [G]^2, h \neq l} u_n \mathbb E  \sum_{g,g' \in [-l]^2}  \left(
  \Gamma_g'B_{h,g}' V_{h} V_h' B_{h,g} P_{g,l,-gh}P_{l,g',-g'h}B_{h,g'}' V_{h} V_h'B_{h,g'} \Gamma_{g'}  \right)  \\
& + \sum_{l \in [G]} u_n  \sum_{g,g',h,h' \in [-l]^4}  \left( \Gamma_g'B_{h,g}' \Omega_{V,h} B_{h,g} P_{g,l,-gh}P_{l,g',-g'h'} B_{h',g'}' \Omega_{V,h'}B_{h',g'} \Gamma_{g'} \right)   \\
& \lesssim \underbrace{ \sum_{g,g',h, k \in [G]^4, h \neq k} u_n^2 \left( \Gamma_{g'}' B_{h,g'}' B_{h,g} \Gamma_g \right) tr\left(   B_{k,g} \left( \sum_{l \in [-gg'hk]} P_{g,l,-ghk} P_{l,g',-g'hk} \right)  B_{k,g'}' \Omega_{V,k} \right) }_{R_{6,1}} \\
& + \underbrace{  \sum_{g,g',h \in [G]^3} u_n \mathbb E  \left(
  \Gamma_g'B_{h,g}' V_{h} V_h' B_{h,g} \left( \sum_{l \in [-gg'h]}
  P_{g,l,-gh}P_{l,g',-g'h} \right) B_{h,g'}' V_{h} V_h'B_{h,g'} \Gamma_{g'}  \right) }_{R_{6,2}}  \\
& +  \underbrace{  u_n  \sum_{g,g',h,h' \in [G]^4}  \left( \Gamma_g'B_{h,g}' \Omega_{V,h} B_{h,g} \left( \sum_{l \in [-gg'hh']} P_{g,l,-gh}P_{l,g',-g'h'} \right) B_{h',g'}' \Omega_{V,h'}B_{h',g'} \Gamma_{g'} \right) }_{R_{6,3}}.
\end{align*}

For $R_{6,1}$, we have
\begin{align*}
R_{6,1} & \lesssim \underbrace{ \left\vert \sum_{g,g',h, k \in [G]^4, h \neq k} u_n^2 \left( \Gamma_{g'}' B_{h,g'}' B_{h,g} \Gamma_g \right) tr\left(   B_{k,g} P_{g,g'}  B_{k,g'}' \Omega_{V,k} \right)  \right\vert }_{R_{6,1,1}} \\
& + \underbrace{  \left\vert \sum_{g,g',h, k \in [G]^4, h \neq k} u_n^2 \left( \Gamma_{g'}' B_{h,g'}' B_{h,g} \Gamma_g \right) tr\left(   B_{k,g} \left( (\tilde  P_{ghk})_{g,ghk} P_{ghk,g'hk} (\tilde  P_{g'hk})_{g'hk,g'}  \right)  B_{k,g'}' \Omega_{V,k} \right)  \right\vert }_{R_{6,1,2}} \\
& + \underbrace{  \left\vert \sum_{g,g',h, k \in [G]^4, h \neq k} u_n^2 \left( \Gamma_{g'}' B_{h,g'}' B_{h,g} \Gamma_g \right) tr\left(   B_{k,g} \left( \sum_{l \in (g,h,k) \cap (g',h,k)} P_{g,l,-ghk} P_{l,g',-g'hk} \right)  B_{k,g'}' \Omega_{V,k} \right)  \right\vert }_{R_{6,1,3}}.
\end{align*}

Denote $\mathcal P \in \Re^{n\times n}$ with its $(g,g')$ block defined as $\mathcal P_{g,g'} = \left( \Gamma_{g'}' (B'B)_{g',g} \Gamma_g \right) P_{g,g'}$. Then, we have
\begin{align*}
    ||\mathcal P||_{op} & = \sup_{a \in \Re^n, ||a||_2 = 1}  \left\vert a'\mathcal P a \right\vert \\
    & = \sup_{a \in \Re^n, ||a||_2 = 1}  \left\vert \sum_{g,g' \in [G]^2}  \left( \Gamma_{g'}' (B'B)_{g',g} \Gamma_g \right) a_g' P_{g,g'} a_{g'} \right\vert \\
    & \leq \sup_{a \in \Re^n, ||a||_2 = 1} n_G^{1/2} \left(\sum_{g,g'\in [G]^2} ||(B'B)_{g',g} \Gamma_g||_2^2 ||a_g||_2^2\right)^{1/2}\left(\sum_{g,g'\in [G]^2} ||P_{g,g'} a_{g'}||_2^2\right)^{1/2}    \\
    & \leq \sup_{a \in \Re^n, ||a||_2 = 1} n_G^{1/2} \left(\sum_{g \in [G]} \Gamma_g'(B'B B'B)_{g,g} \Gamma_g ||a_g||_2^2\right)^{1/2}\left(\sum_{g'\in [G]} a_{g'}' P_{g',g'} a_{g'}\right)^{1/2}    \\
    & \leq n_G \lambda_n.
\end{align*}

Therefore, we have
\begin{align*}
R_{6,1,1} & \lesssim \left\vert \sum_{g,g' \in [G]^2} u_n^2 \left( \Gamma_{g'}' (B'B)_{g',g} \Gamma_g \right) tr\left( P_{g,g'} (B'\Omega_V B)_{g',g} \right)  \right\vert \\
& + \left\vert \sum_{g,g', k \in [G]^3} u_n^2 \left( \Gamma_{g'}' B_{k,g'}' B_{k,g} \Gamma_g \right) tr\left(   B_{k,g} P_{g,g'}  B_{k,g'}' \Omega_{V,k} \right)  \right\vert\\
& \lesssim  \left\vert  u_n^2   tr\left( \mathcal P (B'\Omega_V B) \right) \right\vert \\
& + u_n^3 n_G \lambda_n  \sum_{g,g', k \in [G]^3} ||B_{k,g}||_{op}||B_{k,g'}||_{op}  ||B_{k,g'}|_F ||B_{k,g}  ||_F\\
& \lesssim u_n^2 n_G \lambda_n  tr(B'\Omega_V B) +  u_n^3 n_G \lambda_n \left(\sum_{g,g', k \in [G]^3}||B_{k,g}||_{op}^2 ||B_{k,g'}||_F^2 \right) \\
& \lesssim u_n^3 n_G \lambda_n \kappa_n + u_n^3 n_G \lambda_n (\phi_n + \lambda_n^2)\kappa_n \\
& = o ((\mu_n^2 + \kappa_n + \tilde \mu_n^2)^2) = o(\omega_n^4).
\end{align*}


To bound $R_{6,1,2}$, we only need to bound
\begin{align*}
  \left\vert \sum_{g,g',h, k \in [G]^4, h \neq k} u_n^2 \left( \Gamma_{g'}' B_{h,g'}' B_{h,g} \Gamma_g \right) tr\left(   B_{k,g} \left( (\tilde  P_{ghk})_{g,s} P_{s,s'} (\tilde  P_{g'hk})_{s',g'}  \right)  B_{k,g'}' \Omega_{V,k} \right)  \right\vert
\end{align*}
for $(s,s') = (g,g')$, $(s,s') = (g,h)$, $(s,s') = (g,k)$, $(s,s') = (h,h)$, $(s,s') = (h,k)$, and $(s,s') = (k,k)$. The other cases can be bounded in the same manner.

\textbf{For case $(s,s') = (g,g')$}, by \Cref{lem:P_l3o_1}, we have
\begin{align*}
& (\tilde  P_{ghk})_{g,g} P_{g,g'} (\tilde  P_{g'hk})_{g',g'}  \\
& = \underbrace{ \Xi_{1,g}' P_{g,g'} \Xi_{1,g'} }_{T_{1,gg'}}\\
& + \underbrace{ \Xi_{1,g}' P_{g,g'} \Xi_{2,g'h} + \Xi_{2,gh}'P_{g,g'}\Xi_{1,g'} + \Xi_{2,gh}'P_{g,g'}\Xi_{2,g'h}  }_{T_{2,gg'h}} \\
& + \underbrace{ \Xi_{1,g}' P_{g,g'} \Xi_{3,g'k} + \Xi_{3,gk}'P_{g,g'}\Xi_{1,g'} + \Xi_{3,gk}'P_{g,g'} \Xi_{3,g'k}  }_{T_{3,gg'k}} \\
& + \underbrace{  \begin{pmatrix}
    & \Xi_{1,g}' P_{g,g'} \Xi_{4,g'hk} + \Xi_{2,gh}'P_{g,g'}(\Xi_{3,g'k} + \Xi_{4,g'hk}) \\
    & +  \Xi_{3,gk}'P_{g,g'} (\Xi_{2,g'h} + \Xi_{4,g'hk}) + \Xi_{4,ghk}P_{g,g'}  (\tilde  P_{g'hk})_{g',g'}
\end{pmatrix}  }_{T_{4,gg'hk}},
\end{align*}
where
\begin{align*}
    & ||T_{1,gg'}||_F \lesssim ||\Xi_{1,g}||_{op} ||P_{g,g'}||_F ||\Xi_{1,g'}||_{op} \lesssim \lambda_n^2 ||P_{g,g'}||_F, \\
    & ||T_{2,gg'h}||_{op} \lesssim \lambda_n ||P_{g,g'}||_{op} ( ||P_{g,h}||_{op} + ||P_{g',h}||_{op}), \\
    & ||T_{3,gg'k}||_{op} \lesssim \lambda_n ||P_{g,g'}||_{op} ( ||P_{g,k}||_{op} + ||P_{g',k}||_{op}), \\
    & ||T_{4,gg'hk}||_{op} \lesssim \lambda_n^2 ||P_{g,g'}||_{op} (
        ||P_{g,h}||_{op} + ||P_{g,k}||_{op}) (  ||P_{g',h}||_{op}+||P_{g',k}||_{op}
    ).
\end{align*}
Then, we have
\begin{align*}
& \left\vert \sum_{g,g',h, k \in [G]^4, h \neq k} u_n^2 \left( \Gamma_{g'}' B_{h,g'}' B_{h,g} \Gamma_g \right) tr\left(   B_{k,g} \left( (\tilde  P_{ghk})_{g,g} P_{g,g'} (\tilde  P_{g'hk})_{g',g'}  \right)  B_{k,g'}' \Omega_{V,k} \right)  \right\vert    \\
& \lesssim \underbrace{ \left\vert \sum_{g,g',h, k \in [G]^4, h \neq k} u_n^2 \left( \Gamma_{g'}' B_{h,g'}' B_{h,g} \Gamma_g \right) tr\left(   B_{k,g} T_{1,gg'}  B_{k,g'}' \Omega_{V,k} \right)  \right\vert }_{R_{6,1,2,1}}  \\
& + \underbrace{  \left\vert \sum_{g,g',h, k \in [G]^4, h \neq k} u_n^2 \left( \Gamma_{g'}' B_{h,g'}' B_{h,g} \Gamma_g \right) tr\left(   B_{k,g} T_{2,gg'h}  B_{k,g'}' \Omega_{V,k} \right)  \right\vert }_{R_{6,1,2,2}} \\
& + \underbrace{  \left\vert \sum_{g,g',h, k \in [G]^4, h \neq k} u_n^2 \left( \Gamma_{g'}' B_{h,g'}' B_{h,g} \Gamma_g \right) tr\left(   B_{k,g} T_{3,gg'k}  B_{k,g'}' \Omega_{V,k} \right)  \right\vert }_{R_{6,1,2,3}} \\
& + \underbrace{  \left\vert \sum_{g,g',h, k \in [G]^4, h \neq k} u_n^2 \left( \Gamma_{g'}' B_{h,g'}' B_{h,g} \Gamma_g \right) tr\left(   B_{k,g} T_{4,gg'hk}  B_{k,g'}' \Omega_{V,k} \right)  \right\vert }_{R_{6,1,2,4}},
\end{align*}
where
\begin{align*}
R_{6,1,2,1}  = o(\omega_n^4)
\end{align*}
following the same argument as $R_{6,1,1}$,
\begin{align*}
R_{6,1,2,2} & \lesssim  \sum_{g,g',h, k \in [G]^4} u_n^3 n_G \lambda_n  ||B_{h,g'}||_{op} ||B_{h,g}||_{op}  ||  B_{k,g}||_F  ||P_{g,g'}||_{op} ( ||P_{g,h}||_{op} + ||P_{g',h}||_{op}) ||  B_{k,g'}||_F \\
& \lesssim  u_n^3 n_G \lambda_n \left( \sum_{g,g',h, k \in [G]^4}  ||B_{h,g'}||_{op}^2  ||  B_{k,g}||_F^2 ||P_{g,h}||_{op}^2\right)^{1/2}\left( \sum_{g,g',h, k \in [G]^4}  ||B_{h,g}||_{op}^2  ||  B_{k,g'}||_F^2 ||P_{g,g'}||_{op}^2\right)^{1/2} \\
& + u_n^3 n_G \lambda_n \left( \sum_{g,g',h, k \in [G]^4}  ||B_{h,g'}||_{op}^2  ||  B_{k,g}||_F^2 ||P_{g',g}||_{op}^2\right)^{1/2}\left( \sum_{g,g',h, k \in [G]^4}  ||B_{h,g}||_{op}^2  ||  B_{k,g'}||_F^2 ||P_{g',h}||_{op}^2\right)^{1/2}\\
& \lesssim u_n^3 n_G \phi_n(\phi_n + \lambda_n^2)\lambda_n  \kappa_n = o(\omega_n^4),
\end{align*}
and
\begin{align*}
    R_{6,1,2,3} = o(\omega_n^4), \quad \text{and} \quad R_{6,1,2,4} = o(\omega_n^4)
\end{align*}
following the same argument as $R_{6,1,2,2}$. This implies
\begin{align*}
R_{6,1,2} = o(\omega_n^4).
\end{align*}
we can show $R_{6,1,3} = o(\omega_n^4) $ following the same argument for $R_{1,3}$, which concludes that $R_{6,1} = o(\omega_n^4) $. We can show $R_{6,2} = o(\omega_n^4) $ in the same manner as $R_{6,1}$.

Last, we have
\begin{align*}
R_{6,3} & \lesssim  \underbrace{ \left\vert u_n  \sum_{g,g',h,h' \in [G]^4}  \left( \Gamma_g'B_{h,g}' \Omega_{V,h} B_{h,g} P_{g,g'}  B_{h',g'}' \Omega_{V,h'}B_{h',g'} \Gamma_{g'} \right) \right\vert }_{R_{6,3,1}}\\
& +  \underbrace{ \left\vert u_n  \sum_{g,g',h,h' \in [G]^4}  \left( \Gamma_g'B_{h,g}' \Omega_{V,h} B_{h,g} \left( (\tilde  P_{gh})_{g,gh} P_{gh,g'h'} (\tilde  P_{g'h'})_{g'h',g'} \right)  B_{h',g'}' \Omega_{V,h'}B_{h',g'} \Gamma_{g'} \right) \right\vert }_{R_{6,3,2}} \\
& +  \underbrace{ \left\vert u_n  \sum_{g,g',h,h' \in [G]^4}  \left( \Gamma_g'B_{h,g}' \Omega_{V,h} B_{h,g} \left( \sum_{l \in (g,h) \cap (g',h')} P_{g,l,-gh} P_{l,g',-g'h'} \right)  B_{h',g'}' \Omega_{V,h'}B_{h',g'} \Gamma_{g'} \right) \right\vert }_{R_{6,3,3}},
\end{align*}
where
\begin{align*}
    R_{6,3,1} & \lesssim u_n \sum_{g \in [G]} \left\Vert \sum_{h \in [G]}  \Gamma_g'B_{h,g}' \Omega_{V,h} B_{h,g} \right\Vert_2^2 \\
    & = u_n \sum_{g \in [G]} \Gamma_g' (B' \Omega_V B)_{g,g}^2 \Gamma_g \\
    & \lesssim u_n^3 n_G tr ((B' B)_{g,g}^2) \\
    & \lesssim u_n^3 n_G \lambda_n \kappa_n = o(\omega_n^4),
\end{align*}
\begin{align*}
R_{6,3,2}   & \lesssim u_n \sum_{g \in [G]} \left\Vert \sum_{h \in [G]}  \Gamma_g'B_{h,g}' \Omega_{V,h} B_{h,g} (\tilde  P_{gh})_{g,g}\right\Vert_2^2 +  u_n \sum_{h \in [G]} \left\Vert \sum_{g \in [G]}  \Gamma_g'B_{h,g}' \Omega_{V,h} B_{h,g} (\tilde  P_{gh})_{g,h} \right\Vert_2^2  \\
& \lesssim u_n \sum_{g \in [G]} \left\Vert \sum_{h \in [G]}  \Gamma_g'B_{h,g}' \Omega_{V,h} B_{h,g} \Xi_{1,g}\right\Vert_2^2 +  u_n \sum_{g \in [G]} \left\Vert \sum_{h \in [G]}  \Gamma_g'B_{h,g}' \Omega_{V,h} B_{h,g} \Xi_{2,gh}\right\Vert_2^2 \\
& + u_n \sum_{h \in [G]} \left\Vert \sum_{g \in [G]}  \Gamma_g'B_{h,g}' \Omega_{V,h} B_{h,g} (\tilde  P_{gh})_{g,h} \right\Vert_2^2   \\
& \lesssim o(\omega_n^4) +  u_n^3 \sum_{g \in [G]} \left(  \sum_{h \in [G]}  ||B_{h,g} \Gamma_g||_{2} || B_{h,g}||_{op} ||P_{g,h}||_{op}\right)^2 \\
& + u_n^3 \sum_{h \in [G]} \left(  \sum_{g \in [G]}  ||B_{h,g} \Gamma_g||_{2} || B_{h,g}||_{op} ||P_{g,h}||_{op}\right)^2 \\
& \lesssim o(\omega_n^4) +  u_n^3 \sum_{g \in [G]} \left(  \sum_{h \in [G]}  ||B_{h,g} \Gamma_g||_{2}^2 \right) \left(|| B_{h,g}||_{op} ||P_{g,h}||_{op}\right)^2 \\
& + u_n^3 \sum_{h \in [G]} \left(  \sum_{g \in [G]}  ||B_{h,g} \Gamma_g||_{2} || B_{h,g}||_{op} ||P_{g,h}||_{op}\right)^2 \\
& = o(\omega_n^4) + u_n^3 \phi_n \lambda_n^2 \sum_{g,h \in [G]^2}   ||B_{h,g} \Gamma_g||_{2}^2  \\
& = o(\omega_n^4) + u_n^3 n_G \phi_n \lambda_n^2 \kappa_n = o(\omega_n^4),
\end{align*}
and $R_{6,3,3} = o(\omega_n^4)$ following the same argument for $R_{1,3}$. This implies $R_{6,3} = o(\omega_n^4)$, which concludes the proof.


\subsection{Bound for \texorpdfstring{$R_8$}{R8}}
We have
\begin{align*}
    R_8 & \lesssim  \mathbb V \left( \sum_{g, h, k \in [G]^3, h \neq k, l \in [-ghk]} \left(U_g'B_{h,g}' V_h  \right) \left(V_k ' B_{k,g} P_{g,l,-ghk} U_l \right)  \right) \\
    & +  \mathbb V \left( \sum_{g, h \in [G]^2, l \in [-gh]} \left(U_g'B_{h,g}' (V_h  V_h ' - \Omega_{V,h}) B_{h,g} P_{g,l,-gh} U_l \right)  \right) \\
    & + \mathbb V \left( \sum_{g, h \in [G]^2, l \in [-gh]} \left(U_g'B_{h,g}' \Omega_{V,h} B_{h,g} P_{g,l,-gh} U_l \right)  \right) \\
    & \lesssim   \sum_{g, h, k \in [G]^3, h \neq k, l \in [-ghk]} \mathbb V \left( \left(U_g'B_{h,g}' V_h  \right) \left(V_k ' B_{k,g} P_{g,l,-ghk} U_l \right)  \right)  \\
    & +  \sum_{g, h \in [G]^2, l \in [-gh]}  \mathbb V \left(\left(U_g'B_{h,g}' (V_h  V_h ' - \Omega_{V,h}) B_{h,g} P_{g,l,-gh} U_l \right)  \right) \\
    & + \sum_{g,l \in [G]^2, g \neq l} \mathbb V \left( \sum_{h \in [-l]} \left(U_g'B_{h,g}' \Omega_{V,h} B_{h,g} P_{g,l,-gh} U_l \right)  \right) \\
    & \lesssim \underbrace{ \sum_{g, h, k \in [G]^3, h \neq k, l \in [-ghk]}  u_n^4 \left\Vert B_{h,g} \right\Vert_F^2 \left\Vert B_{k,g} P_{g,l,-ghk}\right\Vert_F^2 }_{R_{8,1}} +  \underbrace{ \sum_{g, h \in [G]^2, l \in [-gh]} u_n^2 \mathbb E \left\Vert B_{h,g}' V_h  V_h '  B_{h,g} P_{g,l,-gh}  \right\Vert_F^2 }_{R_{8,2}}  \\
    &  + \underbrace{  \sum_{g,l \in [G]^2, g \neq l} u_n^2 \left\Vert \sum_{h \in [-l]} B_{h,g}' \Omega_{V,h} B_{h,g} P_{g,l,-gh}  \right\Vert_F^2 }_{R_{8,3}},
\end{align*}
where
\begin{align*}
R_{8,1} & \lesssim \sum_{g, h,k \in [G]^3}  u_n^4 \left\Vert B_{h,g} \right\Vert_F^2 tr\begin{pmatrix}
B_{k,g}' B_{k,g} \left( \sum_{l \in [-ghk]} P_{g,l,-ghk} P_{l,g,-ghk} \right)
\end{pmatrix}   \\
& \lesssim \sum_{g, h,k \in [G]^3}  u_n^4 \left\Vert B_{h,g} \right\Vert_F^2 tr\begin{pmatrix}
B_{k,g}' B_{k,g} P_{g,g,-ghk}
\end{pmatrix}   \\
& \lesssim \sum_{g, h,k \in [G]^3}  u_n^4 \lambda_n  \left\Vert B_{h,g} \right\Vert_F^2 ||B_{k,g}||_F^2    \\
& \lesssim  u_n^4 \lambda_n \sum_{g \in [G]} \left( tr([BB']_{[g,g]}) \right)^2 \\
& \lesssim  u_n^4 n_G \lambda_n^2 \kappa_n  = o(\omega_n^4)
\end{align*}
and
\begin{align*}
R_{8,2} & \lesssim  \sum_{g, h \in [G]^2, l \in [-gh]} u_n^2 \lambda_n^2  \mathbb E tr\left(  P_{l,g,-gh}  B_{h,g}' V_{h} ||V_{h}||_2^2 V_{h}' B_{h,g} P_{g,l,-gh}\right) \\
& \lesssim \sum_{g, h \in [G]^2,l \in [-gh]} u_n^{\frac{3q-4}{q-1}} n_G^{\frac{q}{q-1}} \lambda_n^2   tr\left(  B_{h,g}' B_{h,g} P_{g,l,-gh}P_{l,g,-gh}  \right) \\
& \lesssim \sum_{g, h \in [G]^2} u_n^{\frac{3q-4}{q-1}} n_G^{\frac{q}{q-1}}  \lambda_n^2  tr\left(  B_{h,g}' B_{h,g} P_{g,g,-gh}  \right) \\
& \lesssim u_n^{\frac{3q-4}{q-1}} n_G^{\frac{q}{q-1}} \lambda_n^3 \kappa_n = o(\omega_n^4).
\end{align*}
For $R_{8,3}$, we have
\begin{align*}
R_{8,3} & \lesssim \sum_{g,h,h' \in [G]^3} u_n^2 tr\left( B_{h,g}' \Omega_{V,h} B_{h,g} \left( \sum_{l \in [-ghh']} P_{g,l,-gh}P_{l,g,-gh'} \right)  B_{h',g}' \Omega_{V,h'} B_{h',g} \right) \\
& \lesssim  \underbrace{   \sum_{g,h,h' \in [G]^3} u_n^2 tr\left( B_{h,g}' \Omega_{V,h} B_{h,g} P_{g,g} B_{h',g}' \Omega_{V,h'} B_{h',g} \right) }_{R_{8,3,1}}   \\
& + \underbrace{   \sum_{g,h,h' \in [G]^3} u_n^2 tr\left( B_{h,g}' \Omega_{V,h} B_{h,g} \left((\tilde  P_{gh})_{g,gh} P_{gh,gh'} (\tilde  P_{gh'})_{gh',g}  \right)  B_{h',g}' \Omega_{V,h'} B_{h',g} \right) }_{R_{8,3,2}}    \\
& + \underbrace{   \sum_{g,h,h' \in [G]^3} u_n^2 tr\left( B_{h,g}' \Omega_{V,h} B_{h,g} \left(\sum_{l \in (g,h) \cap (g,h')} P_{g,l,-gh} P_{l,g,-gh'}  \right)  B_{h',g}' \Omega_{V,h'} B_{h',g} \right) }_{R_{8,3,3}},
\end{align*}
where
\begin{align*}
R_{8,3,1} & \lesssim u_n^2  \sum_{g \in [G]} tr( (B' \Omega_V B)_{g,g} P_{g,g}(B' \Omega_V B)_{g,g}) \\
& \lesssim  u_n^4 \lambda_n^2 \kappa_n =  o(\omega_n^4),
\end{align*}
\begin{align*}
R_{8,3,2} & \lesssim u_n^2  \sum_{g \in [G]} tr\left(\left( \sum_{h \in [G]} B_{h,g}' \Omega_{V,h} B_{h,g} (\tilde  P_{gh})_{g,g} \right) P_{g,g} \left( \sum_{h \in [G]} B_{h,g}' \Omega_{V,h} B_{h,g} (\tilde  P_{gh})_{g,g} \right)' \right) \\
& +  \sum_{g,h' \in [G]^2} u_n^2 tr\left( \left( \sum_{h \in [G]} B_{h,g}' \Omega_{V,h} B_{h,g} (\tilde  P_{gh})_{g,g} \right) P_{g,h'} \left( (\tilde  P_{gh'})_{h',g}   B_{h',g}' \Omega_{V,h'} B_{h',g}\right)  \right) \\
& +  \sum_{h,h' \in [G]^2} u_n^2 tr\left( \left( \sum_{g \in [G]} B_{h,g}' \Omega_{V,h} B_{h,g} (\tilde  P_{gh})_{g,h} \right) P_{h,h'} \left( (\tilde  P_{gh'})_{h',g}   B_{h',g}' \Omega_{V,h'} B_{h',g}\right)  \right) \\
& \lesssim u_n^2 \sum_{g \in [G]}  \left\Vert \sum_{h \in [G]} B_{h,g}' \Omega_{V,h} B_{h,g} (\tilde  P_{gh})_{g,g} \right\Vert_F^2 + u_n^2 \sum_{h \in [G]} \left\Vert \sum_{g \in [G]} B_{h,g}' \Omega_{V,h} B_{h,g} (\tilde  P_{gh})_{g,h} \right\Vert_F^2 \\
& \lesssim u_n^2 \sum_{g \in [G]}  \left\Vert \sum_{h \in [G]} B_{h,g}' \Omega_{V,h} B_{h,g} \Xi_{1,g} \right\Vert_F^2 + u_n^2 \sum_{g \in [G]}  \left\Vert \sum_{h \in [G]} B_{h,g}' \Omega_{V,h} B_{h,g} \Xi_{9,gh} \right\Vert_F^2 \\
& + u_n^2 \sum_{h \in [G]} \left\Vert \sum_{g \in [G]} B_{h,g}' \Omega_{V,h} B_{h,g} (\tilde  P_{gh})_{g,h} \right\Vert_F^2 \\
& \lesssim   o(\omega_n^4) + u_n^4 \lambda_n^2 \sum_{g \in [G]}  \left( \sum_{h \in [G]} ||B_{h,g}||_F  ||P_{g,h}||_{op} \right)^2 + u_n^4 \lambda_n^2 \sum_{h \in [G]}  \left( \sum_{g \in [G]} ||B_{h,g}||_F  ||P_{g,h}||_{op} \right)^2 \\
& \lesssim   o(\omega_n^4) + u_n^4 \phi_n \lambda_n^2 \kappa_n =  o(\omega_n^4)
\end{align*}
Last, we can show $R_{8,3,3} = o(\omega_n^4) $ following the same argument for $R_{1,3}$, which concludes that $R_8 = o(\omega_n^4)$.

\subsection{Bound for \texorpdfstring{$R_9$}{R9}}
By \Cref{lem:P_l3o_1}.6, we have
\begin{align*}
R_9 & \lesssim \sum_{l \in [G]} u_n \left\Vert \sum_{g,h \in [-l]^2} \sum_{k \in [-ghl]}  \left( P_{l,h,-ghk} B_{g,h}' \Pi_g \right) \left(\Gamma_h' B_{g,h}' P_{g,k,-gh} \Pi_k \right) \right\Vert_2^2 \\
& = \sum_{ \substack{ g,g',h,h' \in [G]^4\\ k\in [-gh], k' \in [-g'h']}} u_n \begin{bmatrix}
& \left(\Gamma_h' B_{g,h}' P_{g,k,-gh} \Pi_k \right) \Pi_g' B_{g,h} \left( \sum_{l \in [-gg'hh'kk']}P_{h,l,-ghk} P_{l,h',-g'h'k'} \right) \\
& \times B_{g',h'}' \Pi_{g'} \left(\Gamma_{h'}' B_{g',h'}' P_{g',k',-g'h'} \Pi_k' \right)
\end{bmatrix}
 \\
& = \sum_{ \substack{ g,g',h,h' \in [G]^4\\ k\in [-gh], k' \in [-g'h']}} u_n\begin{bmatrix}
& \left(\Gamma_g' B_{h,g}' P_{h,k,-hg} \Pi_k \right) \Pi_h' B_{h,g} \left( \sum_{l \in [-gg'hh'kk']}P_{g,l,-ghk} P_{l,g',-g'h'k'} \right) \\
& \times B_{h',g'}' \Pi_{h'}\left(\Gamma_{g'}' B_{h',g'}' P_{h',k',-h'g'} \Pi_k' \right)
\end{bmatrix} \\
& \lesssim \underbrace{ \left\vert \sum_{ \substack{ g,g',h,h' \in [G]^4\\ k\in [-gh], k' \in [-g'h']}} u_n\left(\Gamma_g' B_{h,g}' P_{h,k,-hg} \Pi_k \right) \Pi_h' B_{h,g} P_{g,g'} B_{h',g'}' \Pi_{h'} \left(\Gamma_{g'}' B_{h',g'}' P_{h',k',-h'g'} \Pi_k' \right) \right\vert}_{R_{9,1}} \\
& + \underbrace{ \left\vert  \sum_{ \substack{ g,g',h,h' \in [G]^4\\ k\in [-gh], k' \in [-g'h']}} u_n\begin{bmatrix}
& \left(\Gamma_g' B_{h,g}' P_{h,k,-hg} \Pi_k \right) \Pi_h' B_{h,g} \left( (\tilde  P_{ghk})_{g,ghk} P_{ghk,g'h'k'} (\tilde  P_{g'h'k'})_{g'h'k',g'} \right) \\
& \times B_{h',g'}' \Pi_{h'} \left(\Gamma_{g'}' B_{h',g'}' P_{h',k',-h'g'} \Pi_k' \right)
\end{bmatrix}  \right\vert}_{R_{9,2}} \\
& + \underbrace{ \left\vert  \sum_{ \substack{ g,g',h,h' \in [G]^4\\ k\in [-gh], k' \in [-g'h']}} u_n\begin{bmatrix}
& \left(\Gamma_g' B_{h,g}' P_{h,k,-hg} \Pi_k \right) \Pi_h' B_{h,g} \left(\sum_{l \in (g,h,k) \cap (g',h',k')} P_{g,l,-ghk} P_{l,g',-g'h'k'}  \right) \\
& \times B_{h',g'}' \Pi_{h'} \left(\Gamma_{g'}' B_{h',g'}' P_{h',k',-h'g'} \Pi_k' \right)
\end{bmatrix}  \right\vert}_{R_{9,3}},
\end{align*}
where
\begin{align*}
    R_{9,1} & = \left\vert \sum_{ \substack{ g,g',h,h' \in [G]^4}} u_n\left(\Gamma_g' B_{h,g}'  \Pi_h \right) \Pi_h' B_{h,g} P_{g,g'}   B_{h',g'}' \Pi_{h'}\left(\Gamma_{g'}' B_{h',g'}' \Pi_h' \right) \right\vert \\
    & \lesssim \sum_{g \in [G]} u_n||\sum_{h \in [G]}\left(\Gamma_g' B_{h,g}'  \Pi_h \right) \Pi_h' B_{h,g}||_2^2 \\
    & \lesssim u_n n_G^2 \lambda_n^2 \sum_{g,h \in [G]} ||\Pi_h' B_{h,g}||_2^2 \\
    & \lesssim u_n n_G^3 \lambda_n^2 \kappa_n = o(\omega_n^4).
\end{align*}

Next, following the same argument in $R_{1,2}$, in order to bound $R_{9,2}$, it suffices to show
\begin{align}
&    \sum_{g \in [G]} u_n \left\Vert \sum_{h \in [G], k \in [-gh]} \left(\Gamma_g' B_{h,g}' P_{h,k,-hg} \Pi_k \right) \Pi_h' B_{h,g}  (\tilde  P_{ghk})_{g,g} \right\Vert_2^2  = o(\omega_n^4), \label{eq:R921}\\
&    \sum_{h \in [G]} u_n  \left\Vert \sum_{g \in [G], k \in [-gh]} \left(\Gamma_g' B_{h,g}' P_{h,k,-hg} \Pi_k \right) \Pi_h' B_{h,g}  (\tilde  P_{ghk})_{g,h} \right\Vert_2^2  = o(\omega_n^4), \label{eq:R922} \\
&    \sum_{k \in [G]} u_n  \left\Vert \sum_{g,h \in [-k]^2} \left(\Gamma_g' B_{h,g}' P_{h,k,-hg} \Pi_k \right) \Pi_h' B_{h,g}  (\tilde  P_{ghk})_{g,k} \right\Vert_2^2  = o(\omega_n^4). \label{eq:R923}
\end{align}
For \cref{eq:R921}, by \Cref{lem:P_l3o_1}.5, we have
\begin{align*}
&   \sum_{g \in [G]} u_n  \left\Vert \sum_{h \in [G], k \in [-gh]} \left(\Gamma_g' B_{h,g}' P_{h,k,-hg} \Pi_k \right) \Pi_h' B_{h,g}  (\tilde  P_{ghk})_{g,g} \right\Vert_2^2 \\
& \lesssim \sum_{g \in [G]} u_n  \left\Vert \sum_{h \in [G], k \in [-gh]} \left(\Gamma_g' B_{h,g}' P_{h,k,-hg} \Pi_k \right) \Pi_h' B_{h,g} (\Xi_{1,g}+ \Xi_{2,gh})  \right\Vert_2^2 \\
& + \sum_{g \in [G]} u_n  \left\Vert \sum_{h \in [G], k \in [-gh]} \left(\Gamma_g' B_{h,g}' P_{h,k,-hg} \Pi_k \right) \Pi_h' B_{h,g} (\Xi_{3,gk}+ \Xi_{4,ghk})  \right\Vert_2^2 \\
& \lesssim \sum_{g \in [G]} u_n  \left\Vert \sum_{h \in [G]} \left(\Gamma_g' B_{h,g}' \Pi_h \right) \Pi_h' B_{h,g} (\Xi_{1,g}+ \Xi_{2,gh})  \right\Vert_2^2 \\
& + \sum_{g \in [G]} \left\Vert \sum_{h \in [G], k \in [-gh]} \left(\Gamma_g' B_{h,g}' P_{h,k,-hg} \Pi_k \right) \Pi_h' B_{h,g} (\Xi_{3,gk}+ \Xi_{4,ghk})  \right\Vert_2^2 \\
& \lesssim \sum_{g \in [G]}  u_n n_G \lambda_n^2  \left( \sum_{h \in [G]} ||\Gamma_g'B_{g,h}||_{2} ||\Pi_h' B_{h,g} ||_2  \right)^2 \\
& + \sum_{g \in [G]} u_n n_G \lambda_n^2 \left( \sum_{h,k \in [G]^2} ||\Gamma_g'B_{h,g}||_{op} (||P_{h,k}||_{op} + ||P_{g,k}||_{op})^2  ||\Pi_h' B_{h,g}||_2 \right )^2 \\
& \lesssim \sum_{g \in [G]}  u_n n_G \lambda_n^2  \left( \sum_{h \in [G]} ||\Gamma_g'B_{g,h}||_{2} ||\Pi_h' B_{h,g} ||_2  \right)^2  + u_n n_G \phi_n^2 \lambda_n^2\sum_{g \in [G]}  \left( \sum_{h \in [G]} ||\Gamma_g' B_{h,g}||_{2} ||\Pi_h' B_{h,g}||_2 \right )^2 \\
& \lesssim
u_n n_G^3 (1+\phi_n^2) \lambda_n^3 \kappa_n  = o(\omega_n^4).
\end{align*}
Similarly, we can show \cref{eq:R922} holds. Next, for $\cref{eq:R923}$, we have
\begin{align*}
 & \sum_{k \in [G]} u_n  \left\Vert \sum_{g,h \in [-k]^2} \left(\Gamma_g' B_{h,g}' P_{h,k,-hg} \Pi_k \right) \Pi_h' B_{h,g}  (\tilde  P_{ghk})_{g,k} \right\Vert_2^2 \\
 & \lesssim \sum_{k \in [G]} u_n  \left\Vert \sum_{g,h \in [-k]^2} \left(\Gamma_g' B_{h,g}' P_{h,k,-hg} \Pi_k \right) \Pi_h' B_{h,g}  \Xi_{7,gk} \right\Vert_2^2 \\
 & + \sum_{k \in [G]} u_n  \left\Vert \sum_{g,h \in [-k]^2} \left(\Gamma_g' B_{h,g}' P_{h,k,-hg} \Pi_k \right) \Pi_h' B_{h,g}  \Xi_{8,ghk} \right\Vert_2^2 \\
 & \lesssim \sum_{k \in [G]} u_n  \left\Vert \sum_{g,h \in [-k]^2} \left(\Gamma_g' B_{h,g}' P_{h,k} \Pi_k \right) \Pi_h' B_{h,g}  \Xi_{7,gk} \right\Vert_2^2 \\
 & + \sum_{k \in [G]} u_n  \left\Vert \sum_{g,h \in [-k]^2} \left(\Gamma_g' B_{h,g}' (\tilde P_{gh})_{h,gh} P_{gh,k}   \Pi_k \right) \Pi_h' B_{h,g}  \Xi_{7,gk} \right\Vert_2^2 \\
 & + \sum_{k \in [G]} u_n  \left\Vert \sum_{g,h \in [-k]^2} \left(\Gamma_g' B_{h,g}' P_{h,k,-hg} \Pi_k \right) \Pi_h' B_{h,g}  \Xi_{8,ghk} \right\Vert_2^2 \\
& \lesssim \sum_{k \in [G]} u_n  n_G \left( \sum_{g,h \in [G]^2} || B_{h,g} \Gamma_g||_{2}^2 ||P_{g,k}||_{op}^2  \right) \left(\sum_{g,h \in [G]^2}  || \Pi_h' B_{h,g}||_2^2 ||P_{h,k}||_{op}^2  \right) \\
& + \sum_{k \in [G]} u_n n_G \lambda_n^2 \left( \sum_{g,h \in [G]^2} ||B_{h,g} \Gamma_g||_{2}^2 (||P_{g,k}||_{op}^2+||P_{h,k}||_{op}^2)   \right) \left(\sum_{g,h \in [G]^2}  || \Pi_h' B_{h,g}||_2^2  ||P_{g,k}||_{op}^2 \right) \\
& + \sum_{k \in [G]} u_n n_G \lambda_n^2 \left( \sum_{g,h \in [G]^2} ||B_{h,g} \Gamma_g ||_{2}^2 (||P_{h,k}||_{op}^2 + ||P_{g,k}||_{op}^2) \right) \left(\sum_{g,h \in [G]^2}  || \Pi_h' B_{h,g}||_2^2  (||P_{h,k}||_{op}^2 +||P_{g,k}||_{op}^2) \right) \\
& \lesssim \sum_{k \in [G]} u_n n_G^2 \phi_n \lambda_n \left(\sum_{g,h \in [G]^2}  || \Pi_h' B_{h,g}||_2^2 ||P_{h,k}||_{op}^2  \right) \\
& + \sum_{k \in [G]} u_n n_G^2 \phi_n(\phi_n + \lambda_n^2) \lambda_n^2 \left(\sum_{g,h \in [G]^2}  || \Pi_h' B_{h,g}||_2^2  (||P_{h,k}||_{op}^2 + ||P_{g,k}||_{op}^2) \right) \\
& \lesssim u_n n_G^3 \phi_n^2 \lambda_n \kappa_n  + u_n n_G^3 \phi_n^2(\phi_n + \lambda_n^2) \lambda_n^2 \kappa_n = o(\omega_n^4).
\end{align*}
This leads to the desired result that
\begin{align*}
    R_{9,2} = o(\omega_n^4).
\end{align*}
Last, we can show $R_{9,3} = o(\omega_n^4) $ following the same argument for $R_{1,3}$, which concludes that $R_9 = o(\omega_n^4)$.

\subsection{Bounds for \texorpdfstring{$R_{10}, R_{11} \text{ and }R_{12}$}{R10, R11, and R12}}
We have
\begin{align*}
    R_{10} & \lesssim \sum_{g,l \in [G]^2, l \neq g} u_n^2 \left\Vert \sum_{h \in [-l]} \sum_{ k \in [-ghl]}  P_{l,h,-ghk} B_{g,h}'  \left(\Gamma_h' B_{g,h}' P_{g,k,-gh} \Pi_k \right) \right\Vert_F^2 \\
    &  = \sum_{g,l \in [G]^2, l \neq g} \sum_{h,h' \in [-l]^2} \sum_{ k \in [-ghl],k' \in [-gh'l] }  u_n^2 \begin{bmatrix}
      &  \left(\Gamma_h' B_{g,h}' P_{g,k,-gh} \Pi_k \right) tr\left( B_{g,h}  P_{h,l,-ghk}P_{l,h',-gh'k'} B_{g,h'}'  \right) \\
      & \times \left(\Gamma_{h'}' B_{g,h'}' P_{g,k',-gh'} \Pi_{k'} \right)
    \end{bmatrix} \\
    & = \sum_{g,h,h' \in [G]^3} \sum_{k \in [-gh], k' \in [-gh']}u_n^2 \begin{bmatrix}
      &  \left(\Gamma_h' B_{g,h}' P_{g,k,-gh} \Pi_k \right) \\
      &  tr\left( B_{g,h} \left( \sum_{l \in [-ghh'kk']} P_{h,l,-ghk}P_{l,h',-gh'k'} \right) B_{g,h'}'  \right) \\
      & \times \left(\Gamma_{h'}' B_{g,h'}' P_{g,k',-gh'} \Pi_{k'} \right)
    \end{bmatrix} \\
        & = \sum_{g,g',h \in [G]^3} \sum_{k \in [-g'h], k' \in [-g'h]}u_n^2 \begin{bmatrix}
      &  \left(\Gamma_g' B_{h,g}' P_{h,k,-hg} \Pi_k \right) \\
      &  tr\left( B_{h,g} \left( \sum_{l \in [-gg'hkk']} P_{g,l,-ghk}P_{l,g',-gh'k'} \right) B_{h,g'}'  \right) \\
      & \times \left(\Gamma_{g'}' B_{h,g'}' P_{h,k',-g'h} \Pi_{k'} \right)
    \end{bmatrix} \\
    & \lesssim \underbrace{ \left\vert \sum_{g,g',h \in [G]^3} \sum_{k \in [-gh], k' \in [-g'h]}u_n^2 \begin{bmatrix}
      \left(\Gamma_g' B_{h,g}' P_{h,k,-hg} \Pi_k \right) tr\left( B_{h,g} P_{g,g'} B_{h,g'}'  \right)  \left(\Gamma_{g'}' B_{h,g'}' P_{h,k',-g'h} \Pi_{k'} \right)
    \end{bmatrix} \right\vert }_{R_{10,1}} \\
    & +  \underbrace{ \left\vert \sum_{g,g',h \in [G]^3} \sum_{k \in [-gh], k' \in [-g'h]}u_n^2 \begin{bmatrix}
      &  \left(\Gamma_g' B_{h,g}' P_{h,k,-hg} \Pi_k \right) \\
      &  tr\left( B_{h,g} \left( (\tilde  P_{ghk})_{g,ghk} P_{ghk,g'hk'} (\tilde  P_{g'hk'})_{g'hk',g'} \right) B_{h,g'}'  \right) \\
      & \times \left(\Gamma_{g'}' B_{h,g'}' P_{h,k',-g'h} \Pi_{k'} \right)
    \end{bmatrix} \right\vert  }_{R_{10,2}} \\
    & +  \underbrace{ \left\vert \sum_{g,g',h \in [G]^3} \sum_{k \in [-gh], k' \in [-g'h]}u_n^2 \begin{bmatrix}
      &  \left(\Gamma_g' B_{h,g}' P_{h,k,-hg} \Pi_k \right) \\
      &  tr\left( B_{h,g} \left(  \sum_{l \in (g,h,k) \cap (g',h,k')} P_{g,l,-ghk} P_{l,g',-g'h'k} \right) B_{h,g'}'  \right) \\
      & \times \left(\Gamma_{g'}' B_{h,g'}' P_{h,k',-g'h} \Pi_{k'} \right)
    \end{bmatrix} \right\vert  }_{R_{10,3}},
\end{align*}
where
\begin{align*}
    R_{10,1} & = \left\vert \sum_{g,g',h \in [G]^3} u_n^2 \begin{bmatrix}
      \left(\Gamma_g' B_{h,g}'  \Pi_h \right) tr\left( B_{h,g} P_{g,g'} B_{h,g'}'  \right)  \left(\Gamma_{g'}' B_{h,g'}'  \Pi_{h} \right)
    \end{bmatrix} \right\vert \\
    & \lesssim \sum_{g,h \in [G]^2} u_n^2 ||\left(\Gamma_g' B_{h,g}'  \Pi_h \right)B_{h,g}||_F^2 \\
    & \lesssim u_n^2 n_G^2 \lambda_n^2 \kappa_n = o(\omega_n^4).
\end{align*}

Next, by a slight abuse of notation, denote
\begin{align*}
& \mathcal A_{h,g} = \sum_{k \in [-gh]} \left(\Gamma_g' B_{h,g}' P_{h,k,-hg} \Pi_k \right)  B_{h,g} (\tilde  P_{ghk})_{g,g}, \\
& \mathcal B_{h,k} = \sum_{g \in [-hk]} \left(\Gamma_g' B_{h,g}' P_{h,k,-hg} \Pi_k \right)  B_{h,g} (\tilde  P_{ghk})_{g,k}, \quad \text{and} \\
& \mathcal C_h = \sum_{g \in [G]} \sum_{k \in [-gh]} \left(\Gamma_g' B_{h,g}' P_{h,k,-hg} \Pi_k \right)  B_{h,g} (\tilde  P_{ghk})_{g,h}.
\end{align*}

To show $R_{10,2} = o(\omega_n^4) $ following the same argument as $R_{2,2}$, we only need to show
\begin{align}
 &   u_n^2 \sum_{g,h \in [G]^2}||\mathcal A_{h,g}||_F^2 = o(\omega_n^4), \label{eq:R10A}\\
 &   u_n^2 \sum_{h,k \in [G]^2}||\mathcal B_{h,k}||_F^2 = o(\omega_n^4), \label{eq:R10B}\\
 &   u_n^2 \sum_{h \in [G]}||\mathcal C_{h}||_F^2 = o(\omega_n^4). \label{eq:R10C}
\end{align}
For \cref{eq:R10A}, we have
\begin{align*}
u_n^2 \sum_{g,h \in [G]^2}||\mathcal A_{h,g}||_F^2 & \lesssim \underbrace{ u_n^2 \sum_{g,h \in [G]^2} \left\Vert \sum_{k \in [-gh]} \left(\Gamma_g' B_{h,g}' P_{h,k,-hg} \Pi_k \right)  B_{h,g} (\Xi_{1,g}+\Xi_{2,gh})\right\Vert_F^2}_{R_{10,2,1}} \\
& + \underbrace{ u_n^2 \sum_{g,h \in [G]^2} \left\Vert \sum_{k \in [-gh]} \left(\Gamma_g' B_{h,g}' P_{h,k,-hg} \Pi_k \right)  B_{h,g} (\Xi_{3,gk}+\Xi_{4,ghk})\right\Vert_F^2 }_{R_{10,2,2}},
\end{align*}
where
\begin{align*}
R_{10,2,1} & =    u_n^2 \sum_{g,h \in [G]^2} \left\Vert \left(\Gamma_g' B_{h,g}' \Pi_h \right)  B_{h,g} (\Xi_{1,g}+\Xi_{2,gh})\right\Vert_F^2  \lesssim u_n^2 n_G^2 \lambda_n^4 \kappa_n = o(\omega_n^4)
\end{align*}
and by \Cref{lem:P_l3o_1},
\begin{align*}
R_{10,2,2} & \lesssim      u_n^2 n_G^2 \lambda_n^2 \sum_{g,h \in [G]^2} \left( \sum_{k \in [G]} \left(||P_{h,k}||_{op} + ||P_{g,k}||_{op}\right)^2 \right) ||B_{h,g}||_F^2 \\
& \lesssim   u_n^2 n_G^2 \phi_n \lambda_n^2 \kappa_n = o(\omega_n^4).
\end{align*}

For \cref{eq:R10B}, we have
\begin{align*}
 u_n^2 \sum_{h,k \in [G]^2}||\mathcal B_{h,k}||_F^2 & \lesssim   u_n^2 n_G^2 \lambda_n^2 \sum_{h,k \in [G]^2} \left(  \sum_{g \in [G]} ||B_{h,g}'||_{op}  \left(||P_{h,k}||_{op} + ||P_{g,k}||_{op}\right) ||  B_{h,g}||_F \right)^2 \\
 & \lesssim u_n^2 n_G^2 (\phi_n + \lambda_n^2) \lambda_n^2 \left( \sum_{g,h,k \in [G]^2} \left(||P_{h,k}||_{op} + ||P_{g,k}||_{op}\right)^2 ||  B_{h,g}||_F^2 \right) \\
 & \lesssim u_n^2 n_G^2 \phi_n(\phi_n + \lambda_n^2)  \lambda_n^2  \kappa_n  = o(\omega_n^4).
\end{align*}

For \cref{eq:R10C}, we have
\begin{align*}
 u_n^2 \sum_{h \in [G]}||\mathcal C_{h}||_F^2 & \lesssim  \underbrace{ u_n^2 \sum_{h \in [G]} \left\Vert   \sum_{g \in [G]} \sum_{k \in [-gh]}  \left(\Gamma_g' B_{h,g}' P_{h,k,-hg} \Pi_k \right)  B_{h,g} \Xi_{5,gh} \right\Vert_F^2 }_{R_{10,2,3}}\\
 & + \underbrace{   u_n^2 \sum_{h \in [G]} \left\Vert   \sum_{g \in [G]} \sum_{k \in [-gh]}   \left(\Gamma_g' B_{h,g}' P_{h,k,-hg} \Pi_k \right)  B_{h,g} \Xi_{6,ghk} \right\Vert_F^2 }_{R_{10,2,4}},
\end{align*}
where
\begin{align*}
R_{10,2,3} & = u_n^2 \sum_{h \in [G]} \left\Vert   \sum_{g \in [G]} \left(\Gamma_g' B_{h,g}'  \Pi_h \right)  B_{h,g} \Xi_{5,gh} \right\Vert_F^2 =o(\omega_n^4)
\end{align*}
and
\begin{align*}
R_{10,2,4} & \lesssim      u_n^2 n_G^2 \sum_{h \in [G]} \left(    \sum_{g,k \in [G]^2}   ||B_{h,g}||_{op}  \left(||P_{h,k}||_{op} + ||P_{g,k}||_{op}\right)  || B_{h,g}||_F  \begin{pmatrix}
    & ||P_{g,h}||_{op} ||P_{g,k}||_{op} \\
    &+ ||P_{h,k}||_{op}||P_{g,h}||_{op} \\
    &+||P_{h,k}||_{op}||P_{g,k}||_{op}
\end{pmatrix} \right)^2 \\
& \lesssim     u_n^2 n_G^2 \lambda_n^2  \sum_{h \in [G]} \left[\sum_{g,k \in [G]^2}   ||B_{h,g}||_{op}^2  \left(||P_{h,k}||_{op} + ||P_{g,k}||_{op}\right)^2  \right] \\
& \times \left[\sum_{g,k \in [G]^2} || B_{h,g}||_F^2  \left(||P_{h,k}||_{op} + ||P_{g,k}||_{op}\right)^2 \right] \\
& \lesssim    u_n^2 n_G^2 \phi_n^2 ( \phi_n + \lambda_n^2)  \lambda_n^2 \kappa_n =o(\omega_n^4).
\end{align*}
This concludes that
\begin{align*}
    R_{10,2} = o(\omega_n^4).
\end{align*}
Last, following the same argument for $R_{1,3}$, we can show that
\begin{align*}
    R_{10,3} =  o(\omega_n^4).
\end{align*}
This concludes the proof for $R_{10}$. We can show
\begin{align*}
    R_{11} = o(\omega_n^4) \quad \text{and} \quad    R_{12} = o(\omega_n^4)
\end{align*}
by combining the above arguments with those for $R_3$ and $R_4$.

\subsection{Bound for \texorpdfstring{$R_{13}$}{R13}}
By \Cref{lem:P_l3o_1}, for $k \neq h$, we have
\begin{align*}
P_{l,g,-ghk} & = P_{l,g} + P_{l,ghk} (\tilde  P_{ghk})_{ghk,g} \\
& = \underbrace{ \begin{pmatrix}
P_{l,g} + P_{l,g} \Xi_{1,g} + P_{l,g} \Xi_{2,gh} + P_{l,h} \Xi_{5,gh}'
\end{pmatrix} }_{\Theta_{l,g,gh}} \\
& +  \underbrace{ \begin{pmatrix}
P_{l,g} (\Xi_{3,gk} + \Xi_{4,ghk})     + P_{l,h} \Xi_{6,ghk}' + P_{l,k}(\Xi_{7,gk} + \Xi_{8,ghk})'
\end{pmatrix} }_{\tilde \Theta_{l,g,ghk}}
\end{align*}
such that
\begin{align*}
\left\Vert    \sum_{l \in [-gh]} \Theta_{l,g,gh}' \Theta_{l,g,gh} \right\Vert_{op} \lesssim \lambda_n
\end{align*}
and for $k,k' \in [-gh]$
\begin{align*}
\left( \sum_{l \in [-ghkk']} \tilde \Theta_{l,g,ghk}' \tilde \Theta_{l,g,ghk} \right) \lesssim \lambda_n (||P_{g,k}||_{op} + ||P_{h,k}||_{op})^2.
\end{align*}

We have
\begin{align*}
R_{13} & \lesssim \sum_{g, h \in [G]^2} \sum_{l \in [-gh]}  \mathbb V\left(  \sum_{k \in [-ghl]} \left( U_l' P_{l,h,-ghk} B_{g,h}' V_g \right) \left(U_h' B_{g,h}' P_{g,k,-gh} \Pi_k \right) \right) \\
& = \sum_{g, h \in [G]^2} \sum_{l \in [-gh]}  \mathbb V\left(  \sum_{k \in [-ghl]} \left( U_l' P_{l,g,-ghk} B_{h,g}' V_h \right) \left(U_g' B_{h,g}' P_{h,k,-gh} \Pi_k \right) \right) \\
& \lesssim  \underbrace{ \sum_{g, h \in [G]^2} \sum_{l \in [-gh]}  \mathbb V\left(  \sum_{k \in [-ghl]} \left( U_l' \Theta_{l,g,gh} B_{h,g}' V_h \right) \left(U_g' B_{h,g}' P_{h,k,-gh} \Pi_k \right) \right)}_{R_{13,1}} \\
& +  \underbrace{  \sum_{g, h \in [G]^2} \sum_{l \in [-gh]}  \mathbb V\left(  \sum_{k \in [-ghl]} \left( U_l' \tilde \Theta_{l,g,ghk} B_{h,g}' V_h \right) \left(U_g' B_{h,g}' P_{h,k,-gh} \Pi_k \right) \right)}_{R_{13,2}},
\end{align*}
where
\begin{align*}
R_{13,1} & = \sum_{g, h \in [G]^2} \sum_{l \in [-gh]}  \mathbb V\left( \left( U_l' \Theta_{l,g,gh} B_{h,g}' V_h \right) \left(U_g' B_{h,g}' (\Pi_h - P_{h,l,-gh} \Pi_l) \right) \right) \\
& \lesssim \sum_{g, h \in [G]^2} \sum_{l \in [-gh]} u_n^3 n_G \lambda_n^2 ||\Theta_{l,g,gh} B_{h,g}'||_F^2 \\
& = \sum_{g, h \in [G]^2} u_n^3 n_G \lambda_n^2 tr\left(B_{h,g}  \left( \sum_{l \in [-gh]} \Theta_{l,g,gh}'\Theta_{l,g,gh} \right) B_{h,g}'\right) \\
& \lesssim u_n^3 n_G \lambda_n^3 \kappa_n = o(\omega_n^4)
\end{align*}
and
\begin{align*}
R_{13,2} & = \sum_{g, h \in [G]^2} \sum_{l \in [-gh]}  \sum_{k,k' \in [-ghl]^2} \mathbb E \begin{bmatrix}
&    \left( U_l' \tilde \Theta_{l,g,ghk} B_{h,g}' V_h \right)  \left( U_l' \tilde \Theta_{l,g,ghk'} B_{h,g}' V_h \right)   \\
& \times  \left(U_g' B_{h,g}' P_{h,k,-gh} \Pi_k \right) \left(U_g' B_{h,g}' P_{h,k',-gh} \Pi_{k'} \right)
\end{bmatrix}\\
& \lesssim \sum_{g, h \in [G]^2} \sum_{l \in [-gh]}  \sum_{k,k' \in [-ghl]^2} \mathbb E \left( U_l' \tilde \Theta_{l,g,ghk} B_{g,h}' V_g \right)^2  \left(U_h' B_{g,h}' P_{g,k',-gh} \Pi_{k'} \right)^2\\
& \lesssim \sum_{g, h \in [G]^2} \sum_{l \in [-gh]}  \sum_{k,k' \in [-ghl]^2} u_n^3 \left\Vert \tilde \Theta_{l,g,ghk} B_{g,h}' \right\Vert^2  \left\Vert B_{g,h}' P_{g,k',-gh} \Pi_{k'} \right\Vert_2^2\\
& = \sum_{g, h \in [G]^2}   \sum_{k,k' \in [-gh]^2}  u_n^3 tr\left( B_{h,g}  \left( \sum_{l \in [-ghkk']} \tilde \Theta_{l,g,ghk}' \tilde \Theta_{l,g,ghk} \right) B_{h,g}' \right)  \left\Vert B_{h,g}' P_{h,k',-gh} \Pi_{k'} \right\Vert_2^2\\
& \lesssim \sum_{g, h \in [G]^2}   \sum_{k,k' \in [-gh]^2}  u_n^3 n_G \lambda_n || B_{h,g}||_F^2  (||P_{g,k}||_{op} + ||P_{h,k}||_{op})^2 || B_{h,g}||_{op}^2 ||P_{h,k',-gh}||_{op}^2\\
& \lesssim u_n^3 n_G \phi_n^2 \lambda_n^3 \kappa_n = o(\omega_n^4).
\end{align*}
This leads to the desired result $R_{13} = o(\omega_n^4)$.

\subsection{Bounds for \texorpdfstring{$R_{14} \text{ and }R_{15}$}{R14 and R15}}

We have
\begin{align*}
 R_{14}    & \lesssim  \sum_{g, k \in [G]^2, g\neq k } \sum_{l \in [-gk]}   \mathbb V\left(\sum_{h \in [-gkl]}\left( U_l' P_{l,h,-ghk} B_{g,h}' V_g \right) \left(\Gamma_h' B_{g,h}' P_{g,k,-gh} V_k \right) \right) \\
 & =  \sum_{h, k \in [G]^2, h \neq k } \sum_{l \in [-hk]}   \mathbb V\left(\sum_{g \in [-hkl]}\left( U_l' P_{l,g,-ghk} B_{h,g}' V_h \right) \left(\Gamma_g' B_{h,g}' P_{h,k,-gh} V_k \right) \right) \\
& \lesssim \sum_{h, k \in [G]^2, h \neq k } \sum_{l \in [-hk]}  u_n^2 \mathbb E \left\Vert \sum_{g \in [-hkl]} P_{l,g,-ghk} B_{h,g}' V_h \Gamma_g' B_{h,g}' P_{h,k,-gh}  \right\Vert_F^2 \\
& = \sum_{g,g',h, \in [G]^3 } \sum_{ k \in [-gg'h]} u_n^2 \begin{bmatrix}
& tr\left(\Omega_{V,h} B_{h,g} \left( \sum_{l \in [-gg'hk]}  P_{g,l,-ghk}  P_{l,g',-g'hk} \right) B_{h,g'}'\right) \\
& \times (\Gamma_g' B_{h,g}' P_{h,k,-gh}  P_{k,h,-g'h}  B_{h,g'} \Gamma_{g'})
\end{bmatrix} \\
& \lesssim \underbrace{ \left\vert \sum_{g,g',h, \in [G]^3 } \sum_{ k \in [-gg'h]} u_n^2\begin{bmatrix}
& tr\left(\Omega_{V,h} B_{h,g} P_{g,g'} B_{h,g'}'\right) \\
& \times (\Gamma_g' B_{h,g}' P_{h,k,-gh}  P_{k,h,-g'h}  B_{h,g'} \Gamma_{g'})
\end{bmatrix}
 \right\vert }_{R_{14,1}}\\
& + \underbrace{  \left\vert \sum_{g,g',h, \in [G]^3 } \sum_{ k \in [-gg'h]} u_n^2\begin{bmatrix}
& tr\left(\Omega_{V,h} B_{h,g} (\tilde  P_{ghk})_{g,ghk} P_{ghk,g'hk} (\tilde  P_{g'hk})_{g'hk,g'} B_{h,g'}'\right) \\
& \times (\Gamma_g' B_{h,g}' P_{h,k,-gh}  P_{k,h,-g'h}  B_{h,g'} \Gamma_{g'})
\end{bmatrix}  \right\vert  }_{R_{14,2}}\\
& + \underbrace{  \left\vert \sum_{g,g',h, \in [G]^3 } \sum_{ k \in [-gg'h]} u_n^2\begin{bmatrix}
& tr\left(\Omega_{V,h} B_{h,g} \left( \sum_{l \in (g,h,k) \cap (g',h,k)} P_{g,l,-ghk} P_{l,g',-g'hk}\right) B_{h,g'}'\right) \\
& \times (\Gamma_g' B_{h,g}' P_{h,k,-gh}  P_{k,h,-g'h}  B_{h,g'} \Gamma_{g'})
\end{bmatrix}  \right\vert  }_{R_{14,3}}.
\end{align*}
Note that
\begin{align*}
\left\Vert    \sum_{k \in [-gg'h]} P_{h,k,-gh}  P_{k,h,-g'h} \right\Vert_{op} \lesssim \lambda_n
\end{align*}
and
\begin{align*}
R_{14,1} & \lesssim    u_n^3 n_G \lambda_n^2 \sum_{g,g',h \in [G]^3} ||B_{h,g}||_F||B_{h,g'}||_F ||B_{h,g}||_{op}||B_{h,g'}||_{op} \\
& \lesssim   u_n^3 n_G \lambda_n^2 \sum_{g,g',h \in [G]^3} ||B_{h,g}||_F^2 ||B_{h,g'}||_{op}^2 \\
& \lesssim  u_n^3 n_G (\phi_n + \lambda_n^2)\lambda_n^2 \kappa_n = o(\omega_n^4).
\end{align*}

For $R_{14,2}$, we have
\begin{align*}
R_{14,2} & \lesssim \sum_{g,g',h, \in [G]^3 } \sum_{ k \in [-gg'h]} u_n^3 n_G \lambda_n^3 || B_{h,g}||_F || B_{h,g'}||_F||B_{h,g}||_{op} ||P_{h,k,-gh}||_{op}  ||P_{k,h,-g'h}||_{op} || B_{h,g'}||_{op} \\
& \lesssim \sum_{g,g',h,k \in [G]^4 } \begin{pmatrix}
   & u_n^3 n_G \lambda_n^3 || B_{h,g}||_F || B_{h,g'}||_F||B_{h,g}||_{op} \\
   & \times (||P_{h,k}||_{op}+||P_{g,k}||_{op}) (||P_{h,k}||_{op}+||P_{g',k}||_{op}) || B_{h,g'}||_{op}
\end{pmatrix}  \\
& \lesssim \sum_{g,g',h \in [G]^3 } u_n^3 n_G \phi_n \lambda_n^3 || B_{h,g}||_F || B_{h,g'}||_F||B_{h,g}||_{op} B_{h,g'}||_{op} \\
& \lesssim u_n^3 n_G \phi_n (\phi_n + \lambda_n^2) \lambda_n^3 \kappa_n = o(\omega_n^4).
\end{align*}
Last, following the same argument for $R_{1,3}$, we can show that
\begin{align*}
    R_{14,3} =  o(\omega_n^4).
\end{align*}
This concludes the proof for $R_{14}$.  We can show $R_{15} = o(\omega_n^4)$ in the same manner.

\subsection{Bound for \texorpdfstring{$R_{16}$}{R16}}
We have
\begin{align*}
    R_{16} & \lesssim  \sum_{g, h \in [G]^2} \sum_{k \in [-gh]} \sum_{l \in [-ghk]} u_n^4  \left\Vert P_{l,h,-ghk} B_{g,h}'  \right\Vert_F^2 \left\Vert B_{g,h}' P_{g,k,-gh} \right\Vert_F^2 \\
    & = \sum_{g, h \in [G]^2} \sum_{k \in [-gh]}  u_n^4 tr( B_{g,h} P_{h,h,-ghk}) B_{g,h}')\left\Vert B_{g,h}' P_{g,k,-gh} \right\Vert_F^2 \\
    & \lesssim  \sum_{g, h \in [G]^2} u_n^4 \lambda_n ||B_{g,h} ||_F^2 tr\left( B_{g,h}' \left( \sum_{k \in [-gh]} P_{g,k,-gh} P_{k,g,-gh}\right) B_{g,h}   \right) \\
    & \lesssim u_n^4 \lambda_n^2  \sum_{g, h \in [G]^2} |B_{g,h} ||_F^4 \\
    & \lesssim  u_n^4 n_G \lambda_n^4 \kappa_n = o(\omega_n^4).
\end{align*}

\section{Proof of Corollary~\ref{cor:L3CO}}
Recall that
\begin{align*}
& \hat  \omega^2_{n,\rm L3CO} = \tilde  \omega^2_{n,\rm L3CO,1}  + \tilde  \omega^2_{n,\rm L3CO,2}, \\
& \mathbb E   \left(  \frac{\tilde  \omega^2_{n,\rm L3CO,1}  + \tilde  \omega^2_{n,\rm L3CO,2}}{ \omega_n^2} \right) = 1,\\
& \mathbb E   \left(  \frac{ \tilde  \omega^2_{n,\rm L3CO,1}}{ \omega_n^2} \right) \geq 0, \quad \mathbb V   \left(  \frac{ \tilde  \omega^2_{n,\rm L3CO,1}}{ \omega_n^2} \right) = o(1)\\
& \mathbb E   \left(  \frac{ \tilde  \omega^2_{n,\rm L3CO,2}}{ \omega_n^2} \right) \geq 0, \quad  \mathbb V   \left(  \frac{ \tilde  \omega^2_{n,\rm L3CO,2}}{ \omega_n^2} \right) = o(1).
\end{align*}
Consider a truncated version of the L3CO estimator defined as
\begin{align*}
& \breve \omega^2_{n,\rm L3CO} = (\tilde \omega^2_{n,\rm L3CO,1})^+ + (\tilde \omega^2_{n,\rm L3CO,2})^+.
\end{align*}
Then, we have
\begin{align*}
 & \mathbb V \left( \frac{\breve  \omega^2_{n,\rm L3CO}}{ \omega_n^2}\right) \lesssim   \mathbb V \left( \frac{(\tilde \omega^2_{n,\rm L3CO,1})^+}{ \omega_n^2}\right) + \mathbb V \left( \frac{(\tilde \omega^2_{n,\rm L3CO,2})^+}{ \omega_n^2}\right) \leq  \mathbb V \left( \frac{(\tilde \omega^2_{n,\rm L3CO,1})}{ \omega_n^2}\right) + \mathbb V \left( \frac{(\tilde \omega^2_{n,\rm L3CO,2})}{ \omega_n^2}\right)  = o(1), \\
 & \mathbb E \left( \frac{\breve  \omega^2_{n,\rm L3CO}}{ \omega_n^2}\right) \geq \mathbb E \left( \frac{\hat \omega^2_{n,\rm L3CO}}{ \omega_n^2}\right) = 1,
\end{align*}
and
\begin{align*}
 & \mathbb E \left( \frac{\breve  \omega^2_{n,\rm L3CO}}{ \omega_n^2}\right) \\
 & =  \mathbb E \left( \frac{(\tilde \omega^2_{n,\rm L3CO,1} - \mathbb E \tilde \omega^2_{n,\rm L3CO,1} + \mathbb E \tilde \omega^2_{n,\rm L3CO,1})^+ }{ \omega_n^2}\right) + \mathbb E \left( \frac{(\tilde \omega^2_{n,\rm L3CO,2} - \mathbb E \tilde \omega^2_{n,\rm L3CO,2} + \mathbb E \tilde \omega^2_{n,\rm L3CO,2})^+ }{ \omega_n^2}\right) \\
 & \leq \mathbb E \left( \frac{(\tilde \omega^2_{n,\rm L3CO,1} - \mathbb E \tilde \omega^2_{n,\rm L3CO,1})^+ + \mathbb E \tilde \omega^2_{n,\rm L3CO,1} }{ \omega_n^2}\right) + \mathbb E \left( \frac{(\tilde \omega^2_{n,\rm L3CO,2} - \mathbb E \tilde \omega^2_{n,\rm L3CO,2})^+ + \mathbb E \tilde \omega^2_{n,\rm L3CO,2} }{ \omega_n^2}\right) \\
 & \leq 1 + \left[\mathbb V \left( \frac{\tilde \omega^2_{n,\rm L3CO,1} }{ \omega_n^2}\right)\right]^{1/2} + \left[\mathbb V \left( \frac{\tilde \omega^2_{n,\rm L3CO,2} }{ \omega_n^2}\right)\right]^{1/2} \\
 & \leq 1 + o(1).
\end{align*}
Therefore, we have
\begin{align*}
 \mathbb E \left( \frac{\breve  \omega^2_{n,\rm L3CO}}{ \omega_n^2}-1\right)^2 \leq    \mathbb V \left( \frac{\breve  \omega^2_{n,\rm L3CO}}{ \omega_n^2}\right) + \left[\mathbb E  \left( \frac{\breve  \omega^2_{n,\rm L3CO}}{ \omega_n^2}\right)-1\right]^2 = o(1),
\end{align*}
which implies
\begin{align*}
\frac{\breve  \omega^2_{n,\rm L3CO}}{ \omega_n^2} \convP 1.
\end{align*}

Further note that
\begin{align*}
\frac{\hat \omega^2_{n,\rm L3CO} }{\omega_n^2}  = \frac{  (\tilde \omega^2_{n,\rm L3CO,1})^+ + (\tilde \omega^2_{n,\rm L3CO,2})^+}{\omega_n^2} - \frac{  (\tilde \omega^2_{n,\rm L3CO,1})^- + (\tilde \omega^2_{n,\rm L3CO,2})^-}{\omega_n^2},
\end{align*}
implying that
\begin{align*}
 0 \leq   \mathbb E \left( \frac{  (\tilde \omega^2_{n,\rm L3CO,1})^- + (\tilde \omega^2_{n,\rm L3CO,2})^-}{\omega_n^2} \right) = \mathbb E \left(\frac{  (\tilde \omega^2_{n,\rm L3CO,1})^+ + (\tilde \omega^2_{n,\rm L3CO,2})^+}{\omega_n^2} \right) - 1 = o(1),
\end{align*}
and thus,
\begin{align*}
 \frac{  (\tilde \omega^2_{n,\rm L3CO,1})^- + (\tilde \omega^2_{n,\rm L3CO,2})^-}{\omega_n^2}  \convP 0.
\end{align*}
Therefore, we have
\begin{align*}
\frac{\tilde \omega^2_{n,\rm L3CO}}{\omega_n^2} =  \frac{\breve \omega^2_{n,\rm L3CO}}{\omega_n^2} +  \frac{(\tilde \omega^2_{n,\rm L3CO,1})^- + (\tilde \omega^2_{n,\rm L3CO,2})^-}{\omega_n^2} \convP 1.
\end{align*}

\section{Proof of Theorem~\ref{thm:var_l2co}}

We split the proof into two parts. \Cref{sec:proof-validity} shows the first
part of the claim that \cref{eq:l2o_conservative} holds.
\Cref{sec:proof-consistency} shows the second part of the claim that
\cref{eq:l2o_conservative} holds with equality under additional conditions.

\subsection{Proof of validity}\label{sec:proof-validity}

Observe we can write
\begin{equation*}
  \hat \omega^2_{n, \rm L2CO} =
  \hat \omega^2_{n,\rm L2CO, 1} + 2 \hat \omega^2_{n,\rm L2CO,2} + \hat \omega^2_{n,\rm L2CO, 3} \geq 0,
\end{equation*}
where
\begin{align*}
    \hat \omega^2_{n,\rm L2CO, 1} &= \sum_{g, h, k \in [G]^3} \left(X_h ' B_{h,g} \tilde Y_{g,-h} \right) \left(X_k ' B_{k,g} \tilde Y_{g,-k} \right), \\
    \hat \omega^2_{n,
    \rm L2CO,2} &= \sum_{g, h, k \in [G]^3} \left(X_h ' B_{h,g} \tilde Y_{g,-h} \right) \left(Y_k ' B_{g,k}' \tilde X_{g,-k} \right), \\
    \hat \omega^2_{n,\rm L2CO,3} &= \sum_{g, h, k \in [G]^3} \left(Y_h ' B_{g,h}' \tilde X_{g,-h} \right) \left(Y_k ' B_{g,k}' \tilde X_{g,-k} \right),
\end{align*}
For $g \neq h$, we have
\begin{align*}
    \tilde Y_{g,-h} &= U_g + \tilde U_{g,-h}, \quad \tilde U_{g,-h} = \sum_{l \in [-g]} \tilde M_{g,l,-gh} U_l, \\
    \tilde X_{g,-h} &= V_g + \tilde V_{g,-h}, \quad \tilde V_{g,-h} = \sum_{l \in [-g]} \tilde M_{g,l,-gh} V_l,
\end{align*}
and similarly for $g \neq k$, where we use the fact that $\sum_{l \in [G]} \tilde M_{g,l,-gh} W_l = 0_{n_g \times d}$, and note that
\begin{align*}
    \mathbb E_{hk} U_g \tilde U_{g,l}' = \mathbb E_{hk} V_g \tilde V_{g,l}' = \mathbb E_{hk} V_g \tilde U_{g,l}' = \mathbb E_{hk} U_g \tilde V_{g,l}'  = 0, \quad g \notin [h,k], \quad l \in [h,k].
\end{align*}
We can write
\begin{align*}
    \mathbb E \hat \omega^2_{n,\rm L2CO, 1} & = \sum_{g, h, k \in [G]^3} \mathbb E \left(X_h ' B_{h,g} \tilde Y_{g,-h} \right) \left(X_k ' B_{k,g} \tilde Y_{g,-k} \right) \\
    &= H' \Omega_U H + \sum_{g,h \in [G]^2} tr  \left( B_{h,g} \Omega_{U,g} B_{h,g}' \Omega_{V,h} \right) \\
    &+ \sum_{g, h, k \in [G]^3} \mathbb E \left(X_h ' B_{h,g} \tilde U_{g,-h} \right) \left(X_k ' B_{k,g} \tilde U_{g,-k} \right),
\end{align*}
and
\begin{align*}
    \mathbb E \hat \omega^2_{n,\rm L2CO, 2} & = \sum_{g, h, k \in [G]^3} \mathbb E \left(X_h ' B_{h,g} \tilde Y_{g,-h} \right) \left(Y_k ' B_{g,k}' \tilde X_{g,-k} \right) \\
    &= H' \Omega_{U,V} \tilde H + \sum_{g, h \in [G]^2}  tr   \left(B_{h,g} \Omega_{U,V,g} B_{g,h} \Omega_{U,V,h}  \right) \\
    &+ \sum_{g, h, k \in [G]^3} \mathbb E \left(X_h ' B_{h,g} \tilde U_{g,-h} \right) \left(Y_k ' B_{g,k}' \tilde V_{g,-k} \right),
\end{align*}
and
\begin{align*}
    \mathbb E \hat \omega^2_{n, \rm L2CO, 3} & = \sum_{g, h, k \in [G]^3} \mathbb E \left(Y_h ' B_{g,h}' \tilde X_{g,-h} \right) \left(Y_k ' B_{g,k}' \tilde X_{g,-k} \right) \\
    &=  \tilde H'\Omega_{V} \tilde H +\sum_{g, h \in [G]^2} tr \left(B_{g,h}' \Omega_{V,g} B_{g,h} \Omega_{U,h} \right) \\
    &+ \sum_{g, h, k \in [G]^3} \mathbb E \left(Y_h ' B_{g,h}' \tilde V_{g,-h} \right) \left(Y_k ' B_{g,k}' \tilde V_{g,-k} \right),
\end{align*}
and thus
\begin{align*}
    \mathbb E \hat \omega^2_{n,\rm L2CO} &= \omega^2_n +  \sum_{g, h \in [G]^2}  \mathbb E \left( V_h'B_{h,g} U_g + V_g' B_{g,h} U_h \right)^2  \\
    &+ \sum_{g \in [G]} \mathbb E \left\{ \left(\sum_{h \in [G]} X_h ' B_{h,g} \tilde U_{g,-h} \right) + \left( \sum_{k \in [G]} Y_k ' B_{g,k}' \tilde V_{g,-k} \right) \right\}^2 \\
    & \geq \omega_n^2.
\end{align*}

To proceed, we note that
\begin{align*}
    \hat \omega_{n,\rm L2CO}^2 = \hat \omega_{n,1}^2 + \hat \omega_{n,2}^2 + \hat \omega_{n,3}^2,
\end{align*}
where
\begin{align*}
    \hat \omega_{n,1}^2 &= \sum_{g, h, k \in [G]^3} \left(X_h ' B_{h,g} U_g \right) \left(X_k ' B_{k,g} U_g \right) \\
    &+ 2 \sum_{g, h, k \in [G]^3} \left(X_h ' B_{h,g} U_g \right) \left(Y_k ' B_{g,k}' V_g \right) \\
    &+ \sum_{g, h, k \in [G]^3} \left(Y_h ' B_{g,h}' V_g \right) \left(Y_k ' B_{g,k}' V_g \right),
\end{align*}
and
\begin{align*}
    \hat \omega_{n,2}^2 &= 2 \sum_{g, h, k \in [G]^3} \left(X_h ' B_{h,g} U_g \right) \left(X_k ' B_{k,g} \tilde U_{g,-k} \right) + 2 \sum_{g, h, k \in [G]^3} \left(X_h ' B_{h,g} U_g \right) \left(Y_k ' B_{g,k}' \tilde V_{g,-k} \right) \\
    &+ 2\sum_{g, h, k \in [G]^3} \left(Y_h ' B_{g,h}' V_g \right) \left(Y_k ' B_{g,k}' \tilde V_{g,-k} \right) + 2 \sum_{g, h, k \in [G]^3} \left(X_h ' B_{h,g} \tilde U_{g,-h} \right) \left(Y_k ' B_{g,k}' V_g \right),
\end{align*}
and
\begin{align*}
    \hat \omega_{n,3}^2 &= \sum_{g, h, k \in [G]^3} \left(X_h ' B_{h,g} \tilde U_{g,-h} \right) \left(X_k ' B_{k,g} \tilde U_{g,-k} \right) \\
    &+ 2 \sum_{g, h, k \in [G]^3} \left(X_h ' B_{h,g} \tilde U_{g,-h} \right) \left(Y_k ' B_{g,k}' \tilde V_{g,-k} \right) \\
    &+ \sum_{g, h, k \in [G]^3} \left(Y_h ' B_{g,h}' \tilde V_{g,-h} \right) \left(Y_k ' B_{g,k}' \tilde V_{g,-k} \right).
\end{align*}
We have $\mathbb E \left( \hat \omega_{n,1}^2 + \hat \omega_{n,2}^2 \right) \geq \omega_n^2$ and $\hat \omega_{n,3}^2 \geq 0$, and it suffices to show that $\mathbb V \left(\hat \omega^2_{n,1} \right) = o(\omega_n^4)$ and $\mathbb V \left(\hat \omega^2_{n,2} \right) = o(\omega_n^4)$. Note also that $\mathbb V \left(\hat \omega^2_{n,1} \right) = o(\omega_n^4)$ can be established similarly as in \Cref{lem:V1}, and we focus on establishing $\mathbb V \left(\hat \omega^2_{n,2} \right) = o(\omega_n^4)$ in what follows. Let $g \neq g' \neq g''$ denote the event that $g \neq g', g \neq g'', g' \neq g''$, and recall that
\begin{align*}
    \tilde M_{g,g',-gg''} = - P_{g,g',-gg''}, \quad g \neq g' \neq g''.
\end{align*}
The proof relies heavily on \Cref{lem:P_l3o_1} (with $l = g'$ and $h = k = g''$) to calculate the summation of $P_{g,g',-gg''}$ across indexes.

We shall only compute the variance of the first term in $\hat \omega_{n,2}^2$, as the other terms can be handled similarly. We have
\begin{align*}
    \mathbb V \left( \sum_{g, h, k \in [G]^3} \left(X_h ' B_{h,g} U_g \right) \left(X_k ' B_{k,g} \tilde U_{g,-k} \right) \right) &\lesssim \underbrace{\mathbb V \left( \sum_{g, h, k \in [G]^3} \left(\Pi_h ' B_{h,g} U_g \right) \left(\Pi_k ' B_{k,g} \tilde U_{g,-k} \right) \right)}_{R_{17}} \\
    &+ \underbrace{\mathbb V \left( \sum_{g, h, k \in [G]^3} \left(\Pi_h ' B_{h,g} U_g \right) \left(V_k ' B_{k,g} \tilde U_{g,-k} \right) \right)}_{R_{18}} \\
    &+ \underbrace{\mathbb V \left( \sum_{g, h, k \in [G]^3} \left(V_h ' B_{h,g} U_g \right) \left(\Pi_k ' B_{k,g} \tilde U_{g,-k} \right) \right)}_{R_{19}} \\
    &+ \underbrace{\mathbb V \left( \sum_{g, h, k \in [G]^3} \left(V_h ' B_{h,g} U_g \right) \left(V_k ' B_{k,g} \tilde U_{g,-k} \right) \right)}_{R_{20}}.
\end{align*}

For $R_{17}$, we have
\begin{align*}
    R_{17} &= \mathbb V \left( \sum_{g, h, k \in [G]^3, g \neq h \neq k} \left(H_g' U_g \right) \left(\Pi_k ' B_{k,g} \tilde M_{g,h,-gk} V_h \right) \right) \\
    &= \mathbb V \left( \sum_{g, h \in [G]^2, g \neq h} \left(H_g' U_g \right) \left( \left(\sum_{k \in [-gh]} \Pi_k ' B_{k,g} P_{g,h,-gk} \right) V_h \right) \right) \\
    &\lesssim u_n^2 \zeta_{H,n} \sum_{g, h \in [G]^2, g \neq h} \sum_{k, k' \in [-gh]^2}\Pi_k ' B_{k,g}  P_{g,h,-gk} P_{h,g,-gk'} B_{k',g}' \Pi_{k'} \\
    &\lesssim u_n^2 \zeta_{H,n} \sum_{g \in [G]} \sum_{k, k' \in [-g]^2}\Pi_k ' B_{k,g} \left( \sum_{h \in [-kk']} P_{g,h,-gk} P_{h,g,-gk'} \right) B_{k',g}' \Pi_{k'} \\
    &\lesssim u_n^2 \zeta_{H,n} \left| \sum_{g \in [G]} \sum_{k, k' \in [-g]^2}\Pi_k ' B_{k,g} \left( \sum_{h \in [G]} P_{g,h,-gk} P_{h,g,-gk'} \right) B_{k',g}' \Pi_{k'} \right|  \\
    &+ u_n^2 \zeta_{H,n} \left| \sum_{g \in [G]} \sum_{k, k' \in [-g]^2}\Pi_k ' B_{k,g} P_{g,k,-gk} P_{k,g,-gk'}  B_{k',g}' \Pi_{k'} \right|\\
    &+ u_n^2 \zeta_{H,n} \left| \sum_{g \in [G]} \sum_{k, k' \in [-g]^2}\Pi_k ' B_{k,g} P_{g,k',-gk} P_{k',g,-gk'}  B_{k',g}' \Pi_{k'} \right| \\
    &+ u_n^2 \zeta_{H,n} \left| \sum_{g \in [G]} \sum_{k \in [-g]}\Pi_k ' B_{k,g} P_{g,k,-gk} P_{k,g,-gk}  B_{k,g}' \Pi_{k} \right|,
\end{align*}
where
\begin{align*}
    &u_n^2 \zeta_{H,n} \left| \sum_{g \in [G]} \sum_{k, k' \in [-g]^2}\Pi_k ' B_{k,g} \left( \sum_{h \in [G]} P_{g,h,-gk} P_{h,g,-gk'} \right) B_{k',g}' \Pi_{k'} \right| \\
    &\lesssim u_n^2 \zeta_{H,n} \sum_{g \in [G]} \sum_{k, k' \in [G]^2}\Pi_k ' B_{k,g} P_{g,g} B_{k',g}' \Pi_{k'} \\
    &+ u_n^2 \zeta_{H,n} \sum_{g \in [G]} \sum_{k, k' \in [G]^2}\Pi_k ' B_{k,g} (\tilde P_{gk})_{g,g} P_{g,g} (\tilde P_{gk'})_{g,g} B_{k',g}' \Pi_{k'} \\
    &+ u_n^2 \zeta_{H,n} \sum_{g \in [G]} \sum_{k, k' \in [G]^2}\Pi_k ' B_{k,g} (\tilde P_{gk})_{g,k} P_{k,k'} (\tilde P_{gk'})_{k',g} B_{k',g}' \Pi_{k'} \\
    &\lesssim u_n^2 \zeta_{H,n} \lambda_n \mu_n^2 + u_n^2 n_G \zeta_{H,n} \phi_n \lambda_n^3 \kappa_n + u_n^2 n_G \zeta_{H,n} \lambda_n^2 \kappa_n = o(\omega_n^4),
\end{align*}
and
\begin{align*}
    &u_n^2 \zeta_{H,n} \left| \sum_{g \in [G]} \sum_{k, k' \in [-g]^2}\Pi_k ' B_{k,g} P_{g,k,-gk} P_{k,g,-gk'}  B_{k',g}' \Pi_{k'} \right| \\
    &= u_n^2 \zeta_{H,n} \left| \sum_{g \in [G]} \sum_{k \in [G]}\Pi_k ' B_{k,g} P_{g,k,-gk} \left( \sum_{k' \in [G]} P_{k,g,-gk'}  B_{k',g}' \Pi_{k'} \right) \right| \\
    &\lesssim u_n^2 \zeta_{H,n} \sum_{g \in [G]} \sum_{k \in [G]}\Pi_k ' B_{k,g} P_{g,k,-gk} P_{k,g,-gk} B_{k,g}' \Pi_k \\
    &+ u_n^2 \zeta_{H,n} \sum_{g \in [G]} \sum_{k \in [G]} \left\Vert \sum_{k' \in [G]} P_{k,g,-gk'}  B_{k',g}' \Pi_{k'} \right\Vert_2^2 \\
    &\lesssim u_n^2 n_G \zeta_{H,n} \lambda_n^2 \kappa_n + o(\omega_n^4) = o(\omega_n^4),
\end{align*}
and similarly for
\begin{align*}
    u_n^2 \zeta_{H,n} \left| \sum_{g \in [G]} \sum_{k, k' \in [-g]^2}\Pi_k ' B_{k,g} P_{g,k',-gk} P_{k',g,-gk'}  B_{k',g}' \Pi_{k'} \right|,
\end{align*}
and
\begin{align*}
    u_n^2 \zeta_{H,n} \left| \sum_{g \in [G]} \sum_{k \in [-g]}\Pi_k ' B_{k,g} P_{g,k,-gk} P_{k,g,-gk}  B_{k,g}' \Pi_{k} \right|.
\end{align*}

For $R_{18}$, we have
\begin{align*}
    R_{18} &= \mathbb V \left( \sum_{g, h, k \in [G]^3, g \neq h \neq k} \left(H_g' U_g \right) \left(V_k ' B_{k,g} \tilde M_{g,h,-gk} V_h \right) \right) \\
    &\lesssim u_n^3 \zeta_{H,n} \sum_{g, h, k \in [G]^3, g \neq h \neq k} \tr \left( B_{k,g} P_{g,h,-gk} P_{h,g,-gk} B_{k,g}' \right) \\
    &\lesssim u_n^3 \zeta_{H,n} \sum_{g, h, k \in [G]^3, g \neq h \neq k} \tr \left( B_{k,g} P_{g,h} P_{h,g} B_{k,g}' \right) \\
    &+ u_n^3 \zeta_{H,n} \sum_{g, h, k \in [G]^3, g \neq h \neq k} \tr \left( B_{k,g} (\tilde P_{gk})_{g,g} P_{g,h} P_{h,g} (\tilde P_{gk})_{g,g} B_{k,g}' \right) \\
    &+ u_n^3 \zeta_{H,n} \sum_{g, h, k \in [G]^3, g \neq h \neq k} \tr \left( B_{k,g} (\tilde P_{gk})_{gk} P_{k,h} P_{h,k} (\tilde P_{gk})_{kg} B_{k,g}' \right) \\
    &\lesssim u_n^3 \zeta_{H,n} \lambda_n \kappa_n = o(\omega_n^4).
\end{align*}

For $R_{19}$, we have
\begin{align*}
    R_{19} &= \mathbb V \left( \sum_{g, h, k, l \in [G]^4, l \neq g, l \neq k} \left(V_h ' B_{h,g} U_g \right) \left(\Pi_k ' B_{k,g} \tilde M_{g,l,-gk} U_l \right) \right) \\
    &\lesssim \underbrace{\mathbb V \left( \sum_{g, h, k \in [G]^3, h \neq g, h \neq k} \left(V_h ' B_{h,g} U_g \right) \left(\Pi_k ' B_{k,g} P_{g,h,-gk} U_h \right) \right)}_{R_{19,1}} \\
    &+ \underbrace{\mathbb V \left( \sum_{g, h, k, l \in [G]^4, l \neq g, l \neq h, l \neq k} \left(V_h ' B_{h,g} U_g \right) \left(\Pi_k ' B_{k,g} P_{g,l,-gk} U_l \right) \right)}_{R_{19,2}},
\end{align*}
where
\begin{align*}
    R_{19,1} &\lesssim \mathbb V \left( \sum_{g, h, k \in [G]^3, h \neq g, h \neq k} \Pi_k ' B_{k,g} P_{g,h,-gk} \Omega_{U,V,h} B_{h,g} U_g \right) \\
    &+ \mathbb V \left( \sum_{g, h, k \in [G]^3, h \neq g, h \neq k} \Pi_k ' B_{k,g} P_{g,h,-gk} (U_h V_h' - \Omega_{U,V,h}) B_{h,g} U_g \right) \\
    &\lesssim u_n \sum_{g \in [G]} \left \Vert  \sum_{h, k \in [-g]^2, h \neq k}  B_{h,g}' \Omega_{U,V,h}' P_{h,g,-gk} B_{k,g}' \Pi_k \right \Vert_2^2 \\
    &+ u_n^{\frac{2q-3}{q-1}} n_G^{\frac{q}{q-1}} \sum_{g,h \in [G]^2, g \neq h} || B_{h,g} ||_F^2 \left \Vert \sum_{k \in [-gh]} P_{h,g,-gk} B_{k,g}' \Pi_k  \right \Vert_2^2,
\end{align*}
with
\begin{align*}
    &u_n \sum_{g \in [G]} \left \Vert \sum_{h, k \in [-g]^2, h \neq k}  B_{h,g}' \Omega_{U,V,h}' P_{h,g,-gk} B_{k,g}' \Pi_k \right \Vert_2^2 \\
    &\lesssim u_n \sum_{g \in [G]} \left \Vert \sum_{h, k \in [-g]^2, h \neq k}  B_{h,g}' \Omega_{U,V,h}' P_{h,g} B_{k,g}' \Pi_k \right \Vert_2^2 \\
    &+ u_n \sum_{g \in [G]} \left \Vert \sum_{h, k \in [-g]^2, h \neq k}  B_{h,g}' \Omega_{U,V,h}' P_{h,g} (\tilde P_{gk})_{g,g} B_{k,g}' \Pi_k \right \Vert_2^2 \\
    &+ u_n \sum_{g \in [G]} \left \Vert \sum_{h, k \in [-g]^2, h \neq k}  B_{h,g}' \Omega_{U,V,h}' P_{h,k} (\tilde P_{gk})_{k,g} B_{k,g}' \Pi_k \right \Vert_2^2 \\
    &\lesssim u_n^3 \lambda_n^2 \mu_n^2 + u_n^3 n_G \phi_n \lambda_n^2 \kappa_n = o(\omega_n^4),
\end{align*}
and
\begin{align*}
    &u_n^{\frac{2q-3}{q-1}} n_G^{\frac{q}{q-1}} \sum_{g,h \in [G]^2} || B_{h,g} ||_F^2 \left \Vert \sum_{k \in [-h]} P_{h,g,-gk} B_{k,g}' \Pi_k \right \Vert_2^2 \\
    &\lesssim u_n^{\frac{2q-3}{q-1}} n_G^{\frac{q}{q-1}} \sum_{g,h \in [G]^2} || B_{h,g} ||_F^2 \left \Vert \sum_{k \in [-h]} P_{h,g} B_{k,g}' \Pi_k \right \Vert_2^2 \\
    &+ u_n^{\frac{2q-3}{q-1}} n_G^{\frac{q}{q-1}} \sum_{g,h \in [G]^2} || B_{h,g} ||_F^2 \left \Vert \sum_{k \in [-h]} P_{h,g} (\tilde P_{gk})_{g,g} B_{k,g}' \Pi_k \right \Vert_2^2 \\
    &+ u_n^{\frac{2q-3}{q-1}} n_G^{\frac{q}{q-1}} \sum_{g,h \in [G]^2} || B_{h,g} ||_F^2 \left \Vert \sum_{k \in [-h]}  P_{h,k} (\tilde P_{gk})_{k,g} B_{k,g}' \Pi_k \right \Vert_2^2 \\
    &\lesssim u_n^{\frac{2q-3}{q-1}} n_G^{\frac{q}{q-1}} \zeta_{H,n} \lambda_n^2 \kappa_n + u_n^{\frac{2q-3}{q-1}} n_G^{\frac{2q-1}{q-1}} \lambda_n^4 \kappa_n + u_n^{\frac{2q-3}{q-1}} n_G^{\frac{2q-1}{q-1}} \phi_n^2 \lambda_n^2 \kappa_n = o(\omega_n^4).
\end{align*}
We also have
\begin{align*}
    R_{19,2} &\lesssim u_n^3 \sum_{g,h \in [G]^2, g \neq h} ||B_{h,g}||_F^2 \sum_{l \in [G], l \neq g, l \neq h} \sum_{k, k' \in [-gl]^2} \Pi_k ' B_{k,g} P_{g,l,-gk} P_{l,g,-gk'} B_{k',g}' \Pi_{k'} \\
    &\lesssim u_n^3 \sum_{g,h \in [G]^2} ||B_{h,g}||_F^2 \sum_{k, k' \in [G]^2} \Pi_k ' B_{k,g} \left(\sum_{l \in [-kk']} P_{g,l,-gk} P_{l,g,-gk'} \right) B_{k',g}' \Pi_{k'} \\
    &\lesssim u_n^3 \sum_{g,h \in [G]^2} ||B_{h,g}||_F^2 \sum_{k, k' \in [G]^2} \Pi_k ' B_{k,g} \left(\sum_{l \in [G]} P_{g,l,-gk} P_{l,g,-gk'} \right) B_{k',g}' \Pi_{k'} \\
    &+ u_n^3 \sum_{g,h \in [G]^2} ||B_{h,g}||_F^2 \left| \sum_{k, k' \in [G]^2} \Pi_k ' B_{k,g} P_{g,k,-gk} P_{k,g,-gk'} B_{k',g}' \Pi_{k'} \right| \\
    &+ u_n^3 \sum_{g,h \in [G]^2} ||B_{h,g}||_F^2 \left| \sum_{k, k' \in [G]^2} \Pi_k ' B_{k,g} P_{g,k',-gk} P_{k',g,-gk'} B_{k',g}' \Pi_{k'} \right| \\
    &+ u_n^3 \sum_{g,h \in [G]^2} ||B_{h,g}||_F^2 \sum_{k \in [G]} \Pi_k ' B_{k,g} P_{g,k,-gk} P_{k,g,-gk} B_{k,g}' \Pi_{k},
\end{align*}
where
\begin{align*}
    &u_n^3 \sum_{g,h \in [G]^2} ||B_{h,g}||_F^2 \sum_{k, k' \in [G]^2} \Pi_k ' B_{k,g} \left(\sum_{l \in [G]} P_{g,l,-gk} P_{l,g,-gk'} \right) B_{k',g}' \Pi_{k'} \\
    &\lesssim u_n^3 \sum_{g,h \in [G]^2} ||B_{h,g}||_F^2 \sum_{k, k' \in [G]^2} \Pi_k ' B_{k,g} P_{g,g} B_{k',g}' \Pi_{k'} \\
    &+ u_n^3 \sum_{g,h \in [G]^2} ||B_{h,g}||_F^2 \sum_{k, k' \in [G]^2} \Pi_k ' B_{k,g} (\tilde P_{gk})_{g,g} P_{g,g} (P_{gk'})_{g,g} B_{k',g}' \Pi_{k'} \\
    &+ u_n^3 \sum_{g,h \in [G]^2} ||B_{h,g}||_F^2 \sum_{k, k' \in [G]^2} \Pi_k ' B_{k,g} (\tilde P_{gk})_{g,k} P_{k,k'} (P_{gk'})_{k',g} B_{k',g}' \Pi_{k'} \\
    &\lesssim u_n^3 \zeta_{H,n} \lambda_n \kappa_n + u_n^3 n_G \phi_n (\phi_n + \lambda_n^2) \lambda_n \kappa_n + u_n^3 n_G \phi_n \lambda_n^2 \kappa_n = o(\omega_n^4),
\end{align*}
and
\begin{align*}
    &u_n^3 \sum_{g,h \in [G]^2} ||B_{h,g}||_F^2 \left| \sum_{k, k' \in [G]^2} \Pi_k ' B_{k,g} P_{g,k,-gk} P_{k,g,-gk'} B_{k',g}' \Pi_{k'} \right| \\
    &\lesssim u_n^3 \sum_{g,h \in [G]^2} ||B_{h,g}||_F^2 \sum_{k \in [G]} \Pi_k ' B_{k,g} P_{g,k,-gk} P_{k,g,-gk} B_{k,g}' \Pi_{k} \\
    &+ u_n^3 \sum_{g,h \in [G]^2} ||B_{h,g}||_F^2 \sum_{k \in [G]} \left\Vert \sum_{k' \in [G]} P_{k,g,-gk'} B_{k',g}' \Pi_{k'} \right\Vert_2^2 \\
    &\lesssim u_n^3 n_G (\phi_n + \lambda_n^2) \lambda_n^2 \kappa_n + o(\omega_n^4) = o(\omega_n^4),
\end{align*}
and similarly for
\begin{align*}
    u_n^3 \sum_{g,h \in [G]^2} ||B_{h,g}||_F^2 \left| \sum_{k, k' \in [G]^2} \Pi_k ' B_{k,g} P_{g,k',-gk} P_{k',g,-gk'} B_{k',g}' \Pi_{k'} \right|,
\end{align*}
and
\begin{align*}
    u_n^3 \sum_{g,h \in [G]^2} ||B_{h,g}||_F^2 \sum_{k \in [G]} \Pi_k ' B_{k,g} P_{g,k,-gk} P_{k,g,-gk} B_{k,g}' \Pi_{k}.
\end{align*}

For $R_{20}$, we have
\begin{align*}
    R_{20} &= \mathbb V \left( \sum_{g, h, k, l \in [G]^4, l \neq g, l \neq k} \left(V_h ' B_{h,g} U_g \right) \left(V_k ' B_{k,g} \tilde M_{g,l,-gk} U_l \right) \right) \\
    &\lesssim \underbrace{\mathbb V \left( \sum_{g, h, k \in [G]^3, g \neq h \neq k} \left(V_h ' B_{h,g} U_g \right) \left(V_k ' B_{k,g} P_{g,h,-gk} U_h \right) \right)}_{R_{20,1}} \\
    &+ \underbrace{\mathbb V \left( \sum_{g, h, k \in [G]^3, g \neq h \neq k} \left(V_h ' B_{h,g} U_g \right) \left(V_h ' B_{h,g} P_{g,k,-gh} U_k \right) \right)}_{R_{20,2}} \\
    &+ \underbrace{\mathbb V \left( \sum_{g, h, k, l \in [G]^4, g \neq h \neq k \neq l} \left(V_h ' B_{h,g} U_g \right) \left(V_k ' B_{k,g} P_{g,l,-gk} U_l \right) \right)}_{R_{20,3}},
\end{align*}
where
\begin{align*}
    R_{20,1} &\lesssim \mathbb V \left( \sum_{g, h, k \in [G]^3, g \neq h \neq k} V_k ' B_{k,g} P_{g,h,-gk} \Omega_{U,V,h} B_{h,g} U_g \right) \\
    &+ \mathbb V \left( \sum_{g, h, k \in [G]^3, g \neq h \neq k} V_k ' B_{k,g} P_{g,h,-gk} (U_h V_h' - \Omega_{U,V,h}) B_{h,g} U_g \right) \\
    &\lesssim u_n^2  \sum_{g, k \in [G]^2, g \neq k} \left\Vert \sum_{h \in [-gk]} B_{k,g} P_{g,h,-gk} \Omega_{U,V,h} B_{h,g} \right \Vert_F^2 \\
    &+ u_n^{\frac{3q-4}{q-1}} u_G^{\frac{q}{q-1}} \lambda_n^2 \sum_{g, h, k \in [G]^3, g \neq h \neq k} \tr \left( B_{k,g} P_{g,h,-gk} P_{h,g,-gk} B_{k,g}' \right) \\
    &\lesssim u_n^2  \sum_{g, k \in [G]^2, g \neq k} \left\Vert \sum_{h \in [-gk]} B_{k,g} P_{g,h} \Omega_{U,V,h} B_{h,g} \right \Vert_F^2 \\
    &+ u_n^2 \sum_{g, k \in [G]^2, g \neq k} \left\Vert \sum_{h \in [-gk]} B_{k,g} (\tilde P_{gk})_{g,g} P_{g,h} \Omega_{U,V,h} B_{h,g} \right \Vert_F^2 \\
    &+ u_n^2 \sum_{g, k \in [G]^2, g \neq k} \left\Vert \sum_{h \in [-gk]} B_{k,g} (\tilde P_{gk})_{g,k} P_{k,h} \Omega_{U,V,h} B_{h,g} \right \Vert_F^2 \\
    &+ u_n^{\frac{3q-4}{q-1}} u_G^{\frac{q}{q-1}} \lambda_n^3 \kappa_n \\
    &\lesssim u_n^4 \lambda_n^2 \kappa_n + u_n^{\frac{3q-4}{q-1}} u_G^{\frac{q}{q-1}} \lambda_n^3 \kappa_n = o(\omega_n^4).
\end{align*}
We also have
\begin{align*}
    R_{20,2} &\lesssim \mathbb V \left( \sum_{g, h, k \in [G]^3, g \neq h \neq k} U_g ' B_{h,g}' \Omega_{U,V,h} B_{h,g} \tilde M_{g,k,-gh} U_k \right) \\
    &+ \mathbb V \left( \sum_{g, h, k \in [G]^3, g \neq h \neq k} U_g ' B_{h,g}' (U_h V_h' - \Omega_{U,V,h}) B_{h,g} \tilde M_{g,k,-gh} U_k \right) \\
    &\lesssim u_n^2 \sum_{g,k \in [G]^2, g \neq k} \left \Vert \sum_{h \in [-gk]} B_{h,g}' \Omega_{U,V,h} B_{h,g} P_{g,k,-gh} \right \Vert_F^2 \\
    &+ u_n^{\frac{3q-4}{q-1}} u_G^{\frac{q}{q-1}} \lambda_n^2 \sum_{g, h \in [G]^2, g \neq h} || B_{h,g} ||_F^2 || \sum_{k \in [-gh]} P_{g,k,-gh} P_{k,g,-gh} ||_{op} \\
    &\lesssim u_n^2 \sum_{g,k \in [G]^2, g \neq k} \left \Vert \sum_{h \in [-gk]} B_{h,g}' \Omega_{U,V,h} B_{h,g} P_{g,k} \right \Vert_F^2 \\
    &+ u_n^2 \sum_{g,k \in [G]^2, g \neq k} \left \Vert \sum_{h \in [-gk]} B_{h,g}' \Omega_{U,V,h} B_{h,g} (\tilde P_{gh})_{g,g} P_{g,k} \right \Vert_F^2 \\
    &+ u_n^2 \sum_{g,k \in [G]^2, g \neq k} \left \Vert \sum_{h \in [-gk]} B_{h,g}' \Omega_{U,V,h} B_{h,g} (\tilde P_{gh})_{g,h} P_{h,k} \right \Vert_F^2 \\
    &+u_n^{\frac{3q-4}{q-1}} u_G^{\frac{q}{q-1}} \lambda_n^3 \kappa_n \\
    &\lesssim u_n^4 (\phi_n + \lambda_n^2) \lambda_n \kappa_n +u_n^{\frac{3q-4}{q-1}} u_G^{\frac{q}{q-1}} \lambda_n^3 \kappa_n = o(\omega_n^4),
\end{align*}
since $\phi_n \lesssim \lambda_n n_G$. Finally, we have
\begin{align*}
    R_{20,3} \lesssim u_n^4 \sum_{g, h, k, l \in [G]^4, g \neq h \neq k \neq l} ||B_{h,g}||_F^2 ||B_{k,g} P_{g,l,-gk}||_F^2 \lesssim u_n^4 n_G \lambda_n^2 \kappa_n = o(\omega_n^4).
\end{align*}

To conclude the proof, we note that $\left(\hat \omega_{n, 1}^2 + \hat \omega_{n, 2}^2 - \mathbb E \left( \hat \omega_{n,1}^2 + \hat \omega_{n,2}^2 \right) \right) / \omega_n^2 \convP 0$ and $\mathbb E \left( \hat \omega_{n,1}^2 + \hat \omega_{n,2}^2 \right) \geq \omega_n^2$, so that $\left(\hat \omega_{n, 1}^2 + \hat \omega_{n, 2}^2 \right) > 0$ with probability approaching one. Therefore, using the fact that $\hat \omega_{n,3}^2 \geq 0$, we have
\begin{align*}
    \mathbb P \left( \left(  \frac{\hat \theta_{\rm LO} - \theta}{\hat \omega_{n,\rm L2CO}} \right)^2 \geq {z}_{1-\alpha/2}^2 \right) &\leq \mathbb P \left( \frac{\left(\hat \theta_{\rm LO} - \theta\right)^2}{\left(\hat \omega_{n,1}^2 + \hat \omega_{n,2}^2\right)} \geq {z}_{1-\alpha/2}^2 \right) + \mathbb P \left( \left(\hat \omega_{n, 1}^2 + \hat \omega_{n, 2}^2 \right) \leq 0 \right) \\
    &\leq \mathbb P \left( \frac{\left(\hat \theta_{\rm LO} - \theta\right)^2}{\left(\hat \omega_{n,1}^2 + \hat \omega_{n,2}^2 - \mathbb E \left( \hat \omega_{n,1}^2 + \hat \omega_{n,2}^2 \right)\right) + \omega_n^2} \geq {z}_{1-\alpha/2}^2 \right) \\
    &+\mathbb P \left( \left(\hat \omega_{n, 1}^2 + \hat \omega_{n, 2}^2 - \mathbb E \left( \hat \omega_{n,1}^2 + \hat \omega_{n,2}^2 \right) \right) \leq - \omega_n^2 \right) + \mathbb P \left( \left(\hat \omega_{n, 1}^2 + \hat \omega_{n, 2}^2 \right) \leq 0 \right),
\end{align*}
and the proof is completed upon taking the limit for both sides and using \Cref{thm:clt}.

\subsection{Proof of consistency}\label{sec:proof-consistency}

The proof consists of two parts: in the first part we show that $\mathbb V \left(\hat \omega^2_{n,3} \right) = o(\omega_n^4)$ which, combined with the proof in \Cref{thm:var_l2co}, implies that $\left(\hat \omega_{n, \rm L2CO}^2 - \mathbb E \hat \omega_{n, \rm L2CO}^2 \right) / \omega_n^2 \convP 0$; in the second part we show that $\lim_{n \to \infty} \mathbb E \hat \omega_{n,\rm L2CO}^2/\omega_n^2 = 1$.

For the first part, we shall only compute the variance of the first term in $\hat \omega_{n,3}^2$, as the other terms can be handled similarly. We have
\begin{align*}
    \mathbb V \left( \sum_{g, h, k \in [G]^3} \left(X_h ' B_{h,g} \tilde U_{g,-h} \right) \left(X_k ' B_{k,g} \tilde U_{g,-k} \right) \right) &\lesssim \underbrace{\mathbb V \left( \sum_{g, h, k \in [G]^3} \left(\Pi_h ' B_{h,g} \tilde U_{g,-h} \right) \left(\Pi_k ' B_{k,g} \tilde U_{g,-k} \right) \right)}_{R_{21}} \\
    &+ \underbrace{\mathbb V \left( \sum_{g, h, k \in [G]^3} \left(\Pi_h ' B_{h,g} \tilde U_{g,-h} \right) \left(V_k ' B_{k,g} \tilde U_{g,-k} \right) \right)}_{R_{22}} \\
    &+ \underbrace{\mathbb V \left( \sum_{g, h, k \in [G]^3} \left(V_h ' B_{h,g} \tilde U_{g,-h} \right) \left(V_k ' B_{k,g} \tilde U_{g,-k} \right) \right)}_{R_{23}}.
\end{align*}

For $R_{21}$, we have
\begin{align*}
    R_{21} &=  \mathbb V \left( \sum_{g, h, k, l, j \in [G]^5, l \neq g, l \neq h,  j \neq g, j \neq k} \left(\Pi_h ' B_{h,g} \tilde M_{g,l,-gh} U_l \right) \left(\Pi_k ' B_{k,g} \tilde M_{g,j,-gk} U_j \right) \right) \\
    &\lesssim \underbrace{\mathbb V \left( \sum_{g, h, k \in [G]^3, l \in [-ghk]} \left(\Pi_h ' B_{h,g} P_{g,l,-gh} U_l \right) \left(\Pi_k ' B_{k,g} P_{g,l,-gk} U_l \right) \right)}_{R_{21,1}} \\
    &+ \underbrace{\mathbb V \left( \sum_{g, h, k \in [G]^3, l \in [-gh], j \in [-gk], l \neq j} \left(\Pi_h ' B_{h,g} P_{g,l,-gh} U_l \right) \left(\Pi_k ' B_{k,g} P_{g,j,-gk} U_j \right) \right)}_{R_{21,2}},
\end{align*}
where
\begin{align*}
    R_{21,1} &\lesssim u_n^{\frac{q-2}{q-1}} n_G^{\frac{q}{q-1}} \sum_{l \in [G]} \left \Vert \sum_{g, h, k \in [-l]^3} P_{l,g,-gh} B_{h,g}' \Pi_h \Pi_k ' B_{k,g} P_{g,l,-gk} \right \Vert_F^2 \\
    &\lesssim  u_n^{\frac{q-2}{q-1}} n_G^{\frac{q}{q-1}} \sum_{l \in [G]} \left \Vert \sum_{g, h, k \in [-l]^3} P_{l,g,-gh} B_{h,g}' \Pi_h \Pi_k ' B_{k,g} P_{g,l,-gk} \right \Vert_{op} \sum_{g, h, k \in [-l]^3} \Pi_k ' B_{k,g} P_{g,l,-gk} P_{l,g,-gh} B_{h,g}' \Pi_h\\
    &\lesssim u_n^{\frac{q-2}{q-1}} n_G^{\frac{q}{q-1}} (\phi_n \zeta_{H,n} + n_G \lambda_n^3 + n_G^2 \phi_n^2 \lambda_n) \sum_{l \in [G]}  \sum_{g, h, k \in [-l]^3} \Pi_k ' B_{k,g} P_{g,l,-gk} P_{l,g,-gh} B_{h,g}' \Pi_h \\
    &\lesssim u_n^{\frac{q-2}{q-1}} n_G^{\frac{q}{q-1}} (\phi_n \zeta_{H,n} + n_G \lambda_n^3 + n_G^2 \phi_n^2 \lambda_n) (\lambda_n \mu_n^2 + n_G \lambda_n^2 \kappa_n + n_G \phi_n \lambda_n^3 \kappa_n ) + o(\omega_n^4) = o(\omega_n^4),
\end{align*}
and
\begin{align*}
    R_{21,2} &\lesssim u_n^2 \sum_{l,j \in [G]^2, l \neq j} \left \Vert \sum_{g \in [-lj], h \in [-l], k \in [-j]} P_{l,g,-gh} B_{h,g}' \Pi_h \Pi_k ' B_{k,g} P_{g,j,-gk} \right \Vert_F^2 \\
    &\lesssim u_n^2 \sum_{l,j \in [G]^2, l \neq j} \left \Vert \sum_{g, h, k \in [G]^3} P_{l,g,-gh} B_{h,g}' \Pi_h \Pi_k ' B_{k,g} P_{g,j,-gk} \right \Vert_F^2 \\
    &+ u_n^2 \sum_{l,j \in [G]^2, l \neq j} \left \Vert \sum_{g \in [-lj], k \in [-j]} P_{l,g,-gl} B_{l,g}' \Pi_l \Pi_k ' B_{k,g} P_{g,j,-gk} \right \Vert_F^2 \\
    &+ u_n^2 \sum_{l,j \in [G]^2, l \neq j} \left \Vert \sum_{g \in [-lj], h \in [G]} P_{l,g,-gh} B_{h,g}' \Pi_h \Pi_j ' B_{j,g} P_{g,j,-gj} \right \Vert_F^2 \\
    &+ u_n^2 \sum_{l,j \in [G]^2, l \neq j} \left \Vert \sum_{h,k \in [G]^2} P_{l,l,-lh} B_{h,l}' \Pi_h \Pi_k ' B_{k,l} P_{l,j,-lk} \right \Vert_F^2 \\
    &+ u_n^2 \sum_{l,j \in [G]^2, l \neq j} \left \Vert \sum_{h,k \in [G]^2} P_{l,j,-jh} B_{h,j}' \Pi_h \Pi_k ' B_{k,j} P_{j,j,-jk} \right \Vert_F^2.
\end{align*}
We have
\begin{align*}
    &u_n^2 \sum_{l,j \in [G]^2, l \neq j} \left \Vert \sum_{g, h, k \in [G]^3} P_{l,g,-gh} B_{h,g}' \Pi_h \Pi_k ' B_{k,g} P_{g,j,-gk} \right \Vert_F^2 \\
    &\lesssim u_n^2 \sum_{g,g' \in [G]^2} \left( \sum_{h,k,l \in [G]^3} \Pi_k ' B_{k,g} P_{g,l,-gk} P_{l,g',-g'h} B_{h,g'}' \Pi_h \right)^2 \\
    &\lesssim u_n^2 \sum_{g,g' \in [G]^2} \left( H_g' P_{g,g'} H_{g'} \right)^2 + u_n^2 \sum_{g,g' \in [G]^2} \left(\sum_{h \in [G]} H_g' P_{g,g'} (\tilde P_{g'h})_{g',g'} B_{h,g'}' \Pi_h\right)^2 \\
    &+ u_n^2 \sum_{g,g' \in [G]^2} \left( \sum_{h,k \in [G]^2} \Pi_k ' B_{k,g} (\tilde P_{gk})_{g,g} P_{g,g'} (\tilde P_{g'h})_{g',g'} B_{h,g'}' \Pi_h \right)^2 \\
    &+ u_n^2 \sum_{g,g' \in [G]^2} \left( \sum_{h,k \in [G]^2} \Pi_k ' B_{k,g} (\tilde P_{gk})_{g,g} P_{g,h} (\tilde P_{g'h})_{h,g'} B_{h,g'}' \Pi_h \right)^2 \\
    &+ u_n^2 \sum_{g,g' \in [G]^2} \left( \sum_{h,k \in [G]^2} \Pi_k ' B_{k,g} (\tilde P_{gk})_{g,k} P_{k,h} (\tilde P_{g'h})_{h,g'} B_{h,g'}' \Pi_h \right)^2 \\
    &\lesssim u_n^2 \zeta_{H,n} \lambda_n \mu_n^2 + u_n^2 n_G^2 \zeta_{H_n} \phi_n \lambda_n \kappa_n + u_n^2 n_G^3 \phi_n^3 \lambda_n \kappa_n = o(\omega_n^4).
\end{align*}
We also have
\begin{align*}
    &u_n^2 \sum_{l,j \in [G]^2, l \neq j} \left \Vert \sum_{g \in [-lj], k \in [-j]} P_{l,g,-gl} B_{l,g}' \Pi_l \Pi_k ' B_{k,g} P_{g,j,-gk} \right \Vert_F^2  \\
    &\lesssim u_n^2 \sum_{l,j \in [G]^2}  \left \Vert \sum_{g,k \in [G]^2} P_{l,g,-gl} B_{l,g}' \Pi_l \Pi_k ' B_{k,g} P_{g,j,-gk} \right \Vert_F^2 \\
    &+ u_n^2 \sum_{l,j \in [G]^2, l \neq j} \left \Vert \sum_{g \in [-j]} P_{l,g,-gl} B_{l,g}' \Pi_l \Pi_j ' B_{j,g} P_{g,j,-gj} \right \Vert_F^2 \\
    &+ u_n^2 \sum_{l,j \in [G]^2, l \neq j}  \left \Vert \sum_{k \in [G]} P_{l,j,-jl} B_{l,j}' \Pi_l \Pi_k ' B_{k,j} P_{j,j,-jk} \right \Vert_F^2,
\end{align*}
where
\begin{align*}
    &u_n^2 \sum_{l,j \in [G]^2}  \left \Vert \sum_{g,k \in [G]^2} P_{l,g,-gl} B_{l,g}' \Pi_l \Pi_k ' B_{k,g} P_{g,j,-gk} \right \Vert_F^2 \\
    &\lesssim u_n^2 \sum_{l,j \in [G]^2}  \left \Vert \sum_{g,k \in [G]^2} P_{l,g,-gl} B_{l,g}' \Pi_l \Pi_k ' B_{k,g} P_{g,j} \right \Vert_F^2 \\
    &+ u_n^2 \sum_{l,j \in [G]^2}  \left \Vert \sum_{g,k \in [G]^2} P_{l,g,-gl} B_{l,g}' \Pi_l \Pi_k ' B_{k,g} (\tilde P_{gk})_{g,g} P_{g,j} \right \Vert_F^2 \\
    &+ u_n^2 \sum_{l,j \in [G]^2}  \left \Vert \sum_{g,k \in [G]^2} P_{l,g,-gl} B_{l,g}' \Pi_l \Pi_k ' B_{k,g} (\tilde P_{gk})_{g,k} P_{k,j} \right \Vert_F^2  \\
    &\lesssim u_n^2 n_G \zeta_{H,n} \phi_n \lambda_n \kappa_n + u_n^2 n_G^3 \phi_n^3 \lambda_n \kappa_n = o(\omega_n^4),
\end{align*}
and
\begin{align*}
    &u_n^2 \sum_{l,j \in [G]^2, l \neq j} \left \Vert \sum_{g \in [-j]} P_{l,g,-gl} B_{l,g}' \Pi_l \Pi_j ' B_{j,g} P_{g,j,-gj} \right \Vert_F^2 \\
    &\lesssim u_n^2 \sum_{l,j \in [G]^2} \left \Vert \sum_{g \in [G]} P_{l,g,-gl} B_{l,g}' \Pi_l \Pi_j ' B_{j,g} P_{g,j,-gj} \right \Vert_F^2 \\
    &\lesssim u_n^2 n_G^3 \phi_n \lambda_n^3 \kappa_n = o(\omega_n^4),
\end{align*}
and
\begin{align*}
    &u_n^2 \sum_{l,j \in [G]^2, l \neq j}  \left \Vert \sum_{k \in [G]} P_{l,j,-jl} B_{l,j}' \Pi_l \Pi_k ' B_{k,j} P_{j,j,-jk} \right \Vert_F^2 \\
    &\lesssim u_n^2 \sum_{l,j \in [G]^2}  \left \Vert \sum_{k \in [G]} P_{l,j,-jl} B_{l,j}' \Pi_l \Pi_k ' B_{k,j} P_{j,j} \right \Vert_F^2 \\
    &+ u_n^2 \sum_{l,j \in [G]^2}  \left \Vert \sum_{k \in [G]} P_{l,j,-jl} B_{l,j}' \Pi_l \Pi_k ' B_{k,j} (\tilde P_{jk})_{j,j}P_{j,j} \right \Vert_F^2 \\
    &+ u_n^2 \sum_{l,j \in [G]^2}  \left \Vert \sum_{k \in [G]} P_{l,j,-jl} B_{l,j}' \Pi_l \Pi_k ' B_{k,j} (\tilde P_{jk})_{j,k}P_{k,j} \right \Vert_F^2 \\
    &\lesssim u_n^2 n_G \zeta_{H,n} \lambda_n^4 \kappa_n + u_n^2 n_G^3 \phi_n \lambda_n^4 \kappa_n = o(\omega_n^4),
\end{align*}
and similarly
\begin{align*}
    &u_n^2 \sum_{l,j \in [G]^2, l \neq j} \left \Vert \sum_{g \in [-lj], h \in [G]} P_{l,g,-gh} B_{h,g}' \Pi_h \Pi_j ' B_{j,g} P_{g,j,-gj} \right \Vert_F^2 \\
    &\lesssim u_n^2 \sum_{l,j \in [G]^2} \left \Vert \sum_{g, h \in [G]^2} P_{l,g,-gh} B_{h,g}' \Pi_h \Pi_j ' B_{j,g} P_{g,j,-gj} \right \Vert_F^2 \\
    &+ u_n^2 \sum_{l,j \in [G]^2} \left \Vert \sum_{h \in [G]} P_{l,l,-lh} B_{h,l}' \Pi_h \Pi_j ' B_{j,l} P_{l,j,-lj} \right \Vert_F^2 = o(\omega_n^4).
\end{align*}
Lastly, we have
\begin{align*}
    &u_n^2 \sum_{l,j \in [G]^2, l \neq j} \left \Vert \sum_{h,k \in [G]^2} P_{l,l,-lh} B_{h,l}' \Pi_h \Pi_k ' B_{k,l} P_{l,j,-lk} \right \Vert_F^2 \\
    &\lesssim u_n^2 \sum_{l,j \in [G]^2} \left \Vert \sum_{h,k \in [G]^2} P_{l,l,-lh} B_{h,l}' \Pi_h \Pi_k ' B_{k,l} P_{l,j} \right \Vert_F^2 \\
    &+ u_n^2 \sum_{l,j \in [G]^2} \left \Vert \sum_{h,k \in [G]^2} P_{l,l,-lh} B_{h,l}' \Pi_h \Pi_k ' B_{k,l} (\tilde P_{lk})_{l,l} P_{l,j} \right \Vert_F^2 \\
    &+ u_n^2 \sum_{l,j \in [G]^2} \left \Vert \sum_{h,k \in [G]^2} P_{l,l,-lh} B_{h,l}' \Pi_h \Pi_k ' B_{k,l} (\tilde P_{lk})_{l,k} P_{k,j} \right \Vert_F^2 \\
    &\lesssim u_n^2 \zeta_{H,n} \lambda_n^3 \mu_n^2 + u_n^2 n_G \zeta_{H,n} \phi_n \lambda_n^3 \kappa_n + u_n^2 n_G \zeta_{H,n} \phi_n^2 \lambda_n^2 \kappa_n+ u_n^2 n_G^3 \phi_n^3 \lambda_n^3 \kappa_n = o(\omega_n^4),
\end{align*}
and similarly for
\begin{align*}
    u_n^2 \sum_{l,j \in [G]^2, l \neq j} \left \Vert \sum_{h,k \in [G]^2} P_{l,j,-jh} B_{h,j}' \Pi_h \Pi_k ' B_{k,j} P_{j,j,-jk} \right \Vert_F^2.
\end{align*}

For $R_{22}$, we have
\begin{align*}
    R_{22} &=  \mathbb V \left( \sum_{g, h, k, l, j \in [G]^5, l \neq g, l \neq h, j \neq g, j \neq k} \left(\Pi_h ' B_{h,g} \tilde M_{g,l,-gh} U_l \right) \left(V_k ' B_{k,g} \tilde M_{g,j,-gk} U_j \right) \right) \\
    &\lesssim \underbrace{\mathbb V \left( \sum_{g, h \in [G]^2, j \in [-g], k \in [-ghj]} \left(\Pi_h ' B_{h,g} P_{g,k,-gh} U_k \right) \left(V_k ' B_{k,g} P_{g,j,-gk} U_j \right) \right)}_{R_{22,1}} \\
    &+ \underbrace{\mathbb V \left( \sum_{g, h, k \in [G]^3, l \in [-ghk]} \left(\Pi_h ' B_{h,g} P_{g,l,-gh} U_l \right) \left(V_k ' B_{k,g} P_{g,l,-gk} U_l \right) \right)}_{R_{22,2}} \\
    &+ \underbrace{\mathbb V \left( \sum_{g, h, k \in [G]^3, l \in[-ghk], j \in [-gkl]} \left(\Pi_h ' B_{h,g} P_{g,l,-gh} U_l \right) \left(V_k ' B_{k,g} P_{g,j,-gk} U_j \right) \right)}_{R_{22,3}}.
\end{align*}
Consider first
\begin{align*}
    R_{22,1} &\lesssim \mathbb V \left( \sum_{g, h \in [G]^2, j \in [-g], k \in [-hj]} \Pi_h ' B_{h,g} P_{g,k,-gh} \Omega_{U,V,k} B_{k,g} P_{g,j,-gk} U_j \right) \\
    &+ \mathbb V \left( \sum_{g, h \in [G]^2, j \in [-g], k \in [-hj]} \Pi_h ' B_{h,g} P_{g,k,-gh} (U_k V_k' - \Omega_{U,V,k}) B_{k,g} P_{g,j,-gk} U_j \right) \\
    &\lesssim u_n \sum_{j \in [G]} \left \Vert \sum_{g, h, k \in [G]^3, g \neq j, k \neq j, k \neq h} P_{j,g,-gk} B_{k,g}' \Omega_{U,V,k}' P_{k,g,-gh} B_{h,g}' \Pi_h \right \Vert_2^2 \\
    &+ \sum_{k,j \in [G]^2, k \neq j} \mathbb E \left(\sum_{g, h \in [G]^2, g \neq j, h \neq k}\Pi_h ' B_{h,g} P_{g,k,-gh} U_k V_k' B_{k,g} P_{g,j,-gk} U_j \right)^2.
\end{align*}
We have
\begin{align*}
    &u_n \sum_{j \in [G]} \left \Vert \sum_{g, h, k \in [G]^3, g \neq j, k \neq j, k \neq h} P_{j,g,-gk} B_{k,g}' \Omega_{U,V,k}' P_{k,g,-gh} B_{h,g}' \Pi_h \right \Vert_2^2 \\
    &\lesssim u_n \sum_{j \in [G]} \left \Vert \sum_{g, h, k \in [G]^3} P_{j,g,-gk} B_{k,g}' \Omega_{U,V,k}' P_{k,g,-gh} B_{h,g}' \Pi_h \right \Vert_2^2 \\
    &+ u_n \sum_{j \in [G]} \left \Vert \sum_{h, k \in [G]^2, k \neq j, k \neq h} P_{j,j,-jk} B_{k,j}' \Omega_{U,V,k}' P_{k,j,-jh} B_{h,j}' \Pi_h \right \Vert_2^2 \\
    &+ u_n \sum_{j \in [G]} \left \Vert \sum_{g, h \in [G]^2} P_{j,g,-gj} B_{j,g}' \Omega_{U,V,j}' P_{j,g,-gh} B_{h,g}' \Pi_h \right \Vert_2^2 \\
    &+ u_n \sum_{j \in [G]} \left \Vert \sum_{g, h \in [G]^2} P_{j,g,-gh} B_{h,g}' \Omega_{U,V,h}' P_{h,g,-gh} B_{h,g}' \Pi_h \right \Vert_2^2 \\
    &+ u_n \sum_{j \in [G]} \left \Vert \sum_{g \in [G]} P_{j,g,-gj} B_{j,g}' \Omega_{U,V,j}' P_{j,g,-gj} B_{j,g}' \Pi_j \right \Vert_2^2,
\end{align*}
where
\begin{align*}
    & u_n \sum_{j \in [G]} \left \Vert \sum_{g, h, k \in [G]^3} P_{j,g,-gk} B_{k,g}' \Omega_{U,V,k}' P_{k,g,-gh} B_{h,g}' \Pi_h \right \Vert_2^2 \\
    &\lesssim u_n \sum_{j \in [G]}\sum_{g, h, k \in [G]^3} \sum_{g',h',k' \in [G]^3} \Pi_h ' B_{h,g} P_{g,k,-gh} \Omega_{U,V,k} B_{k,g} P_{g,j,-gk} P_{j,g',-g'k'} B_{k',g'}' \Omega_{U,V,k'}' P_{k',g',-g'h'} B_{h'g'}' \Pi_{h'} \\
    &\lesssim u_n \sum_{g, h, k \in [G]^3} \sum_{g',h',k' \in [G]^3} \Pi_h ' B_{h,g} P_{g,k,-gh} \Omega_{U,V,k} B_{k,g} \left( \sum_{j \in [G]} P_{g,j} P_{j,g'} \right)  B_{k',g'}' \Omega_{U,V,k'}' P_{k',g',-g'h'} B_{h'g'}' \Pi_{h'} \\
    &+ u_n \sum_{g, h, k \in [G]^3} \sum_{g',h',k' \in [G]^3} \Pi_h ' B_{h,g} P_{g,k,-gh} \Omega_{U,V,k} B_{k,g} (\tilde P_{gk})_{g,g} \\
    &\left( \sum_{j \in [G]} P_{g,j} P_{j,g'} \right) (\tilde P_{g'k'})_{g',g'}  B_{k',g'}' \Omega_{U,V,k'}' P_{k',g',-g'h'} B_{h'g'}' \Pi_{h'} \\
    &+ u_n \sum_{g, h, k \in [G]^3} \sum_{g',h',k' \in [G]^3} \Pi_h ' B_{h,g} P_{g,k,-gh} \Omega_{U,V,k} B_{k,g} (\tilde P_{gk})_{g,k} \\
    &\left( \sum_{j \in [G]} P_{k,j} P_{j,k'} \right) (\tilde P_{g'k'})_{k',g'}  B_{k',g'}' \Omega_{U,V,k'}' P_{k',g',-g'h'} B_{h'g'}' \Pi_{h'} \\
    &\lesssim u_n \sum_{g \in [G]} \left \Vert \sum_{h,k \in [G]^2} B_{k,g}' \Omega_{U,V,k}' P_{k,g,-gh} B_{h,g}' \Pi_{h} \right \Vert_2^2 + u_n \sum_{g \in [G]} \left \Vert \sum_{h,k \in [G]^2} (\tilde P_{gk})_{g,g} B_{k,g}' \Omega_{U,V,k}' P_{k,g,-gh} B_{h,g}' \Pi_{h} \right \Vert_2^2 \\
    &+ u_n \sum_{k \in [G]}  \left \Vert \sum_{g,h \in [G]^2} (\tilde P_{gk})_{k,g} B_{k,g}' \Omega_{U,V,k}' P_{k,g,-gh} B_{h,g}' \Pi_{h} \right \Vert_2^2 \\
    &\lesssim u_n^3 \lambda_n^2 \mu_n^2 + u_n^3 n_G \phi_n \lambda_n^2 \kappa_n + u_n^3 \zeta_{H,n} \phi_n \lambda_n^2 \kappa_n + u_n^3 n_G^2 \phi_n^2 \lambda_n^3 \kappa_n = o(\omega_n^4),
\end{align*}
and
\begin{align*}
    &u_n \sum_{j \in [G]} \left \Vert \sum_{h, k \in [G]^2, k \neq j, k \neq h} P_{j,j,-jk} B_{k,j}' \Omega_{U,V,k}' P_{k,j,-jh} B_{h,j}' \Pi_h \right \Vert_2^2 \\
    &\lesssim u_n \sum_{j \in [G]} \left \Vert \sum_{h, k \in [G]^2} P_{j,j,-jk} B_{k,j}' \Omega_{U,V,k}' P_{k,j,-jh} B_{h,j}' \Pi_h \right \Vert_2^2 \\
    &+ u_n \sum_{j \in [G]} \left \Vert \sum_{h \in [G]} P_{j,j,-jh} B_{h,j}' \Omega_{U,V,h}' P_{h,j,-jh} B_{h,j}' \Pi_h \right \Vert_2^2 \\
    &\lesssim u_n^3 (\phi_n + \lambda_n^2)\lambda_n^4 \mu_n^2 + u_n^3 n_G (\phi_n + \lambda_n^2) \phi_n \lambda_n^4 \kappa_n + u_n^3 n_G (\phi_n + \lambda_n^2) \lambda_n^4 \kappa_n = o(\omega_n^4),
\end{align*}
and
\begin{align*}
    &u_n \sum_{j \in [G]} \left \Vert \sum_{g, h \in [G]^2} P_{j,g,-gj} B_{j,g}' \Omega_{U,V,j}' P_{j,g,-gh} B_{h,g}' \Pi_h \right \Vert_2^2 \\
    &\lesssim u_n^3 (\phi_n + \lambda_n^2) \lambda_n^2 \sum_{g,j \in [G]^2} \left \Vert \sum_{h \in [G]} P_{j,g,-gh} B_{h,g}' \Pi_h \right \Vert_2^2 \\
    &\lesssim u_n^3 (\phi_n + \lambda_n^2) \lambda_n^3 \mu_n^2 + u_n^3 n_G (\phi_n + \lambda_n^2) \phi_n \lambda_n^5 \kappa_n + u_n^3 n_G (\phi_n + \lambda_n^2) \lambda_n^4 \kappa_n = o(\omega_n^4),
\end{align*}
and
\begin{align*}
    &u_n \sum_{j \in [G]} \left \Vert \sum_{g, h \in [G]^2} P_{j,g,-gh} B_{h,g}' \Omega_{U,V,h}' P_{h,g,-gh} B_{h,g}' \Pi_h \right \Vert_2^2 \\
    &\lesssim u_n \sum_{g \in [G]} \left \Vert \sum_{h \in [G]} B_{h,g}' \Omega_{U,V,h}' P_{h,g,-gh} B_{h,g}' \Pi_h \right \Vert_2^2 + u_n \sum_{g \in [G]} \left \Vert \sum_{h \in [G]} (\tilde P_{gh})_{g,g} B_{h,g}' \Omega_{U,V,h}' P_{h,g,-gh} B_{h,g}' \Pi_h \right \Vert_2^2 \\
    &+ u_n \sum_{h \in [G]} \left \Vert \sum_{g \in [G]} (\tilde P_{gh})_{h,g} B_{h,g}' \Omega_{U,V,h}' P_{h,g,-gh} B_{h,g}' \Pi_h \right \Vert_F^2 \\
    &\lesssim u_n^3 n_G (\phi_n + \lambda_n^2) \lambda_n^2 \kappa_n  = o(\omega_n^4),
\end{align*}
and
\begin{align*}
    &u_n \sum_{j \in [G]} \left \Vert \sum_{g \in [G]} P_{j,g,-gj} B_{j,g}' \Omega_{U,V,j}' P_{j,g,-gj} B_{j,g}' \Pi_j \right \Vert_2^2 \\
    &\lesssim u_n^3 n_G (\phi_n + \lambda_n^2) \lambda_n^4 \kappa_n = o(\omega_n^4).
\end{align*}
We also have
\begin{align*}
    & \sum_{k,j \in [G]^2, k \neq j} \mathbb E \left(\sum_{g, h \in [G]^2, g \neq j, h \neq k}\Pi_h ' B_{h,g} P_{g,k,-gh} U_k V_k' B_{k,g} P_{g,j,-gk} U_j \right)^2 \\
    &\lesssim \sum_{k,j \in [G]^2, k \neq j} \mathbb E \left(\sum_{g, h \in [G]^2}\Pi_h ' B_{h,g} P_{g,k,-gh} U_k V_k' B_{k,g} P_{g,j,-gk} U_j \right)^2 \\
    &+ \sum_{k,j \in [G]^2, k \neq j} \mathbb E \left(\sum_{g \in [-j]}\Pi_k ' B_{k,g} P_{g,k,-gk} U_k V_k' B_{k,g} P_{g,j,-gk} U_j \right)^2 \\
    &+ \sum_{k,j \in [G]^2, k \neq j} \mathbb E \left(\sum_{h \in [G]}\Pi_h ' B_{h,j} P_{j,k,-jh} U_k V_k' B_{k,j} P_{j,j,-jk} U_j \right)^2,
\end{align*}
where
\begin{align*}
    &\sum_{k,j \in [G]^2, k \neq j} \mathbb E \left(\sum_{g, h \in [G]^2}\Pi_h ' B_{h,g} P_{g,k,-gh} U_k V_k' B_{k,g} P_{g,j,-gk} U_j \right)^2 \\
    &\lesssim u_n \mathbb E\sum_{k \in [G]} \sum_{g, h \in [G]^2} \sum_{g', h' \in [G]^2}  \Pi_h ' B_{h,g} P_{g,k,-gh} U_k V_k' B_{k,g} \left( \sum_{j \in [G]} P_{g,j,-gk} P_{j,g',-g'k} \right) B_{k,g'}' V_k U_k' P_{k,g',-g'h'} B_{h',g'}' \Pi_{h'} \\
    &\lesssim u_n \sum_{k \in [G]} \mathbb E \left\Vert \sum_{g, h \in [G]^2} B_{k,g}' V_k U_k' P_{k,g,-gh} B_{h,g}' \Pi_{h} \right \Vert_2^2 + u_n \sum_{k \in [G]} \mathbb E \left\Vert \sum_{g, h \in [G]^2} (\tilde P_{gk})_{g,g} B_{k,g}' V_k U_k' P_{k,g,-gh} B_{h,g}' \Pi_{h} \right \Vert_2^2 \\
    &+ u_n \lambda_n \sum_{k \in [G]} \mathbb E \left\Vert \sum_{g, h \in [G]^2} (\tilde P_{gk})_{k,g} B_{k,g}' V_k U_k' P_{k,g,-gh} B_{h,g}' \Pi_{h} \right \Vert_2^2 \\
    &\lesssim u_n^{\frac{2q-3}{q-1}} n_G^{\frac{2q-1}{q-1}} \zeta_{H,n} \lambda_n \kappa_n + u_n^{\frac{2q-3}{q-1}} n_G^{\frac{3q-2}{q-1}} \phi_n^2 \lambda_n \kappa_n = o(\omega_n^4),
\end{align*}
and
\begin{align*}
    &\sum_{k,j \in [G]^2, k \neq j} \mathbb E \left(\sum_{g \in [-j]}\Pi_k ' B_{k,g} P_{g,k,-gk} U_k V_k' B_{k,g} P_{g,j,-gk} U_j \right)^2 \\
    &\lesssim \sum_{k,j \in [G]^2, k \neq j} \mathbb E \left(\sum_{g \in [G]}\Pi_k ' B_{k,g} P_{g,k,-gk} U_k V_k' B_{k,g} P_{g,j,-gk} U_j \right)^2 \\
    &+ \sum_{k,j \in [G]^2, k \neq j} \mathbb E \left(\Pi_k ' B_{k,j} P_{j,k,-jk} U_k V_k' B_{k,j} P_{j,j,-jk} U_j \right)^2 \\
    &\lesssim u_n^{\frac{2q-3}{q-1}} n_G^{\frac{2q-1}{q-1}} \lambda_n^4 \kappa_n + u_n^{\frac{2q-3}{q-1}} n_G^{\frac{2q-1}{q-1}} \phi_n \lambda_n^5 \kappa_n + u_n^{\frac{2q-3}{q-1}} n_G^{\frac{2q-1}{q-1}} \lambda_n^6 \kappa_n = o(\omega_n^4),
\end{align*}
and
\begin{align*}
    &\sum_{k,j \in [G]^2, k \neq j} \mathbb E \left(\sum_{h \in [G]}\Pi_h ' B_{h,j} P_{j,k,-jh} U_k V_k' B_{k,j} P_{j,j,-jk} U_j \right)^2 \\
    &\lesssim u_n^{\frac{2q-3}{q-1}} n_G^{\frac{q}{q-1}} \lambda_n^5 \mu_n^2 + u_n^{\frac{2q-3}{q-1}} n_G^{\frac{2q-1}{q-1}} \phi_n \lambda_n^7 \kappa_n + u_n^{\frac{2q-3}{q-1}} n_G^{\frac{2q-1}{q-1}} \lambda_n^6 \kappa_n = o(\omega_n^4).
\end{align*}
Consider next
\begin{align*}
    R_{22,2} &\lesssim \mathbb V \left( \sum_{g, h, k \in [G]^3, l \in [-ghk]} \Pi_h ' B_{h,g} P_{g,l,-gh} \Omega_{U,l} P_{l,g,-gk} B_{k,g}' V_k   \right) \\
    &+ \mathbb V \left( \sum_{g, h, k \in [G]^3, l \in [-ghk]} \Pi_h ' B_{h,g} P_{g,l,-gh} (U_l U_l' - \Omega_{U,l}) P_{l,g,-gk} B_{k,g}' V_k   \right) \\
    &\lesssim u_n \sum_{k \in [G]} \left \Vert \sum_{g, h \in [G]^2, l \in [-ghk]} B_{k,g} P_{g,l,-gk} \Omega_{U,l} P_{l,g,-gh} B_{h,g}' \Pi_h \right \Vert_2^2 \\
    &+ \sum_{k,l \in [G]^2, k \neq l} \mathbb E \left(\sum_{g,h \in [-l]^2} \Pi_h ' B_{h,g} P_{g,l,-gh} U_l U_l' P_{l,g,-gk} B_{k,g}' V_k   \right)^2,
\end{align*}
where
\begin{align*}
    &u_n \sum_{k \in [G]} \left \Vert \sum_{g, h \in [G]^2, l \in [-ghk]} B_{k,g} P_{g,l,-gk} \Omega_{U,l} P_{l,g,-gh} B_{h,g}' \Pi_h \right \Vert_2^2 \\
    &\lesssim u_n \sum_{k \in [G]} \left \Vert \sum_{g, h, l \in [G]^3} B_{k,g} P_{g,l,-gk} \Omega_{U,l} P_{l,g,-gh} B_{h,g}' \Pi_h \right \Vert_2^2 \\
    &+ u_n \sum_{k \in [G]} \left \Vert \sum_{g, h \in [G]^2} B_{k,g} P_{g,g,-gk} \Omega_{U,g} P_{g,g,-gh} B_{h,g}' \Pi_h \right \Vert_2^2 \\
    &+ u_n \sum_{k \in [G]} \left \Vert \sum_{g, h \in [G]^2} B_{k,g} P_{g,h,-gk} \Omega_{U,h} P_{h,g,-gh} B_{h,g}' \Pi_h \right \Vert_2^2 \\
    &+ u_n \sum_{k \in [G]} \left \Vert \sum_{g, h \in [G]^2} B_{k,g} P_{g,k,-gk} \Omega_{U,k} P_{k,g,-gh} B_{h,g}' \Pi_h \right \Vert_2^2 \\
    &+ u_n \sum_{k \in [G]} \left \Vert \sum_{g \in [G]} B_{k,g} P_{g,k,-gk} \Omega_{U,k} P_{k,g,-gk} B_{k,g}' \Pi_k \right \Vert_2^2.
\end{align*}
We have
\begin{multline*}
  u_n \sum_{k \in [G]} \left \Vert \sum_{g, h, l \in [G]^3} B_{k,g} P_{g,l,-gk} \Omega_{U,l} P_{l,g,-gh} B_{h,g}' \Pi_h \right \Vert_2^2 \\
  \lesssim u_n \sum_{k \in [G]} \left \Vert \sum_{g, h, l \in [G]^3} B_{k,g} P_{g,l} \Omega_{U,l} P_{l,g,-gh} B_{h,g}' \Pi_h \right \Vert_2^2\\
  + u_n \sum_{k \in [G]} \left \Vert \sum_{g, h, l \in [G]^3} B_{k,g} (\tilde P_{gk})_{g,g} P_{g,l} \Omega_{U,l} P_{l,g,-gh} B_{h,g}' \Pi_h \right \Vert_2^2 \\
  + u_n \sum_{k \in [G]} \left \Vert \sum_{g, h, l \in [G]^3} B_{k,g} (\tilde P_{gk})_{g,k} P_{k,l} \Omega_{U,l} P_{l,g,-gh} B_{h,g}' \Pi_h \right \Vert_2^2,
\end{multline*}
with
\begin{multline*}
    u_n \sum_{k \in [G]} \left \Vert \sum_{g, h, l \in [G]^3} B_{k,g} P_{g,l,-gk} \Omega_{U,l} P_{l,g,-gh} B_{h,g}' \Pi_h \right \Vert_2^2 \\
    \lesssim u_n \sum_{k \in [G]} \left \Vert \sum_{g \in [G]} B_{k,g} \left( \sum_{l \in [G]}P_{g,l} \Omega_{U,l} P_{l,g}\right) H_g \right \Vert_2^2\\ + u_n \sum_{k \in [G]} \left \Vert \sum_{g, h \in [G]^2} B_{k,g} \left( \sum_{l \in [G]} P_{g,l} \Omega_{U,l} P_{l,g}\right) (\tilde P_{gh})_{g,g} B_{h,g}' \Pi_h \right \Vert_2^2 \\
    + u_n \sum_{k \in [G]} \left \Vert \sum_{g, h \in [G]^2} B_{k,g} \left( \sum_{l \in [G]} P_{g,l} \Omega_{U,l} P_{l,h}\right) (\tilde P_{gh})_{h,g} B_{h,g}' \Pi_h \right \Vert_2^2 \\
    \lesssim u_n^3 \lambda_n^2 \mu_n^2 + u_n^3 n_G \phi_n \lambda_n^2 \kappa_n = o(\omega_n^4),
\end{multline*}
where we use the fact that $\lambda_{\max}(B'B) = O(1)$, and
\begin{multline*}
  u_n \sum_{k \in [G]} \left \Vert \sum_{g, h, l \in [G]^3} B_{k,g} (\tilde P_{gk})_{g,g} P_{g,l} \Omega_{U,l} P_{l,g,-gh} B_{h,g}' \Pi_h \right \Vert_2^2 \\
  \lesssim u_n \sum_{k \in [G]} \left \Vert \sum_{g, h, l \in [G]^3} B_{k,g} (\tilde P_{gk})_{g,g} P_{g,l} \Omega_{U,l} P_{l,g} B_{h,g}' \Pi_h \right \Vert_2^2\\
  + u_n \sum_{k \in [G]} \left \Vert \sum_{g, h, l \in [G]^3} B_{k,g} (\tilde P_{gk})_{g,g} P_{g,l} \Omega_{U,l} P_{l,g} (\tilde P_{gh})_{g,g} B_{h,g}' \Pi_h \right \Vert_2^2 \\
  + u_n \sum_{k \in [G]} \left \Vert \sum_{g, h, l \in [G]^3} B_{k,g} (\tilde P_{gk})_{g,g} P_{g,l} \Omega_{U,l} P_{l,h} (\tilde P_{gh})_{h,g} B_{h,g}' \Pi_h \right \Vert_2^2 \\
  \lesssim u_n^3 \zeta_{H,n} \phi_n \lambda_n^2 \kappa_n + u_n^3 n_G^2 \phi_n^2 \lambda_n^3 \kappa_n = o(\omega_n^4),
\end{multline*}
and
\begin{multline*}
  u_n \sum_{k \in [G]} \left \Vert \sum_{g, h, l \in [G]^3} B_{k,g} (\tilde P_{gk})_{g,k} P_{k,l} \Omega_{U,l} P_{l,g,-gh} B_{h,g}' \Pi_h \right \Vert_2^2 \\
  \lesssim
  u_n \sum_{k \in [G]} \left \Vert \sum_{g, h, l \in [G]^3} B_{k,g} (\tilde P_{gk})_{g,k} P_{k,l} \Omega_{U,l} P_{l,g} B_{h,g}' \Pi_h \right \Vert_2^2 \\
  + u_n \sum_{k \in [G]} \left \Vert \sum_{g, h, l \in [G]^3} B_{k,g} (\tilde P_{gk})_{g,k} P_{k,l} \Omega_{U,l} P_{l,g} (\tilde P_{gh})_{g,g} B_{h,g}' \Pi_h \right \Vert_2^2 \\
  +u_n \sum_{k \in [G]} \left \Vert \sum_{g, h, l \in [G]^3} B_{k,g} (\tilde P_{gk})_{g,k} P_{k,l} \Omega_{U,l} P_{l,h} (\tilde P_{gh})_{h,g} B_{h,g}' \Pi_h \right \Vert_2^2 \\
  \lesssim u_n^3 \zeta_{H,n} \phi_n \lambda_n^2 \kappa_n + u_n^3 n_G^2 \phi_n^2 \lambda_n^3 \kappa_n = o(\omega_n^4).
\end{multline*}
We also have
\begin{multline*}
  u_n \sum_{k \in [G]} \left \Vert \sum_{g, h \in [G]^2} B_{k,g} P_{g,g,-gk} \Omega_{U,g} P_{g,g,-gh} B_{h,g}' \Pi_h \right \Vert_2^2 \\
  \lesssim u_n \sum_{k \in [G]} \left \Vert \sum_{g, h \in [G]^2} B_{k,g} P_{g,g} \Omega_{U,g} P_{g,g,-gh} B_{h,g}' \Pi_h \right \Vert_2^2 \\
  + u_n \sum_{k \in [G]} \left \Vert \sum_{g, h \in [G]^2} B_{k,g} (\tilde P_{gk})_{g,g} P_{g,g} \Omega_{U,g} P_{g,g,-gh} B_{h,g}' \Pi_h \right \Vert_2^2 \\
  +u_n \sum_{k \in [G]} \left \Vert \sum_{g, h \in [G]^2} B_{k,g} (\tilde P_{gk})_{g,k} P_{k,g} \Omega_{U,g} P_{g,g,-gh} B_{h,g}' \Pi_h \right \Vert_2^2 \\
  \lesssim u_n^3 \lambda_n^4 \mu_n^2 + u_n^3 n_G \phi_n \lambda_n^4 \kappa_n + u_n^3 (\phi_n + \lambda_n^2) \phi_n \lambda_n^4 \mu_n^2 + u_n^3 n_G (\phi_n + \lambda_n^2) \phi_n^2 \lambda_n^4 \kappa_n = o(\omega_n^4),
\end{multline*}
and
\begin{multline*}
  u_n \sum_{k \in [G]} \left \Vert \sum_{g, h \in [G]^2} B_{k,g} P_{g,h,-gk} \Omega_{U,h} P_{h,g,-gh} B_{h,g}' \Pi_h \right \Vert_2^2 \\
  \lesssim u_n \sum_{k \in [G]} \left \Vert \sum_{g, h \in [G]^2} B_{k,g} P_{g,h} \Omega_{U,h} P_{h,g,-gh} B_{h,g}' \Pi_h \right \Vert_2^2 \\
  + u_n \sum_{k \in [G]} \left \Vert \sum_{g, h \in [G]^2} B_{k,g} (\tilde P_{gk})_{g,g} P_{g,h} \Omega_{U,h} P_{h,g,-gh} B_{h,g}' \Pi_h \right \Vert_2^2 \\
  + u_n \sum_{k \in [G]} \left \Vert \sum_{g, h \in [G]^2} B_{k,g} (\tilde P_{gk})_{g,k} P_{k,h} \Omega_{U,h} P_{h,g,-gh} B_{h,g}' \Pi_h \right \Vert_2^2 \\
  \lesssim u_n^3 n_G \phi_n \lambda_n^2 \kappa_n + u_n^3 n_G (\phi_n + \lambda_n^2) \phi_n^2 \lambda_n^2 \kappa_n = o(\omega_n^4),
\end{multline*}
and
\begin{align*}
    &u_n \sum_{k \in [G]} \left \Vert \sum_{g, h \in [G]^2} B_{k,g} P_{g,k,-gk} \Omega_{U,k} P_{k,g,-gh} B_{h,g}' \Pi_h \right \Vert_2^2 \\
    &\lesssim u_n \sum_{k \in [G]} \left \Vert \sum_{g, h \in [G]^2} B_{k,g} P_{g,k,-gk} \Omega_{U,k} P_{k,g} B_{h,g}' \Pi_h \right \Vert_2^2 + u_n \sum_{k \in [G]} \left \Vert \sum_{g, h \in [G]^2} B_{k,g} P_{g,k,-gk} \Omega_{U,k} P_{k,g} (\tilde P_{gh})_{g,g} B_{h,g}' \Pi_h \right \Vert_2^2 \\
    &+ u_n \sum_{k \in [G]} \left \Vert \sum_{g, h \in [G]^2} B_{k,g} P_{g,k,-gk} \Omega_{U,k} P_{k,h} (\tilde P_{gh})_{h,g} B_{h,g}' \Pi_h \right \Vert_2^2 \\
    & \lesssim u_n^3 (\phi_n + \lambda_n^2) \lambda_n^3 \mu_n^2 + u_n^3 n_G (\phi_n + \lambda_n^2) \phi_n^2 \lambda_n^2 \kappa_n = o(\omega_n^4),
\end{align*}
and
\begin{align*}
    &u_n \sum_{k \in [G]} \left \Vert \sum_{g \in [G]} B_{k,g} P_{g,k,-gk} \Omega_{U,k} P_{k,g,-gk} B_{k,g}' \Pi_k \right \Vert_2^2 \\
    &\lesssim u_n^3 n_G \phi_n \lambda_n^4 \kappa_n = o(\omega_n^4).
\end{align*}
In addition,
\begin{align*}
    &\sum_{k,l \in [G]^2, k \neq l} \mathbb E \left(\sum_{g,h \in [-l]^2} \Pi_h ' B_{h,g} P_{g,l,-gh} U_l U_l' P_{l,g,-gk} B_{k,g}' V_k   \right)^2 \\
    & \lesssim \sum_{k,l \in [G]^2, k \neq l} \mathbb E \left(\sum_{g,h \in [G]^2} \Pi_h ' B_{h,g} P_{g,l,-gh} U_l U_l' P_{l,g,-gk} B_{k,g}' V_k   \right)^2 \\
    &+ \sum_{k,l \in [G]^2, k \neq l} \mathbb E \left(\sum_{g \in [G]} \Pi_l ' B_{l,g} P_{g,l,-gl} U_l U_l' P_{l,g,-gk} B_{k,g}' V_k   \right)^2 \\
    &+ \sum_{k,l \in [G]^2, k \neq l} \mathbb E \left(\sum_{h \in [G]} \Pi_h ' B_{h,l} P_{l,l,-lh} U_l U_l' P_{l,l,-lk} B_{k,l}' V_k   \right)^2.
\end{align*}
We have
\begin{align*}
    &\sum_{k,l \in [G]^2, k \neq l} \mathbb E \left(\sum_{g,h \in [G]^2} \Pi_h ' B_{h,g} P_{g,l,-gh} U_l U_l' P_{l,g,-gk} B_{k,g}' V_k   \right)^2 \\
    &\lesssim  u_n \mathbb E \sum_{k,l \in [G]^2} \left \Vert \sum_{g,h \in [G]^2} B_{k,g} P_{g,l} U_l U_l' P_{l,g,-gh} B_{h,g}' \Pi_h  \right \Vert_2^2 \\
    &+ u_n \mathbb E \sum_{k,l \in [G]^2} \left \Vert \sum_{g,h \in [G]^2} B_{k,g} (\tilde P_{gk})_{g,g} P_{g,l} U_l U_l' P_{l,g,-gh} B_{h,g}' \Pi_h  \right \Vert_2^2 \\
    &+ u_n \mathbb E \sum_{k,l \in [G]^2} \left \Vert \sum_{g,h \in [G]^2} B_{k,g} (\tilde P_{gk})_{gk} P_{k,l} U_l U_l' P_{l,g,-gh} B_{h,g}' \Pi_h  \right \Vert_2^2,
\end{align*}
where
\begin{align*}
    &u_n \mathbb E \sum_{k,l \in [G]^2} \left \Vert \sum_{g,h \in [G]^2} B_{k,g} P_{g,l} U_l U_l' P_{l,g,-gh} B_{h,g}' \Pi_h  \right \Vert_2^2 \\
    &\lesssim u_n \mathbb E \sum_{g, l \in [G]^2} \left \Vert \sum_{h \in [G]} P_{g,l} U_l U_l' P_{l,g,-gh} B_{h,g}' \Pi_h  \right \Vert_2^2 \\
    &\lesssim u_n \mathbb E \sum_{g, l \in [G]^2} \left \Vert \sum_{h \in [G]} P_{g,l} U_l U_l' P_{l,g} B_{h,g}' \Pi_h  \right \Vert_2^2 + u_n \mathbb E \sum_{g, l \in [G]^2} \left \Vert \sum_{h \in [G]} P_{g,l} U_l U_l' P_{l,g} (\tilde P_{gh})_{g,g} B_{h,g}' \Pi_h  \right \Vert_2^2 \\
    &+  u_n \mathbb E \sum_{g, l \in [G]^2} \left \Vert \sum_{h \in [G]} P_{g,l} U_l U_l' P_{l,h} (\tilde P_{gh})_{h,g} B_{h,g}' \Pi_h  \right \Vert_2^2 \\
    &\lesssim u_n^{\frac{2q-3}{q-1}} n_G^{\frac{q}{q-1}} \lambda_n^3 \mu_n^2 + u_n^{\frac{2q-3}{q-1}} n_G^{\frac{2q-1}{q-1}} \phi_n^2 \lambda_n^2 \kappa_n = o(\omega_n^4),
\end{align*}
and
\begin{align*}
    &u_n \mathbb E \sum_{k,l \in [G]^2} \left \Vert \sum_{g,h \in [G]^2} B_{k,g} (\tilde P_{gk})_{g,g} P_{g,l} U_l U_l' P_{l,g,-gh} B_{h,g}' \Pi_h  \right \Vert_2^2 \\
    &\lesssim u_n \mathbb E \sum_{k,l \in [G]^2} \left \Vert \sum_{g,h \in [G]^2} B_{k,g} (\tilde P_{gk})_{g,g} P_{g,l} U_l U_l' P_{l,g} B_{h,g}' \Pi_h  \right \Vert_2^2 \\
    &+ u_n \mathbb E \sum_{k,l \in [G]^2} \left \Vert \sum_{g,h \in [G]^2} B_{k,g} (\tilde P_{gk})_{g,g} P_{g,l} U_l U_l' P_{l,g} (\tilde P_{gh})_{g,g} B_{h,g}' \Pi_h  \right \Vert_2^2 \\
    &+ u_n \mathbb E \sum_{k,l \in [G]^2} \left \Vert \sum_{g,h \in [G]^2} B_{k,g} (\tilde P_{gk})_{g,g} P_{g,l} U_l U_l' P_{l,h} (\tilde P_{gh})_{h,g} B_{h,g}' \Pi_h  \right \Vert_2^2\\
    &\lesssim u_n^{\frac{2q-3}{q-1}} n_G^{\frac{q}{q-1}} \zeta_{H,n} \phi_n^2 \lambda_n^2 \kappa_n + u_n^{\frac{2q-3}{q-1}} n_G^{\frac{3q-2}{q-1}} \phi_n^3 \lambda_n^3 \kappa_n = o(\omega_n^4),
\end{align*}
and
\begin{align*}
    &u_n \mathbb E \sum_{k,l \in [G]^2} \left \Vert \sum_{g,h \in [G]^2} B_{k,g} (\tilde P_{gk})_{gk} P_{k,l} U_l U_l' P_{l,g,-gh} B_{h,g}' \Pi_h  \right \Vert_2^2 \\
    &\lesssim u_n \mathbb E \sum_{k,l \in [G]^2} \left \Vert \sum_{g,h \in [G]^2} B_{k,g} (\tilde P_{gk})_{gk} P_{k,l} U_l U_l' P_{l,g} B_{h,g}' \Pi_h  \right \Vert_2^2 \\
    &+ u_n \mathbb E \sum_{k,l \in [G]^2} \left \Vert \sum_{g,h \in [G]^2} B_{k,g} (\tilde P_{gk})_{gk} P_{k,l} U_l U_l' P_{l,g} (\tilde P_{gh})_{g,g} B_{h,g}' \Pi_h  \right \Vert_2^2 \\
    &+ u_n \mathbb E \sum_{k,l \in [G]^2} \left \Vert \sum_{g,h \in [G]^2} B_{k,g} (\tilde P_{gk})_{gk} P_{k,l} U_l U_l' P_{l,h} (\tilde P_{gh})_{h,g} B_{h,g}' \Pi_h  \right \Vert_2^2 \\
    &\lesssim u_n^{\frac{2q-3}{q-1}} n_G^{\frac{q}{q-1}} \zeta_{H,n} \phi_n^2 \lambda_n^2 \kappa_n + u_n^{\frac{2q-3}{q-1}} n_G^{\frac{3q-2}{q-1}} \phi_n^3 \lambda_n^3 \kappa_n = o(\omega_n^4).
\end{align*}
We also have
\begin{align*}
    &\sum_{k,l \in [G]^2, k \neq l} \mathbb E \left(\sum_{g \in [G]} \Pi_l ' B_{l,g} P_{g,l,-gl} U_l U_l' P_{l,g,-gk} B_{k,g}' V_k   \right)^2 \\
    &\lesssim u_n \mathbb E \sum_{k,l \in [G]^2} \left\Vert \sum_{g \in [G]} B_{k,g} P_{g,l,-gk} U_l U_l' P_{l,g,-gl} B_{l,g}' \Pi_l \right\Vert_2^2 \\
    &\lesssim u_n^{\frac{2q-3}{q-1}} n_G^{\frac{2q-1}{q-1}} (\phi_n + \lambda_n^2) \phi_n \lambda_n^2 \kappa_n = o(\omega_n^4),
\end{align*}
and
\begin{align*}
    &\sum_{k,l \in [G]^2, k \neq l} \mathbb E \left(\sum_{h \in [G]} \Pi_h ' B_{h,l} P_{l,l,-lh} U_l U_l' P_{l,l,-lk} B_{k,l}' V_k   \right)^2 \\
    &\lesssim u_n \mathbb E \sum_{k,l \in [G]^2} \left\Vert \sum_{h \in [G]} B_{k,l} P_{l,l,-lk} U_l U_l' P_{l,l,-lh} B_{h,l}' \Pi_h \right\Vert_2^2 \\
    &\lesssim u_n^{\frac{2q-3}{q-1}} n_G^{\frac{q}{q-1}} \lambda_n^5 \mu_n^2 + u_n^{\frac{2q-3}{q-1}} n_G^{\frac{2q-1}{q-1}} \phi_n \lambda_n^5 \kappa_n + u_n^{\frac{2q-3}{q-1}} n_G^{\frac{q}{q-1}} \phi_n \lambda_n^8 \mu_n^2 + u_n^{\frac{2q-3}{q-1}} n_G^{\frac{2q-1}{q-1}} \phi_n^2 \lambda_n^8 \kappa_n = o(\omega_n^4).
\end{align*}
Lastly, consider
\begin{align*}
    R_{22,3} &\lesssim \sum_{k, l, j \in [G]^3, k \neq l \neq j} \mathbb E \left( \sum_{h \in [-l], g \in [-lj]}  \Pi_h ' B_{h,g} P_{g,l,-gh} U_l V_k ' B_{k,g} P_{g,j,-gk} U_j  \right)^2 \\
    &\lesssim \sum_{k, l, j \in [G]^3, k \neq l \neq j} \mathbb E \left( \sum_{g,h \in [G]^2}  \Pi_h ' B_{h,g} P_{g,l,-gh} U_l V_k ' B_{k,g} P_{g,j,-gk} U_j  \right)^2 \\
    &+ \sum_{k, l, j \in [G]^3, k \neq l \neq j} \mathbb E \left( \sum_{g \in [-lj]}  \Pi_l ' B_{l,g} P_{g,l,-gl} U_l V_k ' B_{k,g} P_{g,j,-gk} U_j  \right)^2 \\
    &+ \sum_{k, l, j \in [G]^3, k \neq l \neq j} \mathbb E \left( \sum_{h \in [G]}  \Pi_h ' B_{h,l} P_{l,l,-lh} U_l V_k ' B_{k,l} P_{l,j,-lk} U_j  \right)^2 \\
    &+ \sum_{k, l, j \in [G]^3, k \neq l \neq j} \mathbb E \left( \sum_{h \in [G]}  \Pi_h ' B_{h,j} P_{j,l,-jh} U_l V_k ' B_{k,j} P_{j,j,-jk} U_j  \right)^2.
\end{align*}
We have
\begin{align*}
    &\sum_{k, l, j \in [G]^3, k \neq l \neq j} \mathbb E \left( \sum_{g,h \in [G]^2}  \Pi_h ' B_{h,g} P_{g,l,-gh} U_l V_k ' B_{k,g} P_{g,j,-gk} U_j  \right)^2 \\
    &\lesssim u_n \mathbb E \sum_{k, l \in [G]^2, k \neq l} \sum_{g,h \in [G]^2} \sum_{g',h' \in [G]^2} \Pi_h ' B_{h,g} P_{g,l,-gh} U_l V_k ' B_{k,g} \\
    &\left( \sum_{j \in [-kl]}P_{g,j,-gk} P_{j,g',-gk} \right) B_{k,g'}' V_k U_l' P_{l,g',-g'h'} B_{h',g'}' \Pi_{h'} \\
    &\lesssim u_n \mathbb E \sum_{g, k, l \in [G]^3, k \neq l} \left \Vert \sum_{h \in [G]} B_{k,g}' V_k U_l' P_{l,g,-gh} B_{h,g}' \Pi_h \right \Vert_2^2 \\
    &+ u_n \mathbb E \sum_{g, k, l \in [G]^3, k \neq l} \left \Vert \sum_{h \in [G]} (\tilde P_{gk})_{g,g} B_{k,g}' V_k U_l' P_{l,g,-gh} B_{h,g}' \Pi_h \right \Vert_2^2 \\
    &+ u_n \lambda_n \mathbb E \sum_{k, l \in [G]^2, k \neq l} \left \Vert \sum_{g, h \in [G]^2} (\tilde P_{gk})_{k,g} B_{k,g}' V_k U_l' P_{l,g,-gh} B_{h,g}' \Pi_h \right \Vert_2^2 + o(\omega_n^4),
\end{align*}
where
\begin{align*}
    &u_n \mathbb E \sum_{g, k, l \in [G]^3k \neq l} \left \Vert \sum_{h \in [G]} B_{k,g}' V_k U_l' P_{l,g,-gh} B_{h,g}' \Pi_h \right \Vert_2^2 \\
    &\lesssim u_n \mathbb E \sum_{g, k, l \in [G]^3k \neq l} \left \Vert \sum_{h \in [G]} B_{k,g}' V_k U_l' P_{l,g} B_{h,g}' \Pi_h \right \Vert_2^2 + u_n \mathbb E \sum_{g, k, l \in [G]^3k \neq l} \left \Vert \sum_{h \in [G]} B_{k,g}' V_k U_l' P_{l,g} (\tilde P_{gh})_{g,g} B_{h,g}' \Pi_h \right \Vert_2^2 \\
    &+ u_n \mathbb E \sum_{g, k, l \in [G]^3k \neq l} \left \Vert \sum_{h \in [G]} B_{k,g}' V_k U_l' P_{l,h} (\tilde P_{gh})_{h,g} B_{h,g}' \Pi_h \right \Vert_2^2 \\
    &\lesssim u_n^3 \zeta_{H,n} \lambda_n \kappa_n + u_n^3 n_G^2 \phi_n \lambda_n^4 \kappa_n + u_n^3 n_G^2 \lambda_n^3 \kappa_n  = o(\omega_n^4),
\end{align*}
and
\begin{align*}
    &u_n \mathbb E \sum_{g, k, l \in [G]^3, k \neq l} \left \Vert \sum_{h \in [G]} (\tilde P_{gk})_{g,g} B_{k,g}' V_k U_l' P_{l,g,-gh} B_{h,g}' \Pi_h \right \Vert_2^2 \\
    &\lesssim u_n \mathbb E \sum_{g, k, l \in [G]^3, k \neq l} \left \Vert \sum_{h \in [G]} (\tilde P_{gk})_{g,g} B_{k,g}' V_k U_l' P_{l,g} B_{h,g}' \Pi_h \right \Vert_2^2 \\
    &+ u_n \mathbb E \sum_{g, k, l \in [G]^3, k \neq l} \left \Vert \sum_{h \in [G]} (\tilde P_{gk})_{g,g} B_{k,g}' V_k U_l' P_{l,g} (\tilde P_{gh})_{g,g} B_{h,g}' \Pi_h \right \Vert_2^2 \\
    &+ u_n \mathbb E \sum_{g, k, l \in [G]^3, k \neq l} \left \Vert \sum_{h \in [G]} (\tilde P_{gk})_{g,g} B_{k,g}' V_k U_l' P_{l,h} (\tilde P_{gh})_{h,g} B_{h,g}' \Pi_h \right \Vert_2^2 \\
    &\lesssim u_n^3 \zeta_{H,n} \lambda_n^3 \kappa_n + u_n^3 n_G^2 \phi_n \lambda_n^6 \kappa_n + u_n^3 n_G^2 \lambda_n^5 \kappa_n  = o(\omega_n^4),
\end{align*}
and
\begin{align*}
    &u_n \lambda_n \mathbb E \sum_{k, l \in [G]^2, k \neq l} \left \Vert \sum_{g, h \in [G]^2} (\tilde P_{gk})_{k,g} B_{k,g}' V_k U_l' P_{l,g,-gh} B_{h,g}' \Pi_h \right \Vert_2^2 \\
    &\lesssim u_n \lambda_n \mathbb E \sum_{k, l \in [G]^2, k \neq l} \left \Vert \sum_{g, h \in [G]^2} (\tilde P_{gk})_{k,g} B_{k,g}' V_k U_l' P_{l,g} B_{h,g}' \Pi_h \right \Vert_2^2 \\
    &+ u_n \lambda_n \mathbb E \sum_{k, l \in [G]^2, k \neq l} \left \Vert \sum_{g, h \in [G]^2} (\tilde P_{gk})_{k,g} B_{k,g}' V_k U_l' P_{l,g} (\tilde P_{gh})_{g,g} B_{h,g}' \Pi_h \right \Vert_2^2 \\
    &+ u_n \lambda_n \mathbb E \sum_{k, l \in [G]^2, k \neq l} \left \Vert \sum_{g, h \in [G]^2} (\tilde P_{gk})_{k,g} B_{k,g}' V_k U_l' P_{l,h} (\tilde P_{gh})_{h,g} B_{h,g}' \Pi_h \right \Vert_2^2 \\
    &\lesssim u_n^3 \zeta_{H,n} \phi_n \lambda_n^2 \kappa_n + u_n^3 n_G^2 \phi_n^2 \lambda_n^5 \kappa_n + u_n^3 n_G^2 \phi_n \lambda_n^4 \kappa_n = o(\omega_n^4).
\end{align*}
We also have
\begin{align*}
    &\sum_{k, l, j \in [G]^3, k \neq l \neq j} \mathbb E \left( \sum_{g \in [-lj]}  \Pi_l ' B_{l,g} P_{g,l,-gl} U_l V_k ' B_{k,g} P_{g,j,-gk} U_j  \right)^2 \\
    &\lesssim \sum_{k, l, j \in [G]^3, k \neq l \neq j} \mathbb E \left( \sum_{g \in [G]}  \Pi_l ' B_{l,g} P_{g,l,-gl} U_l V_k ' B_{k,g} P_{g,j,-gk} U_j  \right)^2 \\
    &+ \sum_{k, l, j \in [G]^3, k \neq l \neq j} \mathbb E \left(\Pi_l ' B_{l,j} P_{j,l,-jl} U_l V_k ' B_{k,j} P_{j,j,-jk} U_j  \right)^2 \\
    &\lesssim u_n^3 n_G^2 \lambda_n^3 \kappa_n + u_n^3 n_G^2 \phi_n \lambda_n^4 \kappa_n + o(\omega_n^4) = o(\omega_n^4),
\end{align*}
and
\begin{align*}
    &\sum_{k, l, j \in [G]^3, k \neq l \neq j} \mathbb E \left( \sum_{h \in [G]}  \Pi_h ' B_{h,l} P_{l,l,-lh} U_l V_k ' B_{k,l} P_{l,j,-lk} U_j  \right)^2 \\
    &\lesssim u_n^3 n_G \lambda_n^4 \mu_n^2 + u_n^3 n_G^2 \phi_n \lambda_n^4 \kappa_n = o(\omega_n^4),
\end{align*}
and
\begin{align*}
    &\sum_{k, l, j \in [G]^3, k \neq l \neq j} \mathbb E \left( \sum_{h \in [G]}  \Pi_h ' B_{h,j} P_{j,l,-jh} U_l V_k ' B_{k,j} P_{j,j,-jk} U_j  \right)^2 \\
    &\lesssim u_n^3 n_G \lambda_n^4 \mu_n^2 + u_n^3 n_G^2 \phi_n \lambda_n^6 \kappa_n + u_n^3 n_G^2 \lambda_n^5 \kappa_n = o(\omega_n^4).
\end{align*}

For $R_{23}$, we have
\begin{align*}
    R_{23} &=  \mathbb V \left( \sum_{g, h, k \in [G]^3, l \in [-gh], j \in [-gk]} \left(V_h ' B_{h,g} \tilde M_{g,l,-gh} U_l \right) \left(V_k ' B_{k,g} \tilde M_{g,j,-gk} U_j \right) \right) \\
    &\lesssim \underbrace{\mathbb V \left( \sum_{g, h \in [G]^2, l \in [-gh]} \left(V_h ' B_{h,g} P_{g,l,-gh} U_l \right) \left(V_h ' B_{h,g} P_{g,l,-gh} U_l \right) \right)}_{R_{23,1}} \\
    &+ \underbrace{\mathbb V \left( \sum_{g, h, k\in [G]^3, h \neq k} \left(V_h ' B_{h,g} P_{g,k,-gh} U_k \right) \left(V_k ' B_{k,g} P_{g,h,-gk} U_h \right) \right)}_{R_{23,2}} \\
    &+ \underbrace{\mathbb V \left( \sum_{g, h \in [G]^2, l,j \in [-gh]^2, l \neq j} \left(V_h ' B_{h,g} P_{g,l,-gh} U_l \right) \left(V_h ' B_{h,g} P_{g,j,-gh} U_j \right) \right)}_{R_{23,3}} \\
    &+ \underbrace{\mathbb V \left( \sum_{g, h, k \in [G]^3, h \neq k, l \in [-ghk]} \left(V_h ' B_{h,g} P_{g,l,-gh} U_l \right) \left(V_k ' B_{k,g} P_{g,l,-gk} U_l \right) \right)}_{R_{23,4}} \\
    &+ \underbrace{\mathbb V \left( \sum_{g, h, k \in [G]^3, h \neq k, j \in [-ghk]} \left(V_h ' B_{h,g} P_{g,k,-gh} U_k \right) \left(V_k ' B_{k,g} P_{g,j,-gk} U_j \right) \right)}_{R_{23,5}} \\
    &+ \underbrace{\mathbb V \left( \sum_{g, h, k \in [G]^3, h \neq k, l \in [-ghk]} \left(V_h ' B_{h,g} P_{g,l,-gh} U_l \right) \left(V_k ' B_{k,g} P_{g,h,-gk} U_h \right) \right)}_{R_{23,6}} \\
    &+ \underbrace{\mathbb V \left( \sum_{g, h, k \in [G]^3, h \neq k, l,j \in [-ghk]^2, l \neq j} \left(V_h ' B_{h,g} P_{g,l,-gh} U_l \right) \left(V_k ' B_{k,g} P_{g,j,-gk} U_j \right) \right)}_{R_{23,7}}.
\end{align*}
Consider first
\begin{align*}
    R_{23,1} &\lesssim \underbrace{\mathbb V \left( \sum_{g, h \in [G]^2, l \in [-gh]}V_h ' B_{h,g} P_{g,l,-gh} \Omega_{U,l} P_{l,g,-gh} B_{h,g}' V_h \right)}_{R_{23,1,1}} \\
    &+ \underbrace{\mathbb V \left( \sum_{g, h \in [G]^2, l \in [-gh]} U_l' P_{l,g,-gh} B_{h,g}' \Omega_{V,h} B_{h,g} P_{g,l,-gh} U_l \right)}_{R_{23,1,2}} \\
    &+ \underbrace{\mathbb V \left( \sum_{g, h \in [G]^2, l \in [-gh]} \tr
      \left( B_{h,g}' (V_h V_h' - \Omega_{V,h}) B_{h,g} P_{g,l,-gh} (U_{l} U_l'- \Omega_{U,l}) P_{l,g,-gh} \right) \right)}_{R_{23,1,3}},
\end{align*}
where
\begin{align*}
    R_{23,1,1} &\lesssim u_n^{\frac{q-2}{q-1}} n_G^{\frac{q}{q-1}} \sum_{h \in [G]} \left \Vert \sum_{l,g \in [G]^2, l \neq g, l \neq h}B_{h,g} P_{g,l,-gh} \Omega_{U,l} P_{l,g,-gh} B_{h,g}' \right \Vert_F^2 \\
    &\lesssim u_n^{\frac{q-2}{q-1}} n_G^{\frac{q}{q-1}} \sum_{h \in [G]}  \left \Vert \sum_{l,g \in [G]^2, l \neq g, l \neq h}B_{h,g} P_{g,l,-gh} \Omega_{U,l} P_{l,g,-gh} B_{h,g}' \right \Vert_{op} \\
    &\tr \left( \sum_{l,g \in [G]^2, l \neq g, l \neq h}B_{h,g} P_{g,l,-gh} \Omega_{U,l} P_{l,g,-gh} B_{h,g}' \right) \\
    &\lesssim u_n^{\frac{2q-3}{q-1}} n_G^{\frac{q}{q-1}} \lambda_n^2 \sum_{h \in [G]} \tr \left( \sum_{l,g \in [G]^2, l \neq g, l \neq h}B_{h,g} P_{g,l,-gh} \Omega_{U,l} P_{l,g,-gh} B_{h,g}' \right) \\
    &\lesssim u_n^{\frac{3q-4}{q-1}} n_G^{\frac{q}{q-1}} \lambda_n^3 \kappa_n = o(\omega_n^4),
\end{align*}
and similarly
\begin{align*}
    R_{23,1,2} &\lesssim u_n^{\frac{q-2}{q-1}} n_G^{\frac{q}{q-1}} \sum_{l \in [G]} \left \Vert \sum_{h,g \in [G]^2, g \neq l, h \neq l} P_{l,g,-gh} B_{h,g}' \Omega_{V,h} B_{h,g} P_{g,l,-gh} \right \Vert_F^2 \\
    &\lesssim u_n^{\frac{2q-3}{q-1}} n_G^{\frac{q}{q-1}} (\lambda_n^2 + \lambda_n^2 \phi_n^2) \sum_{l \in [G]} \tr \left( \sum_{h,g \in [G]^2, g \neq l, h \neq l} P_{l,g,-gh} B_{h,g}' \Omega_{V,h} B_{h,g} P_{g,l,-gh} \right) \\
    &\lesssim u_n^{\frac{3q-4}{q-1}} n_G^{\frac{q}{q-1}} (\lambda_n^2 + \lambda_n^2 \phi_n^2) \lambda_n \kappa_n = o(\omega_n^4).
\end{align*}
We also have
\begin{align*}
    R_{23,1,3} &\lesssim \sum_{h,l \in [G]^2, h \neq l} \mathbb E\left( \sum_{g \in [-hl]} \left(V_h ' B_{h,g} P_{g,l,-gh} U_l \right)^2 \right)^2 \\
    &\lesssim u_n^{\frac{q-2}{q-1}} n_G^{\frac{q}{q-1}} \sum_{h,l \in [G]^2, h \neq l} \mathbb E \left \Vert \sum_{g \in [-hl]} B_{h,g} P_{g,l,-gh} U_l U_l' P_{l,g,-gh} B_{h,g}' \right \Vert_F^2 \\
    &\lesssim u_n^{\frac{q-2}{q-1}} n_G^{\frac{q}{q-1}} \lambda_n^2 \sum_{h,l \in [G]^2, h \neq l}  \tr \left( \sum_{g \in [-hl]} B_{h,g} P_{g,l,-gh} \mathbb E\left(U_l U_l' ||U_l||_2^2\right)  P_{l,g,-gh} B_{h,g}' \right) \\
    &\lesssim u_n^{\frac{2q-4}{q-1}} n_G^{\frac{2q}{q-1}} \lambda_n^3 \kappa_n = o(\omega_n^4).
\end{align*}
Consider next
\begin{align*}
    R_{23,2} &\lesssim \underbrace{\mathbb V \left( \sum_{g, h, k\in [G]^3, h \neq k} V_h ' B_{h,g} P_{g,k,-gh} \Omega_{U,V,k} B_{k,g} P_{g,h,-gk} U_h \right)}_{R_{23,2,1 }}\\
    &+ \underbrace{\mathbb V \left( \sum_{g, h, k\in [G]^3, h \neq k} V_k' B_{k,g} P_{g,h,-gk} \Omega_{U,V,h} B_{h,g} P_{g,k,-gh} U_k \right)}_{R_{23,2,2}} \\
    &+ \underbrace{\mathbb V \left( \sum_{g, h, k\in [G]^3, h \neq k} \tr \left( (U_k V_k' - \Omega_{U,V,k}) B_{k,g} P_{g,h,-gk} (U_h V_h' - \Omega_{U,V,h}) B_{h,g} P_{g,k,-gh}  \right) \right)}_{R_{23,2,3}},
\end{align*}
where
\begin{align*}
    R_{23,2,1} &\lesssim u_n^{\frac{q-2}{q-1}} n_G^{\frac{q}{q-1}} \sum_{h \in [G]} \left \Vert \sum_{g,k \in [G]^2, k \neq h} B_{h,g} P_{g,k,-gh} \Omega_{U,V,k} B_{k,g} P_{g,h,-gk} \right \Vert_F^2 \\
    &\lesssim u_n^{\frac{3q-4}{q-1}} n_G^{\frac{q}{q-1}} \sum_{h \in [G]} \left(\sum_{g,k \in [G]^2, k \neq h} ||B_{k,g}||_{op} ||P_{g,k,-gh}||_{op} ||P_{g,h,-gk}||_{op} ||B_{h,g}||_{F} \right)^2 \\
    &\lesssim u_n^{\frac{3q-4}{q-1}} n_G^{\frac{q}{q-1}} \phi_n^2 (\phi_n + \lambda_n^2) \kappa_n = o(\omega_n^4),
\end{align*}
and similarly
\begin{align*}
    R_{23,2,2} &\lesssim u_n^{\frac{q-2}{q-1}} n_G^{\frac{q}{q-1}} \sum_{k \in [G]} \left \Vert \sum_{g,h \in [G]^2, h \neq k} B_{k,g} P_{g,h,-gk} \Omega_{U,V,h} B_{h,g} P_{g,k,-gh} \right \Vert_F^2 \\
    &\lesssim u_n^{\frac{3q-4}{q-1}} n_G^{\frac{q}{q-1}} \phi_n^2 (\phi_n + \lambda_n^2) \kappa_n = o(\omega_n^4).
\end{align*}
We also have
\begin{align*}
    R_{23,2,3} &\lesssim \sum_{h,k \in [G]^2, h \neq k} \mathbb E \left( \sum_{g \in [G]} V_h ' B_{h,g} P_{g,k,-gh} U_k V_k ' B_{k,g} P_{g,h,-gk} U_h \right)^2 \\
    &\lesssim u_n^{\frac{q-2}{q-1}} n_G^{\frac{q}{q-1}} \sum_{h,k \in [G]^2, h \neq k} \mathbb E \left \Vert \sum_{g \in [G]} B_{h,g} P_{g,k,-gh} U_k V_k ' B_{k,g} P_{g,h,-gk} \right \Vert_F^2 \\
    &\lesssim u_n^{\frac{q-2}{q-1}} n_G^{\frac{q}{q-1}} \sum_{h,k \in [G]^2, h \neq k} \mathbb E \sum_{g,g' \in [G]^2} V_k ' B_{k,g} P_{g,h,-gk} P_{h,g',-g'k} B_{k,g'}' V_k  U_k' P_{k,g',-g'h} B_{h,g'}' B_{h,g} P_{g,k,-gh} U_k \\
    &\lesssim u_n^{\frac{q-2}{q-1}} n_G^{\frac{q}{q-1}} \sum_{h,k \in [G]^2, h \neq k} \mathbb E \sum_{g,g' \in [G]^2} \left( V_k ' B_{k,g} P_{g,h,-gk} P_{h,g',-g'k} B_{k,g'}' V_k \right)^2 \\
    &+ u_n^{\frac{q-2}{q-1}} n_G^{\frac{q}{q-1}} \sum_{h,k \in [G]^2, h \neq k} \mathbb E \sum_{g,g' \in [G]^2} \left( U_k' P_{k,g',-g'h} B_{h,g'}' B_{h,g} P_{g,k,-gh} U_k \right)^2 \\
    &\lesssim u_n^{\frac{2q-4}{q-1}} n_G^{\frac{2q}{q-1}} \phi_n^2 \lambda_n^2 \kappa_n = o(\omega_n^4).
\end{align*}
Consider next
\begin{align*}
    R_{23,3} &\lesssim \underbrace{\mathbb V \left( \sum_{g, h \in [G]^2, l,j \in [-gh]^2, l \neq j} U_l' P_{l,g,-gh} B_{h,g}' \Omega_{V,h} B_{h,g} P_{g,j,-gh} U_j \right)}_{R_{23,3,1}} \\
    &+ \underbrace{\mathbb V \left( \sum_{g, h \in [G]^2, l,j \in [-gh]^2, l \neq j} U_l' P_{l,g,-gh} B_{h,g}' (V_h V_h' - \Omega_{V,h}) B_{h,g} P_{g,j,-gh} U_j \right)}_{R_{23,3,2}}.
\end{align*}
We have, letting $\ddot{P}_{gg'hh'}= \sum_{j \in [-gg'hh']} P_{g,j,-gh} P_{j,g',-g'h'}$
\begin{align*}
    R_{23,3,1} &\lesssim u_n^2 \sum_{l,j \in [G]^2, l \neq j} \left \Vert \sum_{g,h \in [-lj]^2} P_{l,g,-gh} B_{h,g}' \Omega_{V,h} B_{h,g} P_{g,j,-gh} \right \Vert_F^2 \\
    &\lesssim u_n^2 \sum_{l \in [G]} \sum_{g,h \in [-l]^2} \sum_{g',h' \in [-l]^2} \tr \left( P_{l,g,-gh} B_{h,g}' \Omega_{V,h} B_{h,g}\ddot{P}_{gg'hh'} B_{h',g'}' \Omega_{V,h'} B_{h',g'} P_{g',l,-g'h'} \right) \\
    &\lesssim u_n^2 \sum_{g,l \in [G]^2, g \neq l} \left \Vert \sum_{h \in [-l]} P_{l,g,-gh} B_{h,g}' \Omega_{V,h} B_{h,g} \right \Vert_F^2 + u_n^2 \sum_{g,l \in [G]^2, g \neq l} \left \Vert \sum_{h \in [-l]} P_{l,g,-gh} B_{h,g}' \Omega_{V,h} B_{h,g} (\tilde P_{gh})_{g,g} \right \Vert_F^2 \\
    &+ u_n^2 \sum_{h,l \in [G]^2, h \neq l} \left \Vert \sum_{g \in [-l]} P_{l,g,-gh} B_{h,g}' \Omega_{V,h} B_{h,g} (\tilde P_{gh})_{g,h} \right \Vert_F^2 + o(\omega_n^4) \\
    &\lesssim u_n^4 (\phi_n+\lambda_n^2) \lambda_n \kappa_n + o(\omega_n^4) = o(\omega_n^4),
\end{align*}
and
\begin{align*}
  R_{23,3,2} &\lesssim \sum_{h,l,j \in [G]^3, h \neq l \neq j} \mathbb E \left( \sum_{g \in [-lj]} U_l' P_{l,g,-gh} B_{h,g}' V_h V_h'  B_{h,g} P_{g,j,-gh} U_j \right)^2 \\
             &\lesssim u_n^2 \sum_{h,l,j \in [G]^3, h \neq l \neq j} \mathbb E \left \Vert \sum_{g \in [-lj]} P_{l,g,-gh} B_{h,g}' V_h V_h'  B_{h,g} P_{g,j,-gh} \right \Vert_F^2 \\
             &\lesssim u_n^2 \sum_{g, h,l \in [G]^3, g \neq l, h \neq l} \mathbb E \left \Vert P_{l,g,-gh} B_{h,g}' V_h V_h'  B_{h,g} \right \Vert_F^2\\
             &\quad + u_n^2 \sum_{g, h,l \in [G]^3, g \neq l, h \neq l} \mathbb E \left \Vert P_{l,g,-gh} B_{h,g}' V_h V_h'  B_{h,g} (\tilde P_{gh})_{g,g} \right \Vert_F^2 \\
             &\quad + u_n^2 \lambda_n \sum_{h,l \in [G]^2, h \neq l} \mathbb E \left \Vert \sum_{g \in [G], g \neq l} P_{l,g,-gh} B_{h,g}' V_h V_h'  B_{h,g} (\tilde P_{gh})_{g,h} \right \Vert_F^2 + o(\omega_n^4) \\
             &\lesssim u_n^{\frac{3q-4}{q-1}} n_G^{\frac{q}{q-1}} \lambda_n^3 \kappa_n + u_n^{\frac{3q-4}{q-1}} n_G^{\frac{q}{q-1}} \phi_n \lambda_n^6 \kappa_n + o(\omega_n^4) = o(\omega_n^4).
\end{align*}
Consider next
\begin{align*}
    R_{23,4} &\lesssim \underbrace{\mathbb V \left( \sum_{g, h, k \in [G]^3, h \neq k, l \in [-ghk]} V_h ' B_{h,g} P_{g,l,-gh} \Omega_{U,l} P_{l,g,-gk} B_{k,g}' V_k  \right)}_{R_{23,4,1}} \\
    &+ \underbrace{\mathbb V \left( \sum_{g, h, k \in [G]^3, h \neq k, l \in [-ghk]} V_h ' B_{h,g} P_{g,l,-gh} (U_l U_l' - \Omega_{U,l}) P_{l,g,-gk} B_{k,g}' V_k  \right)}_{R_{23,4,2}}.
\end{align*}
We have
\begin{align*}
    R_{23,4,1} &\lesssim u_n^2 \sum_{h,k \in [G]^2, h \neq k} \left \Vert \sum_{g \in [G]} B_{h,g} \left(\sum_{l \in [-ghk]} P_{g,l,-gh} \Omega_{U,l} P_{l,g,-gk}\right) B_{k,g}' \right \Vert_F^2 \\
    &\lesssim u_n^2 \sum_{h,k \in [G]^2} \left \Vert \sum_{g \in [G]} B_{h,g} \left(\sum_{l \in [G]} P_{g,l,-gh} \Omega_{U,l} P_{l,g} \right) B_{k,g}' \right \Vert_F^2 \\
    &+ u_n^2 \sum_{h,k \in [G]^2} \left \Vert \sum_{g \in [G]} B_{h,g} \left(\sum_{l \in [G]} P_{g,l,-gh} \Omega_{U,l} P_{l,g} \right)  ( \Xi_{1,g} + \Xi_{9,gk}) B_{k,g}' \right \Vert_F^2 \\
    &+ u_n^2 \sum_{h,k \in [G]^2} \left \Vert \sum_{g \in [G]} B_{h,g} \left(\sum_{l \in [G]} P_{g,l,-gh} \Omega_{U,l} P_{l,k} \right)  (\tilde P_{gk})_{k,g} B_{k,g}' \right \Vert_F^2 + o(\omega_n^4),
\end{align*}
where
\begin{align*}
    &u_n^2 \sum_{h,k \in [G]^2} \left \Vert \sum_{g \in [G]} B_{h,g} \left(\sum_{l \in [G]} P_{g,l,-gh} \Omega_{U,l} P_{l,g} \right) B_{k,g}' \right \Vert_F^2 \\
    &= u_n^2 \sum_{h \in [G]} \tr \left( \sum_{g,g' \in [G]^2} B_{h,g} \left(\sum_{l \in [G]} P_{g,l,-gh} \Omega_{U,l} P_{l,g} \right) \left(\sum_{k \in [G]} B_{k,g}' B_{k,g'}' \right) \left(\sum_{l' \in [G]} P_{g',l'} \Omega_{U,l'} P_{l',g',-gh} \right) B_{h,g'}'  \right) \\
    &\lesssim u_n^2 \sum_{g,h \in [G]^2} \left \Vert B_{h,g} \left(\sum_{l \in [G]} P_{g,l,-gh} \Omega_{U,l} P_{l,g} \right) \right \Vert_F^2 \\
    &\lesssim u_n^2 \sum_{g,h \in [G]^2} ||B_{h,g}||_F^2 \left \Vert  \left(\sum_{l \in [G]} P_{g,l,-gh} \Omega_{U,l} P_{l,g} \right) \right \Vert_{op}^2 \\
    &\lesssim u_n^4 \lambda_n^2 \kappa_n = o(\omega_n^4),
\end{align*}
and similarly for
\begin{align*}
    u_n^2 \sum_{h,k \in [G]^2} \left \Vert \sum_{g \in [G]} B_{h,g} \left(\sum_{l \in [G]} P_{g,l,-gh} \Omega_{U,l} P_{l,g} \right) \Xi_{1,g} B_{k,g}' \right \Vert_F^2,
\end{align*}
and
\begin{align*}
    &u_n^2 \sum_{h,k \in [G]^2} \left \Vert \sum_{g \in [G]} B_{h,g} \left(\sum_{l \in [G]} P_{g,l,-gh} \Omega_{U,l} P_{l,k} \right)  (\tilde P_{gk})_{k,g} B_{k,g}' \right \Vert_F^2 \\
    &\lesssim u_n^2 \phi_n \sum_{g,h,k \in [G]^3} ||B_{h,g}||_F^2 ||B_{k,g}||_{op}^2 \left \Vert \sum_{l \in [G]} P_{g,l,-gh} \Omega_{U,l} P_{l,k} \right \Vert_{op}^2 \\
    &\lesssim u_n^4 \phi_n^2 \lambda_n^2 \kappa_n = o(\omega_n^4),
\end{align*}
and similarly for
\begin{align*}
    u_n^2 \sum_{h,k \in [G]^2} \left \Vert \sum_{g \in [G]} B_{h,g} \left(\sum_{l \in [G]} P_{g,l,-gh} \Omega_{U,l} P_{l,g} \right)  \Xi_{9,gk} B_{k,g}' \right \Vert_F^2.
\end{align*}
We also have
\begin{align*}
    R_{23,4,2} &\lesssim \sum_{h,k,l \in [G]^3, h \neq k \neq l} \mathbb E \left( \sum_{g \in [-l]} V_h'B_{h,g} P_{g,l,-gh} U_l U_l' P_{l,g,-gk} B_{k,g}' V_k \right)^2 \\
    &\lesssim u_n^2 \sum_{h,k,l \in [G]^3} \mathbb E \left \Vert \sum_{g \in [-l]} B_{h,g} P_{g,l,-gh} U_l U_l' P_{l,g,-gk} B_{k,g}' \right \Vert_F^2 \\
    &\lesssim u_n^2 \sum_{l \in [G]} \sum_{g,g' \in [-l]^2} \mathbb E \left(\sum_{h \in [G]} U_l' P_{l,g,-gh} B_{h,g}' B_{h,g'} P_{g',l,-g'h} U_l \right)^2 \\
    &\lesssim u_n^{\frac{3q-4}{q-1}} n_G^{\frac{q}{q-1}} \sum_{l \in [G]} \sum_{g,g' \in [-l]^2} \left \Vert \sum_{h \in [G]} P_{l,g,-gh} B_{h,g}' B_{h,g'} P_{g',l,-g'h} \right \Vert_F^2 \\
    &\lesssim u_n^{\frac{3q-4}{q-1}} n_G^{\frac{q}{q-1}} \lambda_n^3 \kappa_n + u_n^{\frac{3q-4}{q-1}} n_G^{\frac{q}{q-1}} \phi_n \lambda_n^4 \kappa_n = o(\omega_n^4).
\end{align*}
Consider next
\begin{align*}
    R_{23,5} &\lesssim \underbrace{\mathbb V \left( \sum_{g, h, k \in [G]^3, h \neq k, j \in [-ghk]} V_h ' B_{h,g} P_{g,k,-gh} \Omega_{U,V,k} B_{k,g} P_{g,j,-gk} U_j \right)}_{R_{23,5,1}} \\
    &+ \underbrace{\mathbb V \left( \sum_{g, h, k \in [G]^3, h \neq k, j \in [-ghk]} V_h ' B_{h,g} P_{g,k,-gh} (U_k V_k ' - \Omega_{U,V,k}) B_{k,g} P_{g,j,-gk} U_j \right)}_{R_{23,5,2}},
\end{align*}
where
\begin{align*}
    R_{23,5,1} &\lesssim u_n^2 \sum_{h,j \in [G]^2, h \neq j} \left \Vert \sum_{g,k \in [-hj]^2, g \neq k} B_{h,g} P_{g,k,-gh} \Omega_{U,V,k} B_{k,g} P_{g,j,-gk} \right \Vert_F^2 \\
    &\lesssim u_n^2 \sum_{g,h \in [G]^2} \left \Vert \sum_{k \in [-h]} B_{h,g} P_{g,k,-gh} \Omega_{U,V,k} B_{k,g} \right \Vert_F^2 + u_n^2 \sum_{g,h \in [G]^2} \left \Vert \sum_{k \in [-h]} B_{h,g} P_{g,k,-gh} \Omega_{U,V,k} B_{k,g} (\tilde P_{gk})_{g,g} \right \Vert_F^2 \\
    &+ u_n^2 \lambda_n \sum_{h,k \in [G]^2, h \neq k} \left \Vert \sum_{g \in [G]} B_{h,g} P_{g,k,-gh} \Omega_{U,V,k} B_{k,g} (\tilde P_{gk})_{g,k} \right \Vert_F^2 + o(\omega_n^4) \\
    &\lesssim u_n^4 \lambda_n^2 \kappa_n + u_n^4 n_G \phi_n \lambda_n^3 \kappa_n + o(\omega_n^4) = o(\omega_n^4),
\end{align*}
and
\begin{align*}
    R_{23,5,2} &\lesssim \sum_{h,k,j \in [G]^3, h \neq k \neq j} \mathbb E \left( \sum_{g \in [-j]} V_h ' B_{h,g} P_{g,k,-gh} U_k V_k ' B_{k,g} P_{g,j,-gk} U_j \right)^2 \\
    &\lesssim u_n^2 \sum_{h,k,j \in [G]^3} \mathbb E \left \Vert \sum_{g \in [-j]} B_{h,g} P_{g,k,-gh} U_k V_k ' B_{k,g} P_{g,j,-gk} \right \Vert_F^2 \\
    &\lesssim u_n^2 \sum_{g,h,k \in [G]^3} \mathbb E || B_{h,g} P_{g,k,-gh} U_k V_k ' B_{k,g} ||_F^2 + u_n^2 \sum_{g,h,k \in [G]^3} \mathbb E || B_{h,g} P_{g,k,-gh} U_k V_k ' B_{k,g} (\tilde P_{gk})_{g,g} ||_F^2 \\
    &+u_n^2 \lambda_n \sum_{h,k \in [G]^2} \mathbb E ||\sum_{g \in [G]} B_{h,g} P_{g,k,-gh} U_k V_k ' B_{k,g} (\tilde P_{gk})_{g,k} ||_F^2  + o(\omega_n^4) \\
    &\lesssim u_n^{\frac{3q-4}{q-1}} n_G^{\frac{q}{q-1}} \lambda_n^3 \kappa_n + u_n^{\frac{3q-4}{q-1}} n_G^{\frac{q}{q-1}} \phi_n \lambda_n^4 \kappa_n + o(\omega_n^4) = o(\omega_n^4).
\end{align*}
Similarly, we have $R_{23,6} = o(\omega_n^4)$. Lastly, for $R_{23,7}$ we have
\begin{align*}
    R_{23,7} &\lesssim \sum_{h,k,l,j \in [G]^4, h \neq k \neq l \neq j} \mathbb E \left( \sum_{g \in [-lj]} V_h ' B_{h,g} P_{g,l,-gh} U_l V_k ' B_{k,g} P_{g,j,-gk} U_j \right)^2 \\
    &\lesssim u_n^2 \sum_{h,k,l,j \in [G]^4, h \neq k \neq l \neq j} \mathbb E \left \Vert \sum_{g \in [-lj]} B_{h,g} P_{g,l,-gh} U_l V_k ' B_{k,g} P_{g,j,-gk} \right \Vert_F^2 \\
    &\lesssim u_n^2 \sum_{g, h,k,l \in [G]^4, g \neq h \neq k \neq l } \mathbb E ||  B_{h,g} P_{g,l,-gh} U_l V_k ' B_{k,g} ||_F^2 \\
    &+ u_n^2 \sum_{g, h,k,l \in [G]^4, g \neq h \neq k \neq l } \mathbb E ||  B_{h,g} P_{g,l,-gh} U_l V_k ' B_{k,g} (\tilde P_{gk})_{g,g} ||_F^2 \\
    &+ u_n^2 \lambda_n \sum_{h,k,l \in [G]^3, h \neq k \neq l}\mathbb E ||\sum_{g \in [-l]} B_{h,g} P_{g,l,-gh} U_l V_k ' B_{k,g} (\tilde P_{gk})_{g,k} ||_F^2 + o(\omega_n^4) \\
    &\lesssim u_n^4 n_G \lambda_n^2 \kappa_n + u_n^4 n_G \phi_n \lambda_n^3 \kappa_n + o(\omega_n^4) = o(\omega_n^4).
\end{align*}

For the second part, it suffices to show that
\begin{gather*}
    \sum_{g, h \in [G]^2} \mathbb E \left( V_h'B_{h,g} U_g + V_g' B_{g,h} U_h \right)^2 = o(\omega_n^2), \\
    \sum_{g \in [G]} \mathbb E \left\{ \left(\sum_{h \in [G]} X_h ' B_{h,g} \tilde U_{g,-h}\right)  + \left( \sum_{k \in [G]} Y_k ' B_{g,k}' \tilde V_{g,-k} \right)  \right\}^2 = o(\omega_n^2).
\end{gather*}
For the first term, we have
\begin{align*}
 \sum_{g, h \in [G]^2} \mathbb E \left( V_h'B_{h,g} U_g + V_g' B_{g,h} U_h \right)^2 \lesssim \sum_{g, h \in [G]^2} \mathbb E \left( V_h'B_{h,g} U_g \right)^2 \lesssim u_n^2 \kappa_n = o(\omega_n^2),
\end{align*}
since $\kappa_n = o(\mu_n^2 + \tilde \mu_n^2)$. For the second term, we have
\begin{align*}
    &\sum_{g \in [G]} \mathbb E \left\{ \left(\sum_{h \in [G]} X_h ' B_{h,g} \tilde U_{g,-h} \right)  + \left( \sum_{k \in [G]} Y_k ' B_{g,k}' \tilde V_{g,-k} \right)  \right\}^2 \\
    &\lesssim \sum_{g \in [G]} \mathbb E  \left(\sum_{h \in [G]} X_h ' B_{h,g} \tilde U_{g,-h}\right)^2  + \sum_{g \in [G]} \mathbb E \left( \sum_{k \in [G]} Y_k ' B_{g,k}' \tilde V_{g,-k} \right)^2,
\end{align*}
and we only compute the first part as the second part can be handled similarly. We have
\begin{align*}
    \sum_{g \in [G]} \mathbb E  \left(\sum_{h \in [G]} X_h ' B_{h,g} \tilde U_{g,-h}\right)^2  \lesssim \underbrace{\sum_{g \in [G]} \mathbb E  \left(\sum_{h \in [G]} \Pi_h ' B_{h,g} \tilde U_{g,-h}\right)^2}_{R_{24}} + \underbrace{\sum_{g \in [G]} \mathbb E  \left(\sum_{h \in [G]} V_h ' B_{h,g} \tilde U_{g,-h}\right)^2}_{R_{25}},
\end{align*}
where
\begin{align*}
    R_{24} &= \sum_{g \in [G]} \mathbb E  \left(\sum_{h, k \in [G]^2, k \neq g, k \neq h} \Pi_h ' B_{h,g} P_{g,k,-gh} U_{k} \right)^2 \\
    &\lesssim \sum_{g \in [G]} \mathbb E  \left(\sum_{h, k \in [G]^2, k \neq g, k \neq h} \Pi_h ' B_{h,g} P_{g,k} U_{k} \right)^2 \\
    &+ \sum_{g \in [G]} \mathbb E  \left(\sum_{h, k \in [G]^2, k \neq g, k \neq h} \Pi_h ' B_{h,g} (\tilde P_{gh})_{g,g} P_{g,k} U_{k} \right)^2 \\
    &+ \sum_{g \in [G]} \mathbb E  \left(\sum_{h, k \in [G]^2, k \neq g, k \neq h} \Pi_h ' B_{h,g} (\tilde P_{gh})_{g,h} P_{h,k} U_{k} \right)^2 \\
    &\lesssim u_n \lambda_n \mu_n^2 + u_n n_G \phi_n \lambda_n^3 \kappa_n + u_n n_G \lambda_n^2 \kappa_n  = o(\omega_n^2),
\end{align*}
since $u_n \lesssim n_G = O(1)$ and $\lambda_n = o(1)$. We also have
\begin{align*}
    R_{25} &= \sum_{g \in [G]} \mathbb E  \left(\sum_{h, k \in [G]^2, k \neq g, k \neq h} V_h ' B_{h,g} P_{g,k,-gh} U_{k} \right)^2 \\
    &\lesssim \sum_{g \in [G]} \mathbb E  \left(\sum_{h, k \in [G]^2, k \neq g, k \neq h} V_h ' B_{h,g} P_{g,k} U_{k} \right)^2 \\
    &+ \sum_{g \in [G]} \mathbb E  \left(\sum_{h, k \in [G]^2, k \neq g, k \neq h} V_h ' B_{h,g} (\tilde P_{gh})_{g,g} P_{g,k} U_{k} \right)^2 \\
    &+ \sum_{g \in [G]} \mathbb E  \left(\sum_{h, k \in [G]^2, k \neq g, k \neq h} V_h ' B_{h,g} (\tilde P_{gh})_{g,h} P_{h,k} U_{k} \right)^2 \\
    &\lesssim u_n^2 \lambda_n \kappa_n = o(\omega_n^2).
\end{align*}
This concludes the proof.

\section{Proof of Lemma~\ref{lem:P_l3o_1}}\label{sec:P_l3o_1_pf}
\textbf{For the first and second statement}, we have
\begin{align*}
 P_{l,g,-ghk} & = W_{l}'\left[ S_W^{-1} + S_W^{-1} W_{ghk}' M_{ghk,ghk}^{-1}     W_{ghk} S_W^{-1}  \right] W_{g} \\
 & = P_{l,g} + P_{l,ghk} M_{ghk,ghk}^{-1} P_{ghk,g}
\end{align*}
where $S_W = \sum_{g \in [G]} W_g' W_g$. Then, we have
\begin{align*}
\tilde  P_{ghk} & = M_{ghk,ghk}^{-1}   P_{ghk,ghk} = M_{ghk,ghk}^{-1} - I_{ghk,ghk},
\end{align*}
which is symmetric. This also implies \cref{eq:p_l30_1} as
\begin{align*}
    P_{ghk,ghk} \tilde  P_{ghk} = P_{ghk,ghk}  M_{ghk,ghk}^{-1} - P_{ghk,ghk} =  \tilde  P_{ghk} - P_{ghk,ghk}.
\end{align*}

\textbf{For the third statement}, we first consider the case that $g \neq h$, $g \neq k$, and $h \neq k$. Then, we have $M_{g,h} = - P_{g,h}$, $M_{g,k} = - P_{g,k}$, $M_{h,k} = - P_{h,k}$. In addition, we will use $||P_{g,h}||_{op} \lesssim \lambda_n \leq 1$ frequently, which is shown in \Cref{lem:nablaB}.

Note that
\begin{align*}
\tilde P_{ghk} = & M_{ghk,ghk}^{-1} P_{ghk,ghk} \\
& = \begin{bmatrix}
    M_{g,g} & M_{g,hk} \\
    M_{hk,g} & M_{hk,hk}
\end{bmatrix}^{-1}
\begin{bmatrix}
    P_{g,g} & P_{g,hk} \\
    P_{hk,g} & P_{hk,hk}
\end{bmatrix} \\
& = \begin{pmatrix}
  M_{g,g}^{-1} + M_{g,g}^{-1} M_{g,hk} S_{hk,g}^{-1} M_{hk,g} M_{g,g}^{-1}     & - M_{g,g}^{-1}  M_{g,hk} S_{hk,g}^{-1} \\
  - S_{hk,g}^{-1} M_{hk,g} M_{g,g}^{-1} & S_{hk,g}^{-1}
\end{pmatrix} \begin{pmatrix}
    P_{g,g} & P_{g,hk} \\
    P_{hk,g} & P_{hk,hk}
\end{pmatrix},
\end{align*}
where
\begin{align*}
 S_{hk,g} =    M_{hk,hk} - M_{hk,g} M_{g,g}^{-1} M_{g,hk}.
\end{align*}

Further denote
\begin{align*}
 S_{h,g} & =    M_{h,h} - M_{h,g} M_{g,g}^{-1} M_{g,h}, \\
 S_{k,g} & =    M_{k,k} - M_{k,g} M_{g,g}^{-1} M_{g,k}, \\
 F_{hk,g} & = M_{h,k} - M_{h,g} M_{g,g}^{-1} M_{g,k} \\
 \tilde S_{k,gh} & = S_{k,g} - \left(M_{k,h} - M_{k,g} M_{g,g}^{-1} M_{g,h}\right)  S_{h,g} ^{-1} \left( M_{h,k} - M_{h,g} M_{g,g}^{-1} M_{g,k} \right).
 \end{align*}

Then, we have
\begin{align}\label{eq:Hinverse1}
S_{hk,g}^{-1} & = \begin{pmatrix}
 S_{h,g} &  F_{hk,g} \\
F_{hk,g}' & S_{k,g}
\end{pmatrix}^{-1} \notag   \\
= & \begin{pmatrix}
S_{h,g}^{-1} +  S_{h,g}^{-1} F_{hk,g}\tilde S_{k,gh}^{-1} F_{hk,g}' S_{h,g}^{-1} & - S_{h,g}^{-1} F_{hk,g}   \tilde S_{k,gh}^{-1} \\
-    \tilde S_{k,gh}^{-1}F_{hk,g} ' S_{h,g}^{-1} & \tilde S_{k,gh}^{-1}
\end{pmatrix}.
\end{align}

This means
\begin{align*}
(\tilde P_{ghk})_{g,g} & =  M_{g,g}^{-1}P_{g,g} + M_{g,g}^{-1} M_{g,hk} S_{hk,g}^{-1} M_{hk,g} M_{g,g}^{-1} P_{g,g} \\
& - M_{g,g}^{-1}  M_{g,hk} S_{hk,g}^{-1}  P_{hk,g} \\
& =  \underbrace{M_{g,g}^{-1}P_{g,g}}_{\Xi_{1,g}} + \underbrace{ M_{g,g}^{-1} M_{g,h} (S_{hk,g}^{-1})_{h,h} M_{h,g} M_{g,g}^{-1} P_{g,g}}_{T_{1,ghk}} \\
& + \underbrace{ M_{g,g}^{-1} M_{g,h} (S_{hk,g}^{-1})_{h,k} M_{k,g} M_{g,g}^{-1} P_{g,g}}_{T_{2,ghk}} \\
& + \underbrace{ M_{g,g}^{-1} M_{g,k} (S_{hk,g}^{-1})_{k,h} M_{h,g} M_{g,g}^{-1} P_{g,g}}_{T_{3,ghk}} \\
& + \underbrace{ M_{g,g}^{-1} M_{g,k} (S_{hk,g}^{-1})_{k,k} M_{k,g} M_{g,g}^{-1} P_{g,g}}_{T_{4,ghk}} \\
& - \underbrace{ M_{g,g}^{-1}  M_{g,h} (S_{hk,g}^{-1})_{h,h}  P_{h,g}}_{T_{5,ghk}} \\
& - \underbrace{ M_{g,g}^{-1}  M_{g,h} (S_{hk,g}^{-1})_{h,k}  P_{k,g}}_{T_{6,ghk}} \\
& - \underbrace{ M_{g,g}^{-1}  M_{g,k} (S_{hk,g}^{-1})_{k,h}  P_{h,g}}_{T_{7,ghk}} \\
& - \underbrace{ M_{g,g}^{-1}  M_{g,k} (S_{hk,g}^{-1})_{k,k}  P_{k,g}}_{T_{8,ghk}},
\end{align*}
where we have
\begin{align*}
||(S_{hk,g}^{-1})_{h,k}||_{op} & \leq ||S_{h,g}^{-1}||_{op} ||F_{hk,g}||_{op}   ||\tilde S_{k,gh}^{-1} ||_{op} \lesssim \lambda_n, \\
||T_{2,ghk}||_{op} &  \lesssim  \lambda_n^2 ||M_{g,h}||_{op} ||M_{g,k}||_{op} = \lambda_n^2 ||P_{g,h}||_{op} ||P_{g,k}||_{op}, \\
||T_{3,ghk}||_{op} &  \lesssim  \lambda_n^2 ||M_{g,h}||_{op} ||M_{g,k}||_{op} = \lambda_n^2 ||P_{g,h}||_{op} ||P_{g,k}||_{op}, \\
||T_{6,ghk}||_{op} &  \lesssim \lambda_n ||M_{g,h}||_{op} ||P_{g,k}||_{op} = \lambda_n  ||P_{g,h}||_{op} ||P_{g,k}||_{op}, \\
||T_{7,ghk}||_{op} &  \lesssim \lambda_n  ||M_{g,h}||_{op} ||P_{g,k}||_{op} = \lambda_n ||P_{g,h}||_{op} ||P_{g,k}||_{op}.
\end{align*}
By \cref{eq:Hinverse1}, we have
\begin{align*}
T_{1,ghk} & = \underbrace{M_{g,g}^{-1} M_{g,h} S_{h,g}^{-1} M_{h,g} M_{g,g}^{-1} P_{g,g}}_{T_{1,1,gh}} \\
& + \underbrace{ M_{g,g}^{-1} M_{g,h}  S_{h,g}^{-1} F_{hk,g}\tilde S_{k,gh}^{-1} F_{hk,g}' S_{h,g}^{-1}  M_{h,g} M_{g,g}^{-1} P_{g,g} }_{T_{1,2,ghk}},
\end{align*}
where
\begin{align*}
||T_{1,1,gh}||_{op} \lesssim ||M_{g,h}||_{op}^2 \lambda_n = \lambda_n ||P_{g,h}||_{op}^2
\end{align*}
and
\begin{align*}
||T_{1,2,ghk}||_{op} & \lesssim ||M_{g,h}||_{op}^2  ||F_{hk,g}||_{op}^2  \lambda_n \\
& \lesssim ||M_{g,h}||_{op}^2 \left( ||M_{h,k}||_{op} + || M_{h,g} M_{g,g}^{-1} M_{g,k} ||_{op}  \right)^2 \lambda_n \\
& \lesssim ||P_{g,h}||_{op}^2 \left( ||M_{h,k}||_{op}^2 + || M_{g,k} ||_{op}^2  \right) \lambda_n \\
& =  \lambda_n ||P_{g,h}||_{op}^2 \left( ||P_{h,k}||_{op}^2 + || P_{g,k} ||_{op}^2  \right).
\end{align*}
As indexes $h$ and $k$ are symmetric, we also have
\begin{align*}
T_{4,ghk}     = T_{4,1,gk} + T_{4,2,ghk},
\end{align*}
where
\begin{align*}
& ||T_{4,1,gk}||_{op} \lesssim \lambda_n ||P_{g,k}||_{op}^2, \\
& ||T_{4,2,ghk}||_{op} \lesssim \lambda_n ||P_{g,k}||_{op}^2 \left( ||P_{k,h}||_{op}^2 + || P_{g,h} ||_{op}^2  \right).
\end{align*}
Following the same argument, we can show that
\begin{align*}
T_{5,ghk} & = \underbrace{ M_{g,g}^{-1}  M_{g,h} S_{h,g}^{-1}  P_{h,g} }_{T_{5,1,gh}} \\
& + \underbrace{ M_{g,g}^{-1}  M_{g,h}  S_{h,g}^{-1} F_{hk,g}\tilde S_{k,gh}^{-1} F_{hk,g}' S_{h,g}^{-1}   P_{h,g} }_{T_{5,2,ghk}},
\end{align*}
where
\begin{align*}
& ||T_{5,1,gh}||_{op} \lesssim ||P_{g,h}||_{op}^2 \\
& ||T_{5,2,ghk}||_{op} \lesssim ||P_{g,h}||_{op}^2  \left( ||P_{h,k}||_{op}^2 + || P_{g,k} ||_{op}^2  \right).
\end{align*}
Following the same argument as $T_{5,ghk}$, we can show that
\begin{align*}
T_{8,ghk} = T_{8,1,gk} + T_{8,1,ghk},
\end{align*}
where
\begin{align*}
& ||T_{8,1,gk}||_{op} \lesssim ||P_{g,k}||_{op}^2,\\
& ||T_{8,2,ghk}||_{op} \lesssim ||P_{g,k}||_{op}^2  \left( ||P_{h,k}||_{op}^2 + || P_{g,h} ||_{op}^2  \right).
\end{align*}
Combining these results, we have
\begin{align*}
& \Xi_{2,gh} = T_{1,1,gh} + T_{5,1,gh}, \quad \Xi_{3,gk} = T_{4,1,gk} + T_{8,1,gk}, \\
& \Xi_{4,ghk} =    T_{1,2,ghk} + T_{2,ghk} + T_{3,ghk} + T_{4,2,ghk} + T_{5,2,ghk} + T_{6,ghk} + T_{7,ghk} + T_{8,2,ghk},
\end{align*}
where
\begin{align*}
||\Xi_{2,gh}||_{op} & \lesssim ||P_{g,h}||_{op}^2, \quad  ||\Xi_{3,gk}||_{op} \lesssim ||P_{g,k}||_{op}^2, \quad \text{and} \\
||\Xi_{4,ghk}||_{op} & \lesssim  (||P_{g,h}||_{op}^2 + ||P_{g,k}||_{op}^2)  ||P_{h,k}||_{op}^2 + \lambda_n ||P_{g,h}||_{op} ||P_{g,k}||_{op}.
\end{align*}

Next, we have

\begin{align*}
(\tilde P_{ghk})_{g,k} & = M_{g,g}^{-1} P_{g,k} + M_{g,g}^{-1} M_{g,hk} S_{hk,g}^{-1} M_{hk,g} M_{g,g}^{-1} P_{g,k} \\
& - M_{g,g}^{-1}  M_{g,hk} S_{hk,g}^{-1}   P_{hk,k} \\
& = \underbrace{ M_{g,g}^{-1} P_{g,k} }_{T_{9,gk}}+ \underbrace{ M_{g,g}^{-1} M_{g,h} (S_{hk,g}^{-1})_{h,hk} M_{hk,g} M_{g,g}^{-1} P_{g,k} }_{T_{10,ghk}} \\
& + \underbrace{ M_{g,g}^{-1} M_{g,k} (S_{hk,g}^{-1})_{k,h} M_{h,g} M_{g,g}^{-1} P_{g,k} }_{T_{11,ghk}} \\
& + \underbrace{ M_{g,g}^{-1} M_{g,k} (S_{hk,g}^{-1})_{k,k} M_{k,g} M_{g,g}^{-1} P_{g,k} }_{T_{12,ghk}} \\
& - \underbrace{M_{g,g}^{-1}  M_{g,h} (S_{hk,g}^{-1})_{h,h}   P_{h,k} }_{T_{13,ghk}} - \underbrace{ M_{g,g}^{-1}  M_{g,h} (S_{hk,g}^{-1})_{h,k}   P_{k,k} }_{T_{14,ghk}}\\
& - \underbrace{ M_{g,g}^{-1}  M_{g,k} (S_{hk,g}^{-1})_{k,h}   P_{h,k} }_{T_{15,ghk}}- \underbrace{ M_{g,g}^{-1}  M_{g,k} (S_{hk,g}^{-1})_{k,k}   P_{k,k} }_{T_{16,ghk}},
\end{align*}
where
\begin{align*}
& || T_{9,gk} ||_{op} \lesssim ||P_{g,k}||_{op} , \\
& || T_{10,ghk } ||_{op} \lesssim ||P_{g,h}||_{op} ||P_{g,k}||_{op} (||P_{g,h}||_{op} + ||P_{g,k}||_{op}), \\
& || T_{11,ghk } ||_{op} \lesssim \lambda_n ||P_{g,h}||_{op} ||P_{g,k}||_{op}^2, \\
& || T_{13,ghk } ||_{op} \lesssim ||P_{g,h}||_{op} ||P_{h,k}||_{op}, \\
& || T_{15,ghk } ||_{op} \lesssim ||P_{g,k}||_{op} ||P_{h,k}||_{op}.
\end{align*}
In addition, note that
\begin{align*}
S_{hk,g}^{-1} & = \begin{pmatrix}
 S_{h,g} &  F_{hk,g} \\
F_{hk,g}' & S_{k,g}
\end{pmatrix}^{-1} = \begin{pmatrix}
 \tilde S_{h,gk}^{-1} & - \tilde S_{h,gk}^{-1}    F_{hk,g} S_{k,g} ^{-1} \\
 -  S_{k,g} ^{-1}  F_{hk,g}'  \tilde S_{h,gk}^{-1}  &  S_{k,g} ^{-1} + S_{k,g} ^{-1}  F_{hk,g}'   \tilde S_{h,gk}^{-1} F_{hk,g} S_{k,g} ^{-1}
\end{pmatrix},
\end{align*}
where
\begin{align*}
 S_{h,g} & =    M_{h,h} - M_{h,g} M_{g,g}^{-1} M_{g,h}, \\
 S_{k,g} & =    M_{k,k} - M_{k,g} M_{g,g}^{-1} M_{g,k}, \\
 F_{hk,g} & = M_{h,k} - M_{h,g} M_{g,g}^{-1} M_{g,k} \\
 \tilde S_{h,gk} & = S_{h,g} - \left(M_{h,k} - M_{h,g} M_{g,g}^{-1} M_{g,k}\right)  S_{k,g} ^{-1} \left( M_{k,h} - M_{k,g} M_{g,g}^{-1} M_{W,[g,hs]} \right).
 \end{align*}
Then, we have
\begin{align*}
    T_{12,ghk} & = \underbrace{ M_{g,g}^{-1} M_{g,k} S_{k,g} ^{-1} M_{k,g} M_{g,g}^{-1} P_{g,k} }_{T_{12,1,gk}}\\
    & + \underbrace{ M_{g,g}^{-1} M_{g,k} S_{k,g} ^{-1}  F_{hk,g}'   \tilde S_{h,gk}^{-1} F_{hk,g} S_{k,g} ^{-1} M_{k,g} M_{g,g}^{-1} P_{g,k} }_{T_{12,2,ghk}},
\end{align*}
where
\begin{align*}
& ||T_{12,1,gk}||_{op}    \lesssim ||P_{g,k}||_{op}^3 \quad \text{and} \quad ||T_{12,2,ghk}||_{op} \lesssim ||P_{g,k}||_{op}^3 ||F_{hk,g}||_{op}^2   \lesssim ||P_{g,k}||_{op}^3 (||P_{h,k}||_{op}^2 + ||P_{h,g}||_{op}^2).
\end{align*}
Similarly, we have
\begin{align*}
T_{16,ghk}  & = \underbrace{ M_{g,g}^{-1} M_{g,k} S_{k,g} ^{-1}  P_{k,k} }_{T_{16,1,gk}}\\
    & + \underbrace{ M_{g,g}^{-1} M_{g,k} S_{k,g} ^{-1}  F_{hk,g}'   \tilde S_{h,gk}^{-1} F_{hk,g} S_{k,g} ^{-1} P_{k,k} }_{T_{16,2,ghk}},
\end{align*}
where
\begin{align*}
& ||T_{16,1,gk}||_{op}    \lesssim \lambda_n ||P_{g,k}||_{op}, \quad \text{and} \quad ||T_{16,2,ghk}||_{op}  \lesssim \lambda_n ||P_{g,k}||_{op} (||P_{h,k}||_{op}^2 + ||P_{h,g}||_{op}^2).
\end{align*}
Last, we have
\begin{align*}
||T_{14,ghk}||_{op} & =  || M_{g,g}^{-1}  M_{g,h} \tilde S_{h,gk}^{-1}    F_{hk,g} S_{k,g} ^{-1}    P_{k,k}||_{op} \\
& \lesssim  ||  M_{g,h} ||_{op} ||   F_{hk,g}||_{op}  \lambda_n \\
& \lesssim  \lambda_n  ||  P_{g,h} ||_{op} \left( ||P_{h,k}||_{op} + ||P_{h,g}||_{op}||P_{g,k}||_{op} \right) .
\end{align*}

Therefore, we have
\begin{align*}
& \Xi_{7,gk} = T_{9,gk} + T_{12,1,gk} + T_{16,1,gk} \\
& \Xi_{8,ghk} = T_{10,ghk} + T_{11,ghk}+ T_{12,2,ghk} + T_{13,ghk} + T_{14,ghk} + T_{15,ghk} + T_{16,2,ghk},
\end{align*}
such that
\begin{align*}
& || \Xi_{7,gk}||_{op} \lesssim ||P_{g,k}||_{op} \\
& || \Xi_{8,ghk}||_{op} \lesssim \left(||P_{g,h}||_{op} ||P_{g,k}||_{op} + ||P_{h,k}||_{op} \right) \left(||P_{g,h}||_{op} + ||P_{g,k}||_{op}\right).
\end{align*}

Due to the symmetry between the indexes $h$ and $k$, the bounds for $\Xi_{5,gh}$ and $\Xi_{6,ghk}$ can by obtained by exchanging $h$ with $k$ in the bounds for $\Xi_{7,gk}$ and $\Xi_{8,ghk}$, respectively.

\textbf{For the fourth claim}, we have
\begin{align*}
\tilde P_{gh} = & M_{W,[gh,gh]}^{-1} P_{W,[gh,gh]} \\
& = \begin{pmatrix}
  M_{g,g}^{-1} + M_{g,g}^{-1} M_{g,h} S_{h,g}^{-1} M_{h,g} M_{g,g}^{-1}     & - M_{g,g}^{-1}  M_{g,h} S_{h,g}^{-1} \\
  - S_{h,g}^{-1} M_{h,g} M_{g,g}^{-1} & S_{h,g}^{-1}
\end{pmatrix} \\
& \times \begin{pmatrix}
    P_{g,g} & P_{g,h} \\
    P_{h,g} & P_{h,h}
\end{pmatrix},
\end{align*}
where
\begin{align*}
 S_{h,g} =    M_{h,h} - M_{h,g} M_{g,g}^{-1} M_{g,h}.
\end{align*}
This implies
\begin{align*}
(\tilde P_{gh})_{g,g} =  \Xi_{1,g} +  \underbrace{  M_{g,g}^{-1} M_{g,h} S_{h,g}^{-1} M_{h,g} M_{g,g}^{-1} P_{g,g} -  M_{g,g}^{-1}  M_{g,h} S_{h,g}^{-1}  P_{h,g} }_{\Xi_{9,gh}}
\end{align*}
such that
\begin{align*}
||\Xi_{9,gh}||_{op} \lesssim ||P_{g,h}||_{op}^2.
\end{align*}
Next, we have
\begin{align*}
(\tilde P_{gh})_{g,h} & = \left( M_{g,g}^{-1} + M_{g,g}^{-1} M_{g,h} S_{h,g}^{-1} M_{h,g} M_{g,g}^{-1} \right) P_{g,h} \\
& - M_{g,g}^{-1}  M_{g,h} S_{h,g}^{-1} P_{h,h}
\end{align*}
such that
\begin{align*}
||(\tilde P_{gh})_{g,h}||_{op} \lesssim ||P_{g,h}||_{op}.
\end{align*}

\textbf{For the fifth claim}, by \Cref{lem:P_l3o_1}, we have
\begin{align*}
|| P_{h,k, -hg}||_{op} & \lesssim ||P_{h,k}||_{op} + ||P_{k,gh} (\tilde P_{gh})_{gh,h}||_{op} \notag \\
& \lesssim ||P_{h,k}||_{op}+||P_{W, [g,k]}||_{op} ||(\tilde P_{gh})_{g,h}||_{op} \notag  \\
& \lesssim ||P_{h,k}||_{op}+||P_{g,k}||_{op}||P_{g,h}||_{op}.
\end{align*}

\textbf{For the last claim}, we first consider the case that $h \neq k$, $h' \neq k'$. We have
\begin{align*}
& \sum_{l \in [-gg'hh'kk']} P_{g,l,-ghk} P_{l,g',-g'h'k'} \\
& = \sum_{l \in [G]} P_{g,l,-ghk} P_{l,g',-g'h'k'}  -\sum_{l \in (g,h,k)} P_{g,l,-ghk} P_{l,g',-g'h'k'} \\
& -\sum_{l \in (g',h',k')} P_{g,l,-ghk} P_{l,g',-g'h'k'} + \sum_{l \in (g,h,k) \cap (g',h',k')} P_{g,l,-ghk} P_{l,g',-g'h'k'}  \\
& = \sum_{l \in [G]}  \left[P_{g,l} + (\tilde  P_{ghk})_{g,ghk} P_{ghk,l} \right] \left[  P_{l,g'} + P_{l,g'h'k'} (\tilde  P_{g'h'k'})_{g'h'k',g'}  \right] \\
& -\sum_{l \in (g,h,k)} P_{g,l,-ghk} P_{l,g',-g'h'k'} -\sum_{l \in (g',h',k')} P_{g,l,-ghk} P_{l,g',-g'h'k'} \\
& + \sum_{l \in (g,h,k) \cap (g',h',k')} P_{g,l,-ghk} P_{l,g',-g'h'k'}  \\
& = P_{g,g'} +  P_{g,g'h'k'} (\tilde P_{g'h'k'})_{g'h'k',g'} + (\tilde P_{ghk})_{g,ghk} P_{ghk,g'} \\
& + (\tilde P_{ghk})_{g,ghk} P_{ghk,g'h'k'} (\tilde P_{g'h'k'})_{g'h'k',g'} \\
& - \sum_{l \in (g,h,k)}  (\tilde  P_{ghk})_{g,l} \left[  P_{l,g'} + P_{l,g'h'k'} (\tilde  P_{g'h'k'})_{g'h'k',g'}  \right] \\
& - \sum_{l \in (g',h',k')}  \left[P_{g,l} + (\tilde  P_{ghk})_{g,ghk} P_{ghk,l} \right](\tilde  P_{g'h'k'})_{l,g'} + \sum_{l \in (g,h,k) \cap (g',h',k')} P_{g,l,-ghk} P_{l,g',-g'h'k'}  \\
& = P_{g,g'} +  P_{g,g'h'k'} (\tilde P_{g'h'k'})_{g'h'k',g'} + (\tilde P_{ghk})_{g,ghk} P_{ghk,g'}\\
& + (\tilde P_{ghk})_{g,ghk} P_{ghk,g'h'k'} (\tilde P_{g'h'k'})_{g'h'k',g'} \\
& - (\tilde  P_{ghk})_{g,ghk} \left[ P_{ghk,g'} + P_{ghk,g'h'k'} (\tilde  P_{g'h'k'})_{g'h'k',g'}  \right] \\
& - \left[P_{g,g'h'k'} + (\tilde  P_{ghk})_{g,ghk} P_{ghk,g'h'k'} \right](\tilde  P_{g'h'k'})_{g'h'k',g'} + \sum_{l \in (g,h,k) \cap (g',h',k')} P_{g,l,-ghk} P_{l,g',-g'h'k'}  \\
& = P_{g,g'} -  (\tilde  P_{ghk})_{g,ghk} P_{ghk,g'h'k'} (\tilde  P_{g'h'k'})_{g'h'k',g'} + \sum_{l \in (g,h,k) \cap (g',h',k')} P_{g,l,-ghk} P_{l,g',-g'h'k'}.
\end{align*}

Next, we consider the case that $h = k$ but $h' \neq k'$. We have
\begin{align*}
& \sum_{l \in [-gg'hh'k']} P_{g,l,-gh} P_{l,g',-g'h'k'} \\
& = \sum_{l \in [-gg'hh'k']}  \left[P_{g,l} + [\tilde  P_{gh}]_{[g,gh]} P_{W,[gh,l]} \right] \left[  P_{l,g'} + P_{l,g'h'k'} (\tilde  P_{g'h'k'})_{g'h'k',g'} \right] \\
& = \sum_{l \in [G]}  \left[P_{g,l} + [\tilde  P_{gh}]_{[g,gh]} P_{W,[gh,l]} \right] \left[  P_{l,g'} + P_{l,g'h'k'} (\tilde  P_{g'h'k'})_{g'h'k',g'}  \right] \\
& -  [\tilde  P_{gh}]_{[g,gh]}  \left[ P_{W,[gh,g']} + P_{W,[gh,g'h'k']} (\tilde  P_{g'h'k'})_{g'h'k',g'}  \right] \\
& -   \left[P_{g,g'h'k'} + [\tilde  P_{gh}]_{[g,gh]} P_{W,[gh,g'h'k']} \right] (\tilde  P_{g'h'k'})_{g'h'k',g'} \\
& + \sum_{l \in [gh] \cap (g',h',k')} P_{g,l,-gh} P_{l,g',-g'h'k'} \\
& = P_{g,g'} -   [\tilde  P_{gh}]_{[g,gh]} P_{W,[gh,g'h'k']}  (\tilde  P_{g'h'k'})_{g'h'k',g'} \\
& = P_{g,g'} -   [\tilde  P_{ghh}]_{[g,ghh]} P_{W,[ghh,g'h'k']}  (\tilde  P_{g'h'k'})_{g'h'k',g'} \\
& + \sum_{l \in [ghh] \cap (g',h',k')} P_{W,[g,l],-ghh} P_{l,g',-g'h'k'}.
\end{align*}
where the last equality is by the definition of $(\tilde P_{ghk})_{g,ghk}$ and $P_{W,[g,ghk]}$ when $h = k$.

For the same reason, the third statement still holds for cases when $(h' = k', h \neq k)$  and $(h' = k', h = k)$.

\section{Technical Lemmas}
\subsection{Lemma~\ref{lem:T}}
\begin{lem}\label{lem:T}
Let $p = 2^t$ for an arbitrary positive integer $t$, $T$ is the truncation matrix defined in \Cref{sec:clt_pf} and $A \in \Re^{n \times n}$. Then, we have
\begin{align*}
    ||T \circ A||_{op} \leq C p ||A||_p,
\end{align*}
where $C$ is a universal constant and $||A||_p$ is the $p$-th Schatten norm defined as
\begin{equation*}
  \norm{A}_p = \left[tr( (AA')^{p/2})\right]^{1/p} = \left[tr( (A'A)^{p/2})\right]^{1/p} = \left(\sum_{r = 1}^{r_n} \sigma_r^p(A)\right)^{1/p},
\end{equation*}
where $\{\sigma_r(A)\}_{r = 1}^{r_n}$ are the singular values of $A$ in ascending order.
\end{lem}
\begin{proof}
The proof is due to a post on \href{https://mathoverflow.net/questions/177198/norm-of-triangular-truncation-operator-on-rank-deficient-matrices}{mathoverflow.net}. We reproduce the proof here for completeness only.

Note that
\begin{align*}
||T \circ A||_{op} \leq ||T \circ A||_{p}.
\end{align*}
We aim to show $||T \circ A||_{p} \leq C p ||A||_p$ for all $p = 2^t$ by induction. For $t=1$ (i.e., $p = 2$), we have
\begin{align*}
||T \circ A||_{2} = |T \circ A||_{F} \leq ||A||_F = ||A||_2 \quad \text{and} \quad    ||(\iota_n\iota_n' - T) \circ A||_{2} \leq ||A||_2.
\end{align*}
Suppose $||T \circ A||_{p} \leq C p ||A||_p$ and $||(\iota_n\iota_n' - T) \circ A||_{p} \leq C p ||A||_p$. In addition, we have
\begin{align*}
    [(T \circ A)'(T \circ A)]_{jk} & = \sum_{i = 1}^n A_{ij} 1\{i \geq j\} A_{ik}1\{i \geq k\} \\
    & =  \sum_{i = 1}^n A_{ij} 1\{i \geq j\} A_{ik}1\{i \geq k\} 1\{j \geq k\} + \sum_{i = 1}^n A_{ij} 1\{i \geq j\} A_{ik}1\{i \geq k\} 1\{j < k\} \\
    & =  \sum_{i = 1}^n A_{ij} 1\{i \geq j\} A_{ik} 1\{j \geq k\} + \sum_{i = 1}^n A_{ij}  A_{ik}1\{i \geq k\} 1\{j < k\} \\
    & = [T\circ ((T \circ A)' A)]_{jk} + [(\iota_n\iota_n' - T)\circ (A' (T\circ A))]_{jk},
\end{align*}
which implies
\begin{align*}
(T \circ A)'(T \circ A) = T\circ ((T \circ A)' A) + (\iota_n\iota_n' - T)\circ (A' (T\circ A)).
\end{align*}
Therefore, we have
\begin{align*}
||T \circ A||_{2p}^2 & =    ||(T \circ A)'(T \circ A)||_p \\
& \leq  ||T\circ ((T \circ A)' A) ||_p + ||(\iota_n\iota_n' - T)\circ (A' (T\circ A))||_p \\
& \leq C p ||(T \circ A)' A||_p + C p ||A' (T\circ A)||_p \\
& \leq C p ||(T \circ A)||_{2p} ||A||_{2p} + C p ||A||_{2p} ||T\circ A||_{2p},
\end{align*}
where the last inequality is due to the H\"{o}lder's inequality for Schatten norm, i.e., \textcite[EX. IV.2.7]{B13}. This implies
\begin{align*}
||T \circ A||_{2p} \leq C (2p) ||A||_{2p},
\end{align*}
which concludes the proof.
\end{proof}

\subsection{Lemma~\ref{lem:nablaB}}
\begin{lem}\label{lem:nablaB}
Recall $\nabla(B)$ and $\Delta(B)$ defined in \Cref{sec:clt_pf}. Then, we have
\begin{align*}
& ||(BB')_{g,g}||_{op} = O( \lambda_n), \quad ||(B' B)_{g,g}||_{op} = O( \lambda_n), \quad      ||B||_{op} = O(1),  \\
& ||\nabla(B)||_{op} = O \left(\eta_n^{1/2}\right), \quad ||\Delta(B)||_{op} = O \left(\eta_n^{1/2}\right), \quad  ||B_{g,h}||_{op} = O(\lambda_n),\\
& \sum_{h \in [G]}||B_{g,h}||_{op}^2 \lesssim \phi_n + \lambda_n^2 \quad \text{and} \quad \sum_{g \in [G]}||B_{g,h}||_{op}^2 \lesssim \phi_n + \lambda_n^2.
\end{align*}
\end{lem}
\begin{proof}
Note that $B = A - MD$. We have
\begin{align*}
        ||(BB')_{g,g}||_{op} & = ||(A_{\bullet,g} - (D M)_{\bullet,g})'(A_{\bullet,g} - (D M)_{\bullet,g}) ||_{op} \\
  & \leq 2 \norm{A_{\bullet,g}' A_{\bullet,g}}_{op} + 2 ||(D M)_{\bullet,g}' (D M)_{\bullet,g} ||_{op} \\
  & \leq 2 ||(A'A)_{g,g}||_{op} + 2 ||(M D' D M)_{g,g} ||_{op} \\
  & \leq 2 \lambda_n + ||D||_{op}^2 = O(\lambda_n),
\end{align*}
$A_{\bullet,g} $ denote the columns of $A$ corresponding to the $g$-th cluster and
the last inequality is by
\begin{align*}
    ||(AA')_{g,g}||_{op} = ||W_{g,\bullet} (W'W)^{-1/2} A_0 A_0' (W'W)^{-1/2}W_{\bullet,g}||_{op} \leq ||P_{g,g}||_{op} \leq \lambda_n.
\end{align*}

For the second result of \Cref{lem:nablaB}, we have
\begin{align}\label{eq:BB^T}
        ||(B'B)_{g,g}||_{op} = || (A_{g,\bullet} - D_{g,\bullet}M)(A_{g,\bullet} - D_{g,\bullet}M)' ||_{op} \leq  2 ||(AA')_{g,g} ||_{op}+ || (D M D)_{g,g} ||_{op}^2 \lesssim \lambda_n,
\end{align}
where $A_{g,\bullet} $ denote the rows of $A$ corresponding to the $g$-th cluster and the second inequality is by
\begin{align*}
    ||A_{g,g}||_{op} = ||W_{g,\bullet} (W'W)^{-1/2} A_0  (W'W)^{-1/2}W_{\bullet,g}||_{op} \leq ||P_{g,g}||_{op} \leq \lambda_n,
\end{align*}
and
\begin{align*}
||(D M D)_{g,g}||_{op} & \leq ||D||_{op}^2 = \max_{g \in [G]} \left\Vert M_{g,g}^{-1}  A_{g,g} \right\Vert_{op}^2 \lesssim  \max_{g \in [G]} \left\Vert A_{g,g} \right\Vert_{op}^2  \lesssim \lambda_n^2 \leq \lambda_n.
\end{align*}

For the third result of \Cref{lem:nablaB}, we have
\begin{align*}
        ||B||_{op} & \leq ||A||_{op} + ||M||_{op} ||D||_{op} \lesssim 1 + \lambda_n \lesssim 1.
\end{align*}



For the fourth result, we have
\begin{align*}
||\nabla(B)||_{op}  & = ||\nabla(A) + \nabla(M D )||_{op} \leq ||\nabla(A)||_{op} + ||\nabla(M D)||_{op}.
\end{align*}
In addition, for any $p = 2^q$, we have
\begin{align*}
||T \circ A ||_{op} & \leq C p \left(\sum_{r = 1}^{r_n} \sigma_r^p(A)\right)^{1/p} \leq  C p r_n^{1/p},
\end{align*}
where the first inequality is by \Cref{lem:T}, and the second inequality is by $\sigma_1(A) \lesssim 1$.

Choosing $q = \lfloor \log_2 (\log(r_n)) \rfloor + 1$, we have
\begin{align*}
||T \circ A ||_{op} \leq C \log (r_n).
\end{align*}

Similarly, for $p = 2^q$, we have
\begin{align*}
||\nabla(M D)||_{op}    & \leq C p ||MD||_p  \leq C p \lambda_n n^{1/p},
\end{align*}
where the second inequality is by the fact that $||MD||_{op} \leq \lambda_n$ and $\operatorname{rank}(MD) \leq n$.  Choosing $q = \lfloor \log_2 (\log(n)) \rfloor + 1$, we have
\begin{align*}
||T \circ (M D) ||_{op} \leq C \log (n) \lambda_n.
\end{align*}
Combining the two bounds, we have
\begin{align*}
||\nabla(B)||_{op} \lesssim    \log (r_n)  +   \log (n) \lambda_n  \lesssim \eta_n^{1/2}.
\end{align*}

We can show $||\Delta(B)||_{op} = O \left(\eta_n^{1/2}\right)$ in the same manner.

For the sixth result, we note that
\begin{align*}
    ||B_{g,h}||_{op} & \lesssim ||A_{g,h}||_{op} +||M_{g,h}||_{op} ||M_{h,h}^{-1}||  ||A_{h,h}||_{op}  \\
    & \lesssim ||W_{g,\bullet}  (W'W)^{-1} A_0  (W'W)^{-1}W_{\bullet,h}||_{op} + \lambda_n \\
    & \lesssim \lambda_n.
\end{align*}

For the seventh result, we note that
\begin{align*}
    \sum_{h \in [G]} ||B_{g,h}||_{op}^2 & =     \sum_{h \in [G]} ||A_{g,h}||_{op}^2 + \sum_{h \in [G]} ||M_{g,h} M_{h,h}^{-1} A_{h,h}||_{op}^2 \\
    & \lesssim \sum_{h \in [G]} ||P_{g,h}||_{op}^2 + \lambda_n^2 (1 + \sum_{h \in [G]} ||P_{g,h}||_{op}^2) \\
    & \lesssim \phi_n + \lambda_n^2.
\end{align*}
The last result can be established in the same manner.
\end{proof}

\subsection{Lemma~\ref{lem:4mom}}
\begin{lem}\label{lem:4mom}
Suppose \Cref{ass:dgp} holds. Then, we have
\begin{align*}
\left\Vert \mathbb E \left(U_{g} U_{g}' + V_{g}V_{g}'\right)\left(||U_{g}||_2^2 + ||V_{g}||_2^2\right) \right\Vert_{op} \lesssim u_n^{\frac{q-2}{q-1}} n_G^{\frac{q}{q-1}}.
\end{align*}
\end{lem}
\begin{proof}
    We have
    \begin{align*}
        &\left\Vert \mathbb E \left(U_{g} U_{g}' + V_{g}V_{g}'\right)\left(||U_{g}||_2^2 + ||V_{g}||_2^2\right) \right\Vert_{op} \\
        &\leqslant \left\Vert \mathbb E \left(U_{g} U_{g}' ||U_{g}||_2^2 \right)\right\Vert_{op} + \left\Vert \mathbb E \left(V_{g} V_{g}' ||V_{g}||_2^2 \right)\right\Vert_{op} \\ &+ \left\Vert \mathbb E \left(U_{g} U_{g}' ||V_{g}||_2^2 \right)\right\Vert_{op} + \left\Vert \mathbb E \left(V_{g} V_{g}' ||U_{g}||_2^2 \right)\right\Vert_{op}.
    \end{align*}
    For the first term, fix an $x \in \mathbf R^{n_g}$ with $||x||_2 = 1$, then
    \begin{align*}
        &\quad x' \mathbb E \left(U_{g} U_{g}' ||U_{g}||_2^2 \right) x \\
        &= \mathbb E \left( (x' U_{g})^2 ||U_{g}||_2^2 \right) \\
        &= \mathbb E \left( (x' U_{g})^2 ||U_{g}||_2^2 1_{\{||U_{g}||_2 \leqslant M_g \}} \right) + \mathbb E \left( (x' U_{g})^2 ||U_{g}||_2^2 1_{\{||U_{g}||_2 > M_g \}} \right) \\
        &\lesssim u_n M_g^2 + \mathbb E \left(||U_{g}||_2^4 1_{\{||U_{g}||_2 > M_g \}} \right) \\
        &\lesssim u_n M_g^2 + \frac{n_G^q}{M_g^{2q-4}},
    \end{align*}
    where we use the fact that
    \begin{align*}
        \mathbb E \left(||U_{g}||_2^{2q}\right) = \mathbb E \left( \sum_{i \in [n_{g}]} e_{i,g}^2\right)^q \lesssim n_g^{q-1} \sum_{i \in [n_{g}]} \mathbb E e_{i,g}^{2q} \lesssim n_g^q,
    \end{align*}
    and note that the result holds uniformly in $x$; optimizing $M_g$, we obtain
    \begin{align*}
        \left\Vert \mathbb E \left(U_{g} U_{g}' ||U_{g}||_2^2 \right)\right\Vert_{op} \lesssim u_n^{\frac{q-2}{q-1}} n_G^{\frac{q}{q-1}}.
    \end{align*}
    The second term can be handled similarly. For the last two terms, we can use similar argument, but with H\"older's inequality for the second part of the bound.
\end{proof}

\subsection{Lemma~\ref{lem:E4}}
\begin{lem}\label{lem:E4}
Suppose \Cref{ass:reg} holds. Then, we have
\begin{align*}
    & \sum_{g \in [G]} \sum_{h \geq g} \mathbb E (V_{h}' B_{h,g} U_{g})^4 \lesssim u_n^{\frac{2(q-2)}{q-1}} n_G^{\frac{2q}{q-1}} \lambda_n^2 \kappa_n, \\
    & \sum_{g \in [G]} \sum_{h_1,h_2 \geq g, h_1 \neq h_2}\mathbb E (V_{h_1}' B_{h_1g} U_{g})^2(V_{h_2}' B_{h_2g} U_{g})^2 \lesssim  u_n^{\frac{3q-4}{q-1}} n_G^{\frac{q}{q-1}}  \lambda_n \kappa_n .
\end{align*}
\end{lem}

\begin{proof}
We have
\begin{align*}
    \sum_{g \in [G]} \sum_{h \geq g} \mathbb E (V_{h}' B_{h,g} U_{g})^4 & =  \sum_{g \in [G]} \sum_{h \geq g}  \mathbb E tr( U_{g} U_{g}' B_{h,g}' V_{h} V_{h}' B_{h,g}U_{g} U_{g}' B_{h,g}' V_{h} V_{h}' B_{h,g}) \\
    & \lesssim \sum_{g \in [G]} \sum_{h \geq g} \mathbb E tr( B_{h,g}' V_{h} V_{h}' B_{h,g}U_{g} U_{g}' ||U_{g}||_2^2 B_{h,g}' V_{h} V_{h}' B_{h,g}) \\
    & \lesssim \sum_{g \in [G]} \sum_{h \geq g} u_n^{\frac{q-2}{q-1}} n_G^{\frac{q}{q-1}} \mathbb E tr( B_{h,g}' V_{h} V_{h}' B_{h,g}  B_{h,g}' V_{h} V_{h}' B_{h,g}) \\
    & \lesssim \sum_{g \in [G]} \sum_{h \geq g} u_n^{\frac{q-2}{q-1}} n_G^{\frac{q}{q-1}} \left\Vert B_{h,g}  B_{h,g}'\right\Vert_{op} \mathbb E tr( B_{h,g}' V_{h} V_{h}' V_{h} V_{h}' B_{h,g}) \\
    & \lesssim \sum_{g \in [G]} \sum_{h \geq g} u_n^{\frac{2(q-2)}{q-1}} n_G^{\frac{2q}{q-1}} \lambda_n^2 ||B_{h,g}||_F^2 \\
    & \lesssim u_n^{\frac{2(q-2)}{q-1}} n_G^{\frac{2q}{q-1}} \lambda_n^2 \kappa_n,
\end{align*}
where the second inequality is by \Cref{lem:4mom} and the second last inequality is by $\left\Vert B_{h,g}\right\Vert_{op} \lesssim \lambda_n$ as established in \Cref{lem:nablaB}.

For the second result, we have
\begin{align*}
& \sum_{g \in [G]} \sum_{h_1,h_2 \geq g, h_1 \neq h_2}\mathbb E (V_{h_1}' B_{h_1g} U_{g})^2(V_{h_2}' B_{h_2g} U_{g})^2 \\
& =   \sum_{g \in [G]} \sum_{h_1,h_2 \geq g, h_1 \neq h_2} \mathbb E \left(U_{g}' B_{h_1g}' \Omega_{V,h_1} B_{h_1g}U_{g}\right)\left(U_{g}' B_{h_2g}' \Omega_{V,h_2} B_{h_2g}U_{g}\right) \\
& \lesssim   \sum_{g, h_1,h_2 \in [G]^3} u_n^2 \mathbb E \left(U_{g}' B_{h_1g}' B_{h_1g}U_{g}\right)\left(U_{g}' B_{h_2g}' B_{h_2g}U_{g}\right) \\
& \lesssim u_n^2 \sum_{g \in [G]} \mathbb E \left(U_{g}' (B' B)_{g,g}U_{g}\right)^2 \\
& \lesssim u_n^2 \sum_{g \in [G]} \lambda_n \mathbb E U_{g}' (B' B)_{g,g}U_{g} ||U_{g}||_2^2 \\
& \lesssim u_n^{\frac{3q-4}{q-1}} n_G^{\frac{q}{q-1}}  \lambda_n \kappa_n.
    \end{align*}


\end{proof}

\subsection{Lemma~\ref{lem:V1}}
\begin{lem}\label{lem:V1}
    Suppose Assumptions~\ref{ass:dgp}--\ref{ass:var_L3O} hold. Then, we have
\begin{align*}
\mathbb V \left(\sum_{g, h, k \in [G]^3} \left(X_h ' B_{h,g} Y_g \right) \left(X_k ' B_{k,g} U_g \right)  \right) =   o(\omega_n^4).
\end{align*}
\end{lem}
\begin{proof}
For the first term, we have
\begin{align*}
& \mathbb V\left(\sum_{g, h, k \in [G]^3} \left(X_h ' B_{h,g} Y_g \right) \left(X_k ' B_{k,g} U_g \right)  \right) \\
& \lesssim \underbrace{ \mathbb V\left(\sum_{g, h, k \in [G]^3} \left(\Pi_h' B_{h,g} \Gamma_g \right) \left(\Pi_k ' B_{k,g} U_g \right)  \right) }_{R_{\ref{lem:V1},1}} + \underbrace{ \mathbb V\left(\sum_{g, h, k \in [G]^3} \left(\Pi_h' B_{h,g} \Gamma_g \right) \left(V_k ' B_{k,g} U_g \right)  \right) }_{R_{\ref{lem:V1},2}}   \\
& + \underbrace{  \mathbb V\left(\sum_{g, h, k \in [G]^3} \left(V_h' B_{h,g} \Gamma_g \right) \left(\Pi_k ' B_{k,g} U_g \right)  \right) }_{R_{\ref{lem:V1},3}} +   \underbrace{ \mathbb V\left(\sum_{g, h, k \in [G]^3} \left(V_h' B_{h,g} \Gamma_g \right) \left(V_k ' B_{k,g} U_g \right)  \right)  }_{R_{\ref{lem:V1},4}}  \\
& + \underbrace{  \mathbb V\left(\sum_{g, h, k \in [G]^3} \left(\Pi_h' B_{h,g} U_g \right) \left(\Pi_k ' B_{k,g} U_g \right)  \right) }_{R_{\ref{lem:V1},5}}  +  \underbrace{ \mathbb V\left(\sum_{g, h, k \in [G]^3} \left(\Pi_h' B_{h,g} U_g \right) \left(V_k ' B_{k,g} U_g \right)  \right)  }_{R_{\ref{lem:V1},6}}  \\
& + \underbrace{  \mathbb V\left(\sum_{g, h, k \in [G]^3} \left(V_h' B_{h,g} U_g \right) \left(\Pi_k ' B_{k,g} U_g \right)  \right) }_{R_{\ref{lem:V1},7}}  +  \underbrace{ \mathbb V\left(\sum_{g, h, k \in [G]^3} \left(V_h' B_{h,g} U_g \right) \left(V_k ' B_{k,g} U_g \right)  \right)  }_{R_{\ref{lem:V1},8}},
\end{align*}
where
\begin{align*}
 R_{\ref{lem:V1},1} & =  \sum_{g \in [G]}  \mathbb V\left(\left(H_g' \Gamma_g \right) \left(H_g' U_g \right)  \right) \\
 & \lesssim \sum_{g \in [G]} u_n (H_g' \Gamma_g)^2 ||H_g||_2^2 \\
 & \lesssim u_n n_G \zeta_{H,n} \mu_n^2 = o ((\mu_n^2 + \kappa_n + \tilde \mu_n^2)^2) = o(\omega_n^4),
\end{align*}
\begin{align*}
R_{\ref{lem:V1},2} &  = \mathbb V\left(\sum_{g, k \in [G]^2} \left(H_g' \Gamma_g \right) \left(V_k ' B_{k,g} U_g \right)  \right) \\
& \lesssim \sum_{g, k \in [G]^2} \mathbb V\left( \left(H_g' \Gamma_g \right) \left(V_k ' B_{k,g} U_g \right)  \right) \\
& \lesssim \sum_{g, k \in [G]^2} u_n^2 \left(H_g' \Gamma_g \right)^2 ||B_{k,g}||_F^2 \\
& \lesssim u_n^2 n_G \zeta_{H,n} \kappa_n = o ((\mu_n^2 + \kappa_n + \tilde \mu_n^2)^2) = o(\omega_n^4),
\end{align*}
\begin{align*}
 R_{\ref{lem:V1},3} & =    \mathbb V\left(\sum_{g, h\in [G]^2} \left(V_h' B_{h,g} \Gamma_g \right) \left(H_g' U_g \right)  \right) \\
 & \lesssim \sum_{g, h\in [G]^2} u_n^2 ||B_{h,g} \Gamma_g  H_g'||_F^2 \\
 & = \sum_{g \in [G]} u_n^2 \Gamma_g' (B'B)_{g,g} \Gamma_g ||H_g||_2^2 \\
 & \lesssim u_n^2 n_G \zeta_{H,n} \kappa_n = o ((\mu_n^2 + \kappa_n + \tilde \mu_n^2)^2) = o(\omega_n^4),
\end{align*}
\begin{align*}
 R_{\ref{lem:V1},4} & \lesssim \sum_{g, h, k \in [G]^3, h \neq k}   \mathbb V\left(\left(V_h' B_{h,g} \Gamma_g \right) \left(V_k ' B_{k,g} U_g \right)  \right)  +   \mathbb V\left( \sum_{g, h \in [G]^2}  \left(V_h' B_{h,g} \Gamma_g \right) \left(V_h ' B_{h,g} U_g \right)  \right)  \\
 & \lesssim \sum_{g, h, k \in [G]^3, h \neq k}   \mathbb V\left(\left(V_h'
   B_{h,g} \Gamma_g \right) \left(V_k ' B_{k,g} U_g \right)  \right)  +
   \sum_{g, h \in [G]^2} \mathbb V\left(   \left(\Gamma_g' B_{h,g}' (V_{h} V_h' - \Omega_{V,h}) B_{h,g} U_g \right)  \right) \\
 & + \mathbb V\left(  \sum_{g, h \in [G]^2}  \left(\Gamma_g' B_{h,g}' \Omega_{V,h} B_{h,g} U_g \right)  \right) \\
 & \lesssim \sum_{g, h, k \in [G]^3, h \neq k}  u_n^3 \Gamma_g B_{h,g}' B_{h,g}
   \Gamma_g ||B_{k,g}||_F^2 +  \sum_{g, h \in [G]^2} \mathbb E   \left(\Gamma_g'
   B_{h,g}' V_{h} V_h' B_{h,g} B_{h,g}'V_{h} V_h' B_{h,g} \Gamma_g  \right) \\
 & + \sum_{g \in [G]} u_n \left\Vert \Gamma_g' (\sum_{h \in [G]} B_{h,g}' \Omega_{V,h} B_{h,g}) \right\Vert_2^2\\
 & \lesssim u_n^3 n_G \lambda_n \kappa_n + u_n^{\frac{2q-3}{q-1}} n_G^{\frac{2q-1}{q-1}}\lambda_n^2 \kappa_n +  \sum_{g \in [G]} u_n n_G tr \left((B' \Omega_V B)_{g,g}^2\right)  \\
 & \lesssim u_n^3 n_G \lambda_n \kappa_n + u_n^{\frac{2q-3}{q-1}} n_G^{\frac{2q-1}{q-1}}\lambda_n^2 \kappa_n + u_n^3 n_G \lambda_n \kappa_n = o ((\mu_n^2 + \kappa_n + \tilde \mu_n^2)^2) = o(\omega_n^4),
\end{align*}
\begin{align*}
R_{\ref{lem:V1},5} & = \sum_{g \in [G]} \mathbb V\left( \left(H_g' U_g \right)^2  \right) \\
& \lesssim \sum_{g \in [G]} \mathbb E H_{g}' U_{g} U_{g}' ||U_{g}||_2^2 ||H_{g}||_2^2 \\
& \lesssim u_n^{\frac{q-2}{q-1}} n_G^{\frac{q}{q-1}} \zeta_{H,n} \mu_n^2  = o ((\mu_n^2 + \kappa_n + \tilde \mu_n^2)^2) = o(\omega_n^4),
\end{align*}
\begin{align*}
R_{\ref{lem:V1},6}     & = \mathbb V\left(\sum_{g, k \in [G]^2} \left(V_k ' B_{k,g} U_g U_g' H_g\right)  \right) \\
& \lesssim \sum_{g, k \in [G]^2}  \mathbb V\left(\left(V_k ' B_{k,g} (U_g U_g' - \Omega_{U,g}) H_g\right)  \right) + \sum_{k \in [G]} \mathbb V\left(V_k ' (\sum_{g \in [G]}  B_{k,g}  \Omega_{U,g} H_g )\right)  \\
& \lesssim \sum_{g, k \in [G]^2} u_n \mathbb E \left(H_g'U_g U_g'  B_{k,g}' B_{k,g} U_g U_g' H_g \right) +  u_n ||B \Omega_U H||_2^2 \\
& \lesssim u_n^{\frac{2q-3}{q-1}} n_G^{\frac{q}{q-1}} \zeta_{H,n} \kappa_n + u_n^3 \mu_n^2  = o ((\mu_n^2 + \kappa_n + \tilde \mu_n^2)^2) = o(\omega_n^4),
\end{align*}
$R_{\ref{lem:V1},7}   = o(\omega_n^4)$ following the same argument as $ R_{\ref{lem:V1},6} $, and
\begin{align*}
R_{\ref{lem:V1},8} & \lesssim     \mathbb V\left(\sum_{g, h, k \in [G]^3, h \neq k} \left(V_h' B_{h,g} (U_g U_g'-\Omega_{U,g}) B_{k,g}' V_k\right)  \right) +    \mathbb V\left(\sum_{g, h, k \in [G]^3, h \neq k} \left(V_h' B_{h,g} \Omega_{U,g} B_{k,g}' V_k\right)  \right) \\
& + \mathbb V\left(\sum_{g, h \in [G]^2} \left(V_h' B_{h,g} U_g \right)^2\right) \\
& \lesssim u_n^2 \sum_{g, h, k \in [G]^3, h \neq k} \mathbb E tr\left(B_{h,g} (U_g U_g'-\Omega_{U,g}) B_{k,g}'B_{k,g} (U_g U_g'-\Omega_{U,g}) B_{h,g}'  \right) \\
& + u_n^2 \sum_{h,k \in [G]^2} ||(B \Omega_U B')_{h,k}||_F^2  + \mathbb
  V\left(\sum_{g, h \in [G]^2} tr\left((V_{h} V_h' - \Omega_{V,h}) B_{h,g}
  (U_{g} U_g' - \Omega_{U,g})B_{h,g}' \right)\right) \\
&  + \mathbb V\left(\sum_{g, h \in [G]^2} tr\left(\Omega_{V,h}B_{h,g} (U_{g}
  U_g' - \Omega_{U,g})B_{h,g}' \right)\right) +  \mathbb V\left(\sum_{g, h \in
  [G]^2} tr\left((V_{h} V_h' - \Omega_{V,h}) B_{h,g} \Omega_{U,g}B_{h,g}' \right)\right) \\
& \lesssim u_n^2 \sum_{g \in [G]} \mathbb E tr\left( U_{g} U_g' (B'B)_{g,g} U_g U_g' (B'B)_{g,g} \right) + u_n^4 \kappa_n \\
& + \sum_{g, h \in [G]^2} \mathbb E\left( tr\left((V_{h} V_h' - \Omega_{V,h})
  B_{h,g} (U_{g} U_g' - \Omega_{U,g})B_{h,g}' \right)\right)^2 \\
&  + \sum_{g \in [G]} \mathbb E\left( U_g' (B'\Omega_{V}B)_{g,g} U_g \right)^2 +  \sum_{h \in [G]} \mathbb E\left(V_h'(B\Omega_U B')_{h,h}V_h\right)^2 \\
& \lesssim u_n^{\frac{3q-4}{q-1}} n_G^{\frac{q}{q-1}} \lambda_n \kappa_n +u_n^4
  \kappa_n + \sum_{g, h \in [G]^2} \mathbb E\left( tr\left((V_{h} V_h' -
  \Omega_{V,h}) B_{h,g} B_{h,g}'(V_{h} V_h' - \Omega_{V,h}) \right) \right) \\
& + \sum_{g, h \in [G]^2} \mathbb E\left( tr\left((U_{g} U_g' -
  \Omega_{U,g})B_{h,g}' B_{h,g} (U_{g} U_g' - \Omega_{U,g})  \right) \right)  + u_n^{\frac{q-2}{q-1}} n_G^{\frac{q}{q-1}} \sum_{g \in [G]} tr\left((B'\Omega_{V}B)_{g,g}^2 \right) \\
& \lesssim u_n^{\frac{3q-4}{q-1}} n_G^{\frac{q}{q-1}} \lambda_n \kappa_n +u_n^4 \kappa_n + u_n^{\frac{q-2}{q-1}} n_G^{\frac{q}{q-1}} \kappa_n \\
& = o ((\mu_n^2 + \kappa_n + \tilde \mu_n^2)^2) = o(\omega_n^4).
\end{align*}
\end{proof}

\subsection{Lemma~\ref{lem:deco}}
\begin{lem}\label{lem:deco} Let $\{\eps_i\}_{i=1}^n$ be a sequence of random
  variables with values in a measurable space $(S,\mathcal{S})$. Suppose that
  this sequence can be partitioned into $G$ independent clusters, denoted as
  $\{\eps_{[g]}\}_{g=1}^G$, and let $n_g$ denote the cluster size of the $g$-th
  cluster. Let $\{\eps_{[g]}^{(k)}\}_{g=1}^G, k=1, \dotsc, m$ be $m$ independent
  copies of this sequence for some $m \leq G$. Define
  \begin{equation*}
    I_G^m = \{(g_1, \dotsc, g_m): g_j \in \mathbb{N}, 1 \leq g_j \leq G, j \neq
    k \implies g_j \neq g_k \}
  \end{equation*}
  and let
  $h_{g_1 \cdots g_m}\colon S^{n_{g_1}} \times \dotsm \times S^{n_{g_m}} \mapsto
  \mathbb R$ be measurable functions such that
  $\mathbb E |h_{g_1 \cdots g_m}(\eps_{[g_1]}, \dotsc, \eps_{[g_m]})| < \infty$
  for all $(g_1, \dotsc, g_m) \in I_G^m$. Let
  $\Phi\colon [0,\infty) \to [0,\infty)$ be a convex nondecreasing function
  such that
  $\mathbb E \Phi (|h_{g_1 \cdots g_m}(\eps_{[g_1]}, \dotsc, \eps_{[g_m]})|) <
  \infty$ for all $(g_1, \dotsc, g_m) \in I_G^m$, then \begin{align*} \mathbb E
    \Phi (|\sum_{I_G^m} h_{g_1 \cdots g_m}(\eps_{[g_1]}, \dotsc, \eps_{[g_m]})|)
    \leq \mathbb E \Phi (C_m |\sum_{I_G^m} h_{g_1 \cdots
      g_m}(\eps_{[g_1]}^{(1)}, \dotsc, \eps_{[g_m]}^{(m)})|)
    \end{align*}
    where $C_m = 2^m \times (m^m-1) \times ((m-1)^{(m-1)}-1) \times \cdots \times 3$ does not depend on $G$ and $n$.
\end{lem}
The proof of this lemma follows exactly the same lines as in Theorem 3.1.1 of \textcite{DE2012} and is thus omitted. We shall use this result extensively with $\Phi(x) = x^2$ in our proof.

\end{document}